\documentclass[10pt, twocolumn, aps, prd, amssymb, amsmath, superscriptaddress, eqsecnum, nofootinbib, notitlepage, tightenlines, longbibliography]{revtex4-1}

\usepackage[CJKbookmarks, pdftex, bookmarksnumbered, bookmarksopen, colorlinks, citecolor=crimson, linkcolor=blue]{hyperref}

\usepackage{natbib}

\usepackage[utf8]{inputenc}
\usepackage[T1]{fontenc}
\usepackage{graphicx}
\usepackage{hyperref}
\usepackage{verbatim}
\usepackage[english]{babel}
\usepackage{amsthm}
\usepackage{mathrsfs}
\usepackage{mathtools}
\usepackage{amsfonts}
\usepackage{amscd}
\usepackage{dsfont}
\usepackage{soul}
\usepackage{enumitem}
\usepackage{tikz,bbm,bm}
\usepackage{cancel}
\usepackage[subnum]{cases}
\usepackage[all,cmtip,color]{xy}
\usepackage{braket}
\usepackage{color}
\usepackage[many]{tcolorbox}
\usepackage{nccmath} 
\usepackage[normalem]{ulem}
\usepackage{float}
\usepackage{caption}
\usepackage{subcaption}

\allowdisplaybreaks

\newcommand{\Ignore}[1]{}

\newcommand{\be}{\begin{equation}}
\newcommand{\ee}{\end{equation}}
\newcommand{\bes}{\begin{eqnarray}}
\newcommand{\ees}{\end{eqnarray}}

\newcommand{\Hkin}{\mathcal H_{\text{kin}}}

\newcommand{\Piphys}{\Pi_{\text{phys}}}

\newcommand{\Hphys}{\mathcal H_{\text{phys}}}
\newcommand{\Aphys}{\mathcal A_{\text{phys}}}
\newcommand{\Hibar}{\mathcal H_{\overline{i}}}

\renewcommand\bra[1]{{\langle{#1}|}}
\renewcommand\ket[1]{{|{#1}\rangle}}

\newtheorem{definition}{Definition}
\newtheorem{theorem}{Theorem}
\newtheorem{lemma}{Lemma}
\newtheorem{corollary}{Corollary}[lemma]
\newtheorem{example}{Example}
\newtheorem{claim}{Claim}

\newcommand{\tr}{\text{Tr}}

\definecolor{crimson}{rgb}{0.86, 0.08, 0.24}
\definecolor{persianblue}{rgb}{0.11, 0.22, 0.73}
\definecolor{cobalt}{rgb}{0.0, 0.28, 0.67}
\definecolor{green(pigment)}{rgb}{0.0, 0.65, 0.31}
\definecolor{mediumseagreen}{rgb}{0.2, 0.6, 0.65}
\definecolor{asparagus}{rgb}{0.53, 0.66, 0.42}
\definecolor{lightgray}{rgb}{0.83, 0.83, 0.83}
\definecolor{anti-flashwhite}{rgb}{0.95, 0.95, 0.96}

\DeclareMathOperator{\Tr}{Tr}

\newcommand{\mcG}{\mathcal{G}}

\newcommand{\ibar}{\overline{i}}
\newcommand{\jbar}{\overline{j}}
\newcommand{\mcH}{\mathcal{H}}
\newcommand{\mcS}{\mathcal{S}}

\newcommand{\mcR}{\mathcal{R}}

\newcommand{\eff}{{\text{eff}}}
\newcommand{\mdsh}{\mathds{h}}

\newcommand{\mbfU}{\mathbf{U}}
\newcommand{\mbfUX}{\mathbf{U}_{\!X}}

\newcommand{\mcA}{\mathcal{A}}

\newcommand{\mcI}{\mathcal{I}}
\newcommand{\onebar}{\overline{1}}
\newcommand{\twobar}{\overline{2}}

\newcommand{\mbt}{\mathbbm{t}}
\newcommand{\mbd}{\mathbbm{d}}
\newcommand{\Pit}{\hat{\Pi}_{\mbt}}
\newcommand{\Pid}{\hat{\Pi}_{\mbd}}

\newcommand{\PiUX}{\hat{\Pi}_{\ibar}^{X}}
\newcommand{\PiUXperp}{\hat{\Pi}_{\ibar}^{X\perp}}
\newcommand{\PiUXj}{\hat{\Pi}_{\jbar}^{X'}}
\newcommand{\PiUXjperp}{\hat{\Pi}_{\jbar}^{X'\perp}}

\begin{document}

\title{Quantum Frame Relativity of Subsystems, Correlations and Thermodynamics}

\author{Philipp A. H\"ohn} 
\email{philipp.hoehn@oist.jp}
\affiliation{Okinawa Institute of Science and Technology Graduate University, 1919-1 Tancha, Onna-son, Okinawa 904-0495, Japan}

\author{Isha Kotecha}
\email{ikp307@gmail.com}
\affiliation{Basic Research Community for Physics e.V., Germany}

\author{Fabio M. Mele}
\email{fabio.mele@oist.jp}
\affiliation{Okinawa Institute of Science and Technology Graduate University, 1919-1 Tancha, Onna-son, Okinawa 904-0495, Japan}


\begin{abstract}
It was recently noted that different internal quantum reference frames (QRFs) partition a system in different ways into subsystems, much like different inertial observers in special relativity decompose spacetime in different ways into space and time.~Here we expand on this QRF relativity of subsystems and elucidate that it is the source of all novel QRF dependent effects, just like the relativity of simultaneity is the origin of all characteristic special relativistic phenomena.~We show that subsystem relativity, in fact, also arises in special relativity with internal frames and, by implying the relativity of simultaneity, constitutes a generalisation of it.~Physical consequences of the QRF relativity of subsystems, which we explore here systematically, and the relativity of simultaneity may thus be seen in similar light.~We focus on investigating when and how subsystem correlations and entropies, interactions and types of dynamics (open vs.\ closed), as well as quantum thermodynamical processes change under QRF transformations.~We show that thermal equilibrium is generically QRF relative and find that, remarkably, \emph{QRF transformations not only can change a subsystem temperature, but even map positive into negative temperature states}.~While reminiscent of the Unruh effect, this phenomenon is distinct and does not invoke spacetime structure.~We examine how non-equilibrium notions such as heat and work exchange (related to the first law), as well as entropy production and flow (related to the second law) depend on the choice of QRF.~In arriving at these results, we develop the first study of how reduced subsystem states transform under QRF changes, which, in contrast to global states, transform non-unitarily in the generic case.~Focusing on physical insights, we restrict to ideal QRFs associated with finite abelian groups.~Besides being conducive to rigour, the ensuing finite-dimensional setting is where quantum information-theoretic quantities and quantum thermodynamics are best developed.~We anticipate, however, that our results extend qualitatively to more general groups and frames, and even to subsystems in gauge theory and gravity.

\end{abstract}

\maketitle

{\hypersetup{linkcolor=black}
\tableofcontents}

\section{Introduction}\label{Sec:intro}

\begin{figure*}[!t]
\centering
\includegraphics[width=\textwidth]{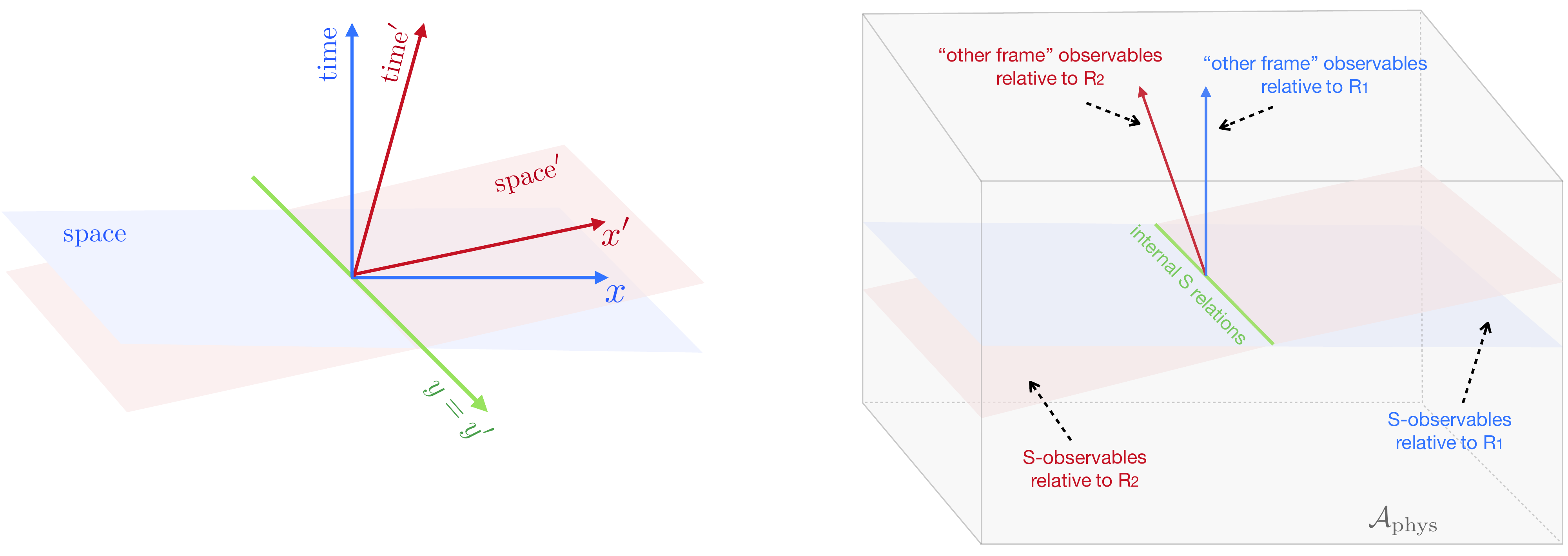}
\caption{Just like different inertial observers in special relativity decompose spacetime in different ways into space and time (relativity of simultaneity, left), different internal quantum reference frames (QRFs) partition the total system in different ways into subsystems (quantum relativity of subsystems, right). More precisely, two QRFs $R_1$ and $R_2$ decompose the total algebra $\Aphys$ of gauge-invariant (external-frame-independent) observables in different ways into subalgebras corresponding to a subsystem $S$ (their joint complement), and ``the other frame''. In fact, we demonstrate that the relativity of subsystems also arises for special relativity with internal frames (tetrads) and \emph{implies} the relativity of simultaneity. The relativity of subsystems is thus a generalisation of the relativity of simultaneity. Just as the latter is \emph{the} source of all the characteristic special relativistic effects and thereby can be viewed as the essence of special covariance, subsystem relativity is \emph{the} source of all the QRF dependence of physical properties and thereby can similarly be regarded as the \emph{essence of QRF covariance}. Given their intimate relation, their physical consequences should be seen in a similar light. While $R_1$ and $R_2$ describe $S$ by distinct invariant subalgebras (blue and red planes on the right), they will agree on the description of internal relational observables of $S$, which comprise the overlap of these subalgebras (green line on the right). This is analogous to how different Lorentz observers agree on scalars (and may agree on some spatial directions). }
\label{Fig:relSR}
\end{figure*}

The notion of a subsystem is one of the most elementary ingredients of modern physics.~Our concepts of interactions, correlations, locality, (de)compositions, exchanges of heat and work, etc.~derive from it.~But what defines a subsystem and relative to what is it defined? 

In many treatments in quantum theory, it is standard that a rigid partitioning into subsystems is given as an external input.~This includes quantum field theory (QFT) in Minkowski space, where locality and subsystems are defined relative to the (for the field) external background spacetime.~Such a rigid and external notion of subsystem is challenged, however, when gravity is taken into account and spacetime becomes an internal and dynamical subsystem itself.~This has sparked much recent debate around suitably defining subsystems associated with subregions in gravitational theories and this involves the introduction of additional reference degrees of freedom \cite{Donnelly:2016auv,Carrozza:2022xut,Freidel:2020xyx,Speranza:2022lxr,Ciambelli:2021nmv,Kabel:2023jve,Chandrasekaran:2022cip,Witten:2023qsv,Jensen:2023yxy,Soni:2023fke}, which can be done in multiple ways \cite{Carrozza:2022xut,Goeller:2022rsx}.~In fact, this even extends to non-gravitational field theories with gauge symmetry, where gauge invariance requires subsystems to be further defined relative to reference degrees of freedom called edge modes \cite{Donnelly:2016auv,SCPHedgemodesasreferenceframes,Geiller:2019bti,Gomes:2022kpm,Riello:2020zbk,Riello:2021lfl,Gomes:2016mwl}.

But this is far from the only circumstance under which a rigid notion of subsystem becomes challenged.~Considerations on the meaning of subsystem in quantum information theory have established the idea that a subsystem ought to be operationally defined in terms of the (algebra of) observables accessible to the observer \cite{Zanardi2001,Zanardi2004,Barnum2003GeneralizationBasedCoherent,barnum_subsystem-independent_2004,viola_barnum_2010,Bartlett:2006tzx}.~This too is suggestive of a subsystem relativity: a subsystem is defined relative to what's accessible and this may change depending on the circumstance and observer.~This also resonates with the ideas underlying renormalisation and coarse-graining in many-body physics and QFT, according to which what is accessible and observable depends on the scale one is looking at.~Moreover, it has been argued that the subsystem structure can be naturally defined in terms of the Hamiltonian (more precisely its spectrum) \cite{Cotler2019} and hence the dynamics of the theory; changing the dynamics could thus lead to a distinct subsystem partitioning.

All these examples indicate that there is no absolute notion of subsystem and that it has to be defined relative to some reference structure, for example a reference frame.~This paper will be concerned with the question of what happens when, rather than choosing our reference frame to be external to a system, we associate it with one of its subsystems, thus invoking an \emph{internal} and thereby dynamical reference frame.~Since what is accessible to an external and internal reference frame can be very different, this will also lead to distinct definitions of subsystems.~In the quantum theory, an internal frame becomes a \emph{quantum reference frame} (QRF) \cite{Aharonov:1967zza,Aharonov:1967zz,Aharonov:1984zz,Angelo_2011,Bartlett:2006tzx,Palmer:2013zza,busch2011position,loveridge2011measurement,loveridge2017relativity,miyadera2016approximating,Loveridge:2019phw,Rovelli:1990pi,rovelliQuantumGravity2004,GiacominiQMcovariance:2019,delaHamette:2021oex,PHquantrel,PHtrinity,Hoehn:2020epv,Vanrietvelde:2018pgb,Vanrietvelde:2018dit,Hohn:2018toe,Hohn:2018iwn,PHinternalQRF,PHQRFassymmetries,delaHamette:2021piz,Giacomini:2021gei,CastroRuizQuantumclockstemporallocalisability,Suleymanov:2023wio,yang2020switching,Giacomini:2018gxh,delaHammetteQRFsforgeneralsymmetrygroups,delaHamette:2021iwx,Ballesteros:2020lgl,streiter2021relativistic,Mikusch:2021kro,Glowacki:2023nnf,Carette:2023wpz,Castro-Ruiz:2021vnq,gour2008resource,smithCommunicatingSharedReference2019,Bojowald:2022caa,Bojowald:2019mas,Favalli:2022wei}.~These become crucial when an external frame is not available in quantum information protocols, the foundations of quantum theory and quantum thermodynamics.~But, in fact, the edge modes used in gauge theory and gravity for gauge-invariantly defining subsystems are also precisely such dynamical or quantum frames \cite{SCPHedgemodesasreferenceframes,Carrozza:2022xut,Kabel:2023jve} (see also \cite{Gomes:2022kpm,Riello:2020zbk,Riello:2021lfl,Gomes:2016mwl}).~Furthermore, dynamical frames are more generally crucial in gravitational theories to define physical observables \cite{Goeller:2022rsx}. Internal frames are thus a widely applicable concept in physics.

Given a quantum notion of reference frame and the observation that there is no unique internal subsystem that can serve as one, the natural question is how one can change between the descriptions relative to two different QRFs.~The question is an enticing one, for now the references frames may be in superposition or entangled with other subsystems.~This has resulted in a recent wave of efforts on developing QRF transformations in several contexts and establishing a novel notion of QRF covariance of physical properties and laws \cite{GiacominiQMcovariance:2019,delaHamette:2021oex,PHquantrel,PHtrinity,Hoehn:2020epv,Vanrietvelde:2018pgb,Vanrietvelde:2018dit,Hohn:2018toe,Hohn:2018iwn,PHinternalQRF,PHQRFassymmetries,delaHamette:2021piz,Giacomini:2021gei,CastroRuizQuantumclockstemporallocalisability,Suleymanov:2023wio,yang2020switching,Giacomini:2018gxh,delaHammetteQRFsforgeneralsymmetrygroups,delaHamette:2021iwx,Ballesteros:2020lgl,streiter2021relativistic,Mikusch:2021kro,Glowacki:2023nnf,Carette:2023wpz,Castro-Ruiz:2021vnq,Bojowald:2010xp,Bojowald:2010qw,Hohn:2011us}.~In other words, the aim is to formulate a quantum relativity principle.

A notion of relativity usually entails that different frames decompose some relevant set of observables or states in different ways, leading to different descriptions of this set.~For example, in Galilean relativity, distinct inertial frames decompose space in different ways into orthogonal directions.~In special relativity, this is extended to the possibility of decomposing also spacetime in different ways into space and time.~This then carries over also to general relativity, however, only in a locally meaningful way.~What then is the set that different QRFs decompose in different ways?

It has been shown \cite{PHquantrel}, that the definition of subsystem depends on the internal QRF and so different QRFs decompose the total system in different ways into subsystems.~More precisely, different QRFs decompose the total algebra of gauge-invariant (or external-frame-independent) observables in different ways into subsystems, much like different Lorentz observers decompose spacetime in different ways into space and time (see Fig.~\ref{Fig:relSR}).~This quantum relativity of subsystems was extended in \cite{delaHamette:2021oex,Castro-Ruiz:2021vnq} to general symmetry groups and explains the earlier observed QRF dependence of entanglement and superpositions \cite{GiacominiQMcovariance:2019,Vanrietvelde:2018pgb,delaHammetteQRFsforgeneralsymmetrygroups}, degree of classicality of a subsystem \cite{Vanrietvelde:2018pgb,le2020blurred,tuziemski2020decoherence}, temporal localisability in relational dynamics \cite{CastroRuizQuantumclockstemporallocalisability,PHtrinity}, etc.~The internal QRF programme is thus really about defining subsystems relative to one another and comparing the ensuing descriptions.

Here, we shall build up on these observations and both refine and simplify the presentation of the quantum relativity of subsystems, as well as systematically explore its physical consequences.~We will demonstrate that subsystem relativity and the relativity of simultaneity are intimately related.~Indeed, subsystem relativity is nothing special to quantum theory and in special relativity with internal frames (tetrads), one obtains it too:~different internal frames decompose the set of invariant observables in different ways into subsystems, and this \emph{implies} the relativity of simultaneity.~(In Galilean relativity, it would imply the relativity of spatial decompositions.)~In other words, subsystem relativity is a generalisation of the relativity of simultaneity (cfr.~Fig.~\ref{Fig:relSR}).~Furthermore, we will show that the framework of QRF covariance \cite{delaHamette:2021oex,PHquantrel,PHtrinity,Hoehn:2020epv,Vanrietvelde:2018pgb,Vanrietvelde:2018dit,Hohn:2018toe,Hohn:2018iwn,PHinternalQRF,PHQRFassymmetries,delaHamette:2021piz,Giacomini:2021gei,CastroRuizQuantumclockstemporallocalisability,Suleymanov:2023wio,yang2020switching} that we will invoke directly mimics the covariance structures underlying special relativity and thereby represents a natural quantum extension of the notion of frame covariance.~Since the relativity of subsystems arises structurally in the same form in special relativity and for QRFs, their consequences should be understood in the same light.~In particular, as will become clear in this work, just like the relativity of simultaneity is \emph{the} root of all the novelties of special compared to Galilean relativity, so is the quantum relativity of subsystems \emph{the} root of all the novel QRF relative physics. 

Now what are the precise physical consequences of the quantum relativity of subsystems?~It does not take much imagination to anticipate that it will render any physical property that depends on the definition of subsystem generally QRF dependent.~However, the challenge lies in characterising the  conditions under which physical properties become contingent on the choice of QRF and to quantify how they may change under a change of QRF.~This is a non-trivial task because not all observables and states transform non-trivially under QRF changes.~Just like all inertial observers at an event in Minkowski space agree on the description of scalars at that event (and two inertial observers with a relative boost in $x$-direction even agree on the other spatial directions), there exist states and observables of subsystems that all QRFs will agree upon, see Fig.~\ref{Fig:relSR} for an illustration.~As we will see, these observables will correspond to internal relational observables of those subsystems.~In fact, in the QRF setting, the situation is even more subtle, because states need not remain exactly invariant under QRF transformations for certain properties to remain invariant.~For example, for subsystem entropies to remain invariant it suffices that QRF transformations act on the relevant states equivalently to multilocal unitaries and, as will become evident in this work, there exists a plethora of ways in which this may happen.~Characterising systematically when and how physical properties change under QRF changes thus is a rather non-trivial endeavour.   

The aim of this article is to precisely launch such a systematic investigation, while unearthing novel physical effects.~Our focus will lie on correlations, entropies, interactions and thermodynamic properties, all of which depend on a partitioning into subsystems.~We will thus also see that what is a dynamically or thermodynamically open, closed or isolated subsystem will depend on the choice of QRF.~We will provide (in some cases partial) characterisations when this happens, both for the full set of dynamical trajectories or even individual trajectories (i.e.\ in an effective treatment).

Reference systems appear ubiquitously in thermodynamics, and in some cases \cite{Aberg2014,lostaglio2015description,Lostaglio2015a,Lostaglio2019a,Marvian2019a,Erker2017,Oppenheim2015a,woods2019autonomous,Woods:2019fzh} these can be considered as QRFs.~However, what has not been previously explored in the literature are the consequences of QRF covariance for thermodynamics.~Given that the framework of QRF covariance applies also to total systems of microscopic size (e.g.\ few-particle systems), it is thus natural to explore \emph{quantum} thermodynamics \cite{binder2019thermodynamics,Goold_2016,e15062100,vinanders} in this context.~In a nutshell, quantum thermodynamics departs from notions of macroscopic reservoirs, thermal baths, etc.~and aims to extend the principles of thermodynamics to quantum microscopic systems, while recovering the standard phenomenological thermodynamics in appropriate macroscopic limits.~It does not rely on weak coupling limits or close-to-equilibrium scenarios, which is crucial for our purposes because, as we will show, QRF transformations can transform non-interacting subsystems into strongly interacting ones, or equilibrium states into non-equilibrium ones.

In particular, we will see that the notion of a Gibbs state of thermal equilibrium for a given subsystem is QRF relative.~Generally, a subsystem that is in thermal equilibrium relative to one frame, will not be in equilibrium relative to another, although we also characterise subsystem Gibbs states that \emph{are} QRF-invariant.~We further  give examples of subsystem Gibbs states that map under QRF transformations into other Gibbs states, however, of \emph{different entropy}. Remarkably, \emph{QRF transformations can not only change the temperature of a thermal equilibrium state, but even flip its sign}.~We illustrate in a qubit example how positive temperature Gibbs states can transform into negative temperature ones.~This frame relativity of equilibrium and temperature may evoke analogies with the Unruh effect.~We emphasise, however, that in our QRF context the effect is of a different origin; no spacetime structure is required and, unlike inertial and Rindler observers in Minkowski space, the different QRFs see the \emph{same} subsystem.~Furthermore, unlike in the Unruh effect, the frames are here not external, but internal and quantum.~The QRF relativity of equilibrium and temperature is a direct result of the relativity of subsystems. 

We further demonstrate that non-equilibrium notions such as work and heat exchange between subsystems (related to the first law of thermodynamics), as well as entropy production and entropy flow (related to the second law) are in general QRF relative.~We supplement these observations with partial characterisations of situations when these notions are QRF-independent. 

In arriving at these results, we also provide the first systematic analysis of how reduced subsystem states transform under QRF changes.~In the previous literature, the focus was rather placed on how global states for the total system transform under QRF changes.~While the global state may transform unitarily, the same is in general not true for subsystem states.~Moreover, given some subsystem state in one QRF perspective, how it will transform into another QRF perspective will generally depend on the global state, though we will also characterise a subclass of subsystem states which transform independently of the global state.~Conversely, generally there will also exist many global states that lead to the same transformation of a subsystem state.~This becomes relevant, for example, when seeking a subclass of global states such that a subsystem state transforms unitarily (so that, e.g., entropies are preserved). Subsystem state transformations are thus a somewhat challenging topic. 

In order to render this work more broadly accessible and avoid clouding the physical insights with technicalities, we impose a few simplifying restrictions on our setup below.~First, we will restrict our attention to \emph{ideal} QRFs, namely those, whose orientation states are perfectly distinguishable.~Second, we restrict to QRFs associated with finite abelian groups.~QRFs that are non-ideal or associated with general symmetry (unimodular Lie) groups, while requiring more machinery, can, however, be treated in essentially the same way and can still yield unitary QRF transformations at the global level \cite{delaHamette:2021oex,PHtrinity,Hoehn:2020epv}.~We thus anticipate that our results here will carry over qualitatively to the more general case.~In fact, coming back to our opening discussion, we even suspect that the insights gained here will extend qualitatively to the case of gauge theories and gravity, where edge mode (quantum) reference frames can be used to gauge-invariantly associate subsystems with subregions \cite{Carrozza:2022xut,Kabel:2023jve,Chandrasekaran:2022cip,Witten:2023qsv,Jensen:2023yxy,Goeller:2022rsx,SCPHedgemodesasreferenceframes,Gomes:2022kpm,Riello:2020zbk,Riello:2021lfl,Gomes:2016mwl}.~Our results may also have bearing on recent efforts \cite{Fahn:2022zql} to explore gravitationally induced decoherence and relational open system dynamics in quantum gravity, as these involve choices of temporal QRFs. 

The restriction to finite abelian groups has the further advantage that it typically leads to finite-dimensional Hilbert spaces, which, besides being conducive to a rigorous treatment, is the context in which quantum information-theoretic quantities of our interest\,---\,like entropies\,---\,and quantum thermodynamics are best developed.~This is also the case of most interest for experiments testing quantum thermodynamics.~In fact, this is the main reason for this restriction.

There are, however, some insights that will not extend in the same form beyond ideal QRFs.~For example, one novel insight surrounding the subsystem relativity that we explain below is that, for ideal QRFs, an internal QRF perspective is formally nothing but a tensor product structure (TPS) on the global external-frame-invariant (and TPS-neutral) Hilbert space.~Thus, a change of ideal QRF is nothing but a change of TPS on that space, as well as on the algebra of observables defined on it.~This observation permits us to develop a simple algebraic approach to characterising when physical properties are or are not QRF relative.~The key ingredient in this endeavour will be our novel notion of TPS-invariant subalgebras, which contain operators that are up to local unitaries invariant under QRF transformations.~There exists a continuum of such subalgebras and for QRF (in-)variance it turns out crucial how states and observables lie in or across such subalgebras.

By contrast, for non-ideal frames, we rather expect QRF perspectives to amount to direct sums of TPSs, and for field theories with gauge symmetry a TPS across subsystems will not in general make sense.~We will comment on this in some detail in the conclusions.~There, we will also explain why the relational notion of subsystems and entropies invoked in this article is qualitatively distinct from the corresponding ones used in gauge theories and gravity.~In particular, our relational definitions of subsystems and entanglement entropies suggest a new alternative for that setting as well.

The rest of the paper is organised as follows.~In Sec.~\ref{sec:qrfphystps}, we review the invoked framework of QRF covariance by directly comparing it to special relativity with internal frames (tetrads).~We highlight that all the covariance structures mimic those underlying special relativity and show that in the latter case the subsystem relativity also arises and implies the relativity of simultaneity.~In Sec.~\ref{Sec:QuantRelTPSs}, we expand on the previously noted quantum relativity of subsystems and explain it in a more transparent manner analogous to the special relativistic case.~We also show that (ideal) internal QRF perspectives and changes are nothing but TPSs and changes of TPS, respectively, on the external-frame-independent Hilbert space. Sec.~\ref{Sec:Uinv&break} introduces the TPS-invariant subalgebras and characterises subsystem and composite operators residing in them, which thus are QRF transformation invariant up to local unitaries.~These algebraic results are key to the subsequent sections.~The transformations of reduced subsystem states, correlations of a subsystem (both with its complement and within it) and entropies under QRF changes are explored in Sec.~\ref{Sec:states&entanglement}.~In Sec.~\ref{Sec:Dyn}, we examine the QRF relativity of interactions in Hamiltonians and, correspondingly, of open and closed subsystem dynamics.~The consequences for quantum thermodynamics and especially equilibrium states, temperature, work and heat exchange, as well as entropy production and flow are investigated in Sec.~\ref{Sec:TD}.~We conclude with a discussion and outlook in Sec.~\ref{conc}.~Technical details have been moved to several appendices for better readability.

\section{Special vs.\ quantum reference frame covariance} \label{sec:qrfphystps}

In this section, we  review the framework for QRF covariance that we shall employ in the rest of this work. It is the so-called \emph{perspective-neutral framework} \cite{delaHamette:2021oex,PHquantrel,PHtrinity,Hoehn:2020epv,Vanrietvelde:2018pgb,Vanrietvelde:2018dit,Hohn:2018toe,Hohn:2018iwn,PHinternalQRF,PHQRFassymmetries,delaHamette:2021piz,Giacomini:2021gei,CastroRuizQuantumclockstemporallocalisability,Suleymanov:2023wio,yang2020switching}. 

A reader familiar with this framework may skip this section.~However, we shall introduce it in a somewhat novel manner, by directly comparing it to the structures underlying special covariance in special relativity. We also highlight the key differences to other recent approaches to QRF covariance and explain that the ability to directly mimic the covariance structures of special relativity in a quantum regime singles out the perspective-neutral framework among the approaches. In this sense, it can be viewed as the natural extension of the relativity principle into the quantum realm. We also use this analogy to clarify that the subsystem relativity underlying much of this article should be seen in the same light as the relativity of simultaneity (cfr.~Fig.~\ref{Fig:relSR}).

The last few years have seen the development of a few other approaches to QRF transformations that are inequivalent in the general case.~Besides the one we will use, these are the purely perspective-dependent approach \cite{GiacominiQMcovariance:2019,Giacomini:2018gxh,delaHammetteQRFsforgeneralsymmetrygroups,delaHamette:2021iwx,Ballesteros:2020lgl,streiter2021relativistic,Mikusch:2021kro}, which initiated this recent line of research, the operational approach of \cite{Glowacki:2023nnf,Carette:2023wpz} and the somewhat related further development of the quantum information-theoretic approach \cite{Castro-Ruiz:2021vnq} (see also the earlier, but distinct \cite{Palmer:2013zza}).~There is also an older effective (semiclassical) approach \cite{Bojowald:2010qw,Bojowald:2010xp,Hohn:2011us} and an algebraic one \cite{Bojowald:2019mas,Bojowald:2022caa}, which (for ideal QRFs) are shown to be equivalent to the perspective-neutral one in \cite{JulianArtur}. Since it may become hard for the novice to maintain an orientation in the face of these different QRF directions, for the convenience of the reader, we thus highlight what distinguishes the perspective-neutral framework in terms of its conceptual motivation, scope and technical achievements, and what its key differences to the other approaches are. 

This is not meant as a detailed comparison of the different approaches, which we leave for elsewhere. We also expressly emphasise that there is no single ``right'' approach to QRF changes which, in fact, can mean different things in different contexts. For example, in operational lab contexts with an observer external to a composite system, who measures this system with external devices, changing internal QRF within this system has a very different meaning from more fundamental scenarios, where no external observer exists in the first place.~The different approaches above come with different motivations and thus also different capabilities. 

That said, as we  now elucidate, the perspective-neutral framework is the fully relational (i.e.\ maximally external-frame-independent) one, it implements QRF covariance in the same structural manner as special covariance, and it is directly amenable to gauge theories and gravity, being formulated in the language of these theories. 

While the perspective-neutral and purely perspective-dependent approaches are inequivalent in general \cite{delaHamette:2021oex}, they \emph{are} equivalent in the setting of ideal QRFs employed throughout this work. We thus emphasise for the reader with sympathy for that approach that all results from Sec.~\ref{Sec:Uinv&break} onward also directly apply to it as well.

\subsection{External vs.\ internal frames}

Suppose we are given a composite system $\tilde S=R_1R_2S$, composed of three subsystems (each of which may itself be a composite system), that is subjected to some (gauge) symmetry group $\mcG$ acting on all subsystems simultaneously.~The basic premise we shall adhere to is that states and observables of $\tilde S$ that are related by such a (gauge) symmetry transformation are physically indistinguishable when $\tilde S$ is considered without access to any external reference structure, or external reference frame for short.~Accordingly, we consider $\mcG$ as a group of external frame transformations and, since these are not accessible when considering $\tilde S$ alone, we shall henceforth simply call them ``gauge transformations''.~Similarly, by the word ``gauge-invariant'', we shall thus mean external frame independence.~In a laboratory situation, $\mcG$ could refer to changes of external laboratory frame, while in the context of gauge theory or gravity, the external frame would usually be fictitious,\footnote{Except in the context of finite boundaries and edge modes \cite{SCPHedgemodesasreferenceframes,Carrozza:2022xut,Gomes:2019otw,Gomes:2021nwt,Gomes:2022kpm,Riello:2020zbk,Riello:2021lfl,Gomes:2016mwl,Kabel:2023jve,Donnelly:2016auv}.} and $\mcG$ would literally be the gauge group, which can be interpreted as implementing changes of fictitious background (gauge) frame.~For example, in gravity, $\mcG$ would be the group of (bulk) diffeomorphisms (i.e.\ essentially the set of general background coordinate changes), while in QED $\mcG$ could correspond to an appropriate set of $\rm{U}(1)$-valued spacetime functions.~This last example already highlights that the notion of frame considered here need not necessarily be a spatiotemporal one, but is more general. 

Next, we make a conceptual jump compared to the standard notion of reference frame, which is indeed the external one that is \emph{not} part of the system to be described.~The aim is now to choose some subsystem of $\tilde S$, say $R_1$, as an \emph{internal} reference frame, relative to which we describe the ``rest'' $R_2S$ in a gauge-invariant (hence external frame-independent) and relational manner.~The central question of internal frame covariance is then how this description changes, if instead we select $R_2$ as the internal frame and describe $R_1S$ relative to it.~What we describe relative to $R_1$ and $R_2$ are thus \emph{distinct} subsystems as seen relative to an external frame; only $S$ is described relative to both (which will be the focus of this paper).

Why would one make such a conceptual jump and move toward internal frames?~One motivation comes from noting that all reference frames used in experiments and observations are physical and that a more fundamental description, devoid of any unphysical background structures, ought to treat them as internal. In quantum theory, moving towards internal frames that are quantum systems themselves can be viewed, in the light of the universality of quantum theory, in some sense as an extension of the Heisenberg cut. This is also the only way in which one can extend the relativity principle genuinely into the quantum regime, i.e.\ in a way that goes beyond the classical notion of frame covariance. Indeed, this step will permit us to change QRF perspective even when the frames can be in relative superposition or entangled with other subsystems.~By contrast, covariance with respect to external frames in quantum theory is already encoded in standard textbook quantum mechanics and quantum field theory (QFT); while external frame transformations (e.g.\ Poincar\'e transformations in QFT) are implemented unitarily, they amount to changes of classical frame and in that sense give nothing new. 

Another motivation comes from the observation that, in theories with gauge symmetry, one in fact implicitly invokes dynamical internal frames when building gauge-invariant descriptions (by means of which one dresses bare, non-invariant quantities into invariant ones) \cite{SCPHedgemodesasreferenceframes,Carrozza:2022xut,Goeller:2022rsx,Gomes:2022kpm,Rovelli:2013fga,Rovelli:2020mpk}.~That is, dealing with internal frames is already baked into such theories.~Furthermore, in operational situations, different agents may not share a common external laboratory frame and may resort instead to internal subsystems as references \cite{Bartlett:2006tzx,PHQRFassymmetries}. 

\subsection{Special and QRF covariance: the perspective-neutral way}

Let us now review the perspective-neutral approach to QRF covariance.~To highlight that its structures are direct quantum versions of those underlying special covariance with internal frames, let us introduce both in parallel.~The relation between the perspective-neutral approach and special relativity is explained in more detail in \cite{delaHamette:2021oex}, though here we add further facets of it.~We also explain the key differences to the other approaches to QRF covariance and use these to clarify that they cannot mimic special covariance in the same sense.

\subsubsection{Description relative to an external frame} 

We begin with special relativity, restricting for simplicity to the Lorentz part of the Poincar\'e transformations (it can be straightforwardly extended to the full Poincar\'e group).~Let $R_1$ and $R_2$ each be given by a tetrad, i.e.\ an orthonormal basis in Minkowski space, $e_a^{\mu}$ and $e'{}_{a'}^{\mu}$, respectively, where $\mu$ labels spacetime indices $t,x,y,z$ and $a$ denotes the frame index $0,1,2,3$.~Owing to normalisation, $\eta_{ab}=e^\mu_ae^\nu_b\eta_{\mu\nu}$, it is clear that the frame orientations $e^\mu_a\in\rm{SO}^+(3,1)$ are Lorentz group valued themselves.\footnote{The restriction to the special and orthonormal component of the group is not necessary, but natural when restricting to equally handed frames.}~Furthermore, let the subsystem $S$ for simplicity be given by some four-velocity $v^\mu$.~The coordinates labeled by $\mu$ are ones associated with some arbitrary background Lorentz frame.~Changes of external background frame are given by spacetime Lorentz transformations $\Lambda^\nu{}_\mu\in\rm{SO}^+(1,3)$ that act on spacetime indices and thus on each of $R_1R_2S$ simultaneously.~These are what we call ``gauge transformations''.~In the case of the internal frames $R_1,R_2$, the group thus acts on itself.~Without access to the external frame, e.g.\ $v^\mu$ and $\Lambda^\mu{}_\nu v^\nu$ are indistinguishable and so the set of possible configurations $\{e_a^\mu,e'{}_{a'}^\mu,v^\mu\}$ comprises the externally distinguishable ones. 

Let us now introduce the analogous structures for QRFs.~As already mentioned in the introduction, in this work we shall restrict to some finite abelian group $\mcG$ ($\mcG\neq\{\mathbf{1}\}$) because it results in a natural finite-dimensional context in which quantum information and quantum thermodynamics are best developed, while still having direct experimental relevance.~This puts us in a position to explore their QRF covariance in a fairly rigorous setting.~The finite abelian group $\mcG$ could, for example, be given by the cyclic group $\mathbb{Z}_n$, i.e.\ the discrete translation group on a circle.\footnote{In fact, owing to the fundamental theorem of finite abelian groups, all of them can be written as direct sums of cyclic groups.}~In the sequel, $\mcG$ will assume the role that the Lorentz group has in special relativity above.~We emphasise, however, that all of the below can be extended to arbitrary (and specifically non-abelian) unimodular Lie groups \cite{delaHamette:2021oex}.~Some of the below material on finite groups is taken from \cite{PHQRFassymmetries,PHinternalQRF}, but QRFs with finite groups have also been discussed in \cite{delaHammetteQRFsforgeneralsymmetrygroups,Glowacki:2023nnf}.

Just like the Lorentz group above constitutes the space of frame orientations, we shall assume the same here.~That is, for both internal QRFs $R_1,R_2$, we take the Hilbert space to be the space of complex functions over the group $\mcH_{i}=\ell^2(\mcG)\simeq\mathbb{C}(\mcG)$ and we take it to be spanned by frame configuration states $\ket{g}_i$, $g\in\mcG$, $i=1,2$, labelling particular QRF orientations.~For simplicity, we therefore assume in this work that the two QRFs are \emph{ideal}, which means that their frame orientation states are perfectly distinguishable $\braket{g|g'}_i=\delta_{g,g'}$.~(Non-ideal frames can also be encompassed \cite{delaHamette:2021oex,PHtrinity,PHquantrel,Hoehn:2020epv}.)~Thus, we can also write $\mcH_{i}\simeq\mathbb{C}^{|\mcG|}$, where $|\mcG|$ denotes the cardinality of $\mcG$.~The frame Hilbert spaces $\mcH_{i}$ furnish a regular representation of $\mcG$, on which it therefore acts as $U_i^g\ket{g'}_i=\ket{gg'}_i$.~As in the tetrad case above, the group thus acts on itself.~All of this is a discrete version of the translation group case $(\mathbb{R},+)$, where the ideal frame Hilbert space would be $L^2(\mathbb{R})$ and position eigenstates would constitute the QRF orientation states \cite{Vanrietvelde:2018pgb,Vanrietvelde:2018dit,delaHamette:2021oex,delaHammetteQRFsforgeneralsymmetrygroups,GiacominiQMcovariance:2019}.~The $\ket{g}_i$ are thus the QRF analogue of a specific tetrad configuration $e^\mu_a$ in special relativity and $\ket{g'}_i\mapsto U_i^g\ket{g'}_i$ is the analogue of the gauge action $e^\mu_a\mapsto\Lambda^\mu{}_\nu e^\nu_a$ on the frame. 

For the subsystem $S$, on the other hand, we shall not assume anything else than that it be described by some (possibly infinite-dimensional) Hilbert space $\mcH_S$ that too carries a unitary representation $U_S^g$ of $\mcG$, though it need not be the regular one. 

In conjunction, we shall thus henceforth assume that the total kinematical Hilbert space of the externally distinguishalbe subsystems $R_1R_2S$, i.e.\ the space of states distinguishable relative to an external frame, is given by a tensor product $\mcH_{\rm kin}=\mcH_{1}\otimes\mcH_{2}\otimes\mcH_S$ and that the external frame transformations are given by a unitary tensor product representation $\mcG\rightarrow\mathcal{U}(\mcH_{\rm kin})$, $g\mapsto {U^g_{12S}:=U_1^g\otimes U_2^g\otimes U_S^g}$, where $\mathcal{U}(\mcH)$ denotes the unitary group on Hilbert space $\mcH$.~These are the ``gauge transformations'' that act globally on all subsystems, just like $\Lambda^\mu{}_\nu$ above acted on all spacetime indices.~Similarly, the kinematical algebra of observables distinguishable relative to an external frame is of the form $\mcA_{\rm kin}=\mcA_{1}\otimes\mcA_{2}\otimes\mcA_S$.

\subsubsection{Internal frame reorientations}

We have thus far only considered gauge transformations.~There are also representations of the group that correspond to physical changes:~these are the \emph{frame reorientations} and only act on the given frame.~Being physical, they commute with the gauge transformations and in this sense can be viewed as ``symmetries''.\footnote{As shown in \cite{SCPHedgemodesasreferenceframes,Carrozza:2022xut}, in the case of gauge theory and gravity, this distinction of transformations literally corresponds to what are called gauge transformations and symmetries in the context of edge modes \cite{Donnelly:2016auv,Geiller:2019bti} (see also \cite{delaHamette:2021oex} for this discussion in the QRF context and \cite{Goeller:2022rsx} for one in gravitational theories).}

In the case of special relativity, this is possible because the tetrads come with a second index: the frame index $a$. This gives rise to a second representation of the Lorentz group $e^\mu_a\mapsto \Lambda_a{}^b e_b^\mu$, with $\Lambda_a{}^b\in\rm{SO}^+(1,3)$, which by construction \emph{only} acts on $R_1$ and reorients this frame (while preserving its orthonormality).~By construction, it commutes with the gauge transformations and, since $R_2S$ are unaffected, the \emph{relation} between $R_1$ and ``the rest'' is changed.~Similarly, we have a Lorentz group action $\Lambda_{a'}{}^{b'}$ on $R_2$. 

In the context of ideal internal QRFs, these frame reorientations are simply
\be \label{qframereorient}
V_{1}^g:=U_1^g\otimes\mathds1_2\otimes\mathds1_S\,,
\ee
which only act on the frame and evidently commute with the gauge transformations $U_{12S}^{g'}$, the group being abelian.\footnote{More precisely, the gauge transformations are in our convention given by a left regular action, while reorientations are given by a right regular action. This requires a group representation that gives rise to both a left and right unitary action, a condition satisfied, for example, by the regular representation used here. This distinction also guarantees in the general (especially non-abelian) case that gauge transformations and reorientations commute \cite{delaHamette:2021oex}. Since everything is abelian here, for ease of notation, we abstain from explicitly implementing this distinction.} It changes the configuration of the QRF $R_1$, while leaving $R_2S$ untouched, thus also changing their relation. Reorientations of $R_2$ are defined similarly.

\subsubsection{``Jumping'' into an internal frame perspective}

Now what does it mean to ``jump'' into the perspective of $R_1$ or $R_2$? Since we are dealing with gauge transformations (external frame rotations), there are two equivalent ways of doing so: (i) gauge fixing to a particular orientation of the frame; (ii) building relational/frame-dressed observables to describe ``the rest'' in a gauge-invariant manner relative to the internal frame. 

In the case of special relativity, (i) means simply aligning the fictitious background Lorentz frame with the tetrad of, say, $R_1$, for example by choosing coordinates such that $e^\mu_a = \delta^\mu_a$.~This breaks the background gauge symmetry encoded in the representation $\Lambda^\mu{}_\nu$ and amounts to describing ``the rest'' relative to $R_1$ being in the ``origin'', or rather identity configuration, recalling that its orientations are group-valued.~Of course, we could also gauge fix $e^\mu_a$ to be any other element of $\rm{SO}^+(1,3)$. 

This is equivalent to (ii), dressing $R_2S$ observables with the $R_1$-tetrad: 
\be \label{SRrelobs}
v_a:=\eta_{\mu\nu}\,e^\mu_a\,v^\nu\,,\qquad \Lambda_{a'}{}^a:=\eta^{ab}\eta_{\mu\nu}e^\mu_be'{}_{a'}^\nu\,.
\ee 
These are \emph{relational observables}: $v_a$ and $\Lambda_{a'}{}^a$ describe $v^\mu$ (so $S$) and the second tetrad $e'{}_{a'}^\mu$ (so $R_2$) in a gauge (i.e.~background coordinate) invariant manner relative to internal frame $R_1$.~For comparison with the QRF case below, notice that, because the dressing variable $e^\mu_a$ is group-valued, we can view the relational observables also as frame-orientation-conditional gauge transformations of the ``bare'' quantities $v^\mu$ and $e'{}^\mu_{a'}$.~Furthermore, note that $\Lambda_{a'}{}^a$, which describes $R_1$ and $R_2$ relative to one another, also takes value in $\rm{SO}^+(1,3)$ because the individual tetrads do.~Equivalence with (i) is clear when gauge fixing to $e^\mu_a=\delta^\mu_a$, in which case the manifestly gauge-invariant relational observables $v_a$ and $\Lambda_{a'}{}^a$ will coincide with the descriptions of the bare $v^\mu$ and $e'{}_{a'}^\mu$ in a coordinate system aligned with $e^\mu_a$.~Finally, we observe that the frame reorientations $\Lambda_a{}^b$ change the values of the relational observables, in line with them being physical transformations.

In the QRF setting, we have to do an additional step in order to implement an \emph{invertible} gauge fixing as in (i).~This is the key difference between the perspective-neutral approach invoked here and the other approaches to QRF covariance, namely the purely perspective-dependent one \cite{GiacominiQMcovariance:2019,Giacomini:2018gxh,delaHammetteQRFsforgeneralsymmetrygroups,PHQRFassymmetries}, the operational one \cite{Glowacki:2023nnf,Carette:2023wpz} and the quantum information-theoretic approach to QRFs \cite{Castro-Ruiz:2021vnq}.~Indeed, gauge fixing to a frame orientation would be naturally implemented via a conditioning of states and observables on, say, the ``identity'' orientation of $R_1$,
$\ket{e}\!\bra{e}_1$.~However, the latter gives rise to a non-trivial and hence non-invertible projector on both $\mcH_{\rm kin}$ and $\mcA_{\rm kin}$.~At this level, it can therefore not be used for a gauge fixing.~Instead, we will have to work at the level of gauge equivalence classes of states.~Which ones should one pick?

One might at first contend that it would be sufficient to consider the set of states $\rho_{12S}$ that are gauge-invariant in the following sense:
\be\label{QIstates} 
\big[\rho_{12S},U_{12S}^g\big]=0\,,\qquad\forall\,g\in\mcG\,.
\ee 
This is the set of states considered in the quantum information-theoretic approach \cite{Castro-Ruiz:2021vnq,Bartlett:2006tzx,PHQRFassymmetries} and, for compact (incl.\ finite) groups, also in the operational approach \cite{Glowacki:2023nnf,Carette:2023wpz}. It is the image of the set of kinematical density matrices $\mcS(\mcH_{\rm kin})$
under the orthogonal projector given by the $G$-twirl, or \emph{incoherent} group averaging over $\mcG$ \cite{PHQRFassymmetries}
\be\label{Gtwirl} 
\hat\Pi_{\rm inv}(\bullet):=\frac{1}{|\mcG|}\sum_{g\in\mcG} U_{12S}^g(\bullet)(U^g_{12S})^\dag\,.
\ee 
However, also on $\hat\Pi_{\rm inv}\left(\mcS(\mcH_{\rm kin})\right)$ (and hence $\hat\Pi_{\rm inv}\left(\mcA_{\rm kin}\right)$), the projector 
\be
\hat{\Pi}_{e,1}(\bullet):=\left(\ket{e}\!\bra{e}_1\otimes\mathds1_2\otimes\mathds1_S\right)(\bullet)\left(\ket{e}\!\bra{e}_1\otimes\mathds1_2\otimes\mathds1_S\right)\nonumber
\ee
is non-trivial and thus not invertible.~To see this, consider the two pure states
\begin{align} 
\ket{\psi}_{12S}&=\ket{e}_1\otimes\ket{e}_2\otimes\ket{\phi}_S\,,\nonumber\\
\ket{\psi'}_{12S}&=U_{12S}^g\ket{\psi}_{12S}\,,\nonumber
\end{align}
for some $\ket{\phi}_S\in\mcH_S$ and $g\neq e$.~These two pure states are clearly gauge equivalent and thus not internally distinguishable.~Now build the two density matrices\footnote{Both of these states are not alignable in the terminology of \cite{PHQRFassymmetries,PHinternalQRF}.}
\begin{align}
\rho_{12S}&=\left(a\ket{\psi}_{12S}+b\ket{\psi'}_{12S}\right)\left(a\bra{\psi}_{12S}+b\bra{\psi'}_{12S}\right)\,,\nonumber\\
    \rho_{12S}'&=|a|^2\ket{\psi}\!\bra{\psi}_{12S}+|b|^2\ket{\psi'}\!\bra{\psi'}_{12S}\,,\nonumber 
\end{align}
for some $a,b\neq 0$. As one can easily check, 
\be\label{extdist} 
\hat\Pi_{\rm inv}(\rho_{12S})\neq\hat\Pi_{\rm inv}(\rho'_{12S})\,,
\ee
but
\be 
\hat\Pi_{e,1}\left(\hat\Pi_{\rm inv}(\rho_{12S})\right)=\hat\Pi_{e,1}\left(\hat\Pi_{\rm inv}(\rho'_{12S})\right)\,.
\ee 
Hence, $\hat\Pi_{e,1}$ can \emph{not} be used as a gauge fixing operator on $\hat\Pi_{\rm inv}(\mcA_{\rm kin})$. In other words, \emph{if} we want to mimic the gauge fixing of special relativity above, we can \emph{not} use the states of \cite{Castro-Ruiz:2021vnq,Bartlett:2006tzx,Glowacki:2023nnf,Carette:2023wpz}. 

Now we can ask about the interpretation of this mathematical observation. One way to interpret this is to argue that this is the case because $\hat\Pi_{\rm inv}\left(\mcS(\mcH_{\rm kin})\right)\subset\hat\Pi_{\rm inv}\left(\mcA_{\rm kin}\right)$, due to~\eqref{extdist}, distinguishes states that should only be distinguishable with the help of the external frame. Indeed, $\rho_{12S}$ and $\rho'_{12S}$ are a superposition and a mixture of two pure states that are internally indistinguishable. While the former is pure and the latter mixed, telling them apart requires the external frame; not having access to the external frame imposes a superselection that renders such superpositions and mixtures indistinguishable. Nevertheless, \eqref{extdist} shows that the two states lie in two distinct equivalence classes of states that are invariant according to~\eqref{QIstates}.
The states satisfying~\eqref{QIstates} are, \emph{in this sense}, not fully external frame-independent. 
(This does \emph{not} invalidate their use for QRF changes in operational situations when some external (e.g.\ lab) frame information is permitted.) 

We thus need a further restriction of permissible states. Instead of~\eqref{QIstates}, the perspective-neutral approach implements gauge invariance in the stronger sense of gauge theories (constrained quantization): we only consider pure states such that
\be \label{physstate}
U_{12S}^g\ket{\psi_{\rm phys}}=\ket{\psi_{\rm phys}}\,,\qquad\forall\,g\in\mcG\,.
\ee 
In gauge theories, these are typically called physical states\footnote{Though in situations when the external frame is not fictitious (hence not an actual gauge theory), any state in $\mcH_{\rm kin}$ will have a physical meaning.} and they can be obtained from the externally distinguishable ones via a \emph{coherent} group averaging
\begin{align}
\Piphys&:\mcH_{\rm kin}\rightarrow\mcH_{\rm phys}\,,\nonumber\\
\Piphys&:=\frac{1}{|\mcG|}\sum_{g\in\mcG}U_{12S}^g\,.\label{eq:Piphys}
\end{align}
In this finite group setting, $\Piphys$ constitutes an orthogonal projector and $\mcH_{\rm phys}$ is a proper subspace of $\mcH_{\rm kin}$ \cite{PHQRFassymmetries,PHinternalQRF}.\footnote{For continuous, non-compact groups, this is not true and $\mcH_{\rm phys}$ has to be understood in the sense of distributions \cite{Giulini:1998kf,Giulini:1998rk,delaHamette:2021oex}.} In this sense, physical states are equivalence classes of externally distinguishable ones. 

Note that the density matrices $\mcS(\mcH_{\rm phys})$ are a strict subset of $\hat\Pi_{\rm inv}\left(\mcS(\mcH_{\rm kin})\right)$.~As shown in \cite[Lemma 7]{PHinternalQRF}, $\mcS(\mcH_{\rm phys})$ has the quantum information-theoretic characterisation of being the maximal subspace of states that can be purified in an external-frame-independent manner.~Moreover, in~\cite{PHQRFassymmetries,PHinternalQRF} it was proven that this set of states, in contrast to $\hat\Pi_{\rm inv}\left(\mcS(\mcH_{\rm kin})\right)$, admits an unambiguous gauge-invariant partial trace.\footnote{It is called the \emph{relational trace} and can be used to resolve \cite{PHQRFassymmetries,PHinternalQRF} the so-called ``paradox of the third particle'' \cite{Angelo_2011}.}

This puts us finally into the position to implement the gauge fixing (i) in analogy to special relativity. Physical states admit sufficient redundancy in their description (in terms of kinematical variables) that conditioning them on the frame orientation \emph{is} invertible (it only removes redundant information). Indeed, we define the reduction or \emph{quantum coordinate map} into $R_1$'s internal perspective by simply conditioning on it being in orientation $g\in\mcG$ \cite{PHinternalQRF} (see also \cite{PHtrinity,Hoehn:2020epv,delaHamette:2021oex,PHquantrel}):\footnote{Compared to~\cite{PHinternalQRF}, we have introduced here the coherent group averaging projector $\Piphys$ into this reduction. This has the slight advantage that $\left(\mcR_1^g\right)^{-1}=\left(\mcR_1^g\right)^\dag$, while preserving all other properties of this map, as can be easily checked.}
\begin{align}  
\mcR_1^g&:\mcH_{\rm phys}\rightarrow\mcH_{2}\otimes\mcH_S\nonumber\\
\mcR_1^g&:=\sqrt{|\mcG|}\left(\bra{g}_1\otimes\mathds1_2\otimes\mathds1_S\right)\Piphys\,.\label{eq:iReduction}
\end{align}
As shown in~\cite[Lemma 21]{PHinternalQRF}, this map is unitary and its inverse is  given by
\be \label{Rinv}
\left(\mcR_1^g\right)^{-1}=\sqrt{|\mcG|}\,\Piphys\left(\ket{g}_1\otimes\mathds1_2\otimes\mathds1_S\right)=\left(\mcR_1^g\right)^\dag\,.
\ee
An attentive reader will notice that~\eqref{eq:iReduction} is the finite group version of the conditioning on a clock reading in the Page-Wootters formalism for a relational quantum time evolution \cite{Page:1983uc,giovannetti2015quantum,Smith:2017pwx,Smith:2019imm}; the invertibility of this conditioning was not noticed until \cite{PHtrinity,Hoehn:2020epv,PHquantrel,delaHamette:2021oex}.
In particular, we have
\be 
\left(U_2^g\otimes U_S^g\right)\,\ket{\psi(g')}_{2S}=\ket{\psi(gg')}_{2S}\,,
\ee
where $\ket{\psi(g)}_{2S}:=\mcR_1^g\,\ket{\psi}_{\rm phys}$. The quantum coordinate map $\mcR_1^e$ is now the quantum analogue of the gauge fixing $e^\mu_a=\delta^\mu_a$ in special relativity above.

The space of physical states is thus invertibly isometric to the two kinematical subsystem tensor factors $\mcH_{2}\otimes\mcH_S$, now describing ``the rest'' relative to frame $R_1$.~This is \emph{only} true for ideal QRFs; for non-ideal QRFs the image of a similarly defined \emph{unitary} quantum coordinate map (incl.~for general unimodular Lie groups) is rather a strict \emph{subspace} of $\mcH_2\otimes\mcH_S$ \cite{delaHamette:2021oex,PHtrinity,Hoehn:2020epv,PHquantrel}.~This is where the perspective-neutral approach invoked here and the purely perspective-dependent approach \cite{GiacominiQMcovariance:2019,Giacomini:2018gxh,delaHammetteQRFsforgeneralsymmetrygroups} to QRF covariance \emph{dis}agree in the general case of non-ideal frames.~The purely perspective-dependent approach does not invoke a global Hilbert space such as $\mcH_{\rm phys}$ and instead directly works at the level of the perspectival Hilbert spaces, $\mcH_2\otimes\mcH_S$, and specifically at the level of what would be the kinematical tensor factors. In the present case of ideal QRFs, this gives equivalent results \cite{PHinternalQRF,PHQRFassymmetries,Vanrietvelde:2018dit,Vanrietvelde:2018pgb}, and, for this reason, all results in this article directly apply to that approach as well.~But for non-ideal frames, the latter will include states of ``the rest'' that are inconsistent with the gauge invariance requirement~\eqref{physstate}.~On this larger set of perspectival states,~\eqref{Rinv} would not be invertible. As in the case of the approaches \cite{Castro-Ruiz:2021vnq,Bartlett:2006tzx,Glowacki:2023nnf,Carette:2023wpz} alluded to above, one could thus argue that the purely perspective-dependent approach contains, only in the case of non-ideal QRFs, perspectival states for ``the rest'' $R_2S$ that are not fully independent of external frame information (see \cite{delaHamette:2021oex} for further discussion).

Given the invertibility of~\eqref{eq:iReduction}, which similarly holds for the reduction into $R_2$'s perspective, we can thus view the global physical Hilbert space $\mcH_{\rm phys}$ as an (internal QRF) \emph{perspective-neutral} Hilbert space.~It contains the information about the internal QRF perspectives and, as we shall see shortly, links them.~It can thus be viewed as a description of $R_1R_2S$ \emph{prior} to choosing an internal frame relative to which the remaining degrees of freedom will be described (in particular, $S$, as a composite system, could contain further internal QRFs).

Next, let us consider the QRF version of (ii), i.e.\ equivalently describing everything in a gauge-invariant manner in terms of relational/frame-dressed observables.~To this end, we recall that the relational observables~\eqref{SRrelobs} in special relativity amount to frame-orientation-conditional gauge transformations of the bare quantities describing $R_2S$.~The QRF analogue is \cite{PHinternalQRF,PHQRFassymmetries} (see also \cite{PHtrinity,Hoehn:2020epv,PHquantrel,delaHamette:2021oex,Chataignier:2019kof,Chataignier:2020fys}):
\be 
\tilde{O}_{f_{2S}|R_1}^{g}=|\mcG|\,\hat\Pi_{\rm inv}(\ket{g}\!\bra{g}_1\otimes f_{2S})\,,
\ee 
for arbitrary $f_{2S}\in\mcA_2\otimes\mcA_S$.~Indeed, comparing with~\eqref{Gtwirl}, this is a gauge transformation of $\mathds1_1\otimes f_{2S}$, conditioned on the frame $R_1$ being in orientation $g$.~By construction, $\tilde{O}_{f_{2S|R_1}}^g$ commutes with all gauge transformations.~This relational observable thus measures in a gauge-invariant manner the operator $f_{2S}$, conditioned on the frame being in orientation $g$.~Since we will only be working with these observables on the perspective-neutral Hilbert space, and since 
\be\label{relationpiphyspiinv}
\hat{\Pi}_{\rm phys}(\bullet):=\hat\Pi_{\rm inv}(\bullet)\Piphys=\Piphys(\bullet)\Piphys\,,
\ee 
we shall henceforth instead work with the following definition of relational observables
\be\label{RelDO} 
O_{f_{2S}|R_1}^{g}:=|\mathcal G|\,\hat{\Pi}_{\rm phys}\bigl(\ket{g}\!\bra{g}_1\otimes f_{2S}\bigr)\,.
\ee
This map is a unital $*$-homomorphism $\mcA_2\otimes\mcA_S\to\mcA_{\rm phys}=\mathcal{L}(\mcH_{\rm phys})$, i.e.\ it preserves the identity, products, linear combinations and the adjoint of the operators we would like to describe relative to $R_1$ \cite{PHinternalQRF,PHQRFassymmetries} (see also \cite{PHtrinity,delaHamette:2021oex}).

Equivalence with the gauge-fixing way (i) of ``jumping into $R_1$'s perspective now follows from the following reduction theorem \cite[Lemma 22]{PHinternalQRF} (see also \cite{PHtrinity,Hoehn:2020epv}):
\begin{equation}\label{eq:obsredthm}
    \hat{\mathcal{R}}_1^{g}\bigl(O^{g}_{f_{2S}|R_1}\bigr)=f_{2S}\,,\qquad\qquad\forall\,g\in\mathcal{G}\,,
\end{equation}
where henceforth we shall use the notation
\be \label{eq:Wconjug}
\hat{W}(\bullet):=W\,(\bullet)\,W^\dag
\ee
for the superoperator corresponding to conjugation with the unitary $W$. In particular, we thus also have
\be 
O^g_{f_{2S}|R_1} = \bigr(\hat{\mcR}_1^g\bigr)^\dag(f_{2S})\,.
\ee
Similarly, we have that expectation values of observables are preserved \cite{PHinternalQRF} (see also \cite{PHtrinity,Hoehn:2020epv,delaHamette:2021oex})
\be\label{expecpres}
\braket{\psi_{\rm phys}|O_{f_{2S}|R_1}^g|\phi_{\rm phys}} = \braket{\psi_{2S}(g)|f_{2S}|\phi_{2S}(g)}\,.
\ee 

In other words, gauge fixing the relational observable to the frame being in orientation $g$ returns the bare observable $f_{2S}$ on the perspectival Hilbert space $\mcH_{2}\otimes\mcH_S$. This is the analogue of gauge-fixing~\eqref{SRrelobs} to $e^\mu_a=\delta^\mu_a$. Furthermore, we see that $\mcA_{\rm phys}$ and $\mcA_2\otimes\mcA_S$ are isomorphic and related by the unitary quantum coordinate map. That is, \emph{every} observable on the perspective-neutral Hilbert space can be written as a relational observable relative to frame $R_1$ (and similarly $R_2$). The special relativity analogue of this will be discussed shortly.

Finally, one can check that the frame reorientations~\eqref{qframereorient} define a flow in the relational observable algebra $\mcA_{\rm phys}$ \cite{delaHamette:2021oex}
\be 
\hat{V}_1^g\bigl(O_{f_{2S}|R_1}^{g'}\bigr) = O_{f_{2S}|R_1}^{gg'}\,,
\ee
and thus change the relation between ``the rest'' and the frame, in analogy to what happens in special relativity.

\subsubsection{Internal frame changes}

\begin{figure*}[!ht]
\centering
\begin{subfigure}[b]{0.4375\textwidth}
\includegraphics[width=\textwidth]{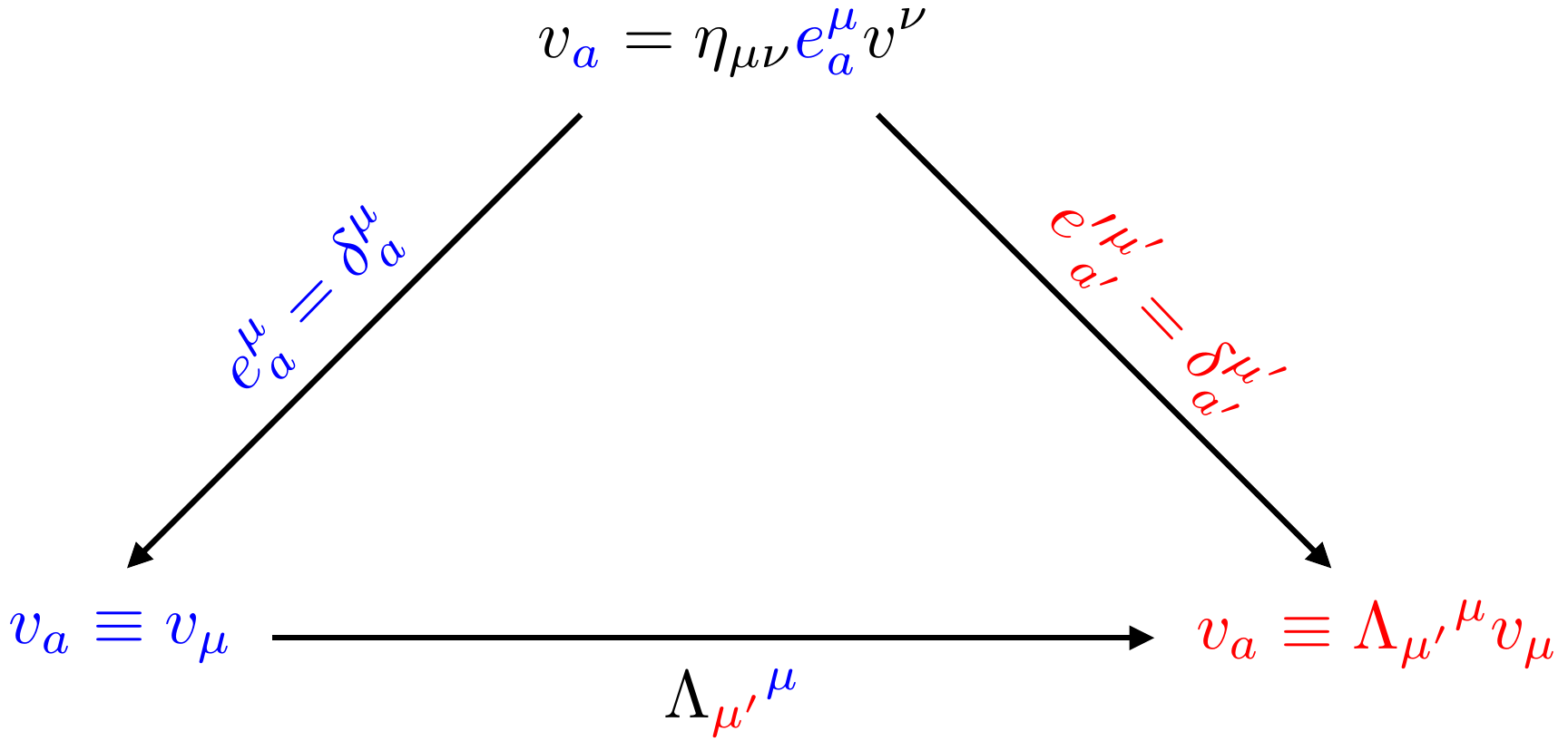}
\caption{}\label{Fig:jumpSR}
\end{subfigure}
\hspace{0.6cm}
\begin{subfigure}[b]{0.5\textwidth}
\includegraphics[width=\textwidth]{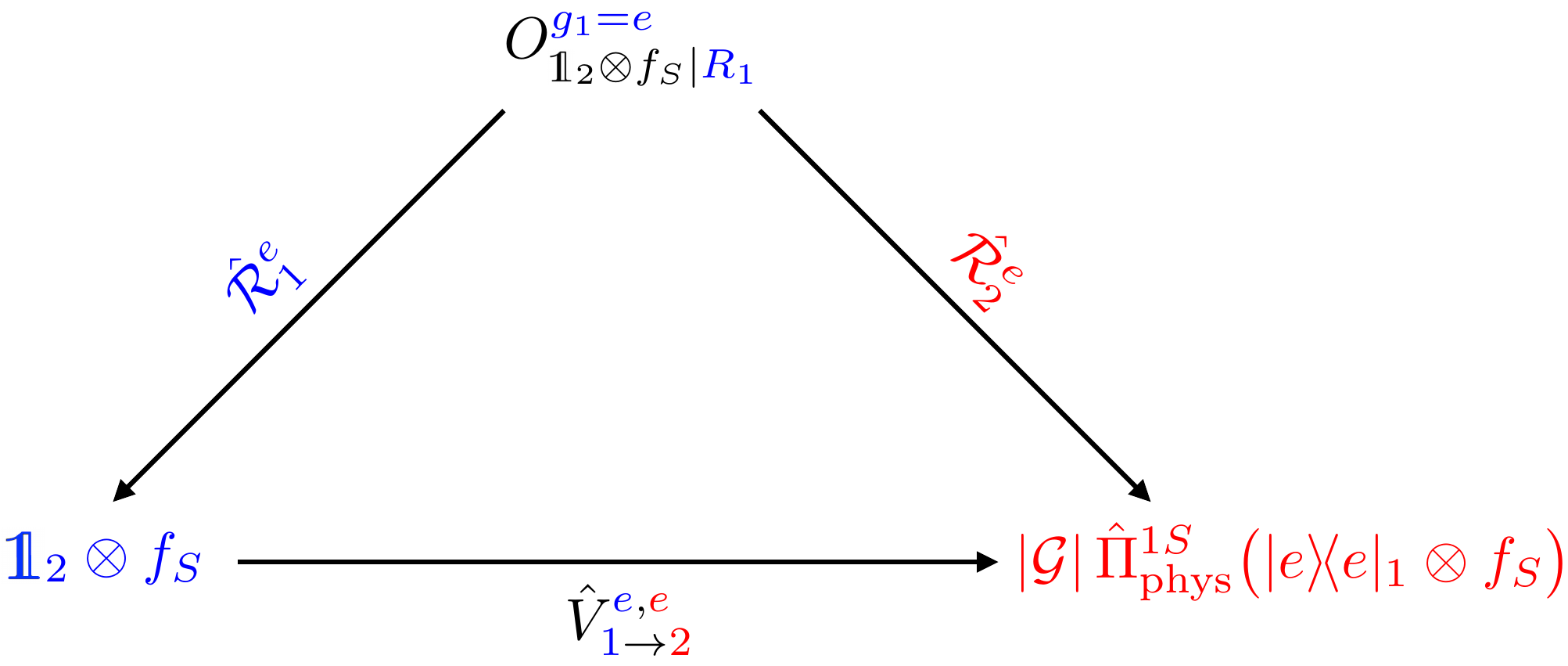}
\caption{}\label{Fig:jumpQRF}
\end{subfigure}
\caption{Internal frame changes in special relativity and in the perspective-neutral approach to QRF covariance:~gauge-fixing the frame-dressed observable $v_a$ to tetrads being in a certain orientation (e.g.~aligning the background coordinates with them) is the analogue of reducing the quantum relational observable $O_{\mathds1_2\otimes f_S|R_1}^{g_1}$ into an internal QRF perspective (down pointing arrows in (a) and (b), respectively).~Here, $f_S$ plays the analogue role of the bare $v_{\mu}$ in the coordinates adapted to the first tetrad and, in analogy with the gauge-fixing of the tetrads in (a), we set $g_1=g_2=e$ in (b).~The QRF counterpart of the Lorentz coordinate change $\Lambda_{\mu'}{}^{\mu}$ is then the quantum coordinate transformation $\hat{V}_{1\to2}^{e,e}$ mapping the observables in $R_1$-perspective to the corresponding observables in $R_2$-perspective.}
\end{figure*}

\begin{figure*}[!ht]
\hspace{-1cm}\includegraphics[width=0.7\textwidth]{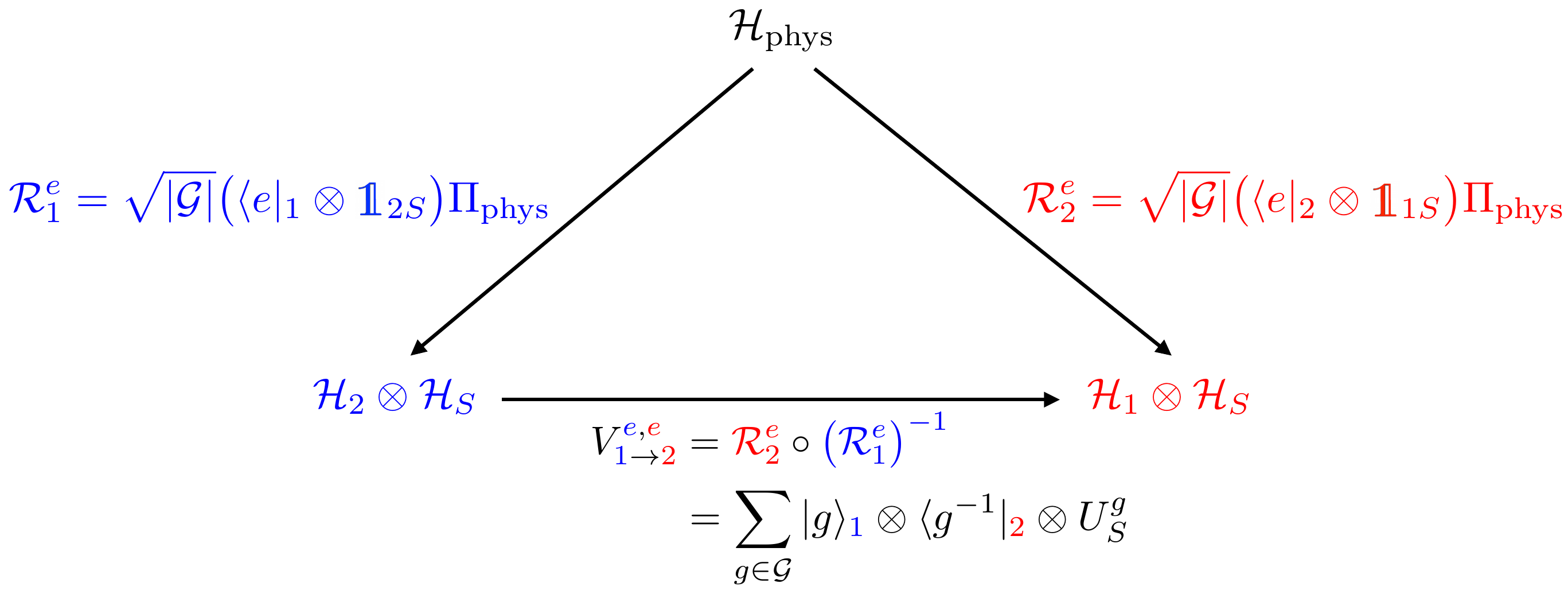}
\caption{Gauge-induced QRF transformation of Hilbert spaces (quantum coordinate change): $V_{1\to2}^{g_1,g_2}$ given in \eqref{eq:Vitoj} maps from $\mathcal H_{2}\otimes\mathcal H_S$ in the perspective of QRF $R_1$ in orientation $g_1$ via the perspective-neutral Hilbert space $\Hphys$ to $\mathcal H_{1}\otimes\mathcal H_S$ in the perspective of $R_2$ in orientation $g_2$.~For comparison with Fig.~\ref{Fig:jumpSR},~\ref{Fig:jumpQRF}, we set $g_1=g_2=e$.}
\label{Fig:hilbertQRFchange}
\end{figure*}

Let us now come to the key question of internal frame covariance: how do we change from the perspective of one internal frame to that of another?~Given that there are two equivalent ways of ``jumping'' into an internal frame perspective, there are also two ways to change internal frame:~(a) by gauge transformation, (b) by a symmetry transformation obtained from the reorientations.

We first consider (a) and again begin with special relativity.~Instead of gauge fixing to tetrad $R_1$ being in a certain orientation, we could equally well do the same for $R_2$, i.e.\ we align the background Lorentz frame with $e'{}_{a'}^\mu$, rather than $e^\mu_a$, e.g.\ by setting $e'{}_{a'}^\mu=\delta^\mu_{a'}$.~It is clear that the two gauge fixings are related by a Lorentz gauge transformation $\Lambda^\mu{}_\nu$, which is nothing but a change of background coordinates.~In particular, originally gauge fixing the relational observable $v_a$ to $e^\mu_a=\delta^\mu_a$ makes it look especially simple: it agrees with the description of the bare $v^\mu$ in this coordinate system (the analogue of~\eqref{eq:obsredthm}, as summarised in Fig.~\ref{Fig:jumpSR} and \ref{Fig:jumpQRF}).~By contrast, gauge fixing it to $e'{}_{a'}^{\mu'}=\delta^{\mu'}_{a'}$ will make it look somewhat more complicated: it will be $v_{\mu'}=\Lambda_{\mu'}{}^{\mu}v_{\mu}$, i.e.\ the description of the bare four-velocity $v_{\mu'}$ in the new coordinate system, expressed by contracting the old $v_{\mu}$ with the Lorentz transformation relating the two background coordinate choices (see Fig.~\ref{Fig:jumpSR}).~Note that, in terms of the numbers, this Lorentz transformation matrix will coincide with the relational observable one in~\eqref{SRrelobs}.~It is thus a gauge transformation by an amount given by the \emph{relation} between the two frames and $\Lambda_{\mu'}{}^{\mu}v_{\mu}$ appears like a composite observable, describing both $R_1$ \emph{and} $S$ in $R_2$'s perspective.

What are now the QRF transformations corresponding to (a)? Instead of conditioning physical states on $R_1$ being in some orientation, we could have done the same with $R_2$ instead.~Invoking the unitarity of these quantum coordinate maps, we obtain the unitary \emph{quantum coordinate changes} from the description relative to frame $R_1$ in orientation $g_1$ to the one relative to frame $R_2$ in orientation $g_2$,
\begin{align}\label{eq:Vitoj}
V_{1\to2}^{g_1,g_2}&:=\mcR_2^{g_2}\circ\left(\mcR_1^{g_1}\right)^{-1}\nonumber\\
&=\sum_{g\in\mathcal G}\ket{gg_1}_1\otimes\bra{g_2g^{-1}}_2\otimes U_S^g\;,
\end{align}
where the latter equality was established in \cite{PHinternalQRF}.~As schematically depicted in Fig.~\ref{Fig:hilbertQRFchange}, this maps us directly from $\mcH_{2}\otimes\mcH_S$ (the perspective of $R_1$) \emph{via} the perspective-neutral Hilbert space $\mcH_{\rm phys}$ to $\mcH_{1}\otimes\mcH_S$ (the perspective of $R_2$).~This QRF change can be generalised to non-ideal frames and general unimodular Lie groups and remains unitary there \cite{delaHamette:2021oex} (see also \cite{PHtrinity,Hoehn:2020epv,PHquantrel}).~This is in contrast to the purely perspective-dependent approach to QRF covariance, whose QRF changes fail to be unitary for non-ideal frames \cite{delaHammetteQRFsforgeneralsymmetrygroups}; as alluded to above, this is essentially because this approach admits perspectival states that can be argued not to be fully free of external frame information.~However, for ideal frames and finite groups, the second line in~\eqref{eq:Vitoj} agrees with the QRF transformations of the perspective-dependent approach \cite{PHinternalQRF,PHQRFassymmetries,delaHammetteQRFsforgeneralsymmetrygroups} (for the equivalence for ideal frames more generally, see \cite{Vanrietvelde:2018pgb,Vanrietvelde:2018dit,delaHammetteQRFsforgeneralsymmetrygroups,delaHamette:2021oex,PHtrinity,GiacominiQMcovariance:2019}).

It is clear that we can also apply these quantum coordinate transformations to observables, mapping the algebra $\mcA_2\otimes\mcA_S$ ($R_1$'s perspective) into $\mcA_1\otimes\mcA_S$ ($R_2$'s perspective).~For example, consider the pure $S$-observable $\mathds1_2\otimes f_S$ relative to $R_1$. One can easily check that
\be 
\hat{V}^{g_1,g_2}_{1\to2}\left(\mathds1_2\otimes f_S\right) = |\mcG|\,\hat{\Pi}_{\rm phys}^{1S}\left(\ket{g_1}\!\bra{g_1}_1\otimes f_S\right)\,,
\ee
where $\Pi^{12}_{\rm phys}$ coincides with~\eqref{eq:Piphys}, except that the $R_2$-tensor factor is removed.~That is, the QRF transformed observable is now a composite operator and in shape agrees with the relational observable $O^{g_1}_{\mathds1_2\otimes f_S|R_1}$ (except that the trivial $R_2$-tensor factor is removed).~This is the analogue of gauge fixing the relational observable $v_a$, describing $S$ relative to $R_1$, to the other frame $R_2$ being in orientation $e'{}_{a'}^{\mu'}=\delta^{\mu'}_{a'}$ above and which appears in coordinate form as $\Lambda_{\mu'}{}^{\mu}v_{\mu}$ (see Figs.~\ref{Fig:jumpSR} and~\ref{Fig:jumpQRF}).~Here, $f_S$ is the analogue of $v_{\mu}$ in the coordinates adapted to the first tetrad and $\Lambda_{\mu'}{}^{\mu}$, the gauge transformation by the relation of the two frames, is the analogue of $|\mcG|\,\hat{\Pi}^{1S}_{\rm phys}(\ket{g_1}\!\bra{g_1}_1\otimes\bullet)$, which appears like a gauge transformation conditioned on frame $R_1$ being in orientation $g_1$ (as seen from the perspective of $R_2$).

The second equivalent way (b) of changing internal frame is simple in special relativity. Specifically, invoking the orthonormality of the tetrads, we can write
\be  \label{relcondreorient}
v_a = \eta_{\mu\nu}e^\mu_a v^\nu = e'^{a'}_\mu e'_{\nu a'}e^\mu_a v^\nu =\Lambda_a{}^{a'}v_{a'}\,,
\ee
where $\Lambda_a{}^{a'}=e'^{a'}_\mu e^\mu_a$ is the $\rm{SO}^+(1,3)$-valued relational observable describing $R_1$ relative to $R_2$ (cfr.~\eqref{SRrelobs}).~This is an actual change of relational observable, mapping between $v_{a'}$, gauge-invariantly describing $S$ relative to $R_2$, to $v_a$, which gauge-invariantly describes $S$ relative to $R_1$.~Since this acts on the frame index and translates by the relation between the two frames, this is a relation-conditional frame reorientation.~Note that the matrix $\Lambda_a{}^{a'}$ coincides with the (inverse) matrix $\Lambda_\mu{}^{\mu'}$ above.

For ideal QRFs, one can also build the analogue of the relation-conditional frame reorientations in~\eqref{relcondreorient} as the second way (b) to implement QRF transformations.~They directly act on relational observables and change them \cite{delaHamette:2021oex}.~However, since we shall not be using them in the main body and they take some space, we summarise them in Appendix~\ref{App:relcondreorientations} instead.

In summary, the perspective-neutral approach to QRF covariance directly mimics the structures underlying special covariance with internal frames and it is the only one that achieves this in the sense illustrated above both for ideal and non-ideal frames \cite{delaHamette:2021oex}. (We emphasise again, however, that the purely perspective-dependent \cite{GiacominiQMcovariance:2019,Giacomini:2018gxh,delaHammetteQRFsforgeneralsymmetrygroups,delaHamette:2021iwx,Ballesteros:2020lgl,streiter2021relativistic,Mikusch:2021kro}, algebraic and effective approaches \cite{JulianArtur,Bojowald:2010qw,Bojowald:2010xp,Bojowald:2019mas,Bojowald:2022caa} are equivalent for ideal frames.)~As such, it can be viewed as implementing the natural quantum extension of the relativity principle.~In particular, the framework is by construction fully relational, i.e.\ independent of external frame information, and furthermore formulated in the language of gauge theory and gravity, thus amenable to that context.~As emphasised above, however, the other approaches to QRF covariance come with different motivations and have a valid applicability in other contexts.~Furthermore, in the ideal frame setting of this article, the perspective-neutral approach is equivalent to the purely perspective-dependent one.~In fact, we shall mostly be working with the QRF transformation~\eqref{eq:Vitoj}; thus, all results below will also apply to the perspective-dependent approach (for ideal frames).

The main aim of this article is to extend previous explorations of the physical implications of the QRF transformations~\eqref{eq:Vitoj}, especially by focusing on what they entail for the subsystem $S$ that is described relative to both frames.~To that end, we first make another observation in special relativity.

\enlargethispage{\baselineskip}

\subsubsection{Relativity of simultaneity from relativity of subsystems}\label{sssec_SrelSR}

\vspace{-1.5mm}

The observation is that $R_1$ and $R_2$ invoke two distinct gauge (external frame) independent, \emph{relational} notions of susbsystem to describe $S$ and this subsystem relativity implies the usual relativity of simultaneity. 

What are all the relational observables describing $S$ relative to $R_1$? There are only $v_a$ and its norm $v_av^a =v_\mu v^\mu$ (and combinations of these).~This constitutes the ``algebra'' of relational observables $\mcA^{\rm phys}_{S|R_1}$, describing $S$ relative to $R_1$.~Similarly, the ``algebra'' of relational observables $\mcA^{\rm phys}_{S|R_2}$ describing $S$ relative to $R_2$ is ``generated'' by $v_{a'}$ and $v_{a'}v^{a'}=v_\mu v^\mu$.~The two relational observable ``algebras'' thus overlap $\mcA^{\rm phys}_{S|R_1}\cap\mcA_{S|R_2}^{\rm phys}= \{\text{functions of }v_\mu v^\mu\}$, but do not coincide.~This is the relativity of subsystems: $R_1$ and $R_2$ describe $S$ with \emph{distinct} gauge-invariant ``algebras''.~These are two distinct, but overlapping ``factors'' in the total ``algebra'' $\Aphys$ of gauge-invariant observables for $R_1R_2S$.~Thanks to~\eqref{relcondreorient}, every gauge-invariant observable can be written as a relational observable either relative to $R_1$ or relative to $R_2$.~This is the special relativistic analogue of the QRF discussion around~\eqref{eq:obsredthm}--\eqref{expecpres}.~As such, the total ``algebra'' $\Aphys$ is generated by $v_a,v^\mu v_\mu$, yielding $\mcA^{\rm phys}_{S|R_1}$, and $\Lambda_a{}^{a'},e'^{\mu}_{a'}e^{b'}_\mu$, yielding $\mcA^{\rm phys}_{R_2|R_1}$, together describing $R_2S$ relative to $R_1$, or, equivalently, by $v_{a'},v_\mu v^\mu$, yielding $\mcA^{\rm phys}_{S|R_2}$, and $\Lambda_a{}^{a'},e^\mu_a e_\mu^b$, yielding $\mcA^{\rm phys}_{R_1|R_2}$, together describing $R_1S$ relative to $R_2$.~In other words, $R_1$ and $R_2$ decompose the total gauge-invariant ``algebra'' in different ways into subsystem $S$ and ``other frame'' (see Fig.~\ref{Fig:relSR}).

For comparison with the QRF case below, we note that the relational observable $v_a$, describing $S$ relative to $R_1$ transforms non-trivially under reorientations $\Lambda_a{}^b$ of $R_1$.~However, it clearly transforms trivially under reorientations $\Lambda_{a'}{}^{b'}$ of the other frame $R_2$ because $v_a$ does not carry any (primed) frame index of $R_2$; it is a scalar for such transformations.~Conversely, $v_{a'}$ transforms non-trivially under $R_2$-reorientations and is a scalar for $R_1$-reorientations.~Clearly, the overlap $v_\mu v^\mu$ of the ``subalgebras'' relationally describing $S$ is thus characterised by observables that are invariant under reorientations of \emph{both} tetrads $R_1$ and $R_2$.~This overlap is therefore independent of both $R_1$ and $R_2$ and  constitutes an \emph{internal} relational observable of $S$.

It is also clear that this subsystem relativity directly implies the usual relativity of simultaneity.~To see this, note that the gauge-invariant frame change map $\Lambda_{a}^{\;\,a'}$ in~\eqref{relcondreorient}, i.e.\ the relational observable encoding the relation between $R_1$ and $R_2$, is \emph{both} the transformation that maps the two non-coincident ``subalgebras'' $\mcA^{\rm phys}_{S|R_2}$ and $\mcA^{\rm phys}_{S|R_1}$ of $\Aphys$ into one another \emph{and} the one transforming the two tetrad frames into one another.\footnote{In fact, more precisely, the ``algebra'' $\mcA^{\rm phys}_{S|R_1}$ contains also $v_b^\Lambda:=\Lambda_b{}^av_a$, where $\Lambda_b{}^a\in\rm{SO}^+(3,1)$ is an arbitrary reorientation of tetrad $R_1$. $v^\Lambda_b$ is then the relational observable answering the question ``what is the four-velocity $v^\mu$ when $R_1$ is in orientation $\Lambda_b{}^a$?''. This can be seen by simply gauge-fixing $v_b^\Lambda$ to $e^\mu_b\equiv \Lambda^\mu{}_b$, where the latter is the inverse matrix of $\Lambda_b{}^a$. This is then the generalisation to arbitrary frame orientations in analogy to how~\eqref{RelDO} describes $f_{2S}$ relative to $R_1$ in orientation $g$ and the $\Lambda$-superscript takes the same role as the $g$-superscript in the quantum theory. The analogous state of affairs holds for $\mcA^{\rm phys}_{S|R_2}$. This means, in particular, that there is no single transformation $\Lambda_a{}^{a'}$ between the elements in $\mcA^{\rm phys}_{S|R_1}$ and those of $\mcA^{\rm phys}_{S|R_2}$, but one for each pair of orientation labels $\Lambda_1$ for $R_1$ and $\Lambda_2$ for $R_2$ on the relational observables before and after transformation; these are just the transformation $\Lambda_a{}^{a'}$ of the main text, contracted with a reorientation each of $R_1$ and $R_2$. This is the same for the symmetry-induced QRF transformations that we alluded to above and discuss in Appendix~\ref{App:relcondreorientations}.} If the two $S$ ``subalgebras'' were coincident, this would not give a non-trivial such transformation. Unless $\Lambda_{a}^{\;\,a'}$ in~\eqref{relcondreorient} is only a spatial rotation, $v_a$ and $v_{a'}$ (i.e.\ the relational $S$-observables that do \emph{not} lie in the overlap $\mcA^{\rm phys}_{S|R_1}\cap\mcA_{S|R_2}^{\rm phys}$ and thus are explicitly $R_1$- and $R_2$-dependent, respectively) decompose into space ($a,a'=1,2,3$) and time ($a,a'=0$) components in distinct ways.~This is the situation summarised in Fig.~\ref{Fig:relSR} in the Introduction.

We shall now discuss the same type of subsystem relativity in the context of QRFs and, like the relativity of simultaneity in special relativity (which it implies), it is \emph{the} reason for the generic frame dependence of physical properties.~As we shall see systematically in this work, it is the reason for why correlations, types of subsystem dynamics, temperature, equilibrium and thermodynamic processes are generally (but not always) QRF-dependent notions.~The purpose of this section is to clarify that this quantum relativity of subsystems is the natural quantum generalisation of the relativity of simultaneity and all its consequences should thus be viewed in the same light.

\section{QRF-dependent subsystems}\label{Sec:QuantRelTPSs}

We had assumed that the kinematical Hilbert space $\mathcal{H}_{\rm kin}=\mathcal{H}_{1}\otimes\mathcal{H}_{2}\otimes\mathcal{H}_S$ comes equipped with a natural tensor product structure (TPS) across the two QRFs and the system $S$, providing a kinematical distinction of subsystems \emph{relative to an external (possibly fictitious) reference frame}.~As we shall now see, this distinction between $R_1,R_2$ and $S$ as independent subsystems does not survive on the perspective-neutral physical Hilbert space $\mathcal{H}_{\rm phys}$.~In Figs.~\ref{Fig:subsystems1} and~\ref{Fig:subsystems2}, we provide pictorial intuition for why \emph{external} (kinematical) and \emph{internal} notions of subsystems are distinct and in Figs.~\ref{Fig:subsystems3} and~\ref{Fig:subsystems4} for why the latter furthermore depend on the choice of internal QRF.~Here in the main text, we focus on the technical argumentation.

A few words about notation are in place: rather than labelling the QRFs in equations by $R_1$ and $R_2$, for simplicity, we shall instead often use indices $i,j=1,2$, and when both $i$ and $j$ appear in expressions, it is assumed that $i\neq j$. When we are in the perspective of frame $R_i$, we also sometimes use the shorthand notation of $\ibar$ for the complement, i.e.\ the ``other frame'' $R_j$ and subsystem $S$. Lastly, operators equipped with a hat $\,\hat{}\,$, denote superoperators; when the hat appears on a unitary, it means conjugation by that unitary as in~\eqref{eq:Wconjug}.

In our finite-group setting, the perspective-neutral Hilbert space $\mathcal{H}_{\rm phys}\subsetneq\mathcal{H}_{\rm kin}$ is a proper subspace that does \emph{not} inherit the kinematical TPS.\footnote{This is a well-known problem even in lattice gauge theories \cite{Casini:2013rba,Donnelly:2011hn,Donnelly:2016auv}, where the non-local nature of the gauge-invariant observables, as e.g.~Wilson loops, prevents us from a straightforward factorisation into kinematical subsystems, as this will generically break gauge-invariance.} This is due to the (gauge-)invariance condition $(U^g_{1}\otimes U^g_2\otimes U^g_{S})\ket{\psi_{\rm phys}}=\ket{\psi_{\rm phys}}$ inducing a redundancy among the kinematical degrees of freedom used to describe $\mathcal{H}_{\rm phys}$. This can be seen from the fact that $\mathcal{H}_{\rm phys}$ is isomorphic to $\mathcal{H}_{i}\otimes\mathcal{H}_S$, but not to $\mathcal{H}_{\rm kin}$ (cfr.~the discussion surrounding~\eqref{eq:iReduction}).\footnote{Note that $\Hphys$ is only isomorphic (under the QRF reduction maps) to the product $\mathcal{H}_{i}\otimes\mathcal{H}_S$ provided the QRF is ideal, i.e.\ its orientation states are perfectly distinguishable (see \cite{delaHamette:2021oex, PHtrinity, Hoehn:2020epv} for the more general case).} Recall that a key idea behind the internal QRF programme is to identify the redundant degrees of freedom with those of the QRF whose perspective one wishes to ``jump'' into. While a TPS on $\mathcal{H}_{\rm kin}$ can be probed using kinematical observables from $\mathcal{A}_{\rm kin}$, a TPS on $\mathcal{H}_{\rm phys}$ has to be defined in terms of gauge-invariant observables from $\mathcal{A}_{\rm phys}$. 
Let us now make precise how ``jumping into a QRF perspective'' (for ideal QRFs) is nothing else than introducing a TPS on the perspective-neutral Hilbert space $\mathcal{H}_{\rm phys}$, which by itself does not come with an \emph{explicit} TPS (and in this sense is TPS-neutral too). We will then sharpen the observation in \cite{PHquantrel,delaHamette:2021oex,Castro-Ruiz:2021vnq} that different choices of QRF amount to inequivalent TPSs (cfr.~Figs.~\ref{Fig:subsystems3} and~\ref{Fig:subsystems4} for illustration), before exploring the ensuing QRF relativity of correlations and thermal properties in later sections in some detail.

\begin{figure*}[!t]
\centering
\begin{subfigure}[b]{0.395\textwidth}
\includegraphics[width=\textwidth]{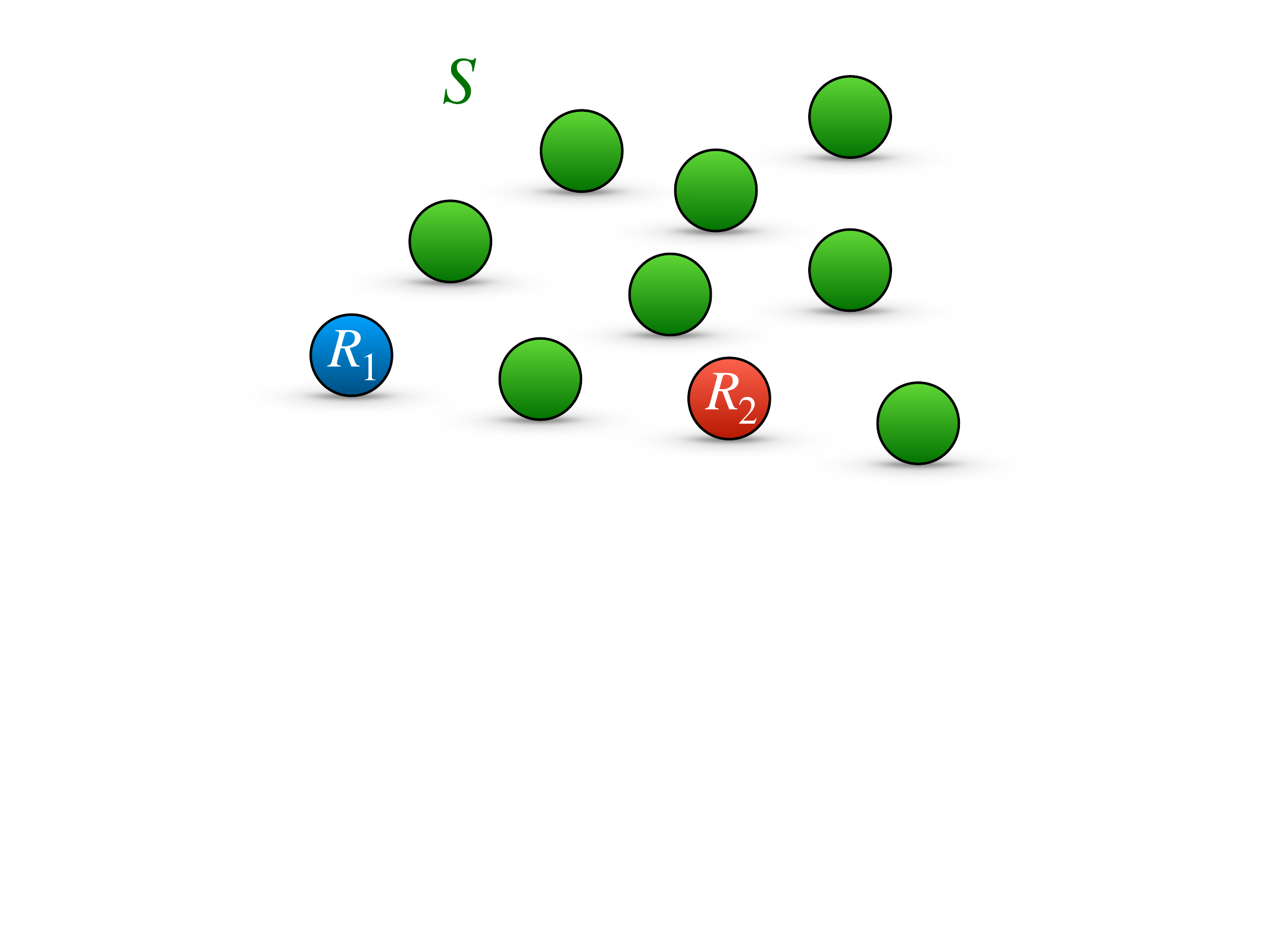}
\caption{}\label{Fig:subsystems1}
\end{subfigure}
\hspace{2cm}
\begin{subfigure}[b]{0.395\textwidth}
\includegraphics[width=\textwidth]{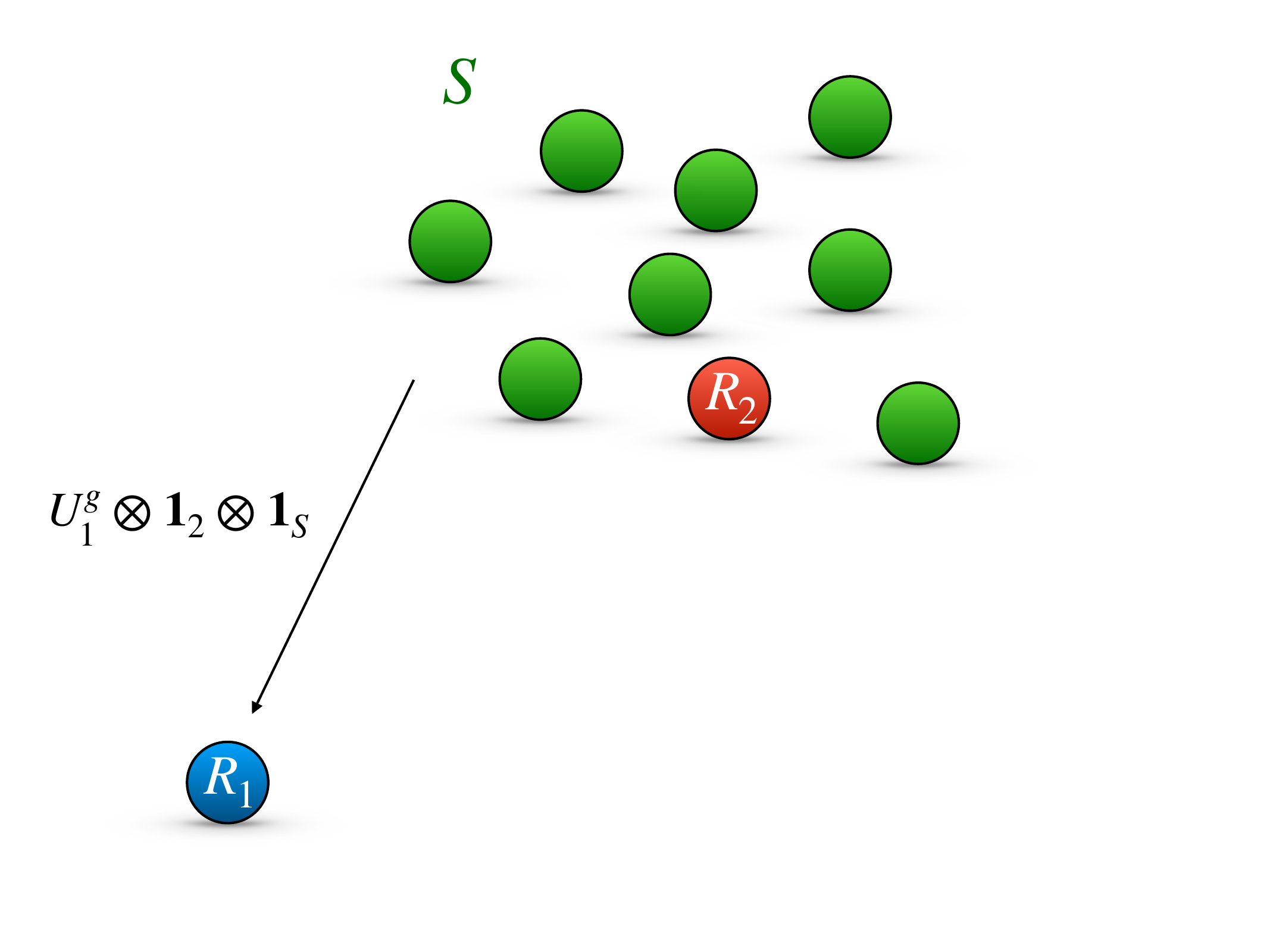}
\caption{}\label{Fig:subsystems2}
\end{subfigure}
\caption{Pictorial illustration of different notions of subsystems and why \emph{external} (kinematical) and \emph{internal} notions of correlations and thermal properties are distinct. (a) A composite subsystem $S$ (green balls) and two internal QRFs $R_1$ (blue ball) and $R_2$ (red ball) as seen from the perspective of an external (possibly fictitious) reference frame. (b) The kinematical notion of subsystem, i.e.\ the one relative to the external frame, is distinct from a relational one, i.e.\ one defined relative to an internal QRF. For example, the description of $S$ is left invariant under a reorientation of $R_1$ relative to the external frame, but changes relative to $R_1$ because the relations between $S$ and $R_1$ change. For this reason, the TPS on the kinematical Hilbert space $\mathcal{H}_{\rm kin}$ between $R_1,R_2$ and $S$ does \emph{not} carry over to the external-frame-independent physical Hilbert space $\mathcal{H}_{\rm phys}\subsetneq\mathcal{H}_{\rm kin}$.
}
\label{Fig:subsystems}
\end{figure*}

\begin{figure*}[!t]
\centering
\begin{subfigure}[b]{0.375\textwidth}
\includegraphics[width=\textwidth]{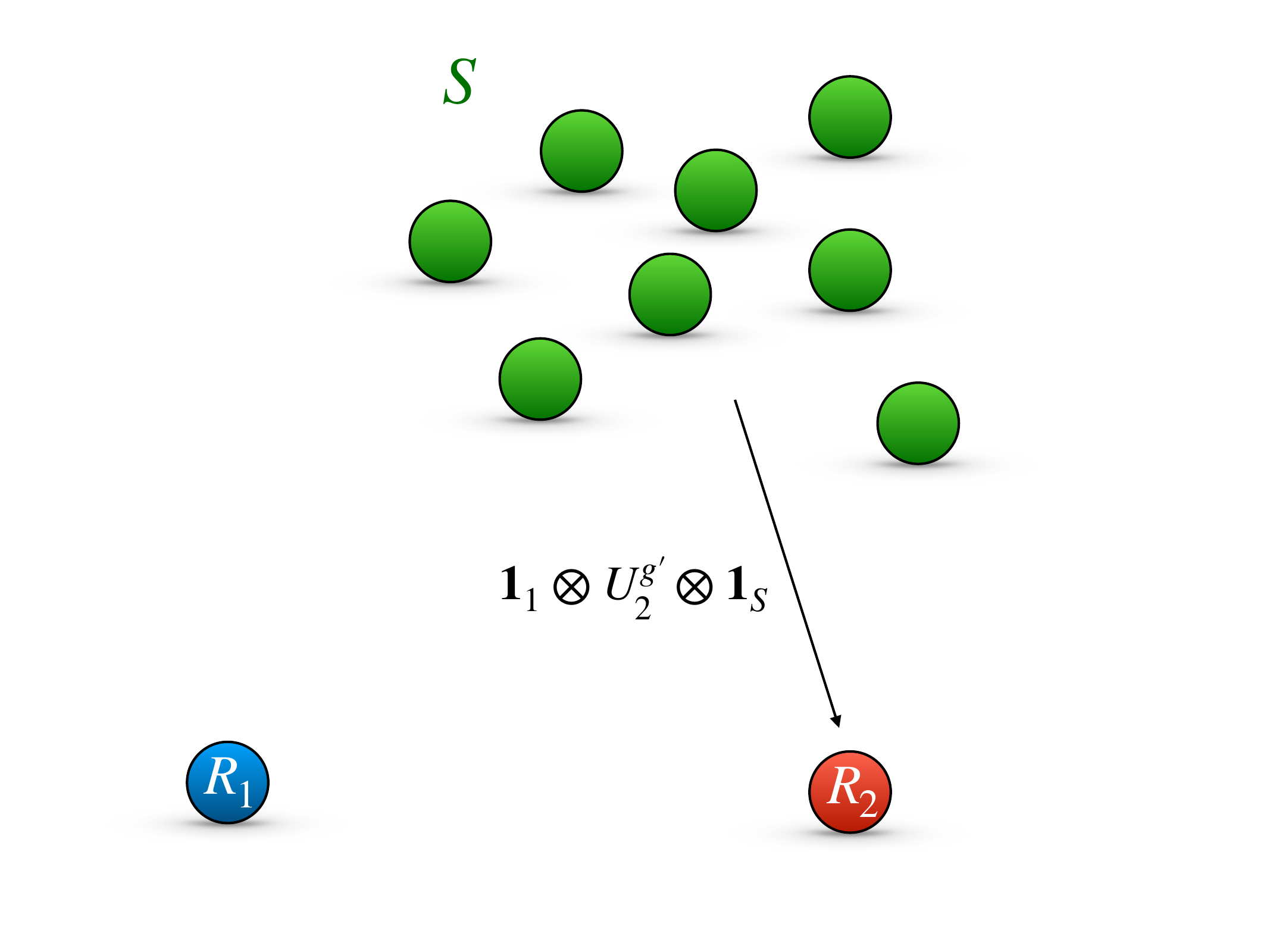}
\caption{}\label{Fig:subsystems3}
\end{subfigure}
\hspace{2cm}
\begin{subfigure}[b]{0.375\textwidth}
\includegraphics[width=\textwidth]{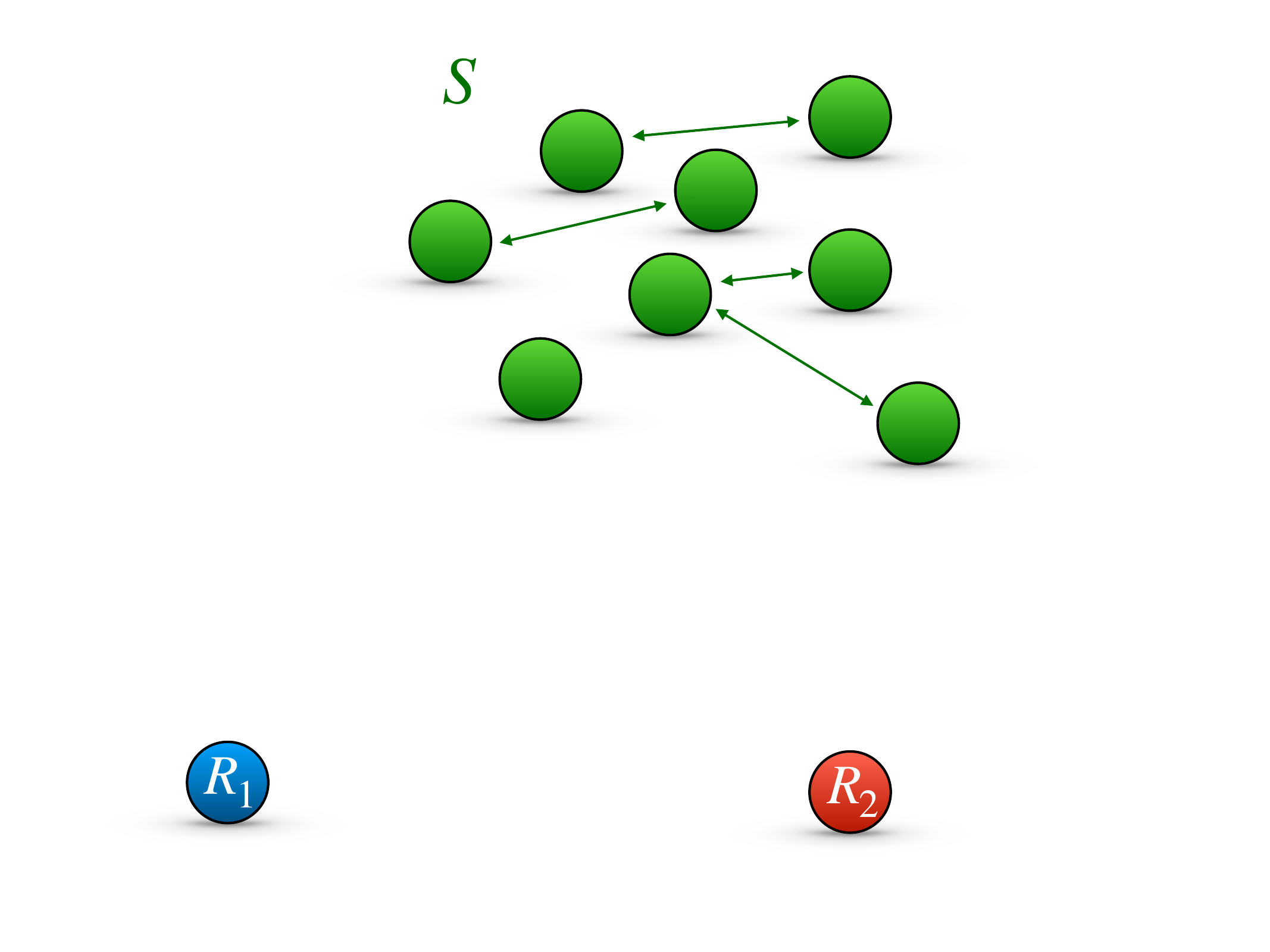}
\caption{}\label{Fig:subsystems4}
\end{subfigure}
\caption{Pictorial illustration of why \emph{internal} notions of subsystem generally depend on the choice of internal QRF. (a) The description of $S$ relative to $R_1$ is invariant under reorientations of $R_2$, but relative to $R_2$ it changes since the relations between $S$ and $R_2$ change. For this reason, the internal perspectives of $R_1$ and $R_2$ induce \emph{inequivalent} TPSs on the internal-frame-neutral $\mathcal{H}_{\rm phys}$ between $R_2,S$ and $R_1,S$, respectively. (b) Relational observables describing $S$ relative to $R_1$ are invariant under reorientations of $R_2$ and vice versa. Hence, relational observables describing $S$ relative to \emph{both} $R_1$ and $R_2$, i.e.\ residing in $\mathcal{A}^{\rm phys}_{S|R_1}\cap\mathcal{A}^{\rm phys}_{S|R_2}$, must be invariant under reorientations of \emph{both} frames. These describe $\mathcal{G}$-invariant properties of $S$ that are independent of the relations between $S$ and both internal frames and correspond to all the relational observables encoding \emph{internal} relations within $S$ (e.g., we could have chosen another internal QRF $R_3$ inside $S$). If $S$ is sufficiently complex, this subalgebra may be large, but is still a proper subalgebra of both $\mathcal{A}^{\rm phys}_{S|R_i}$, $i=1,2$, which also contain observables that depend non-trivially on the relation between $S$ and $R_i$. This is the content of Theorem~\ref{lem_TPSineqalg}.
}
\label{Fig:subsystems22}
\end{figure*}

First, let us recall the definition of a TPS on an abstract Hilbert space $\mathcal{H}$~\cite{Zanardi2001,Zanardi2004,Cotler2019} (see also \cite{Barnum2003GeneralizationBasedCoherent,barnum_subsystem-independent_2004,viola_barnum_2010}). It is well-known that, for \emph{finite-dimensional} quantum systems as here, this definition can be equivalently written in terms of Hilbert spaces or commuting subalgebras of observables. In the latter case, one identifies the subsystems via their observable subalgebras. 

\begin{definition}[\textbf{TPS}]\label{def:TPSreleq}
A TPS $\mathcal{T}$ on $\mathcal{H}$ is an equivalence class of isomorphisms (unitaries) $\mathbf{T}:\mathcal{H}\to\bigotimes_{\alpha=1}^n\mathcal{H}_\alpha$, such that $\mathbf{T}_1\sim\mathbf{T}_2$ if $\mathbf{T}_2\circ\mathbf{T}_1^{-1}$ is a product of multilocal unitaries $\otimes_\alpha U_\alpha$ and permutations of subsystem factors with equal dimension. Equivalently, a TPS $\mathcal{T}$ on $\mathcal{H}$ is a set of subalgebras $\mathcal{A}_\alpha$, $\mathcal{A}_\alpha\in\mathcal{L}(\mathcal{H})$, $\alpha=1,\cdots,n$, with the following properties: (i) subsystem independence: $[\mathcal{A}_{\alpha},\mathcal{A}_{\alpha'}]=0$ and $\mathcal{A}_\alpha\cap\mathcal{A}_{\alpha'}=\mathbb{C}\mathbf{1}$, $\forall\,\alpha\neq\alpha'$; (ii) completness:\footnote{$\bigvee_\alpha\mathcal{A}_\alpha$ denotes the minimal subalgebra of $\mathcal{L}(\mathcal{H})$ (the space of linear operators on $\mathcal{H}$) containing all the $\mathcal{A}_\alpha$, $\alpha=1,\ldots,n$.}  $\bigvee_{\alpha}\mathcal{A}_\alpha=\mathcal{L}(\mathcal{H})$.
\end{definition}
The equivalence of the two definitions follows from the fact that each equivalence class $[\mathbf{T}]$ of isomorphisms identifies a commuting set of subalgebras of $\mathcal{L}(\mathcal{H})$ with the properties (i) \& (ii), and vice versa~\cite{Zanardi2001,Zanardi2004}. 

Comparing the unitary reduction map into the internal perspective of QRF $R_i$, ${\mathcal{R}_i^g:\mathcal{H}_{\rm phys}\to\mathcal{H}_{\bar i}=}\mathcal{H}_{j}\otimes\mathcal{H}_S$, with (the Hilbert space version of) Definition~\ref{def:TPSreleq}, it is clear that we can identify it with a (representative of a) TPS, i.e.\ $\mathbf{T}_i^g=\mathcal{R}_i^g$, on the perspective-neutral Hilbert space $\mathcal{H}_{\rm phys}$.~Since $\mathcal{R}_i^{gg'}=(U^{g}_j\otimes U^{g}_S)\,\mathcal{R}_i^{g'}$, which follows from~\eqref{eq:iReduction} and noting that $(U^{g^{-1}}_i\otimes\mathds1_{j}\otimes\mathds1_S)\,\Piphys=(\mathds1_{i}\otimes U^g_j\otimes U^g_S)\,\Piphys$, we have the equivalence $\mathcal{R}_i^g\sim\mathcal{R}_i^{g'}$, $\forall\,g,g'\in\mathcal{G}$.~Henceforth, we shall call the TPS $[\mathcal{R}_i^{g_i}]$ defined by the reduction map \emph{the natural TPS in $R_i$-perspective}. 

The QRF transformation in~\eqref{eq:Vitoj}, $V_{i\to j}^{g_i,g_j}=\mathcal{R}_j^{g_j}\circ\left(\mathcal{R}_i^{g_i}\right)^{-1}=\mathbf{T}_j^{g_j}\circ\left(\mathbf{T}_i^{g_i}\right)^{-1}$, is thus nothing else than a change of TPS on $\mathcal{H}_{\rm phys}$. Identifying the abstract subsystem labels $\alpha=1$ and $\alpha=2$ in Definition~\ref{def:TPSreleq} with $R_j$ and $S$, respectively, in $R_i$-perspective, as well as with $R_i$ and $S$, respectively, in $R_j$-perspective, it is clear from the explicit form in the second line of~\eqref{eq:Vitoj} that $\mathbf{T}_j^{g_j}\circ\left(\mathbf{T}_i^{g_i}\right)^{-1}=V_{i\to j}^{g_i,g_j}$ is a nonlocal unitary for all $g_i,g_j\in\mathcal{G}$ (it is a conditional unitary).\footnote{One might contend that one is here comparing TPSs between \emph{distinct} subsystems, namely a partition between $R_j$ and $S$ (relative to $R_i$) and one between $R_i$ and $S$ (relative to $R_j$). Owing to the gauge-induced redundancies, $R_1,R_2$ and $S$ are however not independent on $\mathcal{H}_{\rm phys}$. In this sense, the two natural TPSs do correspond to different partitions of the same degrees of freedom.} Hence, we have:
\begin{lemma}\label{lem_TPSineq}
The natural TPS in $R_1$-perspective $[\mathcal{R}_1^{g_1}]$, i.e.\ the equivalence class of isomorphisms $\mathcal{H}_{\rm phys}\to\mathcal{H}_{2}\otimes\mathcal{H}_S$, and the natural TPS in $R_2$-perspective  $[\mathcal{R}_2^{g_2}]$, i.e.\ the equivalence class of isomorphisms $\mathcal{H}_{\rm phys}\to\mathcal{H}_{1}\otimes\mathcal{H}_S$, are inequivalent. 
\end{lemma}
This is the reason why entanglement and superpositions are QRF relative, as noticed in \cite{GiacominiQMcovariance:2019,Vanrietvelde:2018pgb,Vanrietvelde:2018dit,PHtrinity,CastroRuizQuantumclockstemporallocalisability,delaHammetteQRFsforgeneralsymmetrygroups} and then explained in \cite{PHquantrel} and generalised in \cite{delaHamette:2021oex,Castro-Ruiz:2021vnq}.

It is instructive for later purposes to also consider the equivalent algebraic definition of TPSs in the context of QRFs.~The perspective-neutral algebra $\mathcal{A}_{\rm phys}=\mathcal{L}(\mathcal{H}_{\rm phys})$ is TPS-neutral in the same sense as $\mathcal{H}_{\rm phys}$.~Specifically, since the superoperator ${\hat{\mathcal{R}}_i^{g_i}:\mathcal{A}_{\rm phys}\to\mathcal{A}_{\bar i}=\mathcal{A}_{j}\otimes\mathcal{A}_{S}}$, $i=1,2$, from Sec.~\ref{sec:qrfphystps} is unitary, we can write every abstract element of $\mathcal{A}_{\rm phys}$ as a relational observable relative to QRF $R_1$, but \emph{also} as a relational observable relative to QRF $R_2$.~That is, for every $f_{\bar 1}\in\mathcal{A}_{\bar 1}$ and $g_1\in\mathcal{G}$ there exists a $f'_{\bar 2}\in\mathcal{A}_{\bar 2}$ and $g_2\in\mathcal{G}$ such that $O^{g_1}_{f_{\bar 1}|R_1}=O^{g_2}_{f'_{\bar 2}|R_2}$.~This means that any fixed operator in $\mathcal{A}_{\rm phys}$ does not come with an \emph{a priori} physical interpretation; rather it admits a multitude of interpretations and it is the choice of QRF that provides the ``context'' in which to interpret it.~Elements of $\mathcal{A}_{\rm phys}$ are in this sense neutral with respect to a choice of internal QRF perspective.
 
However, relative to a choice of QRF, $\mathcal{A}_{\rm phys}$ decomposes into commuting subalgebras that satisfy properties (i) \& (ii) of Definition~\ref{def:TPSreleq} and so induce a TPS. To see this, recall the reduction theorem of relational observables in~\eqref{eq:obsredthm}, $\hat{\mathcal{R}}_i^{g_i}(O^{g_i}_{f_{\bar i}|R_i}) = f_{\bar i}$. Since this map is unitary, it preserves all relevant algebraic properties. Hence, since $\mathcal{A}_{j}\otimes\mathds1_S$ and $\mathds1_{j}\otimes\mathcal{A}_S$ satisfy (i) \& (ii) as subalgebras of $\mathcal{A}_{\bar i}$ (again under the identification $\alpha=1\leftrightarrow R_j$ and $\alpha=2\leftrightarrow S$), so do $\mathcal{A}_{R_j|R_i}^{\rm phys}:=(\hat{\mathcal{R}}_i^{g_i})^\dag(\mathcal{A}_{j}\otimes\mathds1_S)=O^{g_i}_{\mathcal{A}_{j}\otimes\mathds1_S|R_i}$ and $\mathcal{A}^{\rm phys}_{S|R_i}:=(\hat{\mathcal{R}}_i^{g_i})^\dag(\mathds{1}_{j}\otimes\mathcal{A}_S)=O^{g_i}_{\mathds1_{j}\otimes\mathcal{A}_S|R_i}$ as subalgebras of $\mathcal{A}_{\rm phys}$.~The reduction theorem also makes it clear that the TPS  which the relational observable subalgebras $\mathcal{A}_{R_j|R_i}^{\rm phys}$ and $\mathcal{A}^{\rm phys}_{S|R_i}$ relative to frame $R_i$ induce on $\mathcal{A}_{\rm phys}$ (and thereby on $\mathcal{H}_{\rm phys}$) is precisely given by the natural TPS in $R_i$-perspective $[\mathcal{R}_i^{g_i}]$, as these subalgebras become local with respect to it.

We already know from Lemma~\ref{lem_TPSineq} that the two natural TPSs induced by the relational observables relative to frames $R_1$ and $R_2$ are inequivalent.~However, it is again instructive to formulate this observation in algebraic terms, especially as it admits a simple pictorial interpretation.~The key step in proving the inequivalence of the two TPSs algebraically is the following theorem, illustrated in Figs.~\ref{Fig:subsystems3} and~\ref{Fig:subsystems4}, which 
implies that the relational observables relative to $R_1$ and those relative to $R_2$ define distinct decompositions of $\mathcal{A}_{\rm phys}$ into ``other frame'' and ``system $S$'' (see Fig.\ \ref{Fig:SDOalgebras}). This is the QRF analogue of the same observation that we just made in special relativity in Sec.~\ref{sssec_SrelSR} and which entails that two inertial observers in relative motion in special relativity decompose the set of length observables differently into ``space'' and ``time'' (see Fig.~\ref{Fig:relSR}).~This is also a refinement and simplification of results of \cite{PHquantrel,delaHamette:2021oex,Castro-Ruiz:2021vnq}.
\begin{figure}[!h]
\centering\includegraphics[scale=0.45]{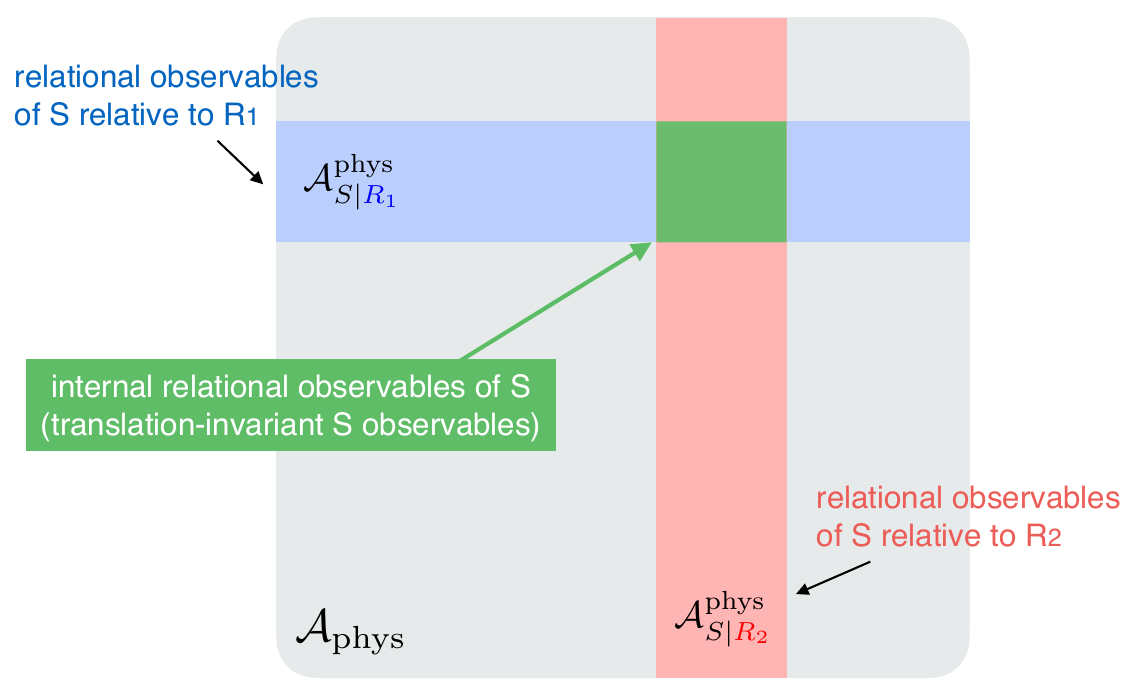}
\caption{Relational observables describing the subsystem $S$ relative to different internal frames identify non-coincident, yet non-trivially overlapping subalgebras of $\Aphys$ (cfr.~Theorem~\ref{lem_TPSineqalg}), leading to distinct decompositions of the latter into ``other frame'' and ``system $S$''.}
\label{Fig:SDOalgebras}
\end{figure}
\begin{theorem}\label{lem_TPSineqalg}
The algebras of relational observables describing solely $S$ relative to the two frames are distinct, but isomorphic subalgebras of $\mathcal{A}_{\rm phys}$, i.e.\ $\mathcal{A}_{S|R_1}^{\rm phys}\neq\mathcal{A}_{S|R_2}^{\rm phys}$. Their overlap $\mathcal{A}_{S|R_1}^{\rm phys}\cap\mathcal{A}_{S|R_2}^{\rm phys}$ is non-trivial and consists of all operators of the form $(\mathds{1}_1\otimes\mathds1_2\otimes f_S)\,\Pi_{\rm phys}=\mathds{1}_1\otimes\mathds1_2\otimes f_S\restriction{\mathcal{H}_{\rm phys}}$,\footnote{$\restriction\mathcal{H}_{\rm phys}$ denotes restriction of the domain to $\mathcal{H}_{\rm phys}$.} where $f_S\in\mathcal{A}_S$ is translation-invariant $[f_S,U^g_S]=0$, $\forall\,g\in\mathcal{G}$.~Hence, observables in the overlap commute with reorientations of \emph{both} frames.
\end{theorem}
The proof is given in Appendix~\ref{App:TPSlocu2} and just a technical implementation of Figs.~\ref{Fig:subsystems3} and~\ref{Fig:subsystems4}. The translation-invariant $S$-observables correspond to all the \emph{internal} relations within $S$ and if $S$ is sufficiently complex they will constitute a non-trivial algebra. 

In fact, the elements in the overlap of the two algebras can alternatively be characterised as those relational observables describing $S$ which are invariant under the second (symmetry-induced) type of QRF transformation analogous to~\eqref{relcondreorient} in special relativity.~This transformation was introduced in \cite{delaHamette:2021oex} and maps the relational observables relative to $R_1$ into the corresponding ones relative to $R_2$.~Since we will not be using this second type of QRF transformation in the remainder, we explain this second characterisation in Appendix~\ref{App:relcondreorientations}. This similarly mimics the discussion of subsystem relativity in special relativity of Sec.~\ref{sssec_SrelSR}.

The theorem thus entails that $\mathcal{A}^{\rm phys}_{S|R_i}\cap\mathcal{A}_{S|R_j}^{\rm phys}\neq\mathbb{C}\mathbf1$, $i\neq j$, overlap non-trivially.~When distinct sets  $\{\mathcal{A}_\alpha\}$ and $\{\mathcal{A}_\beta'\}$ of TPS-inducing subalgebras overlap non-trivially for some $\alpha,\beta$, they induce inequivalent TPSs because they cannot both be local with respect to the same one.~This is the algebraic way to see that the natural TPSs in $R_1$- and $R_2$-perspectives are inequivalent.~The two frames thus give rise to two distinct decompositions of the total algebra $\mathcal{A}_{\rm phys}$ into observables associated with ``the other frame'' and the ``system $S$''.~While these decompositions are isomorphic, probing ``the same'' $S$ properties relative to $R_1$ and $R_2$ means probing the \emph{same} state $\ket{\psi_{\rm phys}}\in\mathcal{H}_{\rm phys}$ with generally distinct observables in $\mathcal{A}_{\rm phys}$, which will generally produce different results.

This can be seen more explicitly in the natural TPSs.~As seen above, in $R_j$-perspective $\mathcal{A}_{S|R_j}^{\rm phys}$ and $\mathcal{A}_{R_i|R_j}^{\rm phys}$ are local, with generic elements of the form $\mathds1_i \otimes A_S\in\mathds{1}_i\otimes\mathcal{A}_{S}$ and $B_i \otimes\mathds1_S\in\mathcal{A}_{i}\otimes\mathds1_S$.~However, their images under the QRF transformation in~\eqref{eq:Vitoj} are non-local in $R_i$-perspective (without further assumptions on $A_S,B_i$, see Sec.~\ref{Sec:Uinv&break}),
 \begin{align}
    &\hat{V}_{j\to i}^{g_j,g_i}\bigl(\mathds 1_{i}\otimes A_{S}\bigr)=\sum_{g\in\mathcal G}\ket{g_j g}\!\bra{g_jg}_j\otimes U^g_S A_{S}U^{g^{-1}}_S\,,\label{localSfromjtoi} \\
     &\hat{V}_{j\to i}^{g_j,g_i}\bigl(B_i \otimes\mathds 1_{S}\bigr)\notag\\
     &=\sum_{g,g'\in\mathcal G}\braket{g_i g^{-1}|B_i|g_i g'}_i\ket{g_jg}\!\bra{g_jg'^{-1}}_j\otimes U^{gg'}_S\,.\label{localifromjtoi}
 \end{align}
 Due to the unitarity of $V_{j\to i}$, the QRF transformed operators \eqref{localSfromjtoi} and \eqref{localifromjtoi} establish commuting subalgebras ${\mathcal{A}_{i|j\to i}:=\hat V_{j\to i}^{g_j,g_i}(\mathcal{A}_{i}\otimes\mathds1_S)}$ and ${\mathcal{A}_{S|j\to i}:=\hat V_{j\to i}^{g_j,g_i}(\mathds1_i\otimes\mathcal{A}_{S})}$ on $\mathcal{H}_{\bar i}$ as well. Here the subscript $j\to i$ is to emphasise that these are the representatives of the $R_i$- and $S$-algebras imported from $R_j$- into the $R_i$-perspective. In particular, these subalgebras are the reductions of the $R_j$ relational observable algebras into  the perspective of ``the other frame'' $R_i$, i.e.\ $\mathcal{A}_{i|j\to i}=\hat{\mathcal{R}}_i^{g_i}\bigl(\mathcal{A}_{R_i|R_j}^{\rm phys}\bigr)$ and $\mathcal{A}_{S|j\to i}=\hat{\mathcal{R}}_i^{g_i}\bigl(\mathcal{A}_{S|R_j}^{\rm phys}\bigr)$, which follows from~\eqref{eq:Vitoj} and~\eqref{eq:obsredthm} (we stress the distribution of $i$- and $j$-labels on the right hand sides). 
 
Later, it will at times be convenient to compare the two natural TPSs ``within a single perspective''.~For example, $\mathcal{A}_{i|j\to i}$ and $\mathcal{A}_{S|j\to i}$ identify the natural TPS of $R_j$-perspective on the Hilbert space $\mathcal{H}_{\bar i}$ of $R_i$-perspective.~Let us write the corresponding explicit tensor product as
$\mathcal H_{i|j\to i}\otimes\mathcal H_{S|j\to i}$.~Thus, we seek a new TPS isomorphism $\mathbf{T}'^{g_j}_i:=\mathbf{U}_{\ibar}^{g_i,g_j}\circ\mathcal{R}_i^{g_i}:\mathcal{H}_{\rm phys}\to\mathcal H_{i|j\to i}\otimes\mathcal H_{S|j\to i}$ with $\mathbf{U}_{\ibar}\in\mathcal{U}(\mathcal{H}_{\bar i})$.~This unitary $\mathbf{U}_{\ibar}$ has to restore the locality of the QRF-transformed subalgebras in Eqs.~\eqref{localSfromjtoi}-\eqref{localifromjtoi}, i.e.\ to disentangle them.~In other words, $\mathbf{U}_{\ibar}$ is essentially the ``inverse of the QRF transformation $V_{j\to i}$ within $R_i$-perspective''.~Comparing with~\eqref{eq:Vitoj}, it is clear that it is given by the conditional unitary\footnote{In the quantum information terminology, the unitary operation \eqref{TPSschangemap} can be understood as a controlled unitary gate implementing a shift of $S$-subsystem with control system $R_j$.~For example, in the case of a system of $N$ two-level ($\mathcal G=\mathbb Z_2$) particles, with $\{\ket{g}\}_{g\in\mathcal G}=\{\ket{0},\ket{1}\}$ being the eigenstates of the Pauli matrix $\sigma_z$ so that $U^{g=0}_S=\mathds1^{\otimes(N-2)}$ and $U_S^{g=1}=\sigma_x^{\otimes(N-2)}$, the expression \eqref{TPSschangemap} reduces for $g_i=g_j=0$ to the familiar \textsf{CNOT} gate $\ket0\!\bra0_2\otimes\mathds1^{\otimes(N-2)}+\ket1\!\bra1_2\otimes\sigma_x^{\otimes(N-2)}$.\label{CNOTfootnote}}
\begin{equation}\label{TPSschangemap}
  \mathbf{U}_{\ibar}^{g_i,g_j}:=  \sum_{g\in\mathcal G}\ket{g_ig}\!\bra{g_jg^{-1}}_j\otimes U_S^g\,.
\end{equation}
Indeed, we then have that the TPSs $[\mathbf{T}_j^{g_j}]$ and $[\mathbf{T}'^{g_j}_i]$ are equivalent because
\begin{align}
    \mathcal{I}_{j\to i}:\!&=\mathbf{T}'^{g_j}_i\circ\left(\mathbf{T}_j^{g_j}\right)^{-1}=\mathbf{U}_{\ibar}^{g_i,g_j}\,V_{j\to i}^{g_j,g_i}\label{IUVrelation}\\
    &=\sum_{g\in\mathcal{G}}\ket{g}_j\otimes\bra{g}_i\otimes\mathds1_S\,\label{eq:Ijtoi}
\end{align}
is a local unitary, namely essentially the ``identity from $R_j$- to $R_i$-perspective''; we shall sometimes call it the \emph{frame swap}.~In particular,  $\mathbf{U}_{\ibar}^{g_j,g_i}$ clearly undoes the nonlocal
transformations \eqref{localSfromjtoi}-\eqref{localifromjtoi} within $R_i$-perspective.

So far we have focused on how gauge-invariant tensor products between ``the other frame'' and ``system $S$'' depend on the choice of internal QRF.~However, if $S$ is a composite system as in Fig.~\ref{Fig:subsystems}, it may come with its own internal TPSs, given by isomorphisms $\mathbf{T}_S:\mathcal{H}_S\to\bigotimes_\alpha\mathcal{H}_\alpha$, or, equivalently, commuting subalgebras $\{\mathcal{A}_\alpha\}$ of $\mathcal{A}_S$ that obey Definition~\ref{def:TPSreleq}.~How is such an internal TPS of $S$ affected by QRF changes? First, it is a refinement of the natural TPSs above, e.g.\ $\tilde{\mathbf{T}}_i^{g_i}:=(\mathds1_j\otimes\mathbf{T}_S)\,\mathcal{R}_i^{g_i}=(\mathds1_j\otimes\mathbf{T}_S)\,\mathbf{T}_i^{g_i}$ yields $\tilde{\mathbf{T}}_i^{g_i}:\mathcal{H}_{\rm phys}\to\mathcal{H}_{j}\otimes\bigotimes_\alpha\mathcal{H}_\alpha$ and we may call it a \emph{refined natural TPS in $R_i$-perspective}.~We can thus not directly compare the internal TPSs of $S$ relative to the two frames, but only via the natural TPSs.~Since the latter are inequivalent, we have that  $\tilde{\mathbf{T}}_j^{g_j}\circ\bigl(\tilde{\mathbf{T}}_i^{g_i}\bigr)^{-1}=(\mathds1_i\otimes\mathbf{T}'_S)\,V_{i\to j}^{g_i,g_j}\,(\mathds1_j\otimes\mathbf{T}^{-1}_S)$ is necessarily a nonlocal unitary across the ``other frame'' and ``system $S$'' partition.~It thus does not matter whether $\mathbf{T}'_S\circ\mathbf{T}_S^{-1}$ is a product of local unitaries and permutations on $\mathcal{H}_S$ and we find:
\begin{corollary}\label{cor:refinedTPS}
Any refined natural TPS in $R_1$-perspective and any refined natural TPS in $R_2$-perspective are inequivalent.
\end{corollary}
In particular, since $S$-observables that one may use in $R_j$-perspective to probe internal correlations of $S$ generally map under QRF transformations into observables probing \emph{both} $R_j$ and $S$ in $R_i$-perspective (cfr.~\eqref{localSfromjtoi}), internal correlations of $S$ will (for the same perspective-neutral state  $\ket{\psi_{\rm phys}}\in\mathcal{H}_{\rm phys}$)  generally depend on the choice of internal QRF. We shall expound on this observation in the following sections.

\section{QRF-independent subalgebras}\label{Sec:Uinv&break}

Having established the QRF relativity of gauge-invariant subsystems (and expanded on \cite{PHquantrel,delaHamette:2021oex,Castro-Ruiz:2021vnq}), a natural question is under which precise conditions properties of ``the other frame'' and subsystem $S$ change or remain invariant under QRF transformations. To this end, we will pursue an algebraic approach, identifying families of subalgebras of gauge-invariant observables, whose elements are independent of the TPS changes induced by the QRF transformations. Since TPSs are defined up to local unitaries, this means we will be looking for observables that are up to local unitaries (and hence local basis changes) invariant under QRF changes and in this sense ``look the same'' relative to $R_1$ and $R_2$.
Clearly, the conditions characterising QRF change invariant observables will simultaneously characterize those observables that necessarily \emph{do} change under such transformations. 

In this section, we will be  studying generic operators (states and observables alike), before investigating the consequences for the (non-)invariance of the \emph{reduced} states of subsystem $S$ under QRF changes.~This will lay the groundwork for exploring the QRF relativity of subsystem correlations, dynamics and thermal properties thereafter, all of which are inextricably tied to the identification of a TPS.~Specifically, interactions in a Hamiltonian depend on the latter.

The analog in special relativity to what we are investigating here is inquiring which properties are invariant under changes of inertial frame.~There are two levels at which we can ask this question: which properties are invariant under the transformations between (i) two inertial frames in \emph{fixed} orientations, and (ii) arbitrary inertial frames, i.e.\ invariant under arbitrary Lorentz transformations?~For example, two observers with relative boost in $x$-direction agree on length measurements in $y$- and $z$-direction, while the set of all inertial observers will only agree on Lorentz scalars (cfr.~Fig.~\ref{Fig:relSR}).~In our setting, the analog of (i) is to ask for operators invariant under $V_{i\to j}^{g_i,g_j}$ for \emph{fixed} orientations $g_i,g_j$ of frames $R_i,R_j$, while (ii) asks for operators invariant under QRF transformations for arbitrary $g_i,g_j$.~For subsystem properties, (ii) is essentially answered by Theorem~\ref{lem_TPSineqalg}. Here, we will mostly focus on (i), though some results will also bear on (ii). 

\enlargethispage{1.75\baselineskip}

\subsection{TPS-invariant subalgebras}\label{Sec:TPSivsubalgebras}

Given the equivalence of the QRF transformation $V_{i\to  j}^{g_i,g_j}$ with $\mathbf{U}_{\ibar}^{g_i,g_j}$ as established in~\eqref{IUVrelation}, we will henceforth work with $\mathbf{U}_{\ibar}^{g_i,g_j}$ given by \eqref{TPSschangemap} instead.~While not necessary, it has the advantage of working within a fixed algebra $\mathcal{A}_{\ibar}$, given that $\mathbf{U}_{\ibar}^{g_i,g_j}\in\mathcal{A}_{\ibar}$.~Indeed, \eqref{IUVrelation} implies that the condition for a given operator to ``look the same'' up to local unitaries in the internal QRF perspectives of $R_1$ and $R_2$ can be equivalently phrased as follows.

\begin{lemma} \label{Lemma:tpseq}
Let $f_{\ibar} \in \mcA_{\ibar}$ be some operator relative to frame $R_i$ and $g_i,g_j \in \mcG$ be fixed orientations of frames $R_i$ and $R_j$, $j\neq i$.~Then, for some bilocal unitary $X = Y_j \otimes Z_S \in \mathcal{U}(\mcH_{\ibar})$,
\be \label{fsw}
\hat{V}_{i\to j}^{g_i,g_j}(f_{\ibar}) = \hat{\mcI}_{i\to j}\hat{X}^\dagger(f_{\ibar}) 
\ee 
if and only if $\hat{\mbfU}_{\ibar}^{g_i,g_j}(f_{\ibar})=\hat{X}^\dagger(f_{\ibar})$.
\end{lemma}
Recall that the tensor factorisation in $\mcA_{\ibar}$ (and hence $X$) is the natural TPS in $R_i$-perspective. 

These conditions define families of  subalgebras of the total algebra $\mathcal{A}_{\ibar}$ (describing gauge-invariant operators in $R_i$-perspective) that are invariant under QRF transformation induced TPS changes.~As shown in Appendix \ref{app:subalgebraAU}, the set of operators satisfying \eqref{fsw} form a unital $*$-subalgebra of $\mathcal{A}_{\ibar}$ for any fixed bilocal unitary $X=Y_{j} \otimes Z_S$ (and pair of orientations $g_i,g_j$ which we drop for notational simplicity)
\be\label{Uinvsubalgebra}
\mcA_{\ibar}^{X} := \{ f_{\ibar} \in \mcA_{\ibar} \;|\; \hat{\mbfU}_{\ibar}^{g_i,g_j}(f_{\ibar}) = \hat{X}^{\dagger}(f_{\ibar}) \}.
\ee
We shall denote then the special case of exact invariance, i.e.\ $X=\mathds1_j\otimes\mathds1_S$, by $\mcA_{\ibar}^{\mathds1}$.~We shall sometimes refer to elements of $\mcA_{\ibar}^{X}$ as TPS- or $\mbfUX$-invariant in what follows.

Being a subalgebra of a finite-dimensional algebra, there exists an orthogonal projector $\PiUX:\mathcal{A}_{\ibar}\to\mcA_{\ibar}^{X}$ so that we can decompose
\be\label{AUXdecofAibar}
\mathcal{A}_{\ibar}=\mcA_{\ibar}^{X}\oplus\mathcal{C}^{X}_{\ibar} 
\ee 
into two Hilbert-Schmidt orthogonal pieces.~QRF transformations preserve this structure, mapping these orthogonal pieces into the corresponding ones in $R_j$-perspective
\be
\mathcal{A}_{\jbar}=\mcA_{\jbar}^{X'}\oplus\mathcal{C}^{X'}_{\jbar}\;,
\ee
where $\mcA_{\jbar}^{X'}$ is defined in the same way as \eqref{Uinvsubalgebra}, but with the $i$ and $j$ labels interchanged and
\be
X'=\hat{\mathcal I}_{i\to j}(X^{\dagger})\,.
\ee
Indeed, one can check that the corresponding projectors are related by QRF transformations as
\be
\begin{aligned}
&\PiUXj\circ\hat{V}_{i\to j}=\hat{V}_{i\to j}\circ\PiUX\;,\\ &\PiUXjperp\circ\hat{V}_{i\to j}=\hat{V}_{i\to j}\circ\PiUXperp\;,
\end{aligned}
\ee
with $\PiUXperp=\hat{\mathds1}_{\ibar}-\PiUX$, and similarly for $\PiUXjperp$.~For later, we also note the following property, proven in App.~\ref{app:subalgebraAU}.
\begin{lemma}\label{lem:pix}
 We have $\PiUX\circ\hat X\circ\hat U_{\ibar}^{g_i,g_j}=\PiUX$.   
\end{lemma}

We can think of the frame orientations $g_i,g_j$ as labelling the families of subalgebras, while the bilocal unitaries $X$ label the members within a family.~We will see below that these subalgebras can intersect non-trivially in multiple ways (besides the identity which clearly belongs to all of them).

Let us first consider whether it depends on the frame orientations if a given $f_{\ibar}$ is part of a TPS-invariant subalgebra. The following lemma shows that this is not the case: if $f_{\ibar}$ resides in a member of one family of subalgebras, then it resides in some member of all families.

\begin{lemma}\label{lem_Agg}
If $f_{\ibar}\in\mathcal{A}_{\ibar}^{X,g_i,g_j}$ for some fixed QRF orientations $g_i,g_j$ and bilocal unitary $X=Y_j\otimes Z_S$, then, for arbitrary $g_i',g'_j\in\mathcal{G}$, also $f_{\ibar}\in\mathcal{A}_{\ibar}^{X',g'_i,g'_j}$ with $X'=Y'_j\otimes Z_S'$ given by
\be
Y'_j=Y_j\,U_j^{(g'_ig'_j){}^{-1}g_ig_j}\,,\qquad Z_S'=Z_S\,U_S^{g_jg'_j{}^{-1}}\,.\nonumber
\ee
\end{lemma}
The proof is found in Appendix~\ref{app:subalgebraAU}.
In other words, if $f_{\ibar}$ is up to local unitaries invariant under the QRF change between $R_1$ in orientation $g_1$ and $R_2$ in orientation $g_2$, then it is also invariant up to (possibly different) local unitaries under QRF transformations for all other orientations.~This further justifies leaving the family labels $g_i,g_j$ implicit in the notation for the subalgebras in \eqref{Uinvsubalgebra}.

We emphasize, however, that \emph{exact} invariance under QRF transformations, i.e.\ $\hat{\mathbf{U}}_{\ibar}^{g_i,g_j}(f_{\ibar})=f_{\ibar}$ for some fixed orientations $g_i,g_j\in\mathcal{G}$, does not in general imply invariance $\hat{\mathbf{U}}_{\ibar}^{g'_i,g'_j}(f_{\ibar})=f_{\ibar}$ for the same observable, but different orientations $g'_i,g'_j$.
This is analogous to (i) in special relativity above, where inertial observers agree on e.g.\ length observables in $y$-direction when they are related by a boost in $x$-direction, but not if they are related by a boost in $y$-direction.

To provide some intution about these TPS-invariant subalgebras $\mcA_{\ibar}^{X}$, we illustrate some of their above mentioned properties in the simplest possible example. 

\begin{example}{\bf(3 qubits).}\label{Ex:AUqubits}
Suppose the two QRFs $R_1,R_2$ and the system $S$ are a qubit each and $\mcG=\mathbb{Z}_2$.~The group basis states for each qubit $\{\ket{g}\}_{g\in\mathcal G}=\{\ket{0},\ket{1}\}$ are given by the computational basis of eigenstates of $\sigma^z$ and we have $U^{g=0}_a=\mathds1_a$ and $U_a^{g=1}=\sigma^{x}_a$, $a=R_1,R_2,S$.~Let us focus on the description of $R_2S$ relative to qubit $R_1$ and, for simplicity, also take the frame orientations of both $R_1,R_2$ to be the identity, i.e.\ $g_1=0=g_2$, before and after QRF change, respectively.~In this case, \eqref{TPSschangemap} reduces to the \textsf{CNOT} gate
\be\label{CNOT} 
\mbfU_{\onebar}:=\mbfU_{\onebar}^{0,0}=\ket0\!\bra0_2\otimes\mathds1_S+\ket1\!\bra1_2\otimes\sigma^x_{S}\,.
\ee
Of the four basis vectors $\ket{00}_{2S},\ket{01}_{2S},\ket{10}_{2S},\ket{11}_{2S}$ in $\mathcal H_{2}\otimes\mathcal H_{S}$, $\ket{00}_{2S},\ket{01}_{2S}$ are left invariant under $\mbfU_{\onebar}$, while $\ket{10}_{2S},\ket{11}_{2S}$ are transformed into each other, so that $\ket{\psi_\pm}_{2S}:=\frac{1}{\sqrt{2}}\left(\ket{10}_{2S}\pm\ket{11}_{2S}\right)$ satisfies $\mbfU_{\onebar}\ket{\psi_\pm}_{2S}=\pm\ket{\psi_\pm}_{2S}$. Hence, the subalgebra that is  invariant under $\mbfU_{\onebar}$ is a 10-dimensional subalgebra of the 16-dimensional algebra $\mcA_{\onebar}$ and given by
\begin{align}
\mcA_{\onebar}^{\mathds1}=&\rm{span}\Big\{\ket{00}\!\bra{00}_{2S},\ket{00}\!\bra{01}_{2S},\ket{01}\!\bra{00}_{2S},\ket{01}\!\bra{01}_{2S}
,\nonumber\\
&\,\,\,\qquad\ket{\psi_+}\!\bra{00}_{2S},\ket{\psi_+}\!\bra{01}_{2S},\ket{00}\!\bra{\psi_+}_{2S},\ket{01}\!\bra{\psi_+}_{2S},\nonumber\\
&\qquad\,\,\,\ket{\psi_+}\!\bra{\psi_+}_{2S},\ket{\psi_-}\!\bra{\psi_-}_{2S}
\Big\}\,.\nonumber
\end{align}

We also obtain a 10-dimensional subalgebra of operators that are $\mbfU_{\onebar}$-invariant up to the bilocal unitary $X=\mathds1_2\otimes\sigma^x_S$.~To this end, observe that $\mbfU_{\onebar}\ket{10}_{2S}=\ket{11}_{2S}={(\mathds1_2\otimes\sigma^x_S)\ket{10}_{2S}}$ and $\mbfU_{\onebar}\ket{11}_{2S}=\ket{10}_{2S}={(\mathds1_2\otimes\sigma^x_S)\ket{11}_{2S}}$.~Moreover,
$\ket{\phi_\pm}_{2S}:=\frac{1}{\sqrt{2}}\left(\ket{00}_{2S}\pm\ket{01}_{2S}\right)$ satisfies $\mbfU_{\onebar}\ket{\phi_\pm}_{2S}=\pm(\mathds1_2\otimes\sigma^x_S)\ket{\phi_\pm}_{2S}=\ket{\phi_\pm}_{2S}$.~The subalgebra thus reads
\begin{align}
\mcA_{\onebar}^{\mathds1\otimes\sigma^x}\!\!=&\rm{span}\Big\{\ket{10}\!\bra{10}_{2S}, \ket{10}\!\bra{11}_{2S}, \ket{11}\!\bra{10}_{2S}, \ket{11}\!\bra{11}_{2S},\nonumber\\
&\,\,\,\qquad
\ket{\phi_+}\!\bra{10}_{2S},\ket{\phi_+}\!\bra{11}_{2S},\ket{01}\!\bra{\phi_+}_{2S},\ket{11}\!\bra{\phi_+}_{2S}\nonumber\\
&\qquad\,\,\,\ket{\phi_+}\!\bra{\phi_+}_{2S},\ket{\phi_-}\!\bra{\phi_-}_{2S}
\Big\}\,.\nonumber
\end{align}

Clearly, the two subalgebras overlap non-trivially: $\mcA_{\onebar}^{\mathds1}\cap\mcA_{\onebar}^{\mathds1\otimes\sigma^x}$ is 6-dimensional and spanned by $\ket{\psi_+}\!\bra{\psi_+}_{2S},\ket{\psi_-}\!\bra{\psi_-}_{2S},\ket{\phi_+}\!\bra{\phi_+}_{2S},\ket{\phi_-}\!\bra{\phi_-}_{2S},\ket{\psi_+}\!\bra{\phi_+}_{2S}$, $\ket{\phi_+}\!\bra{\psi_+}_{2S}$ (which contains the identity).~Together, they thus do not span the entire algebra $\mcA_{\onebar}$. 

Next, let us exhibit  observables that do not reside in any TPS-invariant subalgebra, as well as observables that illustrate Lemma~\ref{lem_Agg}.~Consider the Ising Hamiltonian with longitudinal magnetic field
\be \label{3qham}
H_{\onebar}= A \,\sigma^z_2 \otimes \mathds1_S + B\, \mathds1_2 \otimes \sigma^z_S + C\, \sigma^{z}_2 \otimes \sigma^z_S \,,
\ee
for some parameters $A,B,C\in\mathbb{R}$.~Invoking \eqref{TPSschangemap}, we have $\hat{\mbfU}_{\onebar}^{g_1,g_2}( \sigma^z_2\otimes \mathds1_S)=  (-1)^{g_1+g_2}(\sigma^z_2\otimes \mathds1_S),\; \hat{\mbfU}_{\onebar}^{g_1,g_2}( \mathds1_2 \otimes \sigma^z_S)= (-1)^{g_1}(\sigma_2^z\otimes \sigma^z_S), \; \hat{\mbfU}_{\onebar}^{g_1,g_2}( \sigma_2^z\otimes \sigma^z_S)=  (-1)^{g_2}(\mathds1_2 \otimes \sigma^z_S)$.~Setting $\hat{\mbfU}_{\onebar}^{g_1,g_2}(H_{\onebar})=\hat X^\dag(H_{\onebar})$ and noting the linear independence of the three terms, we see that the bilocal unitary $X=Y_2\otimes Z_S$ needs to satisfy
\begin{eqnarray}
\hat Y^\dag_2(\sigma_2^z)=(-1)^{g_1+g_2}\sigma_2^z\,,\qquad B\hat Z^\dag_S(\sigma_S^z)&=&(-1)^{g_2}C\sigma_S^z\nonumber\\
\text{and}\qquad\qquad C\hat Z^\dag_S(\sigma_S^z)&=&(-1)^{g_2}B\sigma_S^z\,.\nonumber
\end{eqnarray}
For $|B|\neq |C|$, this is impossible.~We conclude that in this case $H_{\onebar}$ does not reside in  $\mcA_{\onebar}^{X}$ for \emph{any} bilocal unitary $X$ and \emph{any} frame orientations $g_1,g_2$.

Let us then consider the standard case $A=B=C=1$, so that the condition can be fulfilled.~We then have
\begin{description}
\item[$g_1=g_2=0$] $H_{\onebar}$ is  $\mbfU_{\onebar}$-invariant and so $H_{\onebar}\in\mcA_{\onebar}^{\mathds1}$,
\item[$g_1=g_2=1$] $H_{\onebar}\in \mcA_{\onebar}^{\mathds1 \otimes \sigma^x}$,
\item[$g_1=0$, $g_2=1$] $H_{\onebar}\in \mcA_{\onebar}^{\sigma^x \otimes \sigma^x}$,
\item[$g_1=1$, $g_2=0$] $H_{\onebar}\in \mcA_{\onebar}^{\sigma^x \otimes \mathds1}$.
\end{description}
Note that the last three cases are related to the first  according to  Lemma~\ref{lem_Agg}.
\end{example}

The example highlights that not every operator will reside in one of the $\mcA_{\ibar}^{X}$.
Those that do not are operators whose locality structure and appearance will be altered by the QRF transformations (such as the interactions in the Ising Hamiltonian $H_{\onebar}$ above for $|B|\neq|C|$). 

On the other hand, TPS-invariant operators have the same degree of locality term-wise by construction.~Indeed, let us write any $f_{\ibar} \in \mcA_{\ibar}$ in the following  manner,
\be 
\label{fdecomp}
f_{\ibar} = f_{j} \otimes \mathds1_{S} + \mathds1_{j} \otimes f_{S} + f_{jS}\,,   
\ee
where $f_{jS}$ is fully non-local across the $jS$ factorisation.~Then for $f_{\ibar} \in \mcA_{\ibar}^{X}$, by Lemma \ref{Lemma:tpseq} we have,
\be \label{fjSwfi}
f_{\jbar} = \tilde{f}_{i} \otimes \mathds1_{S} + \mathds1_{i} \otimes \hat{Z}_{S}^\dagger (f_{S}) + \hat{\mcI}_{i\to j} \hat{X}^\dagger (f_{jS})\,,
\ee
where we have denoted 
\be \label{ftildei}
\tilde{f}_i = \sum \limits_{g,g'} \bra{g} \hat{Y}_{j}^\dagger (f_j) \ket{g'}_j \ket{g}\!\bra{g'}_i \,.
\ee
Up to the unitary local basis change $Y_j$, $\tilde f_i$  has the same matrix elements as $f_j$ did in $R_i$-perspective.
Comparing  \eqref{fdecomp} and \eqref{fjSwfi}, it is clear that the degree of locality is term-wise identical. In the coming subsections, we shall characterise when $R_j$- and $S$-local operators, such as $f_j\otimes\mathds1_S$ and $\mathds1_j\otimes f_S$, reside in some $\mcA_{\ibar}^{X}$ and provide sufficient conditions for general non-local operators thereafter. We will see shortly that the possibilities for $X$ will be rather limited and that \eqref{fjSwfi} will take a much simpler form.
 
Note that the eigenstates of any self-adjoint observable cannot all simultaneously belong to the \emph{same} $\mcA_{\ibar}^{X}$, regardless of whether this observable belongs to this (or any other) subalgebra. If they did, all states in $\mcH_{\ibar}$ would be invariant up to the same bilocal unitary $X$ under QRF changes in contradiction of Lemmas~\ref{lem_TPSineq} and Theorem~\ref{lem_TPSineqalg}. For example, the pure states residing in the spectral decomposition of any density matrix $\rho_{\ibar}$ cannot  all belong to the same TPS-invariant subalgebra.
This will become relevant in Sec.~\ref{Sec:states&entanglement}.

\subsection{TPS-invariant system operators}

Let us now provide necessary and sufficient conditions for $S$-local observables $\mathds1_j\otimes f_S$ to be up to local unitaries QRF change invariant.~This will have many uses later, e.g.\ when characterising QRF-invariant Gibbs states, or the conditions under which a closed subsystem dynamics remains closed under QRF transformation.

To this end, consider the $G$-twirl, or incoherent group averaging, over $S$-translations 
\be \label{pitdef}
\hat{\Pi}_{\mbt} (f_{\ibar}) := \frac{1}{|\mcG|}\sum_{g \in \mcG}  (\mathds1_{j}\otimes U^g_S) \, f_{\ibar} \, (\mathds1_{j}\otimes U^{g^{-1}}_S)
\ee
for any $f_{\ibar} \in \mcA_{\ibar}$. In Appendix \ref{App:algebrasplitting}, we establish the following properties: $\hat{\Pi}_{\mbt}: \mcA_{\ibar} \to \mcA_{\ibar}^\mbt$ is an orthogonal projector, whose image $\mcA_{\ibar}^{\mbt} := \hat{\Pi}_{\mbt}(\mcA_{\ibar})$ is a unital $*$-subalgebra of $\mcA_{\ibar}$. Its orthogonal complement, with respect to the Hilbert-Schmidt inner product, $\mathcal{C}_{\ibar}^{\mbt} := \hat{\Pi}_{\mbt}^{\perp}(\mcA_{\ibar})$, where $\hat{\Pi}_{\mbt}^{\perp}= \hat{\mathds1}_{\ibar} - \hat{\Pi}_{\mbt}$, is a strict subspace of $\mcA_{\ibar}$ but does not close under multiplication as an algebra. 
We may denote their elements respectively by $f_{\ibar}^{\mbt} \in \mcA_{\ibar}^{\mbt}$ and $f_{\ibar}^{\mbt\perp} \in \mathcal{C}_{\ibar}^{\mbt}$. The full algebra can be decomposed into the image and kernel of $\hat{\Pi}_{\mbt}$ as 
\be \label{algtdecomp}
\mcA_{\ibar} = \mcA_{\ibar}^{\mbt} \oplus \mathcal{C}_{\ibar}^{\mbt}\,,
\ee
where the direct sum is with respect to the Hilbert-Schmidt duality.~The subalgebra $\mcA_{\ibar}^{\mbt}$ is completely characterised by invariance under local $S$-translations.~For any $f_{\ibar}\in \mcA_{\ibar}$, it holds that
\be \label{ugscomm}
[\hat{\Pi}_{\mbt}(f_{\ibar}),\mathds1_j \otimes U^g_S] = 0\,, \;\;\; \forall g \in \mcG \,.
\ee 
In particular, $\hat{\Pi}_{\mbt}(f_{\ibar}) = f_{\ibar}$ if and only if $[f_{\ibar}, \mathds1_{j}\otimes U^g_S] = 0, \; \forall g\in \mcG$.~Thus, $\hat{\Pi}_{\mbt}$-invariance is equivalent to commutativity with all local $U^g_S$-translations.~Therefore, any $f_{\ibar} \in \mcA_{\ibar}$ can be written in terms of the $(\mathds1_j\otimes U^g_S)$-invariant and non-invariant components as $f_{\ibar} = f_{\ibar}^{\mbt} + f_{\ibar}^{\mbt\perp}$. 

Interestingly, this algebra decomposition is invariant under QRF transformations, since the projector itself is (see Appendix~\ref{App:algebrasplitting}):
\be\label{mapcomm0}
[\hat{\mbfU}_{\ibar}^{g_i,g_j}, \hat{\Pi}_{\mbt}]=0, 
\ee
or, equivalently,
\be
\hat{\Pi}_{\mbt} \circ \hat{V}^{g_i,g_j}_{i \to j} = \hat{V}^{g_i,g_j}_{i \to j} \circ \hat{\Pi}_{\mbt} \,, \label{pivcomp0} 
\ee
where the projector on the LHS is implicitly understood to act on $\mcA_{\jbar} = \hat{V}_{i \to j}^{g_i,g_j}(\mcA_{\ibar})$, while on the RHS it acts on $\mcA_{\ibar}$.

We are now ready to characterise when a local subsystem observable is invariant under QRF transformation induced TPS changes.
\begin{theorem}[$\mathbf{S}$\textbf{-operators}]\label{thm:Slocals}
Let $\mathds1_j\otimes f_S \in \mcA_{\ibar}$ be some $S$-local operator relative to frame $R_i$.~Then $\mathds1_j\otimes f_S\in\mcA_{\ibar}^{X}$ for some bilocal unitary $X$ (and some orientations $g_i,g_j\in\mcG$) if and only if it lies in the image of $\Pit$, i.e.\ if and only if $f_S$ is translation invariant, $[f_S,U_S^g]=0$, $\forall\,g\in\mcG$.~In particular, $\mathds1_j\otimes f_S\in\mcA_{\ibar}^\mbt$ resides in exactly all those subalgebras $\mcA_{\ibar}^{X}$, with $X$ of the form $X=Y_j\otimes Z_S$, where the $S$-unitary $Z_S$ commutes with $f_S$, $[Z_S,f_S]=0$, and $Y_j$ is an arbitrary $R_j$-unitary.
\end{theorem}
This highlights the role of the projector $\hat\Pi_\mbt$. The proof is given in Appendix~\ref{app:thmproof-S&jlocals}.~In short, $\mathds1_j\otimes f_S$ belongs to \emph{some} $\mcA_{\ibar}^{X}$ if and only if $f_S$ is translation-invariant.~We note that, at the perspective-neutral level, this corresponds precisely to the internal relational observables within $S$ that define the overlap $\mcA^{\rm phys}_{S|R_1}\cap\mcA^{\rm phys}_{S|R_2}$ of the relational observable algebras describing $S$ relative to $R_1$ and $R_2$, as characterised by Theorem~\ref{lem_TPSineqalg}.~These are thus the only $S$-local observables whose locality structure does not change under QRF transformations.

What is remarkable, however, is the fact that all of these observables are not only invariant up to local unitaries under QRF changes, but \emph{exactly} invariant regardless of the frame orientations $g_i,g_j$ (they thus characterise case (ii) from the introduction of this section for $S$-observables). That is, $S$-local observables that are $\mbfU_{X}$-invariant up to non-trivial local unitaries with which they do not commute do not exist; they all belong to $\mcA_{\ibar}^{\mathds1}$ for arbitrary $g_i,g_j$, as well as to any $\mcA_{\ibar}^{X}$ as long as they commute with $X$. Note that this is consistent with Lemma~\ref{lem_Agg}, given that $f_S$ commutes with both $Z_S$ and $U_S^g$. 

Conversely, if $\mathds1_j\otimes f_S$ is not fully translation invariant, then its image $f_{iS}=\hat{V}_{i\to j}^{g_i,g_j}(\mathds1_j\otimes f_S)$ under a QRF transformation will necessarily be non-local with non-trivial contributions from both the $R_i$- and $S$-tensor factors of the natural TPS in $R_j$-perspective.

These observations are essential generalisations of \cite[Lemma 24]{PHinternalQRF} and \cite[Corollary 4]{PHtrinity}.

\subsection{TPS-invariant ``other frame'' operators}

Next, we wish to provide necessary and sufficient conditions for ``the other frame'' observables, i.e.\ $R_j$-local observables $f_j\otimes\mathds1_S$, to be up to local unitaries QRF change invariant.~Also this result will have many uses later, such as when studying how QRF transformations can affect subsystem dynamics.

Consider the incoherent average over conditioning any operator $f_{\ibar} \in \mcA_{\ibar}$ on the ``other frame's'' orientation
\be \label{piddef}
\hat{\Pi}_{\mbd}(f_{\ibar}) := \sum_{g \in \mcG} (\ket{g}\!\bra{g}_j \otimes \mathds1_S) \, f_{\ibar} \, (\ket{g}\!\bra{g}_j \otimes \mathds1_S)\,.
\ee
The action of $\hat{\Pi}_{\mbd}$ is to diagonalise the $R_j$ tensor factor in the group basis, hence the subscript $\mbd$.\footnote{E.g., for density operators, $\hat{\Pi}_{\mbd}(\rho_{\ibar}) = \sum \limits_{g\in \mcG}\braket{g|\rho_{j}|g}_{\!j} \rho_{\ibar}^{R_j|g}$, where $\rho_{j} = \Tr_S(\rho_{\ibar})$ is the reduced state of $R_j$ and $\rho_{\ibar}^{R_j|g}=M_g \rho_{\ibar} M_g^\dagger/\Tr(M_g \rho_{\ibar} M_g^\dagger)$ with $M_g = \ket{g}\!\bra{g}_j \otimes \mathds1_S$ is the state of $\ibar=R_jS$ conditional on $R_j$ being in orientation $g$. Hence, $\hat{\Pi}_{\mbd}(\rho_{\ibar})$ is an average of conditional states weighted by the probability $\braket{g|\rho_{j}|g}_{\!j}$ that $R_j$ is in orientation $g$.\label{PidPVMaverage}}~In complete analogy to $\hat\Pi_\mbt$, we prove the following properties in Appendix \ref{App:algebrasplitting}: $\hat{\Pi}_{\mbd}: \mcA_{\ibar} \to \mcA_{\ibar}^\mbd$, is  an orthogonal projector, whose image $\mcA_{\ibar}^{\mbd} := \hat{\Pi}_{\mbd}(\mcA_{\ibar})$ is a unital $*$-subalgebra of $\mcA_{\ibar}$, while its Hilbert-Schmidt complement $\mathcal{C}_{\ibar}^{\mbd} := \hat{\Pi}_{\mbd}^{\perp}(\mcA_{\ibar})$, where $\hat{\Pi}_{\mbd}^{\perp} = \hat{\mathds1}_{\ibar} - \hat{\Pi}_{\mbd}$, is only a subspace but does not close as an algebra.~Their elements may be respectively denoted by $f_{\ibar}^{\mbd} \in \mcA_{\ibar}^{\mbd}$ and $f_{\ibar}^{\mbd\perp} \in \mathcal{C}_{\ibar}^{\mbd}$ and the full algebra can again be decomposed into the image and kernel of $\hat{\Pi}_{\mbd}$ as
\be \label{algddecomp}
\mcA_{\ibar} = \mcA_{\ibar}^{\mbd} \oplus \mathcal{C}_{\ibar}^{\mbd} \,.
\ee
It is clear that for any $f_{\ibar} \in \mcA_{\ibar}$
\be \label{pvmgcomm}
[\hat{\Pi}_{\mbd}(f_{\ibar}),\ket{g}\!\bra{g}_j \otimes \mathds1_S] = 0\,, \;\;\; \forall g \in \mcG \,.
\ee
We also have $\Pid(f_{\ibar}) = f_{\ibar} $ if and only if $[f_{\ibar}, \ket{g}\!\bra{g}_j \otimes \mathds1_S] = 0$, $\forall g\in \mcG$. Any $f_{\ibar} \in \mcA_{\ibar}$ can be written as a $R_j$-diagonal and off-diagonal piece $f_{\ibar} = f_{\ibar}^{\mbd} + f_{\ibar}^{\mbd\perp}$.

This  decomposition is also invariant under QRF transformations because the projector is (see Appendix \ref{App:algebrasplitting}):
\be \label{mapcomm}
[\hat{\mbfU}_{\ibar}^{g_i,g_j}, \hat{\Pi}_{\mbd}]=0, 
\ee
or, equivalently,
\be
\hat{\Pi}_{\mbd} \circ \hat{V}^{g_i,g_j}_{i \to j} = \hat{V}^{g_i,g_j}_{i \to j} \circ \hat{\Pi}_{\mbd} \,. \label{pivcomp}
\ee

This puts us in the position to identify necessary and sufficient conditions for ``other frame'' observables to be invariant up to local unitaries under QRF changes.
\begin{theorem}[\textbf{Frame-operators}]\label{thm:jlocals}
Let $f_j \otimes \mathds1_S \in \mcA_{\ibar}$ be some $R_j$-local operator relative to frame $R_i$.~Then $f_j\otimes\mathds1_S\in\mcA_{\ibar}^{X}$ for some bilocal unitary $X$ (and some orientations $g_i,g_j\in\mcG$) if and only if it lies in the image of $\Pid$, i.e.\ if and only if $f_j$ is diagonal in the frame orientation basis, $f_j=\sum_g f_g\ket{g}\!\bra{g}_j$.~In particular, $f_j\otimes\mathds1_S\in\mcA_{\ibar}^\mbd$ resides in exactly all those subalgebras $\mcA_{\ibar}^{X}$ with $X=Y_j\otimes Z_S$ such that $Z_S$ is an arbitrary $S$-unitary and the action of the $R_j$-unitary $Y_j$ on $f_j$ is equivalent to that of the (unitary) \emph{parity-swap operator}
\be\label{Pswap}
P_j^{g_i,g_j}:=\sum_g\ket{g_ig}\!\bra{g_jg^{-1}}_j
\ee
in the form 
\be \label{Pswap2}
\hat Y^\dag_j(f_j)=\hat{P}_j^{g_i,g_j}(f_j)\,.
\ee
\end{theorem}
The role of the projector $\Pid$ is clarified by this result. Its proof is given in Appendix~\ref{app:thmproof-S&jlocals}. The parity-swap operator was originally introduced in \cite{GiacominiQMcovariance:2019} and appears frequently in the context of QRF transformations, e.g.\ see \cite{PHtrinity,Vanrietvelde:2018pgb,Vanrietvelde:2018dit,Hohn:2018toe,CastroRuizQuantumclockstemporallocalisability,delaHammetteQRFsforgeneralsymmetrygroups}; it encodes that the relative orientation between the two frames as seen from the perspective of $R_i$ ``changes sign'' when transforming to $R_j$-perspective. Here we generalise it slightly to arbitrary frame orientations $g_i,g_j$. It is the pure ``other frame'' part of \eqref{TPSschangemap}.

The theorem implies that any $f_j\otimes\mathds1_S\in\mcA_{\ibar}^\mbd$ resides, for instance, in $\mcA_{\ibar}^{X}$ (for arbitrary orientations $g_i,g_j$) with $X=(P_j^{g_i,g_j})^\dag\otimes Z_S$ and $Z_S$ arbitrary, but it may belong to more subalgebras.\footnote{For example, consider the qubit case $\mathcal G=\mathbb Z_2$ with frame orientations $g_1=0,g_2=1$, and the $\Pi_{\mbd}$-invariant operator $f_{\onebar}=\sigma_2^z\otimes\mathds1_S$.~In this case, $\mbfU_{\onebar}^{0,1}=\ket{0}\!\bra{1}_2\otimes\mathds1_S+\ket{1}\!\bra{0}_2\otimes\sigma_S^x$ and $P_2^{0,1}=\sigma_2^x$.~Thus, $f_{\onebar}\in\mcA_{\onebar}^{P_2^{0,1}\otimes Z_S}\cap\mcA_{\onebar}^{Y_2\otimes Z_S'}$ with arbitrary $Z_S,Z_S'$, and $Y_2=\sigma_2^y$ (cfr.~\eqref{Pswap2}).}

Notice that, in contrast to the case of TPS-invariant local $S$ observables, covered in Theorem~\ref{thm:Slocals}, which lie in $\mcA_{\ibar}^{\mathds1}$ for \emph{arbitrary} orientations, not all TPS-invariant local $R_j$-observables lie in the same subalgebra $\mcA_{\ibar}^{X}$ for some fixed bilocal unitary $X$ independent of orientations.~In particular, the above result does not imply that all $f_j\otimes\mathds1_S\in\mcA_{\ibar}^\mbd$ also belong to $\mcA_{\ibar}^{\mathds1}$.~That is, not all $R_j$-local observables that are up to local unitaries invariant under QRF transformations (mapping into $R_i$-local observables) are also exactly invariant.~However, the way in which $\mbfU_{\ibar}^{g_i,g_j}(f_j\otimes\mathds1_S)$ and $f_j\otimes\mathds1_S$ can differ is very limited: \eqref{mapcomm} and~\eqref{pivcomp} imply that also the QRF transformed operator is diagonal in the configuration basis.~Since the QRF transformation is unitary, this means that 
$\mbfU_{\ibar}^{g_i,g_j}(f_j\otimes\mathds1_S)$ and $f_j\otimes\mathds1_S$ can only differ by permutation of their eigenvalues on the diagonal in the $\mcG$-basis.~Indeed, this is what one also finds by inserting \eqref{Pswap2} into \eqref{ftildei}.

It is worth stressing that the theorem also implies that the QRF transformation $f'_{iS}=\hat V_{i\to j}^{g_i,g_j}(f_j\otimes\mathds1_S)$ of any $f_j$ that is \emph{not} diagonal in the configuration basis  necessarily results in a non-local operator with non-trivial contributions from both the $R_i$- and $S$-tensor factors of the natural TPS in $R_j$-perspective.

In analogy to the case of $S$-local observables covered by Theorem~\ref{thm:Slocals}, the relational observables $O^{g_i}_{f_j\otimes\mathds1_S|R_i}=(\hat{\mathcal{R}}_i^{g_i})^\dag(f_j\otimes\mathds1_S)$, corresponding to the TPS-invariant operators $f_j\otimes\mathds1_S\in\mcA_{\ibar}^\mbd$ at the perspective-neutral level, lie at the intersection $\mcA_{R_2|R_1}^{\rm phys}\cap\mcA_{R_1|R_2}^{\rm phys}\subset\Aphys$ of the subalgebras  describing the frame $R_2$ relative to frame $R_1$ and vice versa.
This is schematically depicted in Fig.~\ref{Fig:locframeDO}. 

\begin{figure}[!h]
\centering\includegraphics[scale=0.425]{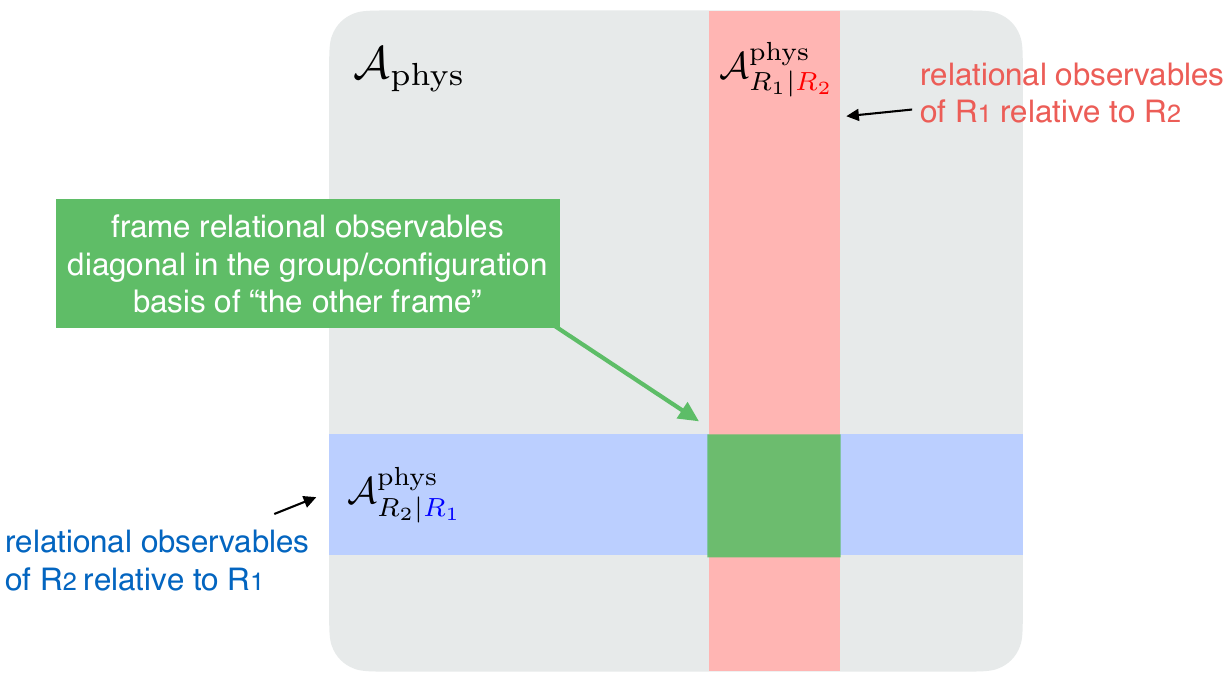}
\caption{Relational operators associated to TPS-invariant $R_j$-local operators relative to $R_i$, as characterised by Theorem~\ref{thm:jlocals} (green).~These are elements of the intersection of the distinct subalgebras of relational observables describing ``the other frame'' $R_j$ relative to $R_i$ and identify the same gauge-invariant notion of complementary d.o.f.~to $S$ in the factorisations relative to the two frames.}
\label{Fig:locframeDO}
\end{figure}

\enlargethispage{\baselineskip}

\subsection{General TPS-invariant operators}

We have characterised when purely $S$-local or $R_j$-local operators are invariant under QRF transformations. Let us now consider general non-local operators as in \eqref{fdecomp}. This will become relevant for density matrices and Hamiltonians later.

We first note that the two projectors defining the algebra decompositions \eqref{algtdecomp} and \eqref{algddecomp} commute as superoperators (see Appendix~\ref{App:algebrasplitting})
\be 
[\hat{\Pi}_{\mbt}, \hat{\Pi}_{\mbd}]=0 \,. \ee
We can thus apply the two decompositions simultaneously, resulting in four components via
\be \label{fulldecomp}
\hat{\mathds1}_{\ibar} = \hat{\Pi}_{\mbd}\circ\hat{\Pi}_{\mbt} + \hat{\Pi}_{\mbd}\circ\hat{\Pi}_{\mbt}^{\perp} + \hat{\Pi}_{\mbd}^{\perp}\circ\hat{\Pi}_{\mbt} +\hat{\Pi}_{\mbd}^{\perp}\circ\hat{\Pi}_{\mbt}^{\perp} \,.
\ee
Next, we observe that, thanks to the identity operator in one of the tensor factors, we have
\begin{align} 
\Pit(\mathds1_j\otimes f_S)&=\Pid\Pit(\mathds1_j\otimes f_S)\,,\\ \Pid(f_j\otimes\mathds1_S)&=\Pit\Pid(f_j\otimes\mathds1_S)\,.
\end{align}
Hence, Theorems~\ref{thm:Slocals} and~\ref{thm:jlocals} imply that the necessary and sufficient condition for $S$- and $R_j$-local operators to be TPS-invariant is to lie in the image of \emph{both} projectors.

The following result shows that this is also a sufficient condition for generic non-local operators.
\begin{lemma} \label{lemma:pinpd-uin}
For any $f_{\ibar} \in \mcA_{\ibar}$ and frame orientations $g_i, g_j \in \mcG$, it holds that
\be \hat{\mbfU}_{\ibar}^{g_i,g_j} \hat{\Pi}_{\mbd} \hat{\Pi}_{\mbt} (f_{\ibar}) = (\hat {P}_{j}^{g_i,g_j}\otimes\hat{\mathds1}_S)\,\hat{\Pi}_{\mbd} \hat{\Pi}_{\mbt} (f_{\ibar}) \,, 
\ee
where $P_j^{g_i,g_j}$ is the parity-swap operator \eqref{Pswap}. In other words, for any $f_{\ibar} \in \mcA_{\ibar}$, we have 
$\hat{\Pi}_{\mbd} \hat{\Pi}_{\mbt} (f_{\ibar})\in\mcA_{\ibar}^{X}$ with $X=(P_j^{g_i,g_j})^\dag\otimes\mathds1_S$.
\end{lemma}
See Appendix \ref{app:thmproof-S&jlocals} for the proof. 

However, residing in the image of both $\Pit,\Pid$ is \emph{not} a necessary condition for non-local observables.~Indeed, consider the following counterexample: $f_{\ibar}=\ket{g}\!\bra{g}_j\otimes f_S$ with $f_S\in\mcA_S$ arbitrary.~Hence, $f_{\ibar}$ lies in the image of $\Pid$, but not necessarily in that of $\Pit$.~Yet, we find 
\be \label{counterexample}
\hat{\mbfU}_{\ibar}^{g_i,g_j}(\ket{g}\!\bra{g}_j\otimes f_S)=\hat X^\dag(\ket{g}\!\bra{g}_j\otimes f_S)\,,
\ee 
i.e.\ $f_{\ibar}\in\mcA_{\ibar}^{X}$, where $X=U_j^{(g_ig_j)^{-1}g^2}\otimes U_S^{g_j^{-1}g}$. 

The interpretation of this example is clear.~Using \eqref{eq:obsredthm} and~\eqref{RelDO}, the corresponding relational observable at the perspective-neutral level is symmetric in $R_i$ and $R_j$,
\begin{eqnarray}
O^{g_i}_{\ket{g}\!\bra{g}_j\otimes f_S|R_i}&=&|\mcG|\,\hat{\Pi}_{\rm phys}(\ket{g_i}\!\bra{g_i}_i\otimes\ket{g}\!\bra{g}_j\otimes f_S) \nonumber\\
&=& O^{g}_{\ket{g_i}\!\bra{g_i}_i\otimes f_S|R_j}\,,\nonumber
\end{eqnarray} 
and can thus equivalently be viewed as a relational observable relative to frame $R_i$ or frame $R_j$. The reduction theorem~\eqref{eq:obsredthm} applied to both frames then explains the observation of QRF transformation invariance up to the bilocal unitary in \eqref{counterexample}.

We leave open what the precise characterising conditions are for generic operators to belong to a TPS-invariant subalgebra.~For later, we note, however, that the condition
 $\hat{\mbfU}_{\ibar}^{g_i,g_j}(f_{\ibar}) = \hat{X}^\dagger (f_{\ibar})$ reads equivalently
 \begin{eqnarray}
 (\hat{\mbfU}_{\ibar}^{g_i,g_j}-\hat X^\dag)(\hat{\Pi}_{\mbd}\hat{\Pi}_{\mbt}^{\perp} + \hat{\Pi}_{\mbd}^{\perp})(f_{\ibar})\qquad\qquad\qquad\qquad&&\label{ugly}\\
 =   (\hat{X}^\dagger - (\hat {P}_{j}^{g_i,g_j}\otimes\hat{\mathds1}_S))\,\hat{\Pi}_{\mbd}\hat{\Pi}_{\mbt}(f_{\ibar})\,.&&\nonumber
 \end{eqnarray} 

Let us illustrate these observations in a simple example, especially the terms one would miss if one restricted to the image of $\Pid\Pit$.
\begin{example}{\bf (3 qubits).}\label{Ex:AUqubits2}~We continue with Example~\ref{Ex:AUqubits} of three qubits. Consider a generic $R_2S$-observable 
\be\label{Ex:2qbobs}
A_{\onebar}=\sum_{\mu,\nu=0}^3 A_{\mu\nu}\sigma_2^\mu\otimes\sigma_S^\nu\;,
\ee
with coefficients $A_{\mu\nu}$ and the conventions $\sigma^0=\mathds1$, and $\{123\}=\{xyz\}$ for the Pauli matrices.~The four basis elements $\sigma_2^\mu\otimes\sigma_S^\nu$ with $\mu\nu=00,01,30,31$ are left invariant under conjugation by $\mbfU_{\onebar}$.~These are precisely the $\Pid\Pit$-invariant terms in \eqref{Ex:2qbobs} which, according to Lemma~\ref{lemma:pinpd-uin}, are always $\mbfU_{\onebar}$-invariant.~Their coefficients are thus unconstrained.~The remaining 12 terms are mapped pairwise into one another and must satisfy the following requirement in order to obey $\hat{\mbfU}_{\onebar}(A_{\onebar})=A_{\onebar}$: 
\begin{align}
&A_{10}=A_{11}\,,\quad A_{12}=A_{23}\,,\quad A_{13}=A_{22}\,,\nonumber\\
&A_{03}=A_{33}\,,\quad A_{02}=A_{32}\,,\quad A_{20}=A_{21}\,.\nonumber
\end{align}
These conditions are precisely those ensuring \eqref{ugly} holds with $X=\mathds1_2\otimes\mathds1_S$.~Note that for $\mathcal G=\mathbb Z_2$ we have $P_j^{0,0}=\mathds1_j$, so that the RHS of \eqref{ugly} vanishes.
\end{example}

\section{QRF-covariance of correlations and entropies}\label{Sec:states&entanglement}

As a first physically relevant application of the abstract framework developed in the previous sections, let us study the interplay between QRF transformations, (reduced) states of the subsystem $S$ at a given instant of time, and frame-induced relativity of quantum correlations among subsystems.~The study of QRF-covariance of a subsystem's properties over time will be the subject of Secs.~\ref{Sec:Dyn},~\ref{Sec:TD}.

The fact that QRF transformations affect entanglement and correlations was first demonstrated through examples in \cite{GiacominiQMcovariance:2019,Angelo_2011,Vanrietvelde:2018pgb,Vanrietvelde:2018dit,PHtrinity,CastroRuizQuantumclockstemporallocalisability,delaHammetteQRFsforgeneralsymmetrygroups} and later systematically explained in \cite{PHquantrel,delaHamette:2021oex,Castro-Ruiz:2021vnq} where the relativity of the notion of quantum subsystems has been formalised, on which we expanded in Sec.~\ref{Sec:QuantRelTPSs}.~What has been missing so far, however, is a thorough investigation into the conditions under which the description of certain properties of a subsystem $S$ remain invariant when changing from one QRF to another.~A characterisation of such conditions in turn clarifies when the same properties \emph{do} change under QRF transformations.~Clearly, different properties require different levels of (non-)invariance.~For example, in order for $S$'s correlations with its complement to remain invariant, its reduced state $\rho_S$ has to be the same up to unitaries in the different QRF perspectives.~For the internal correlations of $S$ to remain invariant requires these unitaries to further be multilocal across the internal TPS of $S$.~For all $S$ properties to be the same in the two perspectives requires $\rho_S$ to be invariant.~The algebraic analysis of the previous sections, in particular the identification of the operator subalgebras $\mcA_{\ibar}^{X}\subset\mcA_{\ibar}$ commuting with the $\mbfU_{\ibar}$-map \eqref{TPSschangemap} up to local unitaries $X=Y_j\otimes Z_S$, allows us to launch such an investigation.

\subsection{Correlations between $S$ and ``the other frame''}

Given any global density operator $\rho_{\text{phys}}\in\mathcal S(\mathcal H_{\text{phys}})$ at the perspective-neutral level, the corresponding state $\rho_{\ibar}(g_i)=\hat{\mathcal{R}}_i^{g_i}(\rho_{\text{phys}})\in\mathcal S(\mathcal H_{\ibar})$ is its representation in the natural TPSs relative to the frame $R_i$.~Due to the inequivalence of the natural TPSs in $R_1$- and $R_2$-perspective (cfr.~Lemma~\ref{lem_TPSineq}), such perspectival representations of the same global states will generically  exhibit different gauge-invariant correlation structures between the subsystem $S$ and its complement.~Generically, the ensuing reduced states $\rho_S$ of $S$ will thus be inequivalent (e.g.\ of distinct rank) relative to the two frames.~However, there are special cases when these correlation structures and subsystem states are equivalent in the two perspectives: as discussed in the previous section (cfr.~Lemma \ref{Lemma:tpseq} and the surrounding discussion), the action of the QRF transformation on elements of $\mcA_{\ibar}^{X}$ amounts to the exchange of the frame subsystems and conjugation by local unitaries. Such operators -- in particular density matrices -- will thus exhibit the same degree of locality and structure of correlations/interactions relative to the two frames.~As we shall see, however, these are not the only ones.

\subsubsection{Pure global states}

Let us start with the characterisation of the pure state case: the requirement of the total state $\rho_{\ibar}$ to belong to a subalgebra $\mcA_{\ibar}^{X}$, for some given bilocal unitary $X=Y_j\otimes Z_S$ provides us with a necessary and sufficient condition for the subsystem states relative to the two frames to be unitarily equivalent.
\begin{theorem}[\textbf{Unitarily related subsystem states}]\label{claim:purestatesrhoSZS}
Let $\rho_{\ibar}\in\mathcal S(\mathcal H_{\ibar})$ be a pure state relative to frame $R_i$ in orientation $g_i$.~Let $\rho_{\jbar}=\hat{V}_{i\to j}^{g_i,g_j}(\rho_{\ibar})\in\mathcal S(\mathcal H_{\jbar})$ be the corresponding QRF-related pure state relative to frame $R_j$ in orientation $g_j$.~Then, the $S$-subsystem states $\rho_{S|R_i}=\Tr_j(\rho_{\ibar})$ and $\rho_{S|R_j}=\Tr_i(\rho_{\jbar})$ in the two frame perspectives are unitarily related, i.e.~$\exists\, Z_S\in\mathcal U(\mathcal H_S)$ s.t.
\be\label{eq:rhoSZSpure}
\rho_{S|R_j}=\hat Z_S^{\dagger}(\rho_{S|R_i})
\ee
if and only if $\rho_{\ibar}\in\mcA_{\ibar}^{X}$ with $X=Y_j\otimes Z_S$ for some local unitary $Y_j \in \mathcal{U}(\mcH_{j})$, i.e.
\be\label{rhoibarinAU}
\hat\mbfU_{\ibar}^{g_i,g_j}(\rho_{\ibar})=(\hat Y_{j}^\dagger\otimes\hat{Z}_S^{\dagger})(\rho_{\ibar})\;.
\ee
\end{theorem}

Note that a generic pure state will not reside in any of the subalgebras $\mcA_{\ibar}^{X}$, so the unitary equivalence of subsystem states in different QRF perspectives is a rather special property (Example~\ref{Ex:Wstate} below will illustrate this).~We also stress that, in contrast to the case of TPS-invariant $S$-local operators in Theorem~\ref{thm:Slocals}, the reduced state $\rho_{S|R_i}$ need \emph{not} necessarily be exactly invariant (i.e.\ it does not necessarily commute with $Z_S$).

The proof of Theorem~\ref{claim:purestatesrhoSZS} is provided in Appendix~\ref{APP:statesandentropy}.~As a special case, we have that, for pure global states, a necessary and sufficient condition for equality of the reduced states of the subsystem $S$ relative to the two frames, i.e.\ $\rho_{S|R_1}=\rho_{S|R_2}$, is that $\rho_{\ibar}\in\mcA_{\ibar}^{X}$ for some local unitary $X=Y_j\otimes\mathds 1_S$ with $Y_j\in\mathcal U(\mathcal H_{j})$\footnote{This encompasses also the case of a non-trivial unitary $Z_S$ belonging to the isotropy group of $\rho_{S|R_i}$ in which case \eqref{rhoibarinAU} tells us that $\rho_{\ibar}\in\mcA_{\ibar}^{Y_j\otimes Z_S}\cap\mcA_{\ibar}^{Y_j\otimes\mathds1_S}$.}.~In this case, the corresponding total states $\rho_{\ibar}$ and $\rho_{\jbar}$ in the two frame perspectives constitute two distinct purifications of the subsystem state $\rho_{S|R_1}=\rho_{S|R_2}$ with the frames $R_1$ and $R_2$ acting as purifying parties. The fact that they are related by a local unitary transformation acting only on the frames $\mathcal I_{i\to j}\circ (Y_j^\dagger\otimes\mathds1_S)$ with $\mathcal I_{i\to j}$ given in~\eqref{eq:Ijtoi} (cfr.~Lemma~\ref{Lemma:tpseq}) is a reflection of the unitarity freedom of purifications \cite{NielsenQIbook2010}: two purifications of the same state are always related by a local unitary transformation between the purifying parties.\footnote{This includes also situations of interest for thermal considerations as e.g.\ thermofield double states in which case the results of Theorem~\ref{claim:purestatesrhoSZS} and Corollary~\ref{claim:Renyipure} below can be applied to the subsystem Gibbs states}

While for any perspective-neutral state $\rho_{\text{phys}}$ the total von Neumann entropy of the composite system is the same relative to the two internal QRFs,  $S_{\rm vN}[\rho_{\bar1}]=S_{\rm vN}[\rho_{\bar2}]=S_{\rm vN}[\rho_{\text{phys}}]$ (the reduction into an internal QRF perspective is unitary, see~\eqref{eq:iReduction}), the entanglement entropy of the $S$-subsystem (which is a measure of bipartite entanglement for pure global states) will generically differ, $S_{\rm vN}[\rho_{S|R_1}]\neq S_{\rm vN}[\rho_{S|R_2}]$, for the same global physical state.
However, under the conditions of Theorem~\ref{claim:purestatesrhoSZS}, the entanglement entropy in the two perspectives will coincide.~More generally, we can ask about (non-)invariance of the $\alpha$-R\'enyi entropies of the subsystem $S$ in the two frame perspectives.~As a direct consequence of Theorem~\ref{claim:purestatesrhoSZS}, we have the following result whose proof can be found in Appendix~\ref{APP:statesandentropy}.
\begin{corollary}[\textbf{QRF-invariant subsystem R\'enyi entropies}]
\label{claim:Renyipure}
Let $\rho_{\ibar}$ be a pure state relative to frame $R_i$ in orientation $g_i$.~Let $\rho_{\jbar}=\hat{V}_{i\to j}^{g_i,g_j}(\rho_{\ibar})$ be the corresponding state relative to frame $R_j$ in orientation $g_j$.~Then, the $\alpha$-R\'enyi entropies of the $S$-subsystem states $\rho_{S|R_i}=\Tr_j(\rho_{\ibar})$ and $\rho_{S|R_j}=\Tr_i(\rho_{\jbar})$ in the two frame perspectives are the same for any value of $\alpha$, i.e.
\be\label{eq:Sstaterenyi}
S_\alpha[\rho_{S|R_i}]=S_\alpha[\rho_{S|R_j}], \;\; \forall \alpha\in(0,1)\cup(1,\infty)
\ee
where
\be\label{alpharenyi}
S_\alpha[\rho_{S|\bullet}]=\frac{1}{1-\alpha}\log\bigl(\Tr_S(\rho_{S|\bullet}^\alpha)\bigr),
\ee
if and only if $\rho_{\ibar} \in \mcA_{\ibar}^{X}$ for some bilocal unitary $X=Y_j\otimes Z_S$.~In other words, the entanglement spectrum agrees in both QRF perspectives, if and only if the global state lies in a TPS-invariant subalgebra. 
\end{corollary}
\noindent
In the limit $\alpha\to1$, R\'enyi entropy reduces to von Neumann entropy \cite{Muller-Lennert:2013} so that, for a pure total state
$\rho_{\ibar}$, its $\mbfUX$-invariance provides us with a sufficient condition for the entanglement entropy $S_{\rm EE}[\rho_{\ibar}] = S_{\rm vN}[\rho_{S|R_i}] =
-\Tr(\rho_{S|R_i}\log\rho_{S|R_i})$ to be invariant under QRF changes, that is $S_{\rm vN}[\rho_{S|R_i}] = S_{\rm vN}[\rho_{S|R_j}]$.

To illustrate the above results in a simple setting, let us consider the following example with $W$ states.~Further examples with physically relevant bipartite and multipartite entangled states and somewhat different properties, such as generalised Bell (GB) and Greenberger–Horne–Zeilinger (GHZ) states, are provided in Appendix~\ref{app:Vex}.

\begin{example}{\bf (W States)}\label{Ex:Wstate}
~Let us consider a system of $N>4$ qubits ($\mathcal G=\mathbb Z_2$) and let $\ket{\psi}_{\onebar}=\ket{\varphi}_{2}\otimes\ket{W}_{N-2}$ be a separable pure state relative to qubit $R_1$ in orientation $g_1=0=e$ (e.g., qubit $R_1$ sits in the origin).~The state $\ket{\varphi}_{2}$ of particle $R_2$ relative to $R_1$ is some superposition $\ket{\varphi}_{2}=a\ket{0}_{2}+b\ket1_{2}$, $a,b\in\mathbb C$, $|a|^2+|b|^2=1$, and the subsystem $S$ consisting of the remaining $N-2$ qubits is in a W state
\be\label{Wqubits}
\ket{W}_{N-2}=\frac{1}{\sqrt{N-2}}\bigl(\ket{10\dots0}+\ket{01\dots0}+\dots+\ket{0\dots01}\bigr).
\ee
Let us investigate the necessary and sufficient condition \eqref{rhoibarinAU} for subsystem state equivalence.~For $\mathcal G=\mathbb Z_2$, we only have $g=0,1$ with $U^{g=0}=\mathds 1$ and $U^{g=1}=\sigma_x$.~The expression \eqref{TPSschangemap} reduces for $g_1=g_2=0$ to the familiar \textsf{CNOT} gate $\mbfU_{\onebar}=\ket0\!\bra0_2\otimes\mathds1^{\otimes(N-2)}+\ket1\!\bra1_2\otimes\sigma_x^{\otimes(N-2)}$.~Thus, we have
\begin{align}
\mbfU_{\onebar}\ket{\psi}_{\onebar}&=a\ket0_2\otimes\ket{W}_{N-2}\nonumber\\
&\;+b\ket1_j\otimes\hspace{-0.75cm}\underset{\frac{1}{\sqrt{N-2}}(\ket{011\dots1}+\ket{101\dots1}+\dots+\ket{11\dots10})}{\underbrace{\sigma_x^{\otimes(N-2)}\ket{W}_{N-2}}}\,.\label{UpsiibarWstate}
\end{align}
The corresponding state $\ket{\psi}_{\twobar}=V_{1\to 2}^{g_1=0,g_2=0}\ket{\psi}_{\onebar}$ relative to $R_2$ then reads 
\be\label{psijbarWstate}
\ket{\psi}_{\twobar}=a\ket0_1\otimes\ket{W}_{N-2}+b\ket1_1\otimes\sigma_x^{\otimes(N-2)}\ket{W}_{N-2}\,.
\ee

\noindent
Clearly, for a generic superposition $a,b\neq0$, condition \eqref{rhoibarinAU} is violated and so $\rho_{\onebar}=\ket{\psi}\!\bra{\psi}_{\onebar}$ does not reside in any subalgebra $\mcA_{\onebar}^{X}$ for any bilocal unitary $X$.~In line with Theorem~\ref{claim:purestatesrhoSZS}, $\ket{\psi}_{\twobar}$ is not a separable state anymore and consequently the subsystem $S$ as seen from $R_2$ is now in a mixed state
\begin{align}\label{Wex:Sjpersp}
\rho_{S|R_2}&=|a|^2\ket W\bra W_{N-2}\nonumber\\
&\;+|b|^2\sigma_x^{\otimes(N-2)}\ket{W}\bra{W}_{N-2}\sigma_x^{\otimes(N-2)}\;.
\end{align}
The subsystem states are thus clearly not unitarily equivalent.~In particular, in agreement with Corollary~\ref{claim:Renyipure}, the $\alpha$-R\'enyi entropy of the $S$ subsystem in the two perspective differs:~in $R_1$ perspective, $\rho_{S|R_2}=\ket{W}\!\bra{W}_{N-2}$ is pure and so $S_\alpha[\rho_{S|R_1}]=0$ for all $\alpha$.~By contrast, in $R_2$-perspective, we have $S_\alpha[\rho_{S|R_2}]=\frac{1}{1-\alpha}\log\left(|a|^{2\alpha}+|b|^{2\alpha}\right)$ and $S_{\rm vN}[\rho_{S|R_2}]=-|a|^2\log|a|^2-|b|^2\log|b|^2$ for the entanglement entropy.

However, condition~\eqref{rhoibarinAU} \emph{is} satisfied for two special cases.~First, for $b=0$,~i.e.~when particle $R_2$ is in the same position as particle $R_1$ ($\ket{\varphi}_2=a\ket0_2$,~$a\in\mathbb C$,~$|a|^2=1$), only the first line of \eqref{UpsiibarWstate} survives and condition \eqref{rhoibarinAU} is met with $Y_2=\mathds 1_2$,~$Z_S=\mathds1^{\otimes(N-2)}$ and so $\rho_{\onebar}\in\mcA_{\onebar}^{\mathds1}$.~Correspondingly, the total state is (up to the exchange of $1$ and $2$ labels) invariant under the QRF transformation (cfr.~\eqref{psijbarWstate} with $b=0$) and $S$ is in a $W$ state in both perspectives (cfr.~\eqref{Wex:Sjpersp} with $b=0$).~Second, for $a=0$,~that is when particle $R_2$ is in opposite but still definite orientation relative to particle $R_1$ ($\ket{\varphi}_2=b\ket1_2$,~$b\in\mathbb C$,~$|b|^2=1$), the total state $\rho_{\onebar}$ belongs to $\mcA_{\onebar}^{X}$ with $X=\mathds 1_2\otimes\sigma_x^{\otimes(N-2)}$ (cfr.~\eqref{UpsiibarWstate} for $a=0$).~The state remains separable under QRF change and we have $\rho_{S|R_2}=\sigma_x^{\otimes(N-2)}\ket{W}\!\bra{W}_{N-2}\sigma_x^{\otimes(N-2)}$ (cfr.~\eqref{Wex:Sjpersp} with $a=0$), so the subsystem states in the two perspectives are unitarily related.~In agreement with Corollary~\ref{claim:Renyipure}, we find that all R\'enyi entropies and the entanglement entropy agree in the two perspectives: indeed, since the subsystem states are pure in both perspectives when $a=0$ or $b=0$ they all vanish.
\end{example}

\enlargethispage{\baselineskip}

\subsubsection{Mixed global states}

Consider the decomposition of a mixed state into a probabilistic mixture of pure states, $\rho_{\ibar}=\sum_{A}p_A\rho_{\ibar}^A$, $(\rho^A_{\ibar})^2=\rho^A_{\ibar}$.~The nonlocality of the QRF transformation $V_{i\to j}$ does not affect the classical probabilities in the ensemble, but it might alter the quantum correlations of the pure states therein, $\rho_{\jbar}=\sum_{A}p_A\hat{V}_{i\to j}(\rho_{\ibar}^A)=:\sum_{A}p_A\rho_{\jbar}^A$.~Whether or not the correlations between the subsystem $S$ and its complement will be the same in the two frame perspectives thus depends on how the correlation structure in the pure states $\rho_{\ibar}^A$ transforms.~The high nonuniqueness of pure state decompositions of mixed states \cite{Hughston93,Nielsen2000} renders this case a challenging one.\footnote{Also, the Schmidt decomposition, underlying the proof of Theorem~\ref{claim:purestatesrhoSZS} in App.~\ref{APP:statesandentropy}, is somewhat less powerful for mixed states; e.g.\ its relationship with entanglement is weaker.}~Nevertheless, the linearity of the partial trace and the QRF transformation imply that Theorem~\ref{claim:purestatesrhoSZS} and Corollary~\ref{claim:Renyipure} have an immediate consequence for mixed states:
\begin{corollary}\label{claim:rhoSZSmixed}
Let $\rho_{\ibar}$ be a generic rank $n$ density operator relative to frame $R_i$ in orientation $g_i\in\mathcal G$, and let $\rho_{\jbar}=\hat{V}_{i\to j}^{g_i,g_j}(\rho_{\ibar})$ be the corresponding state relative to frame $R_j$ in orientation $g_j$.~Then, the reduced $S$-subsystem states $\rho_{S|R_i}=\Tr_j(\rho_{\ibar})$ and $\rho_{S|R_j}=\Tr_i(\rho_{\jbar})$ relative to frames $R_i$ and $R_j$ are unitarily related, that is
\be\label{eq:rhoSZS}
\rho_{S|R_j}=\hat Z_S^{\dagger}(\rho_{S|R_i}) 
\ee
for some $Z_S\in\mathcal U(\mathcal H_S)$, if the total state $\rho_{\ibar}$ can be decomposed into a mixture of $N\geq n$ pure states $\rho_{\ibar}^A$
\be\label{eq:PSdecomp}
\rho_{\ibar}=\sum_{A=1}^N p_A\rho_{\ibar}^A
\ee
with $\sum_A p_A = 1$ and $p_A >0$,~such that $\rho_{\ibar}^A\in\mcA_{\ibar}^{X^{A}}$ with $X^{A}=Y_j^{A}\otimes Z_S$, for some (possibly different) local unitaries $Y_j^{A}\in\mathcal U(\mathcal H_{j})$, i.e.~such that
\be\label{rhoAibarinAU}
\hat\mbfU_{\ibar}^{g_i,g_j}(\rho_{\ibar}^A)=( \hat{Y}_{j}^{A\dagger}\otimes\hat{Z}_S^{\dagger})(\rho_{\ibar}^A).
\ee
In this case, the R\'enyi entropies are the same in both perspectives, $S_\alpha[\rho_{S|R_1}]=S_\alpha[\rho_{S|R_2}]$, $\forall \alpha\in(0,1)\cup(1,\infty)$.
\end{corollary}
This encompasses the situation that the global state $\rho_{\ibar}$ (e.g., the completely mixed state) lies itself in one (or more) of the subalgebras $\mcA_{\ibar}^{X}$.~An example illustrating this and satisfying the conditions \eqref{eq:PSdecomp},~\eqref{rhoAibarinAU} is discussed in Example~\ref{Ex:mixWGHZ} in Appendix~\ref{app:Vex}.

Note that not every pure state decomposition of $\rho_{\ibar}$ will satisfy the condition of the corollary even if one of them does (Example~\ref{ex:AUqubits3} below will demonstrate this point).~It is plausible that the existence of a pure state decomposition of the global state that is a mixture of states from (possibly different) $\mcA_{\ibar}^{X}$ is also a necessary condition for the subsystem states to be unitarily related in the two perspectives.~We have, however, not shown this and leave this question for future work.\footnote{This would require proving existence of a decomposition such that unitary equivalence \eqref{eq:rhoSZS} of the $S$-reduced states implies the (generally mixed) individual $\rho_{S|R_j}^A=\Tr_i(\rho_{\jbar}^A)$ in the subsystem state decomposition to be also unitarily related to $\rho_{S|R_i}^A=\Tr_j(\rho_{\ibar}^A)$ for all $A$.}

\begin{example}\label{ex:AUqubits3}{\bf (3 qubits).}
Consider again the case of three qubits as in Examples~\ref{Ex:AUqubits} and~\ref{Ex:AUqubits2}.~Let the state of $R_2S$ relative to qubit $R_1$  be a maximally mixed state $\rho_{\onebar}=\frac{1}{4}\mathds1_2\otimes\mathds1_S$, so that $\rho_{\onebar}\in\mcA_{\onebar}^{X}$ for any  $g_1,g_2\in\mathcal G$ and any bilocal unitary $X$. The $S$-subsystem state is the same in both  perspectives, namely $\rho_{S|R_i}=\frac{1}{2}\mathds1_S$, $i=1,2$.~Restricting again to $\mathbf{U}_{\onebar}^{g_1=0,g_2=0}=\mathbf{U}_{\onebar}$ (i.e.\ the \textsf{CNOT}), the pure state decomposition of $\rho_{\onebar}$ in the  computational basis
\be\label{maxmixpsdecomp1}
\rho_{\onebar}=\frac{1}{4}\left(\ket{00}\!\bra{00}_{\onebar}+\ket{01}\!\bra{01}_{\onebar}+\ket{10}\!\bra{10}_{\onebar}+\ket{11}\!\bra{11}_{\onebar}\right)\,,
\ee
 does not satisfy the conditions \eqref{eq:PSdecomp},~\eqref{rhoAibarinAU} of Cor.~\ref{claim:rhoSZSmixed} ($\ket{00}\!\bra{00}_{\onebar}, \ket{01}\!\bra{01}_{\onebar}\in\mcA_{\onebar}^{X}$ with $X=\mathds1_2\otimes\mathds1_S$, while $\ket{10}\!\bra{10}_{\onebar}, \ket{11}\!\bra{11}_{\onebar}\in\mcA_{\onebar}^{X}$ with $X=\mathds1_2\otimes\mathds\sigma^x_S$, see Example~\ref{Ex:AUqubits}).~$\rho_{\onebar}$ can however be  rewritten as
\begin{align}
\rho_{\onebar}=\frac{1}{4}\bigl(\ket{\psi_+}\!\bra{\psi_+}_{\onebar}+\ket{\psi_-}\!\bra{\psi_-}_{\onebar}+\ket{\phi_+}\!\bra{\phi_+}_{\onebar}+\ket{\phi_-}\!\bra{\phi_-}_{\onebar}\bigr)\,,\nonumber
\end{align}
Unlike \eqref{maxmixpsdecomp1}, this decomposition  does satisfy the conditions \eqref{eq:PSdecomp},~\eqref{rhoAibarinAU} with all bilocal unitaries $X^A$ equal to the identity (cfr.~Example~\ref{Ex:AUqubits}).
\end{example}

\subsubsection{Subsystem state invariance and translation invariance}\label{Subsec:statesPit}

In Sec.~\ref{Sec:Uinv&break} we saw that $S$-translation invariance is equivalent to exact $\mbfU_{\ibar}$-invariance for local $S$ operators (cfr.~Theorem~\ref{thm:Slocals}), but this is not necessarily true for general operators.~Coming back to the discussion about Theorem~\ref{claim:purestatesrhoSZS}, let us now investigate the situation for states.~Recalling the expression~\eqref{eq:Vitoj} of the QRF transformation, it is easy to see that 
\begin{align}
    \rho_{S|R_j}&=\tr_i(\hat{V}^{g_i,g_j}_{i\to j}(\rho_{\ibar}))\nonumber\\
    &=\sum_{g\in\mathcal G}U_S^g\braket{g_jg^{-1}|\rho_{\ibar}|g^{-1}g_j}_jU_S^{g^{-1}}\nonumber\\
    &=\rho_{S|R_i}\;,\label{rhoSUgSinv}
\end{align}
for $S$-translation invariant states, $[\mathds 1_{j}\otimes U^g_S, \rho_{\ibar}]=0$ for any $g\in\mathcal G$, i.e.\ $\rho_{\ibar}\in\mcA_{\ibar}^{\mbt}$.~Hence, $S$-translation invariance is a sufficient condition for the invariance of the reduced subsystem states under QRF transformations.~Indeed, Lemma~\ref{claim:Uinv&tinvpure} below establishes that $U_S^g$-invariance of pure states implies their $\mbfU_{\ibar}$-invariance up to local unitaries and so links with Theorem~\ref{claim:purestatesrhoSZS}; but the converse does not hold.~We refer to App.~\ref{APP:statesandentropy} for the proof.

\begin{lemma}\label{claim:Uinv&tinvpure}
Let $\rho_{\ibar}$ be any $S$-translation invariant pure state relative to frame $R_i$ in orientation $g_i$, i.e.\ $[\mathds 1_{j}\otimes U^g_S, \rho_{\ibar}]=0$ for any $g\in\mathcal G$ (equivalently, $\rho_{\ibar}\in\mcA_{\ibar}^{\mbt}$).~Then, $\rho_{\ibar}\in\mcA_{\ibar}^{X}$ with bilocal unitary $X=Y_j\otimes\mathds1_S$,
\be\label{YjUgSinv}
Y_j(g_i,g_j)=\sum_{g\in\mathcal G}\overline{q(g)}\ket{g^{-1}g_j}\!\bra{g_ig}_j\,,
\ee
and $q(g)\in\mathbb C$, $|q(g)|^2=1$.~The converse is however not true, i.e.\ $\rho_{\ibar}\in\mathcal{A}_{\ibar}^X$ does not imply $\rho_{\ibar}\in\mcA_{\ibar}^{\mbt}$.
\end{lemma}

Note that, since all elements of $\mcA_{\ibar}^{\mbt}$ are $U_S^g$-invariant, Lemma~\ref{claim:Uinv&tinvpure} also tells us that pure states $\rho_{\ibar}\in\mcA_{\ibar}^{\mbt}$ lie in all subalgebras $\mcA_{\ibar}^{X}$ with  $X=Y_j\otimes U_S^g$ and $Y_j$ as in \eqref{YjUgSinv}, for any $g\in\mathcal G$ too.

$U_S^g$-invariance thus provides us with a sufficient -- but not necessary -- condition for equality of $S$-subsystem states in the two perspectives.~The following example illustrates this observation.
\begin{example}\label{Ex:nontranslinv}
Suppose the total state relative to frame $R_1$ is of the form $\rho_{\ibar}=\ket{g_2}\!\bra{g_2}_2\otimes\rho_S$, for a generic $\rho_S$, so that frame $R_2$ is in a definite orientation relative to $R_1$.~Then, from the second line in \eqref{rhoSUgSinv}, it is clear that the state of $S$ is the same in both frame perspectives even when $[\rho_S,U_S^g]\neq0$ for at least one $g\in\mathcal{G}$.~Indeed, such a state satisfies \eqref{rhoibarinAU} with $Y_2=U_2^{g_1^{-1}g_2}$ and $Z_S=\mathds1_S$ as can be easily checked by using the expression \eqref{TPSschangemap} of $\mbfU_{\onebar}^{g_1,g_2}$.~Correspondingly, the QRF-transformed state relative to $R_2$ reads $\rho_{\twobar}=\ket{g_1}\!\bra{g_1}_1\otimes\rho_S$.~This was the case, for instance, in Example~\ref{Ex:Wstate} for the special cases $a=0$ and $b=0$ where $\rho_S=\ket{W}\!\bra{W}_{N-2}$ was in fact not invariant under conjugation by $U_g^{\otimes(N-2)}$ for any $g$ (cfr.~\eqref{Wqubits} and second line of~\eqref{UpsiibarWstate}).
\end{example}

Conversely, a translation-invariant reduced subsystem state $\rho_{S|R_i}$ does not imply a $S$-translation-invariant global state $\rho_{\ibar}$ from which it results via partial trace.~This is because terms such as $(\ket{g}\!\bra{g'}_j+\ket{g'}\!\bra{g}_j)\otimes\bar{\rho}_S$, with $g\neq g'$ and $[\bar{\rho}_S,U_S^{\bar{g}}]\neq0$ for at least one $\bar{g}\in\mcG$, disappear under the partial trace.~However, owing to the positivity of states, it does imply $S$-translation invariance of $\Pid(\rho_{\ibar})$.
Invoking this observation, we characterise when QRF change induced subsystem state transformations are independent of the global state. 

\begin{theorem}[\bf QRF-invariant subsystem states]\label{thm_globindep}
    Given some subsystem state $\rho_{S|R_i}\in\mcS(\mcH_S)$ in $R_i$-perspective, its counterpart in $R_j$-perspective ($i\neq j$), $\rho_{S|R_j}=\Tr_i\bigl(\hat{V}_{i\to j}^{g_i,g_j}(\rho_{\ibar})\bigr)$, is \emph{independent of the global state} $\rho_{\ibar}$ from which $\rho_{S|R_i}$ originates, if and only if $\rho_{S|R_i}$ is translation-invariant. In this case, we further have exact QRF-invariance, $\rho_{S|R_j}=\rho_{S|R_i}$.
\end{theorem}
The proof is provided in App.~\ref{APP:statesandentropy}. The intuition behind this result is clear: subsystem states that are translation-invariant essentially correspond to relational descriptions of $S$ relative to an $S$-internal QRF (cfr.\ the discussion surrounding Theorem~\ref{lem_TPSineqalg}). These are thus independent of the ``external'' (to $S$) frames $R_1,R_2$ and specifically their quantum states whose information sits inside the global $\rho_{\ibar}$. Conversely, non-translation-invariant subsystem states correspond to descriptions of $S$ that are not fully relational without reference to at least one of $R_1,R_2$  and thereby do depend on their relative state. This result will become useful when characterising QRF-invariant Gibbs states later.

\subsection{Internal $S$ correlations}\label{ssec:intScorr}

The above observations pertain to the QRF relativity or invariance of subsystem states and correlations between $S$ and the ``other frame'' subsystem, but \emph{a priori} not to internal correlations of $S$.~It is, however, straightforward to refine the results such that they encompass also the QRF covariance of internal $S$ correlations.~The local unitaries $Z_S$ entering the above results are local with respect to the natural TPS $\mathbf T_i^{g_i}:\Hphys\to\mathcal H_{j}\otimes\mathcal H_S$ in $R_i$-perspective but are not necessarily local with respect to a refined TPS in which $S$ comes to be a composite system itself carrying an internal TPS $\mathbf T_S:\mathcal H_S\to\bigotimes_\gamma\mathcal H_\gamma$ (cfr.~discussion surrounding Corollary~\ref{cor:refinedTPS} in Sec.~\ref{Sec:QuantRelTPSs}).~It is thus clear that internal correlations of $S$ coincide -- and internal subsystem states of $S$ obtained upon tracing out part of it are unitarily equivalent -- in the two QRF perspectives, whenever $Z_S$ in the above results is a product of fully local unitaries in the internal TPS of $S$ and permutations of internal subsystem factors with equal dimension in $R_i$-perspective, and change otherwise. If $Z_S$ is only local with respect to part of the internal TPS, only some of the internal subsystem states and correlations will be equivalent in the two perspectives.~Similarly, if no subsystem unitary $Z_S$ exists and so $\rho_{\ibar}$ does not reside in any subalgebra $\mcA_{\ibar}^{X}$ (if it is pure) or can be written as a mixture of states that do for $X=Y_j\otimes Z_S$ with possibly distinct $Y_j$, then generically at least part of the internal correlation structure of $S$ will be contingent on the choice of QRF perspective.\footnote{There could still exist some subsystem of $S$ whose correlations with its complement inside $S$ might be equivalent in the two perspectives.}

The following examples will illustrate these observations.

\begin{example}[{\bf Invariant internal $S$ correlations}]\label{Ex:fullW} In Example~\ref{Ex:Wstate} (special cases $a=0$ or $b=0$), as well as Examples~\ref{Ex:GBS}--\ref{Ex:mixWGHZ} discussed in App.~\ref{app:Vex}, $S$ is a multi-particle system and in each case $Z_S$ happens to be a multilocal unitary so that the internal correlations of $S$ are left invariant by the QRF transformation.

An illustration of a global state that resides in a TPS-invariant algebra $\mathcal{A}_{\onebar}^X$ with $X$ involving local permutations within $S$ is further provided in Example~\ref{Ex:permut}.~Furthermore, Example~\ref{Ex:BSinternaltoS} demonstrates the case of a non-TPS-invariant state that changes the entanglement between $S$ and ``the other frame'', but leaves internal entanglement invariant.~Both of these are found in App.~\ref{app:Vex}.
\end{example}

\begin{example}[{\bf Changing internal $S$ correlations}]\label{Ex:totalGHZ}
We shall now provide two examples respectively illustrating a change of classical and quantum internal $S$-correlations.

Consider a discrete translation invariant $\mathbb Z_n$-system of $N>4$ particles.~Let the total $R_2S$ state relative to particle $R_1$ in orientation $g_1=e$ be a $(N-1)$-particle GHZ state
\be
\ket\psi_{\onebar}=\ket{GHZ}_{\onebar}=\frac{1}{\sqrt{|\mathcal G|}}\sum_{g\in\mathcal G}\ket{g}^{\otimes(N-1)}\;.
\ee
Then, setting also $g_2=e$,
\be
\mbfU_{\onebar}\ket\psi_{\onebar}=\Bigl(\frac{1}{\sqrt{|\mathcal G|}}\sum_{g\in\mathcal G}\ket{g}_{2}\Bigr)\otimes\ket{0}^{\otimes(N-2)}
\ee
i.e., $\rho_{\onebar}=\ket\psi\!\bra\psi_{\onebar}$ does not belong to any of the subalgebras $\mcA_{\onebar}^{X}$ and the global state $\ket\psi_{\twobar}=V_{1\to 2}^{e,e}\ket\psi_{\onebar}=\mathcal I_{1\to 2}\mbfU_{\onebar}\ket\psi_{\onebar}$ relative to particle $R_2$ is a product one.~Compatibly with Theorem~\ref{claim:purestatesrhoSZS} and Corollary~\ref{claim:Renyipure}, the $S$-reduced states relative to the two frames are $\rho_{S|R_1}=\frac{1}{|\mathcal G|}\sum_{g\in\mathcal G}\ket{g}\!\bra{g}^{\otimes(N-2)}$ and $\rho_{S|R_2}=\ket{0}\!\bra{0}^{\otimes(N-2)}$, respectively, and their $\alpha$-R\'enyi entropies are different.~Note that in this case not only the correlations between $S$ and the ``other frame'' are modified by the change of frame, but also the \emph{total} correlations within the subsystem $S$ itself.~Specifically, as in a GHZ state the particles in each pair are classically correlated but not entangled \cite{Coffman:1999jd, Dur:2000zz}, the $S$-reduced state is not entangled in both perspectives, but classical correlations among the $N-2$ particles constituting $S$ are present in $R_1$-perspective, not however in $R_2$-perspective.

Consider now the qubits setup of Example~\ref{Ex:Wstate} for $a,b\neq0$ in which case the global state does not belong to any TPS-invariant subalgebra (cfr.~\eqref{UpsiibarWstate}).~The global state $\rho_{\bar 1}$ is a separable pure state and the $S$-subsystem is in a W state \eqref{Wqubits} for which both monogamy and strong monogamy inequalities are saturated (zero residual tangle) \cite{Coffman:1999jd,monogamyWstate1,monogamyWstate2}.~Thus, in $R_1$-perspective, the entanglement between any of the $N-2$ particles in $S$ and the remaining $N-3$ particles as quantified e.g.~by the 1-tangle is completely characterised in terms of bipartite entanglement between that particle and each of the other particles in $S$.~In $R_2$-perspective, for $a,b\neq0$ the global pure state $\rho_{\bar2}$ is entangled and the $S$-subsystem state \eqref{Wex:Sjpersp} is a mixture of a W and a flipped W state.~As discussed in \cite{3tangleWflippedW} for the case of three qubits $S=(ABC)$, even though the residual 3-tangle vanishes for a mixture of W and flipped W states, the 1-tangle between one of the qubits, say $A$, and the other two qubits $(BC)$ is larger than the sum of the squared concurrence of $(AB)$ and $(AC)$.~Thus, in this example not only the external $S$-entanglement is changed by the change of frame but also the quantum correlations internal to $S$.~In $R_2$-perspective, the entanglement between one of the particles in $S$ and the remaining ones is in fact not entirely determined by the bipartite entanglement between that particle and each of the other particles in $S$.
\end{example}

\section{QRF-covariance of quantum dynamics}\label{Sec:Dyn}

How do \emph{thermal} features of a dynamical system depend on the choice of reference frames?~There are well-known examples from spacetime thermodynamics in which thermality is perspectival.~The paradigmatic example are Rindler observers with different accelerations who observe different temperatures.~In these examples, detectors are modelled as internal quantum systems that couple to other degrees of freedom and undergo their own dynamics, not however the different reference frames themselves (though see  \cite{Barbado:2020snx,Foo:2020xqn,Foo:2020dzt}).~Here we will not directly deal with spacetime physics.~But we will ask whether treating reference frames as internal physical systems also leads to a relativity of thermality.~More precisely, how do dynamical and thermal aspects of a physical subsystem of interest $S$ depend on the choice of the QRF subsystems $R_1$ or $R_2$ in our context? We address this question in the rest of this work.

It will be important to distinguish between the notions of closed, isolated or open for a \textit{quantum} system versus the same notions for a \textit{thermodynamic} system, which we recall here for convenience (see for instance \cite{Breuer:2002pc,lebon2008understanding}).~A closed quantum system is one whose dynamics is given by the Liouville-von Neumann equation of motion (e.o.m.) of the form \eqref{vNeom}, as generated by some Hamiltonian that may be time dependent.~An isolated quantum system is one whose Hamiltonian is further time independent.~An open quantum system is one whose e.o.m.~is not given by a Liouville-von Neumann equation of the form \eqref{vNeom}, but has additional non-trivial terms that are typically associated with interactions and dissipation.~A closed thermodynamic system is one which only exchanges energy, but not matter with its surroundings.~An isolated thermodynamic system is one which neither exchanges  matter nor energy with its surroundings.~An open thermodynamic system is one which exchanges both matter and energy with its surroundings.

In this section, we will encounter closed and open quantum systems while studying the total and subsystem dynamics in different QRF perspectives.~We return to the question of energy exchange, i.e.~isolated and closed thermodynamic systems, in the next Sec.~\ref{Sec:TD}.~We will not consider matter exchange.

\subsection{Setup}\label{ssec:setupdyn}

Let us take the \emph{total} physical system at the perspective-neutral level to be an isolated quantum system, with dynamics given by some time-independent Hamiltonian $H_{\rm phys}\in \Aphys$ and states $\rho_{\rm phys}(t) \in \mcS(\Hphys)$ at time $t$.\footnote{In the present work, we consider systems subject to a true Hamiltonian.  However, reparametrisation invariant systems and temporal QRFs can be treated similarly \cite{Hohn:2018toe,PHtrinity,Hoehn:2020epv,CastroRuizQuantumclockstemporallocalisability,Hohn:2018iwn}. We thus expect the QRF dependence of subsystem dynamics and thermality observed below to hold similarly in that case.} 

\begin{figure*}[!t]
\centering
\includegraphics[height=6.0cm,width=0.815\textwidth]{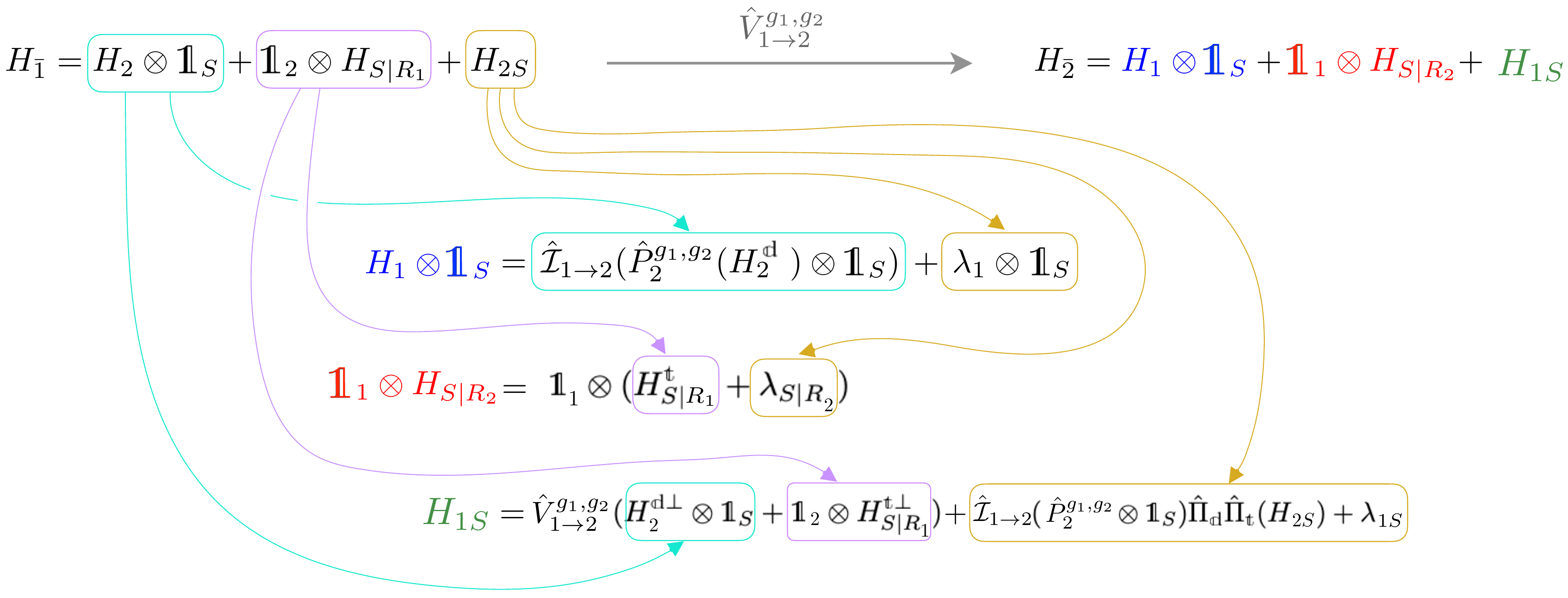}
\caption{Schematic summary of the various contributions \eqref{gla1}--\eqref{intsourcelocal2} to the QRF transformation of the Hamiltonian between the two frame perspectives.}
\label{Fig:QRFHam}
\end{figure*}

The total system as described relative to  internal reference frame $R_i$, say, is then also isolated, since the Schr\"odinger reduction maps \eqref{eq:iReduction} are invertible and time-independent.~The perspectival equation of motion is of standard Liouville-von Neumann form
\be\label{vNeom}
i  \, \dot{\rho}_{\ibar}(t) = [H_{\ibar},\rho_{\ibar}(t)]\,,
\ee
where $\rho_{\ibar}(t) = \hat{\mathcal{R}}_{i}^{g_i} (\rho_{\rm phys}(t))$, $H_{\ibar} = \hat{\mathcal{R}}_{i}^{g_i} (H_{\rm phys})$, and we work in units of $\hbar=1$.~Its solutions $\rho_{\ibar}(t) = e^{-i H_{\ibar}t}\rho_{\ibar}(0) e^{iH_{\ibar}t}$ describe unitary evolution of $\ibar=R_jS$ relative to QRF $R_i$ in orientation $g_i$.~The  Hamiltonian $H_{\ibar}$ can be generically written as (cfr.~\eqref{fdecomp})
\be \label{hibar}
 H_{\ibar} = H_{j} \otimes \mathds1_S + \mathds1_j \otimes H_{S|R_i} + H_{jS}\,,
\ee
where $H_{jS}$ is the interaction between the other frame and the subsystem, and thus is fully non-local across $R_j$ and $S$ degrees of freedom. 

Now in the $R_j$-perspective $(j \neq i)$, the \emph{global} dynamics is still isolated (since the QRF transformation is unitary) and the solutions are similarly given by $\rho_{\jbar}(t) = e^{-i H_{\jbar}t}\rho_{\jbar}(0) e^{iH_{\jbar}t}$ with initial state $\rho_{\jbar}(0)=\hat{V}_{i\to j}^{g_i,g_j}(\rho_{\ibar}(0))$. Writing the Hamiltonian $H_{\jbar} = \hat{V}_{i\to j}^{g_i,g_j} (H_{\ibar})$ similarly in the form \eqref{hibar},
\be\label{eq:jbarHam}
H_{\jbar} =H_i\otimes\mathds1_S+\mathds1_i\otimes H_{S|R_j}+H_{iS}\,,
\ee 
it is interesting to inquire how the local and interaction contributions in $R_j$-perspective are related to the corresponding ones in $R_i$-perspective via QRF transformations.~Owing to the subsystem relativity established in Sec.~\ref{Sec:QuantRelTPSs}, the degree of locality of an operator may change and interaction terms can be relative, i.e.~QRF transformations can map free Hamiltonians into interacting ones and vice versa. 

Invoking what we have learned from the algebraic decompositions in Sec.~\ref{Sec:Uinv&break}, the answer is as follows and schematically illustrated in Fig.~\ref{Fig:QRFHam}:
\begin{align}
H_i\otimes\mathds1_S &= \hat{\mathcal{I}}_{i\to j}  (\hat{P}_{j}^{g_i,g_j}\left(H_{j}^{\mbd}) \otimes \mathds1_S\right) +\lambda_i\otimes\mathds1_S\label{gla1}\\
\mathds1_i\otimes H_{S|R_j}&=\mathds1_i\otimes (H^{\mbt}_{S|R_i}+\lambda_{S|R_j})\label{gla2}\\
H_{iS}&= \hat{V}_{i\to j}^{g_i,g_j}(H_{j}^{\mbd\perp}\otimes \mathds1_S + \mathds1_j\otimes H_{S|R_i}^{\mbt\perp}) \nonumber \\
&\;\;\;\;\;\;+ \hat{\mathcal{I}}_{i\to j}(\hat{P}_j^{g_i,g_j}\otimes\mathds1_S)\Pid \Pit (H_{jS}) + \lambda_{iS}\,.\label{newtotint}
\end{align}
Here we used the following convenient shorthand notation (in $R_i$-perspective)
\begin{align}
    H_{j}^{\mbd} \otimes \mathds1_S&:=\Pid(H_j\otimes\mathds1_S)\label{Hd5}\\
     H_{j}^{\mbd\perp} \otimes \mathds1_S&:=\Pid^\perp(H_j\otimes\mathds1_S)\label{Hdperp}\\
   \mathds1_j\otimes H^{\mbt}_{S|R_i} &:=\Pit(\mathds1_j\otimes H_{S|R_i})\label{Ht5}\\
   \mathds1_j\otimes H^{\mbt\perp}_{S|R_i} &:=\Pit^\perp(\mathds1_j\otimes H_{S|R_i})\,,\label{Htperp}
\end{align}
recalling from \eqref{piddef} and \eqref{pitdef} that $\Pid$ and $\Pit$ only act non-trivially on the $R_j$- and $S$-tensor factor, respectively. Moreover, the $\lambda$-terms come from decomposing the following
\begin{align} \label{intsourcelocal2}
\hat{V}_{i \to j}^{g_i,g_j}&(\Pid \Pit^{\perp} + \Pid^{\perp})(H_{jS}) =: 
\lambda_{i}\otimes \mathds1_S + \mathds1_i \otimes \lambda_{S|R_j} + \lambda_{iS}
\end{align}
into $R_i$-local, $S$-local and $R_iS$-interactions terms.

A few explanations are in place: the first term on the RHS of \eqref{gla1} is implied by Theorem~\ref{thm:jlocals} and Lemma~\ref{Lemma:tpseq} and is also diagonal in the frame orientation basis.~The first term on the RHS of \eqref{gla2} follows from Theorem~\ref{thm:Slocals} and Lemma~\ref{Lemma:tpseq} and is also translation-invariant.~Similarly, in \eqref{newtotint}, Theorems~\ref{thm:Slocals} and~\ref{thm:jlocals} entail that the term ${\hat{V}_{i\to j}^{g_i,g_j}(H_{j}^{\mbd\perp}\otimes \mathds1_S + \mathds1_j\otimes H_{S|R_i}^{\mbt\perp})}$ in the first line (if non-vanishing) necessarily contributes to the interactions between $R_i$ and $S$ in $R_j$-perspective.~Furthermore, that ${\hat{\mathcal{I}}_{i\to j}(\hat{P}^{g_i,g_j}_j\otimes\mathds1_S)\Pid \Pit (H_{jS})}$ likewise is necessarily an $R_iS$-interaction in $R_j$-perspective (if non-vanishing) is a consequence of Lemmas~\ref{lemma:pinpd-uin} and~\ref{Lemma:tpseq}.~Lastly, the last remaining piece in \eqref{intsourcelocal2}, depending on the precise form of $H_{jS}$, may or may not contribute free and/or interacting pieces.~This depends on whether $(\Pid \Pit^{\perp} + \Pid^{\perp})(H_{jS})$ is part of a TPS-invariant subalgebra, which by \eqref{ugly} is possible.~For example, the two-qubit interaction ${H_{2S}=\sigma_2^z\otimes\sigma_S^z}$ turns for $g_1=g_2=0$ into a free term for $S$, ${\hat{V}_{1 \to 2}^{0,0}(\Pid \Pit^{\perp} + \Pid^{\perp})(H_{2S}) =  \mathds1_1 \otimes \sigma_S^z}$ because $(\Pid \Pit^{\perp} + \Pid^{\perp})(H_{2S})=H_{2S}$ does not reside in a TPS-invariant subalgebra (cfr.\ Example~\ref{Ex:AUqubits2}).

As we shall demonstrate in the following, the QRF-dependence of the Hamiltonian and the interactions it contains have immediate consequences for the type of subsystem dynamics (closed vs.\ open), as well as for the type of thermodynamic system (isolated vs.\ closed) in the different frame perspectives.~In particular, we notice that the QRF-transformation can turn a non-interacting or weakly interacting Hamiltonian, say $H_{\ibar} = H_{j} \otimes \mathds1_S + \mathds1_j \otimes H_{S|R_i} + \epsilon H_{jS}$ with $\epsilon\ll1$, into a strongly interacting one, owing to the transformation properties \eqref{gla1}--\eqref{intsourcelocal2} of the individual local and interaction terms.

\subsection{Global dynamics} \label{subsec:globaldyn}

\begin{figure*}[!t]
\centering
\includegraphics[height=6.0cm,width=0.815\textwidth]{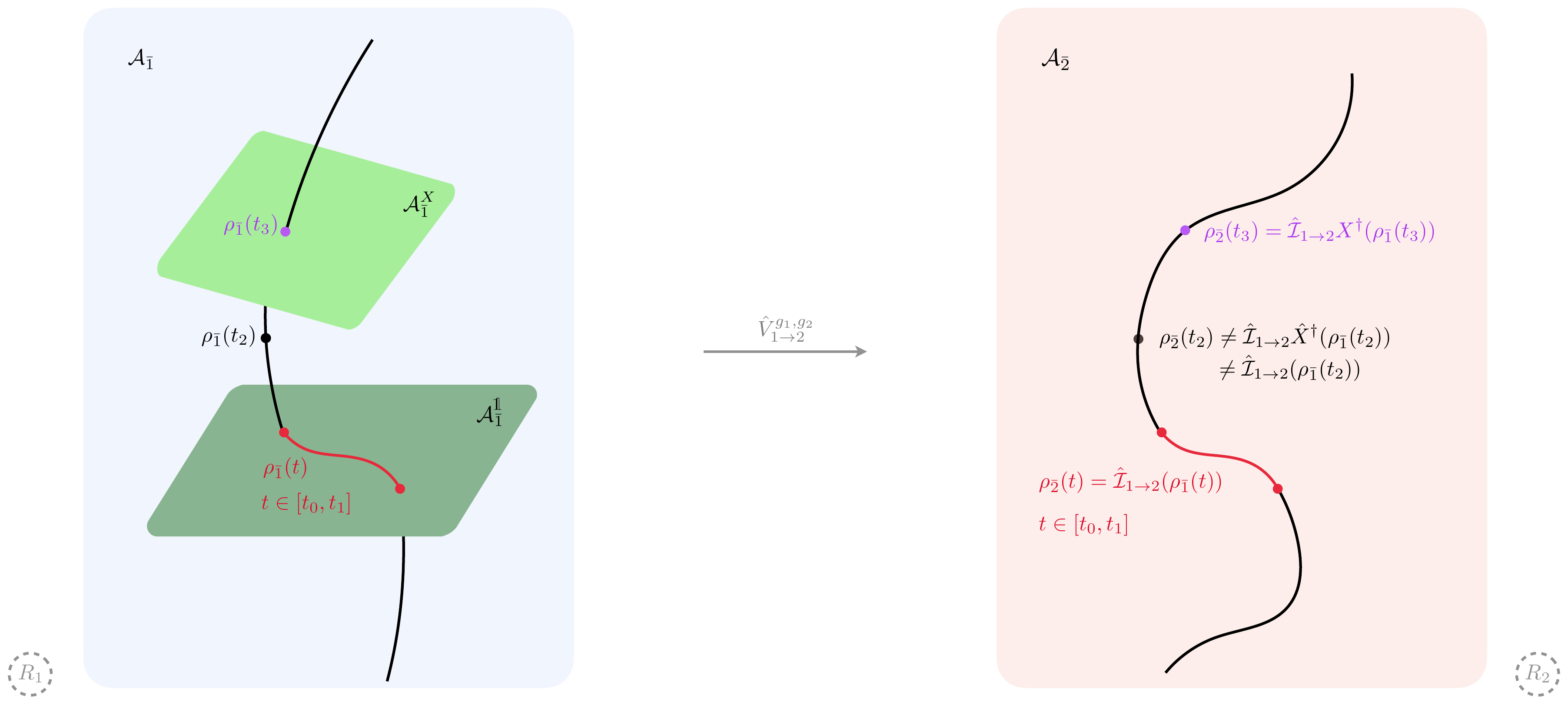}
\caption{Pictorial representation of the dynamical evolution of the global system relative to the frames $R_1$ (left) and $R_2$ (right).~At different times along the dynamical trajectory, the state $\rho_{\bar1}(t)$ relative to $R_1$ may enter different TPS-invariant subalgebras.~If at some instant $t_3$, $\rho_{\bar1}(t_3)\in\mcA_{\bar1}^{X}$ for some bilocal unitary $X$, then as per Lemma \ref{Lemma:tpseq} $\rho_{\bar2}(t_3)$ differs from $\rho_{\bar1}(t_3)$ only by local unitary conjugation and frame swap (shown in purple).~If $\rho_{\bar1}(t)$ lies in the subalgebra $\mcA_{\bar1}^{\mathds1}$ of operators commuting with $\mbfU_{\bar1}^{g_1,g_2}$ during some interval $[t_0,t_1]$, the states in the two perspectives appear identical up to frame swap (shown in red).~When the global state is outside of any of the TPS-invariant subalgebras, say for $t_2 \in (t_1,t_3)$, the QRF transformation acts non-locally and the trajectory appears qualitatively differently in the two perspectives (e.g.\ interacting vs.\ free).}
\label{Fig:AUXdyn}
\end{figure*}

Before exploring how the subsystem dynamics behaves under QRF changes, let us briefly investigate the same for the global dynamics.~From the preceding sections, we can already anticipate that, in order for the two QRF perspectives to agree (possibly up to local unitaries) or disagree on the description of the dynamics, it matters how the global Hamiltonian $H_{\ibar}$ and/or initial state $\rho_{\ibar}(0)$ relate to the various TPS-invariant subalgebras $\mcA_{\ibar}^X$.~Clearly, a variety of cases can occur: either of them could
\begin{itemize}
    \item[(1)] belong to a single (possibly distinct) TPS-invariant subalgebra, or
    \item[(2)] be decomposed into elements of multiple such subalgebras, or 
    \item[(3)] satisfy neither of the two.
\end{itemize}
Let us examine the consequences of such cases for the relation of the dynamics in the two perspectives.

We begin by inquiring when the state trajectory $\rho_{\ibar}(t)$ resides in a TPS-invariant subalgebra $\mcA_{\ibar}^X$ for some time interval $[t_0,t_1]$, so that, up to the bilocal unitary, both QRF perspectives on the dynamics agree (see Fig.~\ref{Fig:AUXdyn}).~To this end, it will turn out useful to consider a second Hamiltonian in $R_i$-perspective, namely 
\be\label{defHimport}
H_{j \to i}:=\hat{X} \hat{\mathcal{I}}_{j \to i}(H_{\jbar})=\hat X\hat U_{\ibar}^{g_i,g_j}(H_{\ibar})\,,
\ee
which, up to frame swap and $X$-conjugation, is identical to $H_{\jbar}$ and we accordingly call it the imported Hamiltonian from $R_j$-perspective.~Lemma~\ref{lem:pix} implies that their components inside $\mcA_{\ibar}^X$ agree,
\be\label{eq:hagree}
\PiUX(H_{\ibar})=\PiUX(H_{j\to i})\,,
\ee 
but the previous discussion in this section already tells us that they will generically differ outside of it $\hat{\Pi}_{\ibar}^{X\perp}(H_{\ibar})\neq\hat{\Pi}_{\ibar}^{X\perp}(H_{j\to i})$. For example, suppose $H_{\ibar}$ is free and of the form $H_{\ibar}=H_j^{\mbd\perp}\otimes\mathds1_S+\mathds1_j\otimes H_S^{\mbt\perp}$, where we used \eqref{Hdperp} and~\eqref{Htperp}, then $H_{j\to i}=\hat{X}\hat{\mbfU}_{\ibar}^{g_i,g_j}(H_j^{\mbd\perp}\otimes\mathds1_S+\mathds1_j\otimes H_{S|R_i}^{\mbt\perp})$ is a pure interaction term. We then have the following simple result, proved in Appendix~\ref{app:dyn}.
\begin{lemma} \label{lemma:finitetimeinvsoln}
Consider a Hamiltonian $H_{\ibar} \in \mcA_{\ibar}$ and an initial state $\rho_{\ibar}(t_0) \in \mathcal{S}(\mcH_{\ibar})$.~Then, $\rho_{\ibar}(t) \in \mcA_{\ibar}^{X}, \;\; \forall t \in[t_0,t_1]
$, for some bilocal unitary $X$ (and frame orientations $g_i,g_j \in \mcG$), if and only if the initial state resides in this subalgebra, $\rho_{\ibar}(t_0)\in\mcA_{\ibar}^X$, and
\be \label{commAUXt2}
[H_{\ibar}-H_{j \to i},\rho_{\ibar}(t)]=0 , \;\; \forall t\in [t_0,t_1]\,,
\ee
which is equivalent to $[\hat{\Pi}_{\ibar}^{X\perp}(H_{\ibar}-H_{j\to i}),\rho_{\ibar}(t)]=0$ for $t\in [t_0,t_1]$.
\end{lemma}

In other words, in order for the dynamical trajectory to lie in some TPS-invariant subalgebra $\mcA_{\ibar}^{X}$ for some interval $[t_0,t_1]$, the Hamiltonian and this state must be such that,  over that interval, its evolution cannot distinguish between the true Hamiltonian $H_{\ibar}$ relative to frame $R_i$ and the imported one $H_{j \to i}$.~Even though $H_{\ibar}$ is \emph{not} necessarily $\mbfUX$-invariant, it behaves as though it were as far as the perspectival description of the global dynamics is concerned.~In such cases, the correlation structure between the subsystem and its complement is preserved under the QRF transformation during the interval $[t_0,t_1]$ (cfr.~Lemma \ref{Lemma:tpseq}).~Clearly, in the special case that also $H_{\ibar}\in\mcA_{\ibar}^X$, condition~\eqref{commAUXt2} is satisfied trivially and time evolution is just an automorphism of this subalgebra. 

Let us now illustrate Lemma~\ref{lemma:finitetimeinvsoln} with an example that yields a $\mbfUX$-invariant trajectory despite the Hamiltonian, depending on a choice of parameters, either not belonging to any TPS-invariant subalgebra or being $\mbfU_{\!X'}$-invariant with $X\neq X'$.

\begin{example}[\bf TPS-invariant trajectory]\label{ex:hamnon}
Consider a Hamiltonian $H_{\onebar} = \sum_{a}E_a \ket{E_a}\bra{E_a}_{\onebar}$ such that $\ket{E_b}\bra{E_b}_{\onebar} \in \mcA_{\onebar}^{X}$ for some $b$ and $X$.\footnote{We recall that the eigenstates of the Hamiltonian cannot all belong to the same subalgebra $\mcA_{\onebar}^{X}$ as otherwise any state in $\mathcal H_{\onebar}$ would be $\mbfUX$-invariant.~This cannot be the case due to the inequivalence of the TPSs relative to the two frames (cfr.~Lemma~\ref{lem_TPSineq}).}~The full $H_{\onebar}$ may or may not belong to this or any other invariant subalgebra.~Now consider an initial state $\rho_{\onebar}(0)= \ket{E_b}\bra{E_b}_{\onebar}\in\mcA_{\onebar}^X$ (or some $\ket{\psi}\bra{\psi}_{\onebar}$, with $\psi$ a superposition of eigenstates belonging to the same $\mcA_{\onebar}^{X}$).~Eq.~\eqref{commAUXt2} is satisfied since $\ket{E_b}_{\onebar}$ is also an eigenstate of $H_{2\to 1}=\hat{X}\hat{\mbfU}_{\onebar}^{g_1,g_2}(H_{\onebar})$, i.e.~$H_{2\to 1}\ket{E_b}\bra{E_b}_{\onebar} = E_b \ket{E_b}\bra{E_b}_{\onebar}$.~In this case, the state is stationary and thus obviously remains in $\mcA_{\onebar}^X$ at all times.

As a concrete example, consider a system of $N=3$ qubits $R_1R_2S$ and frame orientations $g_1=g_2=0$ before and after QRF change, respectively.~Recall that in this case the QRF transformation becomes a \textsf{CNOT}, $\mbfU_{\onebar}^{0,0}=\ket{0}\!\bra0_{2}\otimes\mathds1_S+\ket{1}\!\bra1_{2}\otimes\sigma_S^x \equiv \mbfU_{\onebar}$.~Let the Hamiltonian of $R_2S$ relative to $R_1$ be the following $ZZ$-interacting Hamiltonian
\be\label{HamBJZZ}
H_{\bar 1}=B(\sigma_2^z\otimes\mathds1_S+\mathds1_2\otimes\sigma_S^z)+2J\sigma_2^z\otimes\sigma_S^z\,.
\ee
For $B=2J$, the Hamiltonian \eqref{HamBJZZ} belongs to $\mcA_{\bar 1}^{\mathds1}$, otherwise it does not belong to any TPS-invariant subalgebra for any $X$ (cfr.~\eqref{3qham} in Example~\ref{Ex:AUqubits}).~Consider the initial state
\be\label{ex:2qex2}
\ket{\psi(0)}_{\bar{1}} = \ket{1}_{2}\otimes(a\ket{0}_{S}+b\ket{1}_{S})
\ee
for some $a,b\in\mathbb C$, which satisfies $\rho_{\bar 1}(0)=\ket{\psi(0)}\bra{\psi(0)}_{\bar{1}}\in\mcA_{\bar 1}^{X}$ with $X=\mathds 1_2\otimes\sigma_S^x$ (cfr.~Example~\ref{Ex:AUqubits}).~Hence, even for $B=2J$, the Hamiltonian and initial state do not share the same subalgebra.~However, \eqref{ex:2qex2} is a linear combination of the eigenstates $\ket{10}_{\bar 1}$ and $\ket{11}_{\bar 1}$ of the Hamiltonian \eqref{HamBJZZ}, with eigenvalues $-2J$ and $2(J-B)$, and we have that $X\mbfU_{\onebar}\ket{10}_{\bar 1}=\ket{10}_{\bar 1}$ and $X\mbfU_{\onebar}\ket{11}_{\bar 1}=\ket{11}_{\bar 1}$.~The state $\ket{\psi(t)}_{\bar{1}}$ at any later time is thus given by
\be\label{ex:2qex2globev}
\ket{\psi(t)}_{\bar{1}}=\ket{1}_{2}\otimes(e^{2iJt}a\ket{0}_S+e^{-2i(J-B)t}b\ket{1}_S)
\ee
and so $\rho_{\bar 1}(t)\in\mcA_{\bar 1}^{X}$ with $X=\mathds 1_2\otimes\sigma_S^x$ also at all later times.~Indeed, condition~\eqref{commAUXt2} of Lemma~\ref{lemma:finitetimeinvsoln} is satisfied for all later times, regardless of the parameter values $B$ and $J$, with imported Hamiltonian given by
\be\label{XUHamBJZZ}
H_{2\to1}=B\sigma_2^z\otimes\mathds1_S-2J\mathds1_2\otimes\sigma_S^z-B\sigma_2^z\otimes\sigma_S^z.
\ee
\end{example}

In App.~\ref{app:dyn}, we describe two further illustrations of Lemma~\ref{lemma:finitetimeinvsoln}. Example~\ref{Ex:HrhoAUX} invokes a QRF-invariant Hamiltonian to generate trajectories that are  automorphisms of a TPS-invariant subalgebra or are not TPS-invariant. In Example \ref{ex:inoutAUX}, dynamical trajectories oscillate between different TPS-invariant subalgebras, so that the two QRF perspectives agree only periodically on the correlations between the subsystem $S$ and its complement. 

\subsection{Subsystem dynamics}\label{Subsec:subsysdyn}

Let us now investigate to which extent the dynamical description of the subsystem of interest will differ (or not) in the two frame perspectives.~For instance, we can ask under which conditions a subsystem that appears closed in $R_1$-perspective can appear open in $R_2$-perspective as a direct consequence of the change of reference frame.~Specifically, a subsystem that is uncorrelated with its complement under dynamical evolution in one perspective can exhibit non-trivial dynamical correlations in another.~This is a direct physical consequence of the subsystem relativity established in Sec.~\ref{Sec:QuantRelTPSs} and the ensuing frame dependent degrees of locality (thus, the interaction structure) of the Hamiltonian and of correlations in the total system as discussed in Sections \ref{Sec:Uinv&break}, \ref{Sec:states&entanglement} and \ref{subsec:globaldyn}. 

Clearly, as we will now see in detail, the subsystem equations of motion and solutions, will thereby generically also differ in different perspectives.~As the subsystem equations of motion will generically depend on the global state, we can inquire about the \emph{type} of subsystem dynamics (closed or open) and its frame dependence at two levels:
\begin{itemize}
    \item[(A)] the full space of solutions, or
    \item[(B)] \emph{specific} solutions.
\end{itemize}
Indeed, in a space of solutions of an open subsystem, there will typically exist solutions for which the subsystem is effectively closed.~As we will see in Sec.~\ref{Sec:qrfdepopenclosed} below, whether a subsystem is closed or open at either level generically depends on the frame choice.~We will use this in Sec.~\ref{Sec:dynrhoSandentropies} to further study the frame dependence of correlations between $S$ and the ``other frame'' over time, extending the instantaneous results of Sec.~\ref{Sec:states&entanglement} to the dynamical setting.

\subsubsection{Subsystem equations of motion}\label{sssec:subeom}

To write down the general form of the subsystem equations of motion, let us write any global state $\rho_{\ibar}(t)$ at any time $t$ in the following useful way (see e.g.~\cite{Rivas_2012,rezakhani} and references therein)
\be \label{rhocorr}
 \rho_{\ibar}(t) = \rho_{j}(t) \otimes \rho_{S|R_i}(t) + \omega_{jS}(t)\,, 
\ee
where $\rho_j(t) := \Tr_S \rho_{\ibar}(t)$ and $\rho_{S|R_i}(t) := \Tr_j \rho_{\ibar}(t)$ are the reduced states relative to $R_i$, and $\omega_{jS}(t) :=  \rho_{\ibar}(t) - \rho_j(t) \otimes \rho_{S|R_i}(t)$ is the part encoding correlations (both classical and quantum).~Notice that by construction, we have
\be\label{trcorr}
\Tr_j \omega_{jS}(t) = 0 = \Tr_S \omega_{jS}(t)\,. 
\ee
The QRF transformed state can be written similarly as
\be 
 \rho_{\jbar}(t) = \hat{V}_{i\to j}^{g_i,g_j}(\rho_{\ibar}(t)) =: \rho_{i}(t) \otimes \rho_{S|R_j}(t) + \omega_{iS}(t) \,,\label{rhocorrj}
\ee
where $\rho_i(t) := \Tr_S \rho_{\jbar}(t)$ and $\rho_{S|R_j}(t) := \Tr_i \rho_{\jbar}(t)$ are the respective reduced states relative to $R_j$, and $\omega_{iS}(t):=\rho_{\jbar}(t) - \rho_i(t) \otimes \rho_{S|R_j}(t)$ is the traceless part (satisfying equations analogous to \eqref{trcorr}) encoding correlations.~We should stress that $\rho_i(t) \otimes \rho_{S|R_j}(t)$ and $\omega_{iS}(t)$ in the $R_j$-perspective are generically not the QRF transform respectively of $\rho_j(t) \otimes \rho_{S|R_i}(t)$ and $\omega_{jS}(t)$.

Let us consider a possibly interacting, time-independent Hamiltonian $H_{\ibar}$ written in the form \eqref{hibar}.~The subsystem equation of motion in the $R_i$-perspective ($i=1,2$) is
\begin{align}
i  \, \dot{\rho}_{S|R_i}(t) &= \Tr_j [H_{\ibar},\rho_{\ibar}(t)] \label{sieom2} \\
&= \big[H_{S|R_i} + \tilde{H}_{S|R_i}(t), \rho_{S|R_i}(t) \big] \nonumber \\
& \;\;\; + \Tr_j \!\big[H_{jS},\omega_{jS}(t) \big] \,,   \label{sieom}
\end{align}
where we have denoted,
\be\label{eq:HStilde}
    \tilde{H}_{S|R_i}(t) = \Tr_j\!\big(H_{jS} (\rho_{j}(t) \otimes \mathds1_S) \big).
\ee
Eq.~\eqref{sieom} is the exact, general form of the subsystem equation of motion relative to frame $R_i$ ($i=1,2$), encompassing both closed and open dynamics, respectively, when interactions are zero or not in the total Hamiltonian.~Starting from \eqref{sieom}, we can thus inquire about the frame dependence of the \emph{type} of dynamics that the subsystem undergoes.~To this aim, let us look more carefully at the structure of the various terms on the RHS of \eqref{sieom}. 

The ``other frame''-local Hamiltonian $H_{j} \otimes \mathds1_S$ of course does not play any role for subsystem $S$ dynamics in the old $R_i$-perspective.~However, upon QRF transforming to the new $R_j$-perspective $(j=1,2, j\neq i)$, its $\Pid^{\perp}(H_{j} \otimes \mathds1_S)$ component \emph{does} play a non-trivial role for the $S|R_j$ dynamics,
since it contributes to the interactions $H_{iS}$ as shown in \eqref{newtotint}.~Together, this also means that $\Pid(H_{j} \otimes \mathds1_S)$ is the only component of $H_{\bar i}$ that does not have any role to play for the dynamics of the subsystem $S$ in any of the two perspectives.

The first contribution to the commutator on the RHS of \eqref{sieom} represents unitary subsystem dynamics.~It depends only on $H_{S|R_i}$, hence on interactions internal to the subsystem and not on those between $S$ and its complement.~The second contribution involves $\tilde{H}_{S|R_i}(t)$, which can be understood as a contribution to the subsystem Hamiltonian \emph{induced} by the dynamical coupling of $S$ with the other frame $R_j$ degrees of freedom in the $R_i$-perspective.~Its expression \eqref{eq:HStilde} is clearly state-dependent; only when the total Hamiltonian is non-interacting in a given perspective (hence also $\tilde{H}_{S|R_i}(t)=0$), does the \emph{full} subsystem solution space undergo unitary (hence closed) dynamics in that perspective.~On the contrary, in the presence of interactions, as will be illustrated shortly below, an open subsystem may admit specific solutions that undergo an \emph{effective} unitary (hence closed) evolution (see for instance Example \ref{Ex:closedHint} below).

The last term in \eqref{sieom} is responsible for a genuine non-unitary open dynamics, as expected from the theory of open quantum dynamics and quantum thermodynamics \cite{Rivas_2012,Breuer:2002pc,e15062100,binder2019thermodynamics,Goold_2016,vinanders}.~In other words, this term cannot be brought to a closed commutator form like that of the first term in \eqref{sieom}, nor does it identically compute to zero in general for an interacting global dynamics.~Under a Markovian approximation for example, this term would represent the Lindblad dissipator. 

\subsubsection{QRF-relative closed and open subsystems}\label{Sec:qrfdepopenclosed}
 
We can now proceed to utilise the algebraic results of Sec.~\ref{Sec:Uinv&break} to characterise the conditions under which the $S$-subsystem dynamics will (or will not) be of the same type in the two perspectives.~As we have done previously, our strategy will be to first characterise when the subsystem dynamics appears the same in the two perspectives, which will then guide us as to when it appears differently in these perspectives.~Specifically, we inquire when a closed subsystem in one perspective will remain closed under a change of reference frame, i.e.~when no interactions are present between $S$ and the other frame in \emph{both} perspectives (so that it is closed for all solutions).~This is characterised by the following simple result, for which we recall definitions~\eqref{Hd5}--\eqref{Htperp} (see App.~\ref{app:dyn} for the proof).
\begin{lemma}[\bf Closed-to-closed \& closed-to-open]\label{lemma:closed2closed}
Let the subsystem $S$ be \emph{dynamically closed} relative to frame $R_i$ (in some orientation $g_i$), i.e.~the total Hamiltonian $H_{\ibar}=H_{j}\otimes\mathds1_S+\mathds1_j\otimes H_{S}$ relative to $R_i$ is non-interacting.~Then, the subsystem $S$ is \emph{dynamically closed} also relative to frame $R_j$ (in any orientation $g_j$) if and only if $H_{\ibar}$ is of the form
\be\label{eq:invHclosed2}
H_{\ibar}=H_j^{\mbd}\otimes\mathds1_S+\mathds1_j\otimes H_S^{\mbt}\,.
\ee
In other words, $S$ is \emph{dynamically open} relative to frame $R_j$ if and only if at least one of the following holds
\be\label{cl2opconditions}
H_{j}^{\mbd\perp}\otimes\mathds1_S\neq0\,,\quad\mathds1_j\otimes H_{S}^{\mbt\perp}\neq0\,.
\ee
\end{lemma}
Since we are considering only time-independent total Hamiltonians, Lemma~\ref{lemma:closed2closed} characterises the physical scenarios wherein the subsystem is dynamically isolated in both perspectives and thereby also when it is isolated in one, but open in the other.~Due to Theorems~\ref{thm:Slocals}~and~\ref{thm:jlocals}, condition \eqref{eq:invHclosed2} is equivalent to having a non-interacting $H_{\ibar}$ (i.e.~$H_{jS}=0$) belonging to the subalgebra $\mcA_{\ibar}^{X}$ with $X=P_j^{g_i,g_j}\otimes \mathds1_S$.~Free Hamiltonians belonging to this subalgebra thus fully characterise the physical scenario of having a ``closed-to-closed quantum subsystem'' under the QRF transformation.~This is clearly a strong constraint and for generic free Hamiltonians in $R_i$-perspective at least one of~\eqref{cl2opconditions} will hold, such that the closed subsystem will appear open in $R_j$-perspective ($i\neq j$).

We emphasise that this QRF relativity of dynamical openness or closure of $S$ is not a gauge artefact, which one might at first conclude upon recalling that ``jumping into the perspective'' of QRF $R_i$ is associated with a reduction map $\mathcal{R}_{i}^{g_i}$, given in~\eqref{eq:iReduction}, that really is a gauge-fixing of $R_i$ into orientation $g_i$.~The QRF transformation $V_{i\to j}^{g_i,g_j}$ in~\eqref{eq:Vitoj} is thus associated with a change of gauge.~However, we are gauge fixing manifestly gauge-invariant states in $\Hphys$ and observables in $\Aphys$ and, using the invertibility of the reduction maps, all these statements could be equivalently written at the manifestly gauge-invariant level without any gauge-fixing at all.~In particular, recall from Sec.~\ref{Sec:QuantRelTPSs} (and especially Theorem~\ref{lem_TPSineqalg}) that the relational observables relative to frames $R_1$ and $R_2$ induce two \emph{inequivalent} factorisations of $\Hphys$ and $\Aphys$, thereby leading to two inequivalent \emph{gauge-invariant} relational notions of subsystem $S$ and ``the other frame''.~It is these inequivalent gauge-invariant definitions of $S$ that need not agree on whether $S$ is open or closed with respect to the perspective-neutral Hamiltonian $H_{\rm phys}=(\hat{\mathcal{R}}_i^{g_1})^\dag(H_{\ibar})$.~Recall also from Sec.~\ref{Sec:QuantRelTPSs} that the reduction maps $\mathcal{R}_i^{g_i}$ are nothing but a TPS on the gauge-invariant $\Hphys$ and $\Aphys$ that happen to be ``aligned'' with a frame choice (which by itself is not a choice of gauge), in the sense that the relational observables describing $R_j$ and $S$ relative to $R_i$ each appear local in that TPS, and vice versa.

In the present context, condition~\eqref{eq:invHclosed2} is equivalent to the two contributions to $H_{\rm phys}$, namely ${(\hat{\mathcal{R}}_i^{g_i})^\dag(H_j^\mbd\otimes\mathds1_S)}$ and ${(\hat{\mathcal{R}}_i^{g_i})^\dag(\mathds1_j\otimes H^\mbt_S)}$, lying in the overlaps ${\mathcal{A}^{\rm phys}_{R_2|R_1}\cap\mathcal{A}^{\rm phys}_{R_1|R_2}}$ in Fig.~\ref{Fig:locframeDO} and ${\mathcal{A}^{\rm phys}_{S|R_1}\cap\mathcal{A}^{\rm phys}_{S|R_2}}$ in Fig.~\ref{Fig:SDOalgebras}, respectively.~This follows from Theorems~\ref{lem_TPSineqalg}--\ref{thm:jlocals} and provides a gauge-invariant characterisation of a Hamiltonian that is free with respect to both QRFs.

Let us illustrate the above discussions with the following examples.

\begin{example}\label{ex:closedvsopen}
Consider again the three-qubit total system $R_1R_2S$ and the following two example Hamiltonians
\begin{align}
    &(1)\, H_{\bar 1}=\sigma_2^z\otimes\mathds1_S+\mathds1_2\otimes\sigma_S^x,\label{invHfree}\\
    &(2)\, H_{\bar 1}=\sigma_2^z\otimes\mathds1_S+\mathds1_2\otimes\sigma_S^z.\label{noninvHfree}
\end{align}
Clearly, both Hamiltonians lead to unitary dynamics for $S$ relative to $R_1$.~Using the defining equations \eqref{pitdef} and \eqref{piddef} for the projectors, it is straightforward to check that the Hamiltonian \eqref{invHfree} satisfies \eqref{eq:invHclosed2}, while \eqref{noninvHfree} satisfies \eqref{cl2opconditions}.~The corresponding Hamiltonians $H_{\bar 2}=\hat{V}^{0,0}_{1\to2}(H_{\bar 1})$ (with $g_1=g_2=0$) relative to $R_2$ are
\begin{align}
    &(1)\, H_{\bar 2}=\sigma_1^z\otimes\mathds1_S+\mathds1_1\otimes\sigma_S^x,\nonumber\\
    &(2) \, H_{\bar 2}=\sigma_1^z\otimes\mathds1_S+\sigma_1^z\otimes\sigma_S^z,\nonumber
\end{align}
using $V_{1\to2}^{0,0}=\ket0_1\otimes\bra0_2\otimes\mathds1_S+\ket1_1\otimes\bra1_2\otimes\sigma_S^x$.~Thus, we see that no $R_1S$ interactions are present in $R_2$-perspective for $(1)$, while they are for $(2)$.~That is, $S$ in the new $R_2$-perspective is closed and open in cases $(1)$ and $(2)$, respectively.
\end{example}

In the presence of interactions and dynamical correlations in the total system ($H_{jS}\neq0, H_{iS} \neq 0$),~\eqref{sieom} shows that the subsystem $S$ is typically open relative to both $R_i$ and $R_j$ perspectives.~However, some specific subsystem states can certainly evolve unitarily.~Depending on the interplay between the Hamiltonian and the state preparation, the RHS of~\eqref{sieom} can reduce to a commutator, so that the evolution becomes effectively closed.~From~\cite{Rodriguez_Rosario_2011} it follows that the necessary and sufficient condition for effectively closed dynamics in $R_i$-perspective is 
\be
[\mathds1_j\otimes\rho_{S|R_i},\rho_{\ibar}]=0\,,\nonumber
\ee
which is equivalent to the subsystem state commuting with the correlation part of the global state
\be\label{necsuffeffclosed}
[\mathds1_j\otimes\rho_{S|R_i},\omega_{jS}]=0\,.
\ee
We leave open the question of characterising the precise conditions under which effectively closed dynamics remains invariant under QRF changes, i.e.\ under which conditions \eqref{necsuffeffclosed} is obeyed in \emph{both} perspectives and so holds also with $i$ and $j$ labels interchanged.~Coming back to~Lemma~\ref{lemma:finitetimeinvsoln}, we can, however, provide a sufficient condition for this to occur, namely that~\eqref{commAUXt2} holds.~This is the necessary and sufficient condition for a state to reside in a \emph{fixed} TPS-invariant algebra during an interval, in which then $\rho_{S|R_j}=\hat{Z}_S^\dag(\rho_{S|R_i})$ with time independent unitary $Z_S$.~Hence, ${\dot{\rho}_{S|R_2}=\hat{Z}_S^\dag\left(\dot{\rho}_{S|R_1}\right)}$ and $S$ undergoes effectively open or closed dynamics in $R_2$-perspective, whenever it respectively does so also in $R_1$-perspective.~Let us illustrate this with a simple example.
\begin{example}{\bf (Effectively isolated in both perspectives)} \label{Ex:closedHint}
We continue with Example \ref{ex:hamnon}.~Recall that the interacting Hamiltonian relative to $R_1$ is ${H_{\bar 1}=B(\sigma_2^z\otimes\mathds1_S+\mathds1_2\otimes\sigma_S^z)+2J\sigma_2^z\otimes\sigma_S^z}$.~The corresponding Hamiltonian in  $R_2$-perspective is also interacting, ${H_{\bar 2}=B\sigma_1^z\otimes\mathds1_S+2J\mathds1_1\otimes\sigma_S^z+B\sigma_1^z\otimes\sigma_S^z}$.~The global states at any time $t$ in the two perspectives are given by
\begin{align}
\ket{\psi(t)}_{\bar{1}}&=\ket{1}_{2}\otimes(e^{2iJt}a\ket{0}_S+e^{-2i(J-B)t}b\ket{1}_S),\nonumber\\
\ket{\psi(t)}_{\bar{2}}&=\ket{1}_{1}\otimes(e^{2iJt}a\ket{1}_S+e^{-2i(J-B)t}b\ket{0}_S)\,,\nonumber
\end{align}
and so separable in both perspectives.~In particular, $\omega_{2S}(t)=0=\omega_{1S}(t)$, so that~\eqref{necsuffeffclosed} is trivially fulfilled in both perspectives.~Further, we have $\rho_{S|R_2}(t)=\sigma_S^x\rho_{S|R_1}(t)\sigma_S^x$, such that the third term in the subsystem equation of motion~\eqref{sieom} vanishes.~Moreover, the expression \eqref{eq:HStilde} yields $\tilde{H}_{S|R_1}(t)=-2J\sigma_S^z$ and $\tilde{H}_{S|R_2}(t)=-B\sigma_S^z$.~The evolution of the $S$-subsystem state is effectively isolated in both perspectives, generated by the following time-independent \emph{effective} Hamiltonians $H_{S|R_1} + \tilde{H}_{S|R_1}(t)=(B-2J)\sigma_S^z\equiv-(H_{S|R_2} + \tilde{H}_{S|R_2}(t))$.
\end{example}

\subsubsection{QRF-relative subsystem states and entropies}\label{Sec:dynrhoSandentropies}

Even when the subsystem $S$ is (either effectively or exactly) closed in both perspectives, it is not necessarily the case that the subsystem entropies agree dynamically.~Let us combine our earlier results to briefly discuss conditions for the correlations between $S$ and its complement to remain invariant or change under QRF transformations for individual solutions to the equations of motion.  

The proof of Corollary~\ref{claim:Renyipure} implies that the $\alpha$-R\'enyi (and hence also entanglement) entropies are the same in both perspectives during some time interval, i.e.\ $S_\alpha[\rho_{S|R_1}(t)]=S_\alpha[\rho_{S|R_2}(t)]$,~$\forall\, t\in [t_0,t_1]$ and $\forall\, \alpha\in(0,1)\cup(1,\infty)$, if and only if 
\be\label{alltrhoSZSpure}
\rho_{S|R_2}(t)=\hat Z_S(t)^{\dagger}(\rho_{S|R_1}(t)),\;\; \forall t \in[t_0,t_1]
\ee
with a possibly $t$-dependent unitary $Z_S$.~In the special case of pure global states, Theorem~\ref{claim:purestatesrhoSZS} tells us that this is true if and only if $\rho_{\ibar}(t)\in\mcA_{\ibar}^{X(t)}$, $\forall t \in[t_0,t_1]$, i.e.\ if and only if the global state resides in some (possibly $t$-dependent) TPS-invariant subalgebra at all times during the pertinent interval.~For mixed states, Corollary~\ref{claim:rhoSZSmixed} entails that a sufficient condition for~\eqref{alltrhoSZSpure} to hold is that $\rho_{\ibar}(t)$ admits a decomposition into pure states, each of which lies in some $\mcA_{\ibar}^{X^A(t)}$ for all times in the interval with $X^A(t)= Y^A_j(t)\otimes Z_S(t)$, where $Y^A_j(t)$ may depend on that state.

If we restrict further to the special case that $X$ is independent of time, Lemma~\ref{lemma:finitetimeinvsoln} provides necessary and sufficient conditions for a (pure or mixed) state to reside in a TPS-invariant algebra during an interval, in which case the subsystem entropies then agree in both perspectives.~Specifically, it provides the necessary and sufficient conditions for the pure states in the above decomposition to reside in a fixed TPS-invariant subalgebra.~As already emphasised in Sec.~\ref{Sec:states&entanglement} for the instantaneous setting, a generic pure state will not belong to any $\mcA_{\ibar}^{X}$.~Even if it does initially at $t_0$, it will not remain inside it unless condition~\eqref{commAUXt2} is satisfied, which also depends on the Hamiltonian.~Therefore, the unitary equivalence of subsystem states in different QRF perspectives~\eqref{alltrhoSZSpure} over extended intervals is a rather special property and we may anticipate that for generic dynamical trajectories, the subsystem entropies will not coincide in the two perspectives.

We illustrate these observations with three examples in App.~\ref{app:dyn}.~Example~\ref{Ex:subsystemdyn} uses a QRF-invariant Hamiltonian to generate example trajectories that do or do not feature the same subsystem entropies in both perspectives, depending on whether they reside in a TPS-invariant subalgebra.~Example~\ref{inandout:Sdyn} exhibits global pure state trajectories oscillating between distinct TPS-invariant subalgebras, so that the two QRFs only agree periodically on all subsystem entropies and disagree otherwise.~Finally, Example~\ref{ex:mixedstatedynamics} demonstrates the same for global mixed states.

\section{QRF-covariance of thermodynamics}\label{Sec:TD}

The tools and results developed in the previous sections allowed us to systematically study the dependence of observables, states, correlations, and dynamics on the choice of internal QRF.~All of these are prerequisite for the investigation of the QRF-dependence of thermodynamical properties of physical subsystems, to which this last section is devoted.

To start, a few words about the framework for our analysis are in order.~Let us emphasise that our discussions in this work are not restricted to systems of macroscopic sizes, for which the standard phenomenological thermodynamical setup and interpretations usually apply.~In other words, we have allowed for quantum systems of microscopic sizes,~e.g.~a small number of qubits.~This is the realm of quantum thermodynamics (see for instance \cite{binder2019thermodynamics,Goold_2016,e15062100,vinanders}, and references therein).~It broadly aims to extend the principles of thermodynamics to quantum microscopic systems, and derive the emergence of standard phenomenological thermodynamics from an underlying quantum theory.

Quantum thermodynamics does not rely on conventional and restrictive notions like macroscopic reservoirs, thermal baths, weak couplings and close-to-equilibrium scenarios, thus being particularly valuable for finite-time and non-equilibrium thermodynamics \cite{binder2019thermodynamics,Goold_2016,e15062100,vinanders}.~In the present setting of QRFs, where even a single particle can serve as an internal reference frame, the complement of the subsystem of interest $S$ in the chosen frame perspective can then also consist of a single particle (the ``other frame'').~Moreover, quantum thermodynamics has a rich interplay with quantum information theory, to connect quantum properties of the system with (possibly generalised) thermodynamical ones.~The framework of quantum thermodynamics is thus naturally suited for the present setup.

We begin with a discussion relevant for equilibrium scenarios and thermostatics in Sections~\ref{Sec:thermostatics} and~\ref{Sec:thermostatics2}.~We then move to a discussion of generic thermodynamical processes in Sec.~\ref{Sec:tdynproc}.

\subsection{Subsystem observable averages}\label{Sec:thermostatics}

In the case when $S$ is of  macroscopic size and in thermodynamic equilibrium, the specification of certain expectation values can constitute the specification of a thermodynamic macrostate.~A thermodynamic macrostate is a collection of thermodynamically relevant observable averages (like energy, volume, particle number, etc.), of operationally accessible observables.~In this case, the following results can be understood as statements about thermodynamic macrostates of $S$ described relative to different QRFs.

\begin{lemma}\label{Lem:avgs}
Let $A_S \in \mcA_{S}$ be some $S$-local operator.~Then, for \emph{any} global state $\rho_{\ibar} \in \mathcal{S}(\mcH_{\ibar})$,
\be \label{eqavgs}
\Tr(\rho_{S|R_j} A_S) = \Tr(\rho_{S|R_i} A_S)\,,
\ee
if and only if ${[A_S,U^g_S] = 0, \, \forall g\in \mcG}$.~Here $\rho_{S|R_j} = \Tr_i(\rho_{\jbar})$, $\rho_{S|R_i} = \Tr_j(\rho_{\ibar})$,~and $\rho_{\jbar} = \hat{V}_{i \to j}^{g_i,g_j}(\rho_{\ibar})$ for some frame orientations $g_i,g_j \in \mcG$.
\end{lemma}

We refer to App.~\ref{app:thermo} for the proof.~In light of Theorem~\ref{thm:Slocals} in Sec.~\ref{Sec:Uinv&break}, we recall that the condition of translation invariance of $A_S$ is equivalent to: $\mathds1_j \otimes A_S \in \mcA_{\ibar}^{X}$, with some bilocal unitary $X=Y_j \otimes Z_S$ such that $[Z_S,A_S]=0$.

Lemma \ref{Lem:avgs} determines sufficient conditions for when a coarse-grained, effective description of a subsystem in thermodynamic equilibrium, encoded in a thermodynamic macrostate, appears the same relative to the frames $R_1$ and $R_2$.~Suppose we probe the subsystem in $R_1$-perspective by a finite set $\{\mathds1_2 \otimes A_S\}$ of $U_S^g$-invariant local $S$ observables and that these encompass all thermodynamically relevant ones in a given context.~The macrostate description in the $R_1$-perspective is given by the expectation values $\{\braket{A_S}_{\!\rho_{S|R_1}}\}$.~Being a coarse-grained description, these averages correspond to some unspecified underlying microstate $\rho_{S|R_1}$.~Lemma \ref{Lem:avgs} guarantees that the thermodynamic macrostate of the subsystem is described to be the same relative to both $R_1$- and $R_2$-perspectives.~From Theorem~\ref{thm:Slocals}, we further know that such $S$-local observables are the \emph{only} TPS-invariant local $S$-observables.

More generally, denoting the $G$-twirl over $\mcG$ for $S$ by $\Pit^S(\bullet) := \frac{1}{|\mcG|}\sum \limits_{g \in \mcG} U_S^g(\bullet) U^{g^{-1}}_S$, which is an orthogonal projector \cite{PHQRFassymmetries,PHinternalQRF}, we note that for any subsystem state $\rho_{S|R_i}$ and any translation-invariant $A_S$ we have
\be
\Tr\left(\rho_{S|R_i}A_S\right)=\Tr\left(\Pit^S(\rho_{S|R_i})A_S\right)\,.\nonumber
\ee
Eq.~\eqref{eqavgs} thus becomes
\be \label{eqavgs2}
\Tr(\Pit^S(\rho_{S|R_j}) A_S) = \Tr(\Pit^S(\rho_{S|R_i}) A_S) \,.
\ee
Clearly, when probing $S$ only with translation-invariant local observables, it is not possible to reconstruct $\Pit^{S\perp}$-components of the state in any of the perspectives; translation-invariant $S$-observables are tomographically incomplete for subsystem states $\rho_{S|R_i}$. However, they \emph{are} tomographically complete for the $\Pit^S(\rho_S)$-components of subsystem states.~For any state $\rho_S$, $\Pit^S(\rho_S)$ lies in the cone of positive semidefinite matrices in the space $\mcA_{S}^{\mbt}$ of translation-invariant $S$-observables.~As such, every $\Pit^S(\rho_S)$ can be decomposed into an orthonormal basis of $\mcA_{S}^{\mbt}$ with respect to the Hilbert-Schmidt inner product with coefficients given by the expectation values.\footnote{Note that these need not necessarily be normalised.}
Hence, one can reconstruct $\Pit^S(\rho_{S|R_2})$ from the knowledge only of a quorum $\{\braket{A_S}_{\rho_{S|R_1}}\}$ of expectation values of translation-invariant observables in $R_1$-perspective.

For the purposes of thermodynamics, where both the subsystem $S$ and its complement are involved, it is of interest to look also at observable averages of the ``other frame''.~The closest to a thermodynamic limit for the ``other frame'' subsystem is the limit of large group cardinality $|\mcG|$, in which case its Hilbert space has a large dimension. This permits us to speak about coarse-grained state descriptions and the following provides sufficient conditions for both QRFs to agree on such macrostate descriptions of $S$'s complement.
\begin{lemma} \label{lemma:Rjavgs}
Let $A_j  \in \mcA_{j}$ be some $R_j$-local operator.~Then, for \emph{any} global state $\rho_{\ibar} \in \mathcal{S}(\mcH_{\ibar})$
\be \label{eq:Rjavgs}
\Tr_j \bigl(\rho_j A_j\bigr) = \Tr_i \bigl(\rho_i A_i\bigr)\,, 
\ee
if and only if $A_j$ is diagonal in the group basis,~i.e.~if and only if $\Pid^\perp(A_j \otimes\mathds 1_S)=0$.~Here, $\rho_{j} = \Tr_S(\rho_{\ibar})$, $\rho_{i} = \Tr_S(\rho_{\jbar})$, $\rho_{\jbar} = \hat{V}_{i \to j}^{g_i,g_j}(\rho_{\ibar})$, and $A_i\otimes\mathds1_S=\hat{\mathcal I}_{i\to j}(\hat{P}_j^{g_i,g_j}(A_j)\otimes\mathds1_S)$, where $P_j^{g_i,g_j}$ is the parity-swap operator \eqref{Pswap},~and $\mathcal I_{i\to j}$ is the frame swap given in \eqref{eq:Ijtoi}.
\end{lemma}
The proof of Lemma~\ref{lemma:Rjavgs} follows similar steps to that of Lemma~\ref{Lem:avgs} and is thus omitted.~We recall from Theorem~\ref{thm:jlocals} that any such $R_j$-local observable $A_j\otimes\mathds1_S\in\mcA_{\ibar}^\mbd$ in the image of $\Pid$ resides in the TPS-invariant subalgebra $\mcA_{\ibar}^{X}$ with $X=(P_j^{g_i,g_j})^\dag\otimes Z_S$ for arbitrary unitary $Z_S$.~These are the local observables probing the ``other frame'' subsystem which are invariant under QRF-transformation up to the action of the parity-swap operator.~Lemma~\ref{lemma:Rjavgs} is thus stating that the frame $R_2$ being probed in the $R_1$-perspective by $R_2$-local observables in the image of $\Pid$, is operationally equivalent to the frame $R_1$ being probed in the $R_2$-perspective by the parity-swapped observables.~Moreover, similarly to the above discussion for $S$-local observables, for any state $\rho_{\ibar}\in \mcA_{\ibar}$ and operator $A_j\otimes \mathds1_S \in \mcA_{\ibar}^{\mbd}$, \eqref{eq:Rjavgs} can be written as
\be
\Tr_j\bigl(\Pid^j(\rho_{j})A_j\bigr)=\Tr_i\bigl(\Pid^i(\rho_{i})A_i\bigr)\,,
\ee
where we have denoted, $\Pid^j(\rho_{j})=\sum_{g\in\mathcal G}\bra{g}\rho_{j}\ket{g}_j\ket{g}\!\bra{g}_j$, and similarly for $\Pid^i(\rho_{i})$.~Therefore, when the ``other frame'' subsystem is probed only with $\Pid$-invariant local observables, then also only the $\Pid$-components of the subsystem state can be reconstructed in any of the two perspectives from the knowledge of their expectation values.

\subsection{QRF-relative subsystem equilibrium and temperature}\label{Sec:thermostatics2}

In the above discussion, the expectation values did not necessarily refer to stationary properties, as appropriate for a subsystem in thermal equilibrium. 
Let us now focus on $S$ being in thermal equilibrium relative to one QRF and explore what happens to it under QRF transformations. This brings us to the class of Gibbs states, which are stationary and maximise the von Neumann entropy for a given average energy or temperature.

To begin with, suppose $\rho_{S|R_i}$ is some stationary subsystem state in $R_i$-perspective. In order for the reduced state of $S$ relative to $R_j$ to be equally entropic and stationary, it follows from the discussion in Sec.~\ref{Sec:dynrhoSandentropies} that we must have $\rho_{S|R_j}=\hat{Z}_S^\dag(\rho_{S|R_i})$ with $t$-independent unitary $Z_S$. Provided the stationary $\rho_{S|R_i}$ originates from a global pure state $\rho_{\ibar}$,
Theorem~\ref{claim:purestatesrhoSZS} informs us that this is true if and only if $\rho_{\ibar}$ lives in a fixed TPS-invariant subalgebra $\mcA_{\ibar}^X$ for the considered duration, a case further characterised by Lemma~\ref{lemma:finitetimeinvsoln}. For a global mixed state instead, Corollary~\ref{claim:rhoSZSmixed} offers sufficient conditions.

For the subclass of Gibbs states, the situation is more restrictive.~Even if ${\rho_{S|R_j}=\hat{Z}_S^\dag(\rho_{S|R_i})}$ holds for a Gibbs state $\rho_{S|R_i}=\frac{1}{Z}e^{-\beta H_{S|R_i}}$ with $t$-independent $Z_S$, then $\rho_{S|R_j}$ will be stationary, but not Gibbs, unless the subsystem Hamiltonians are equally related, $H_{S|R_j}=\hat{Z}_S^\dag(H_{S|R_i})$. From~\eqref{gla2} we see that this requires
\be\label{Gibbscond}
H_{S|R_j}=H^\mbt_{S|R_i}+\lambda_{S|R_j}=\hat{Z}_S^\dag(H_{S|R_i})\,.
\ee 
We note that at this stage the subsystem Hamiltonians could be the bare or even effective (global state- and interaction-dependent) ones.~Clearly, \eqref{Gibbscond} is a necessary condition for a subsystem Gibbs state to map to a (unitarily related) subsystem Gibbs state under QRF transformations. But it is not sufficient, as we shall illustrate below in Example~\ref{ex:negbq}.  Generally, it depends on the global state whether a Gibbs state in $R_i$-perspective with Hamiltonian obeying~\eqref{Gibbscond} maps into a Gibbs state also in $R_j$-perspective.~However, there are subsystem Gibbs states which map \emph{independently of the global state} into equally entropic Gibbs states in the other perspective.~These are Gibbs relative to the bare subsystem Hamiltonian $H_S$\footnote{There are situations when the bare subsystem Hamiltonian is appropriate for defining subsystem equilibrium, e.g.\ when setting up initial states or when the interaction term commutes with the subsystem one.} and their characterisation directly follows from Theorem~\ref{thm_globindep}.
\begin{corollary}[\textbf{QRF-invariant Gibbs states}]\label{cor_QRFinvGibbs}
Let $\rho_{S|R_i}=\frac{1}{Z}e^{-\beta H_S}$ be a subsystem Gibbs state relative to frame $R_i$.~Under QRF transformations it maps into an \emph{equally entropic} Gibbs subsystem state relative to $R_j$ (${i\neq j}$) \emph{for all global states} $\rho_{\ibar}\in\mcS(\mcH_{\ibar})$ such that $\rho_{S|R_i}=\Tr_j(\rho_{\ibar})$, if and only if $H_S=H_S^\mbt$ is translation-invariant and $\lambda_{S|R_j}=0$.~In this case, we further have exact QRF invariance, i.e.\ $\rho_{S|R_j}=\rho_{S|R_i}$.
\end{corollary}
This corresponds to the special case that the subsystem Hamiltonians are mapped into one another under QRF transformations. It is compatible with Lemma~\ref{Lem:avgs} above in that, maximising the subsystem entropy under the constraint $\braket{H_S}_{\rho_{S|R_i}}=$ constant, where $H_S$ is translation-invariant, yields the same result in the perspectives of both $R_1$ and $R_2$.\footnote{Similar arguments can also be applied to \emph{generalised Gibbs states} defined in terms of a set of (possibly non-commuting) observables as considered for instance, in quantum thermodynamics \cite{Halpern2016}, reparametrisation-invariant systems \cite{Chirco:2013zwa,Chirco:2015bps,Kotecha:2018gof,Chircomultisympl,Chirco:2019tig,Kotecha:2020oxz}, and quantum gravity \cite{Kotecha:2020oxz,Kotecha:2018gof,Chirco:2018fns,Chirco:2019kez}.~As long as we consider macrostates of the subsystem of interest to be specified by the expectation values of a finite set of $\Pit$-invariant (hence $\mbfU_{\ibar}$-invariant) $S$-local observables $\mathds1_j\otimes A_S^a$, then maximisation of the von Neumann entropy, under the set of constraints $\braket{A_S^a}=$ constant, will yield the same generalised Gibbs states $\rho_{S|R_1}(\beta_a)=\frac{1}{Z(\beta_a)}e^{-\sum_a\beta_a A^a_S}=\rho_{S|R_2}(\beta_a)$ relative to the internal reference frames $R_1$ and $R_2$.}~Translation-invariant Gibbs states can be viewed as describing thermal equilibria of $S$ defined relative to an $S$-internal QRF; these are already fully relational without reference to an ``external'' QRF $R_1$ or $R_2$ and thus independent of their state information contained in $\rho_{\ibar}$. If the $S$-local term coming from the QRF transformation of the interaction $H_{jS}$ in $R_i$-perspective does not vanish, $\lambda_{S|R_j}\neq 0$, one would still obtain ${\rho_{S|R_j}=\frac{1}{Z}e^{H^\mbt_{S|R_i}}=\rho_{S|R_i}}$, but in $R_j$-perspective this would not constitute a Gibbs state since the exponent would not be the complete subsystem Hamiltonian in that perspective.

When the subsystem Hamiltonians are not translation-invariant, and thus do not transform into one another under QRF changes, it is still possible for \emph{some} global states that a subsystem Gibbs state  from $R_1$-perspective maps to an equally entropic Gibbs state in $R_2$-perspective when condition~\eqref{Gibbscond} is fulfilled. In this case, the Gibbs states in the two perspectives need not be equal. For instance, Example~\ref{ex:negbq} (for $\mu=1$) and the discussion below will show that the QRF transformation can induce a sign flip in the temperature of Gibbs states. Furthermore, it is also possible that the subsystem Hamiltonians $H_{S|R_1}$ and $H_{S|R_2}$ have distinct spectra in the two perspectives, in which case their corresponding Gibbs states can also carry \emph{distinct entropy}. Their relation is then non-unitary and, as we shall see, there  exist  circumstances when these map into one another under QRF changes. We will also illustrate this case in Example~\ref{ex:negbq} (for $\mu\neq1$), where this will also lead to a rescaling of the temperature.

In any case, it is clearly a special situation that a Gibbs state of one perspective maps into a Gibbs state of the other. In particular, we can anticipate that a violation of the condition of translation-invariance of the subsystem Hamiltonian in Corollary~\ref{cor_QRFinvGibbs} generically leads to Gibbs states in one perspective, which are mapped to non-equilibrium subsystem states in the other. Let us illustrate this in a simple example.

\begin{example}[\textbf{QRF-relative thermal equilibrium}]\label{ex:relequilibrium}
Consider again three qubits $R_1R_2S$ and suppose in $R_1$-perspective the global Hamiltonian and initial state are given by ${H_{\onebar}=a\sigma_2^x\otimes\mathds1_S+b \mathds1_2\otimes\sigma_S^z}$, with $a,b\neq0$, and ${\rho_{\onebar}(0)=\ket{0}\!\bra{0}_2\otimes\frac{1}{Z}e^{-\beta H_{S|R_1}}}$, where ${H_{S|R_1}=b\sigma_S^z}$. Clearly, we have that \be
\rho_{S|R_1}(t)=\frac{1}{Z}e^{-\beta H_{S|R_1}}
\ee
is an equilibrium Gibbs state at all times.

Let us now fix the frame orientations before and after QRF transformations to be $g_1=g_2=e$.~Then, \eqref{counterexample} tells us that the initial state transforms into $R_2$-perspective as ${\rho_{\twobar}(0)=\ket{0}\!\bra{0}_1\otimes\frac{1}{Z}e^{-\beta H_{S|R_1}}}$, while from Example~\ref{Ex:AUqubits2} we can read off that the Hamiltonian in $R_2$-perspective is a pure interaction piece ${H_{\twobar}=a\sigma_1^x\otimes\sigma_S^x+b\sigma_1^z\otimes\sigma_S^z}$ (thus illustrating the closed-to-open transition of Lemma~\ref{lemma:closed2closed}).~The second contribution to the Hamiltonian commutes with $\rho_{\twobar}(0)$ and we find the time-dependent subsystem state
\be 
\rho_{S|R_2}(t)=\frac{\cos^2at}{Z}e^{-\beta H_{S|R_1}}+\frac{\sin^2at}{Z}e^{\beta H_{S|R_1}}\,,
\ee 
which thus clearly is neither Gibbs nor stationary with respect to the bare Hamiltonian in $R_2$-perspective.~In particular, the bare $S$-Hamiltonian in $R_2$-perspective vanishes such that its associated Gibbs state would be the completely mixed one.~The resulting state is also not Gibbs with respect to a generic prescription for a subsystem effective Hamiltonian (more on this in Sec.~\ref{Sec:TDsetup}).~It is essentially only Gibbs with respect to its modular Hamiltonian.
\end{example}

As the notion of temperature can be defined via Gibbs states, the relativity of subsystem equilibrium states naturally leads us to inquire about the QRF dependence of temperature.~To this aim, let us consider global factorised states of the form $\rho_{\ibar}=\rho_j\otimes\rho_{S}$.\footnote{We note however that this condition of a product state can be relaxed, and Lemma \ref{Lemma:NegBProduct} below can be generalised to generic global states $\rho_{\ibar} = \rho_j \otimes e^{-\beta H_S}/Z + \omega_{jS}$, where $\omega_{jS}$ is not necessarily zero but satisfies $\sum_{g\in\mcG} U^g_S\bra{g_jg^{-1}}\omega_{jS}\ket{g_jg^{-1}}_j U^{g^{-1}}_S = 0$.}~Such kind of states are often considered as initial states of composite systems; for instance when the subsystems are
individually prepared in some equilibrium states via thermalisation with separate reservoirs before putting them in contact with one another.~As a first step, we then have the following observation, generalising the previous example.
\begin{lemma} \label{Lemma:NegBProduct}
Let the total system relative to $R_i$ in some orientation $g_i\in \mcG$ be prepared in a global state $\rho_{\ibar} \in \mathcal{S}(\mcH_{\ibar})$ of the form
\be\label{i-separable-gibbs}
\rho_{\ibar} = \rho_{j} \otimes \frac{1}{Z}e^{-\beta H_{S|R_i}}
\ee
and the Hamiltonian $H_{S|R_i}$ be such that,
\begin{align}
\exists \, \mcG_{\rm a} \subsetneq \mcG \;\;\; &\;{\rm s.t.} \;\, \{U_S^h, H_{S|R_i}\} = 0 \;\; \forall h\in \mcG_{\rm a}, \label{negcond1} \\
&{\rm and} \;\;\; [U_S^g, H_{S|R_i}] \,= 0 \;\; \forall g \in \mcG_{\rm c}  \label{negcond2}
\end{align}
where, $\mcG_{\rm a}$ is a strict\footnote{This is because there always exists one element of $\mcG$, the identity, which only ever satisfies \eqref{negcond2} and thus always belongs to $\mcG_{\rm c}$.} subset of $\mcG$, $\{\cdot,\cdot\}$ denotes an anti-commutator, and $\mcG_{\rm c} = \mcG \backslash \mcG_{\rm a}$.~Then, the subsystem state relative to $R_j$ ($i\neq j$) in some orientation $g_j \in \mcG$ is
\be \label{negmixProduct}
\rho_{S| R_j} = q_{\rm a}^{g_j} \frac{1}{Z} \, e^{+\beta H_{S|R_i}}  + (1 - q_{\rm a}^{g_j}) \frac{1}{Z} \, e^{-\beta H_{S|R_i}} 
\ee
where, $q_{\rm a}^{g_j} = \sum \limits_{h \in \mcG_{\rm a}} \bra{g_j h^{-1}}\rho_{j}\ket{g_j h^{-1}}_j$.
\end{lemma}

We refer to App.~\ref{app:thermo} for the proof.~Even if $H_{S|R_i}$ is the appropriate (bare or effective) Hamiltonian for defining equilibrium in $R_i$-perspective, we can interpret \eqref{negmixProduct} as a mixture of a positive and negative temperature Gibbs state in $R_j$-perspective only if $H_{S|R_i}$ is \emph{also} the appropriate (bare or effective) subsystem Hamiltonian for defining Gibbs states in that perspective. In a generic situation, this will not be true (as, e.g., in Example~\ref{ex:relequilibrium} which satisfies the conditions of this Lemma), but as we shall see shortly in special cases  this condition can be satisfied (cfr.~Example~\ref{ex:negbq}). 

In the case that the bare Hamiltonian is appropriate for defining equilibrium, this means that, instead of~\eqref{Gibbscond}, the following condition has to hold 
\be\label{nonunitaryGibbs} 
H_{S|R_j}=H_{S|R_i}^\mbt+\lambda_{S|R_j} = \pm \mu\,H_{S|R_i}\,,
\ee 
where $\mu>0$ is some positive real number.~In contrast to~\eqref{Gibbscond}, this is generally \emph{not} a unitary transformation between the subsystem Hamiltonians in the two perspectives, as it can change the spectrum.~Since~\eqref{negcond1} entails that $H_{S|R_i}$ is not translation-invariant and, in fact, lies in the kernel of $\Pit$,\footnote{\label{FN:negbham} This can be readily seen from the fact that, using the decomposition $H_S=H_S^{\mbt}+H_S^{\mbt\,\perp}$,~\eqref{negcond1} yields $H_S^{\mbt\,\perp}=-2H_S^{\mbt}-U_S^hH_S^{\mbt\perp}U_S^{h^{-1}}$, for any $h\in\mathcal G_{\rm a}$.~That is, $H_S=H_S^{\mbt}+H_S^{\mbt\perp}=-H_S^{\mbt}-U_S^hH_S^{\mbt\perp}U_S^{h^{-1}}$, from which, applying the projector $\Pi_\mbt^S$ on both sides of the last equality, it follows that $H_S^{\mbt}=0$.~Therefore, $H_S=H_S^{\mbt\perp}=-U_S^hH_S^{\mbt\perp}U_S^{h^{-1}}$, for any $h\in\mathcal G_{\rm a}$.} we see that~\eqref{nonunitaryGibbs} can only be fulfilled provided that ${H_{S|R_j}=\lambda_{S|R_j}=\pm\mu\,H_{S|R_i}}$.~In other words, the subsystem Hamiltonian in $R_j$-perspective results entirely from the interaction term $H_{jS}$ in $R_i$-perspective via QRF transformations (cfr.~\eqref{gla2} and~\eqref{intsourcelocal2}).~When~\eqref{nonunitaryGibbs} holds,~\eqref{negmixProduct} constitutes a convex mixture of two Gibbs states in $R_j$-perspective, one with positive rescaled inverse temperature ${\beta^+_j = \beta/\mu}$ and one with negative rescaled inverse temperature ${\beta^-_j=-\beta/\mu}$.~As regards the global states, it can be checked that the total state $\rho_{\ibar}$ in \eqref{i-separable-gibbs} generically breaks $\mbfUX$-invariance, and $\rho_{\jbar}$ is not necessarily separable. 
 
Further, if the system is prepared in a state $\rho_{j}$ relative to $R_i$ (for some $g_i,g_j$) such that 
\be\label{fullbetainversion}
\bra{g_j g^{-1}}\rho_{j}\ket{g_j g^{-1}}_j = 0 \,, \;\; \forall g \in \mcG_{\rm c}\,, 
\ee 
then $q_{\rm a}^{g_j}=1$ and $\rho_{S|R_j}$ in \eqref{negmixProduct} is purely a Gibbs state associated with negative inverse temperature $\beta^-_j$ and an inverted Hamiltonian,~$-H_{S|R_i}$.~In the special case that $\mu=1$, we moreover have that the subsystem states in the two perspectives are unitarily related, $\rho_{S|R_j}= \hat{U}^g_S(\rho_{S|R_i}),\, \forall g\in \mcG_{\rm a}$.

We note that thermal states with negative temperature only make sense for Hamiltonians that are both upper and lower bounded \cite{afa,*uffink,PhysRev.103.20,BALDOVIN20211}; this is clearly the case in the finite-dimensional setting in which we are working.~Compared to the ground state, negative temperature states can often be understood via a population inversion, where more degrees of freedom of the subsystem occupy higher energy states.~In such a system, any increase in energy corresponds to a decrease in its von Neumann entropy \cite{PhysRev.103.20,BALDOVIN20211}.

Let us now illustrate two multi-qubit scenarios (thus restricting to $\mathcal{G}=\mathbb{Z}_2$) in which condition~\eqref{nonunitaryGibbs} can be satisfied and the QRF transformation thus not only maps between positive and negative temperature Gibbs states, but also rescales the temperature.~The qubit case is of particular interest because negative temperatures have been reported experimentally for spin systems\footnote{Negative temperatures have also been reported experimentally for other types of systems like hydrodynamic vortices, lasers and condensates in optical lattices.~See for instance \cite{BALDOVIN20211,PhysRevE.95.012125}, and references therein.} \cite{PhysRev.103.20,BALDOVIN20211,PhysRevE.95.012125,Oja:1997zz}.

\begin{example}[\textbf{Positive-to-negative temperature QRF change}] \label{ex:negbq}
Consider a system of $N$ qubits ($\mcG = \mathbb{Z}_2$) evolving under a total Hamiltonian  relative to frame $R_1$ of the form
\be\label{ex:NqIZZ}
H_{\onebar}= \nu\,\sigma_2^z\otimes \mathds1_S + \mathds1_2 \otimes H_{S|R_1} + \mu\,\sigma_2^z\otimes H_{S|R_1} \,,
\ee 
such that $\{\sigma_x^{\otimes N-2},H_{S|R_1}\}=0$ and with $\nu,\mu>0$.~The interaction and free terms in $H_{\onebar}$ commute,\footnote{Thus, the dynamical map satisfies the so-called strict energy conservation condition, i.e.~$[e^{-iH_{\onebar}t},H_2\otimes \mathds1_S + \mathds1_2\otimes H_{S|R_1}]=0$ \cite{RevModPhys.93.035008}.~We shall be using commutator conditions like in the main text in various examples that follow.} and conditions \eqref{negcond1}-\eqref{negcond2} 
are satisfied with $\mcG_{\rm a} = \{1\}$ ($U^{h=1}=\sigma_x$) and $\mcG_{\rm c} = \{0\}$ ($U^{g=0}=\mathds1$).~We realize these conditions with the following two subsystem Hamiltonians\footnote{See for instance \cite{BALDOVIN20211,PhysRevE.95.012125,Oja:1997zz} for their experimental relevance in the context of negative temperatures.~Such Hamiltonians also appear in quantum spin chains, although one typically considers two-particle interactions, rather than a long-range many-body interaction of the kind given in example $(b)$ in \eqref{qunegall}.~A contribution to the total Hamiltonian of the form $\sigma_j^z\otimes H_S$ with $H_S$ given by $(a)$ is an example of a two-particle interaction term, though it may be long-range.}
\be
(a)\, H_{S|R_1} = \sum_{\ell=1}^{N-2}  \sigma^z_{\ell} \bigotimes \limits_{\substack{m =1 \\ m \neq \ell}}^{N-2} \mathds1_m, \quad (b)\, H_{S|R_1} =  \sigma_z^{\otimes N-2}. \label{qunegall} \ee
It is straightforward to verify that example $(a)$ satisfies the conditions for any $N$, while $(b)$ does so for odd $N$ only.\footnote{For $N=3$ these two Hamiltonians are identical.}

Given that $\mathds1_2\otimes H_{S|R_1}$ commutes with the interaction, it is a reasonable prescription to use the bare $H_{S|R_1}$ for defining the $S$-Gibbs states in $R_1$-perspective.
Consider thus an initial global state of the form ${\rho_{\onebar}(0)=\rho_{2}(0) \otimes e^{-\beta H_{S|R_1}}/Z}$, such that $\rho_2(0)$ is diagonal in the group basis and hence $\rho_{\onebar}(0)=\rho_{\onebar}(t)$ is stationary and $S$ is in the Gibbs state $\rho_{S|R_1}(t)=\frac{1}{Z}e^{-\beta H_{S|R_1}}$.

Let us fix the orientations of frames $R_1$ and $R_2$ to be $g_1=0$ and $g_2=0$ before and after the QRF transformation, respectively, so that the QRF transformation is given by $V_{1\to2}^{0,0}=\mathcal{I}_{1\to2}\mathbf{U}_{\onebar}$, where $\mathbf{U}_{\onebar}$ is the \textsf{CNOT} gate (cfr.~Example~\ref{Ex:Wstate}) and $\mcI_{1\to2}$ is the frame swap.~In $R_2$-perspective, the subsystem state is then given by the statistical mixture in \eqref{negmixProduct} with time-independent ${q_{\rm a}^{e}(t) = q_{\rm a}^{e}(0) =  \bra{1}\rho_2(0)\ket{1}_{2}}$ (since $\rho_2(t)$ is stationary).~Importantly, condition~\eqref{nonunitaryGibbs} is satisfied with $H_{S|R_1}^\mbt=0$ for both (a) and (b) (with odd $N$ in the latter case).~Indeed, in both cases it can be checked that 
\begin{align} 
\hat{V}_{1\to2}^{0,0}(\sigma_2^z\otimes\mathds1_S) &= \sigma_1^z\otimes \mathds1_S\nonumber\\
\hat{V}_{1\to2}^{0,0}(\mathds1_2\otimes H_{S|R_1}) &= \sigma_1^z\otimes H_{S|R_1}\nonumber\\
\hat{V}_{1\to2}^{0,0}(\sigma_2^z\otimes H_{S|R_1}) &= \mathds1_1\otimes H_{S|R_1}\nonumber\,,
\end{align}
so that the QRF transformation leaves the free ``other frame'' Hamiltonian invariant, while swapping the free subsystem and interaction Hamiltonians (up to frame swap).~Hence, comparing with~\eqref{ex:NqIZZ}, we have $H_{S|R_2} = \mu H_{S|R_1}$ and since this also commutes with the interaction piece in $R_2$-perspective we are again entitled to use this contribution to define subsystem Gibbs states in $R_2$-perspective.~Thus, in this case, the state~\eqref{negmixProduct} is a genuine mixture of postive and negative temperature Gibbs states, where the temperature is rescaled by $\mu$.

Further, if the system in the $R_1$-perspective is prepared such that $\bra{0} \rho_2(0)\ket{0}_2=0$, i.e.~\eqref{fullbetainversion} is satisfied, then $\rho_{S|R_2}(t) = e^{+\beta^-_2 H_{S|R_2}}/Z$ with $\beta_2^-=-\beta/\mu$, and $S$ is in a thermal state with negative temperature relative to $R_2$ at all times.~From \eqref{fullbetainversion} and $\Tr(\rho_2^2)\leq1$, it follows that $\rho_{2}=\ket{1}\!\bra{1}_2$ and so ${\rho_{\onebar}= \ket{1}\!\bra{1}_2\otimes \frac{1}{Z_S}e^{-\beta H_{S|R_1}}}$.

Moreover, in this case, the entropy of the $S$ subsystem is the same in the two perspectives \emph{only} if $\mu=1$ in which case the subsystem states are unitarily related ${\rho_{S|R_2}(t)=\sigma_{x}^{\otimes N-2}\rho_{S|R_1}(t)\sigma_{x}^{\otimes N-2}}$ at all times.~In this case, we also have $\rho_{\onebar}(t)\in\mcA_{\onebar}^{\mbfU_{\!X^1}}$ with $X^1 = \mathds1_2 \otimes \sigma_x^{\otimes N-2}$ and the global state is separable in both perspectives, while the global Hamiltonian $H_{\onebar}$ is QRF-invariant and resides in $\mcA_{\onebar}^{\mathds1}$. This example thus also illustrates Lemma \ref{lemma:finitetimeinvsoln} with Hamiltonian and trajectory belonging to distinct TPS-invariant subalgebras.

By contrast, when $\mu\neq1$, the relation is non-unitary and furnishes examples of QRF transformation related thermal subsystem states with in general different entropy in the two perspectives.~Neither global Hamiltonian, nor trajectory are TPS-invariant in this case.
\end{example}

While reminiscent of the Unruh effect in special relativistic quantum field theory, it is worth emphasising that the above QRF transformation generated change of subsystem temperature is \emph{a priori} distinct in nature.~In the Unruh effect, the transformation proceeds between classical inertial and uniformly accelerated frames and is non-unitary for the total quantum field state because the accelerated frames ``see'' only half of the system.~Here, by construction, the transformation relates two internal QRFs and is unitary at the global level (involving the frames) because both frames ``see'' the same amount of degrees of freedom.~While the same effect may appear in a quantum field theoretic extension with quantum observers (and may recover the Unruh effect in the limit of these observers being sufficiently classical), it appears to be qualitatively distinct and does not require relativistic spacetime structures, instead being entirely rooted in the quantum relativity of subsystems.

In the present case, in contrast to the Unruh effect, it is also possible to flip the sign of the temperature.~In the multi-qubit example with $\mcG=\mathbb{Z}_2$, this is due to the fact that the subsystem Hamiltonian is not flip-invariant, thus reversing the ordering of eigenstates under the only non-trivial $\mathbb{Z}_2$ transformation described by $\sigma^x$.

\subsection{QRF-relative thermodynamic processes}\label{Sec:tdynproc}

We saw in Sec.~\ref{Sec:Dyn} above that subsystem dynamics is QRF-dependent in general; it can be closed in one perspective and open in another.~This is intimately linked to the fact that correlations are generally QRF-dependent due to subsystem relativity (cfr.~Secs.~\ref{Sec:QuantRelTPSs}-\ref{Sec:states&entanglement}).~Thus, {\it both} the dynamical evolution and QRF transformations can generate correlations. This is particularly interesting in the context of quantum thermodynamics, where correlations are known to play an important role, e.g.~they can be understood and used as a thermodynamic resource, like for work \cite{vinanders,Vitagliano2018,RevModPhys.93.035008}.~Dynamical correlations are crucially related to energy exchanges and entropy production, thus generally to non-equilibrium thermodynamical processes in quantum systems \cite{RevModPhys.93.035008,rezakhani,Hossein_Nejad_2015,lebon2008understanding}.~Since QRF transformations can introduce interactions and map equilibrium states into non-equilibrium ones, this is a realm one necessarily needs to take into account when exploring quantum thermodynamics in the context of QRF covariance.

Given the QRF-dependence of correlations and subsystem dynamics, we can already anticipate that thermodynamical processes and their associated state- and process-dependent quantities will also be contingent on the choice of QRF in general.~On the other hand, we expect the laws of thermodynamics to be QRF-independent due to covariance of physical laws under reference frame changes. 

We now discuss the QRF-(in)dependence of thermodynamical processes, within the setting of subsystem dynamics as presented in Sec.~\ref{Sec:Dyn} above, and utilising insights and tools from quantum thermodynamics.~As a first step, our considerations here shall be limited to only a few aspects, placed generally in the context of energy balance (the first law) and entropy balance (related to the second law), in order to uncover interesting physical features in the context of QRFs.~In particular, we shall inquire about the QRF-dependence of heat, work, entropy production and flow. 

\subsubsection{Setup}\label{Sec:TDsetup}

Recall from Sec.~\ref{ssec:setupdyn} that the total $R_jS$ system in $R_i$-perspective is dynamically and thermodynamically isolated.~The energy and entropy exchanges that we consider here are therefore between the subsystem $S$ and ``the other frame'' in a given perspective.\footnote{Within a given perspective, this is akin to considering `autonomous bipartite quantum systems' in the context of quantum thermodynamics, see for instance \cite{rezakhani,weimer,Hossein_Nejad_2015}.} 

Let us begin with the formulation of a setup surrounding the first law. As is well known, there is no general consensus on the definitions of various (generalisations of) thermodynamic quantities like heat and work for quantum systems.~We will thus attempt to keep our discussions independent of any specific prescriptions as much as possible, while still noticing the points of differences or similarities between them.~This requires clarifications on the precise setup which we shall employ in the sequel.

Let $E_{S|R_i}$ denote the subsystem internal energy in $R_i$-perspective 
\be\label{esi}
E_{S|R_i}(t) = \Tr(H_{S|R_i}^{\eff}(t) \rho_{S|R_i}(t))\,, 
\ee
where $H_{S|R_i}^{\eff}$ is some effective subsystem Hamiltonian, which in general includes also contributions arising from non-trivial interactions between $S$ and its complement, i.e.\ here the other frame $R_j$.~How precisely $H_{S|R_i}^{\eff}$ is constructed from the total Hamiltonian $H_{\ibar}$ and global state $\rho_{\ibar}(t)$ is a matter of choice \cite{rezakhani,Hossein_Nejad_2015,weimer,alipour2019,PhysRevLett.124.160601}.~Generically, we can write 
\be \label{effhs}
H_{S|R_i}^{\eff}(t) = H_{S|R_i} + \mathds{h}_{S|R_i}(t) \,,
\ee
where $H_{S|R_i}$ is the subsystem local part of $H_{\ibar}$ (cfr.~Eq.~\eqref{hibar}), and $\mdsh_{S|R_i}$ encodes the effect on $S$ that is \emph{induced by interactions} with its complement $R_j$.~Which specific prescription for $\mdsh_{S|R_i}$ one chooses could be based on a combination of factors, for instance related to the specific (class of) system(s) at hand, e.g.~system size, presence or type of interactions, etc.\footnote{In presence of no interactions though, we naturally expect that $\mdsh_{S|R_i}(t) = 0,\,\forall t$.
}~For example, (i) $\mdsh_{S|R_i}(t) = \tilde{H}_{S|R_i}(t) + \alpha(t) \mathds1_S$ in \cite{rezakhani}, where $\tilde{H}_{S|R_i}(t)$ is given by \eqref{eq:HStilde}, and $\alpha$ could further depend on the interactions and the reduced states (more details on this choice in relation to the energy balance can be found in App.~\ref{App:TDroleofcorrelations}); (ii) $\mdsh_{S|R_i}(t) = \tilde{H}_{S|R_i,1}(t)$ in \cite{weimer,Hossein_Nejad_2015}, where $\tilde{H}_{S|R_i,1}$ is that part of $\tilde{H}_{S|R_i}(t)$ which commutes with $H_{S|R_i}$.\footnote{$\tilde{H}_{S|R_i}$ can be decomposed as $\tilde{H}_{S|R_i} = \tilde{H}_{S|R_i,1} + \tilde{H}_{S|R_i,2}$, such that $[\tilde{H}_{S|R_i,1}, H_{S|R_i}] = 0$ and $[\tilde{H}_{S|R_i,2},H_{S|R_i}] \neq 0$ \cite{weimer,Hossein_Nejad_2015}.~More precisely, $\tilde H_{S|R_i,1}$  is the projection of $\tilde H_{S|R_i}$ into the commutant of the bare $H_{S|R_i}$.}~Our aim is not to compare the different choices for their scope of use (which is a subject of active research in quantum thermodynamics), but to explore thermodynamical properties in different QRF perspectives without necessarily having to commit  to a specific form for $\mdsh_{S|R_i}$.~In the following, we thus work with the general form \eqref{effhs}, except in special cases where it will turn out useful to restrict to a prescription.

However, we do impose some basic restrictions on the construction of effective Hamiltonian contributions (not only for $S$).~This is motivated by the desire to have effective Hamiltonians that are consistent with the original dynamics and total energy and whose expectation values do not depend on the choice of local basis for the subsystems.~(The latter condition essentially means that the effective prescription does not introduce any coupling to anything ``external''.)~Specifically, modified contributions to the total energy arise simply from ``reassigning'' certain bare contributions to different subsystems or interactions.~Modifications of bare Hamiltonians must then arise as (possibly partial) expectation values of pieces of bare Hamiltonians in the relevant state.
\begin{definition}[\textbf{Permissible effective Hamiltonians}]\label{def_prescription}
A prescription for effective Hamiltonian contributions $H_{S|R_i}^{\eff}$, $H_{j}^{\eff}$ and $H_{jS}^{\eff}$ for $S$, $R_j$ and their interaction in $R_i$-perspective ($i\neq j$), respectively, is permissible if 
\begin{itemize}
    \item[(a)] they are built \emph{linearly} from the bare Hamiltonian contributions $H_j$, $H_{S|R_i}$ and $H_{jS}$, as well as the global state $\rho_{\ibar}$; 
    \item[(b)] their sum equals $H_{\ibar}$; and
    \item[(c)] they are covariant with respect to local changes of basis in $S$ and $R_j$.~That is, if we apply a local change of basis $X=Y_j\otimes Z_S$ to the bare Hamiltonian and state, $H_{\ibar}\mapsto\hat X(H_{\ibar})$ and ${\rho_{\ibar}\mapsto \hat X(\rho_{\ibar})}$, then the effective Hamiltonians transform likewise, $H_{S|R_i}^{\eff}\mapsto\hat Z_S(H_{S|R_i}^{\eff})$, $H_{j}^{\eff}\mapsto \hat Y_j(H_j^{\eff})$ and $H_{jS}^{\eff}\mapsto\hat X(H_{jS}^{\eff})$.
    \end{itemize}
    \end{definition}
Both examples (i) and (ii) above satisfy this definition. In particular, the effective Hamiltonian contributions may depend linearly on $\rho_j=\Tr_S\rho_{\ibar}$ and $\rho_{S|R_i}=\Tr_j\rho_{\ibar}$ and thus also $\omega_{jS}$. They may involve traces and even projections, so long as the projectors are defined in terms of the Hamiltonian and state (as in (ii)).  

The total Hamiltonian \eqref{hibar} can now be rewritten as
\be \label{hibareff}
H_{\ibar} = H_{j}^{\eff} \otimes \mathds1_S + \mathds1_j \otimes H_{S|R_i}^{\eff} + H_{jS}^{\eff}\,,
\ee
where $H_{j}^{\eff}$ is defined analogously to \eqref{effhs} with a corresponding $\mdsh_{j}$, and equation \eqref{hibareff} itself is taken as the definition of $H_{jS}^{\eff}$.~Note that, unlike in the case of $H_{jS}$ in~\eqref{hibar}, the effective interaction $H_{jS}^{\eff}$ no longer needs to involve only terms that are fully non-local across the $S$ and $R_j$ partition (see App.~\ref{App:TDroleofcorrelations} for an example).~It must be noted that the individual terms on the right hand side in \eqref{hibareff} are time dependent in general, but the full $H_{\ibar}$ is time independent by construction.~The total internal energy relative to $R_i$ can then be written as
\be \label{totinten}
E_{\ibar} = \Tr(H_{\ibar}\rho_{\ibar}) = E_{j} + E_{S|R_i} + E_{jS}\,,
\ee
where $E_{S|R_i}$ is given by \eqref{esi}, $E_{j}$ is defined analogously to \eqref{esi}, and $E_{jS} = \Tr(H_{jS}^{\eff} \,\rho_{\ibar})$ is the contribution associated with interactions in the total system, thus also with dynamical correlations.~Depending on the specific prescription for defining the effective Hamiltonian, this last energy contribution can be simplified further, e.g.~if the effective Hamiltonian is such that $\Tr_S(H_{jS}^{\eff}\, \rho_{S|R_i}) = 0 = \Tr_j(H_{jS}^{\eff} \, \rho_{j})$, then $E_{jS} = \Tr(H_{jS}^{\eff}\, \omega_{jS|R_i})$ as done in \cite{rezakhani}.\footnote{In fact, the study in \cite{rezakhani} has understood this contribution $E_{jS}$ as the `binding' energy which encodes how much energy is stored within correlation and is not transferred away as heat between the interacting components of a bipartite system.~We refer to App.~\ref{App:TDroleofcorrelations} for further details in our present context.}~Notice that even though the individual terms on the right hand side of \eqref{totinten} are time dependent in general, the total internal energy is time independent, $\dot{E}_{\ibar} = 0$ (conservation of total energy), as the total system is isolated.

The energy balance equation (the first law) for the subsystem of interest is then
\be \label{energybal}
\dot{E}_{S|R_i}(t) = \Tr(\dot{H}_{S|R_i}^{\eff}(t) \rho_{S|R_i}(t)) + \Tr(H_{S|R_i}^{\eff}(t) \dot{\rho}_{S|R_i}(t)).
\ee
Notice that at this level there is no ambiguity in evaluating the change in internal energy as determined by \eqref{energybal}.~The differences arise in the precise microscopic quantum definitions for work and heat from the right hand side of \eqref{energybal}, i.e.~in the precise separation of the RHS into terms that can reasonably well be understood as thermodynamic work and heat.~This can be clarified by simplifying the RHS further using the subsystem equations of motion~\eqref{sieom},
\begin{align} 
\dot{E}_{S|R_i}(t) &= \Tr\! \big(\dot{H}_{S|R_i}^{\eff}(t) \rho_{S|R_i}(t)\big) + \Tr\! \big(H_{S|R_i}^{\eff}(t) \dot{\rho}_{S|R_i}(t)\big) \label{energybal2} \\
&= \Tr\!\big( \dot{H}_{S|R_i}^{\eff}(t) \rho_{S|R_i}(t)\big) \nonumber \\
&\;\;\; -i \Tr\!\big(H_{S|R_i}^{\eff}(t)\big[H_{S|R_i} + \tilde{H}_{S|R_i}(t), \rho_{S|R_i}(t) \big]\big) \nonumber \\
&\;\;\; -i \Tr\!\big(H_{S|R_i}^{\eff}(t) \Tr_j\!\big[H_{jS}, \omega_{jS}(t) \big]\big). \label{energybal3}
\end{align}
The conventional definitions motivated from classical thermodynamics identify the rate of work done on $S$ as the first term on the RHS in \eqref{energybal2} and heat flow into $S$ as the second term in \eqref{energybal2}; see for instance \cite{rezakhani,alicki1979,PhysRevE.75.051118}. Let us denote these by
\begin{align} 
    \dot{w}_{S|R_i} &= \Tr (\dot{H}_{S|R_i}^{\eff}(t) \rho_{S|R_i}(t)), \label{convenWQ1}\\
    \dot{q}_{S|R_i} &= \Tr (H_{S|R_i}^{\eff}(t) \dot{\rho}_{S|R_i}(t)). \label{convenWQ2}
\end{align}
Other studies have argued for the second term in \eqref{energybal3} to instead be understood as a work contribution, and not as heat \cite{weimer,Hossein_Nejad_2015,alipour2019}.~The main conceptual motivation behind this is that only the entropy changing terms should be associated with heat.~Since this second term in \eqref{energybal3} is associated with the unitary part of the subsystem dynamics, and is thus entropy preserving, it should be understood as a work contribution instead.~Let us denote this term by
\be \label{star}
\dot{E}^*_{S|R_i}(t) = -i \Tr\!\big(H_{S|R_i}^{\eff}(t)\big[H_{S|R_i} + \tilde{H}_{S|R_i}(t), \rho_{S|R_i}(t) \big]\big) .
\ee 
We can then write the first law in \eqref{energybal2}-\eqref{energybal3} equivalently as $ \dot{E}_{S|R_i}(t) = \dot{w}_{S|R_i}(t) + \dot{q}_{S|R_i}(t) = \dot{W}_{S|R_i}(t) + \dot{Q}_{S|R_i}(t) $, where $\dot{w}_{S|R_i},\dot{q}_{S|R_i}$ are given by \eqref{convenWQ1}-\eqref{convenWQ2}, and we have denoted
\be\label{eq:QWdef}
\begin{aligned}
 \dot{W}_{S|R_i} &= \dot{w}_{S|R_i} + \dot{E}^*_{S|R_i}, \\
 \dot{Q}_{S|R_i} &= \dot{q}_{S|R_i} - \dot{E}^*_{S|R_i}.
\end{aligned}
\ee

In some cases of interest, this term $\dot{E}^*_{S|R_i}(t)$ could vanish, so that no ambiguity in the identification of heat and work terms arises.~This occurs, for example, when (i) $\rho_{S|R_i}$ is a stationary state of $H_{S|R_i}^{\eff}$, or even of $H_{S|R_i} + \tilde{H}_{S|R_i}$, or (ii) $[H_{S|R_i} + \tilde{H}_{S|R_i}, H_{S|R_i}^{\eff}]  = 0$.~(ii) is naturally satisfied in a weak coupling limit, which is also when the conventional definitions are usually employed.~However, $\dot{E}^*_{S|R_i}(t)=0$ can also happen in a strong coupling, finite size regime.~Indeed, neither (i) nor (ii) imply a weak coupling, macroscopic limit.~For example,~(ii) can be satisfied trivially when $\mdsh_{S|R_i}(t) = \tilde{H}_{S|R_i}(t) + \alpha(t) \mathds1_S$ in a generic non-equilibrium, interacting setting as in \cite{rezakhani} (see App.~\ref{App:TDroleofcorrelations} for further details). 

Whether $\dot{E}^*_{S|R_i}$ is included as heat or work will of course change the amount of heat and work exchanged in a given process \cite{weimer,Hossein_Nejad_2015,alipour2019,ahmadi2019refined}.~However, here we are more interested in whether the heat and work contributions would differ in the two QRF perspectives, regardless of where this contribution is included, as long as this is done consistently in both perspectives. 

Coming now to the setup for exploring entropy exchanges (and thus the second law), let us consider the von Neumann entropy $S_{\rm vN}$ of the subsystem of interest in $R_i$-perspective 
\be \label{vnent} 
S_{\rm vN}[\rho_{S|R_i}(t)] = -\Tr(\rho_{S|R_i}(t) \log \rho_{S|R_i}(t)).
\ee 
The change in $S_{\rm vN}$ must be accounted for by any entropy that is generated internally within the subsystem and any entropy flow into/out of it.~That is,
\be \label{entbal2} 
\Delta S_{S|R_i} = \Sigma_{\ibar} - \Phi_{\ibar} \,,
\ee
where $\Delta S_{S|R_i} = S_{\rm vN}[\rho_{S|R_i}(t)] - S_{\rm vN}[\rho_{S|R_i}(0)]$ (for some initial time $t_0=0$), $\Sigma_{\ibar}$ is the entropy production and $\Phi_{\ibar}$ is the entropy flow out of $S$ \cite{RevModPhys.93.035008,lebon2008understanding,Esposito_2010,spohn,Breuer:2002pc,PhysRevLett.107.140404}. 
Entropy production characterises irreversibility in processes,\footnote{The second law of thermodynamics can be given by the non-negativity of entropy production $\Sigma \geq 0$, where equality holds only for reversible processes; see for instance \cite{RevModPhys.93.035008,lebon2008understanding}.\label{fn:2ndlaw}} and entropy flow encodes the reversible heat exchange between $S$ and its complement $R_j$ \cite{RevModPhys.93.035008,lebon2008understanding}.~Our notation with subscript $\ibar$ refers to the total system because thermodynamical processes  are inherently non-local, thus pertaining to the global system.~This will become clear shortly when discussing possible prescriptions for $\Sigma_{\ibar}$ and $\Phi_{\ibar}$.

The entropy balance equation \eqref{entbal2} is applicable to generic non-equilibrium situations, and at this level is fully general  \cite{RevModPhys.93.035008,lebon2008understanding}.~As with heat and work above, the differences arise in the precise microscopic quantum definitions for $\Sigma_{\ibar}$ and $\Phi_{\ibar}$ and their regimes of validity. For example, in classical thermodynamics, the entropy flow leaving $S$ can be approximated by Clausius' expression $\Phi_{\ibar} = \Delta q_{j}/T_j$, thus relating it to the heat flow entering its complement, ``the other frame''; and $\Sigma_{\ibar}$ is then defined via the entropy balance equation \eqref{entbal2}, see  \cite{RevModPhys.93.035008,lebon2008understanding} for a discussion.~By contrast, in quantum thermodynamics, entropy production can be defined as \cite{Esposito_2010,RevModPhys.93.035008}
\be \label{entprod}
\Sigma_{\ibar} = I[\rho_{\ibar}(t)] + S[\rho_{j}(t) || \rho_j(0)]
\ee
for an initial global product state $\rho_{\ibar}(0) = \rho_j (0) \otimes \rho_{S|R_{i}}(0)$.~Here, $I[\rho_{\ibar}] = S[\rho_{\ibar}|| \rho_{j} \otimes \rho_{S|R_i}]$ is the mutual information between $S$ and frame $R_j$ in $R_i$-perspective, and $S[\rho||\sigma]=\Tr(\rho\log\rho - \rho\log\sigma)$ is the quantum relative entropy.

The two terms on the RHS of \eqref{entprod} quantify the possible sources of irreversibility due to disregarding (tracing over) the complement of $S$.~The mutual information in the first term quantifies the non-local information shared between S and its complement (the ``environment'') which is lost if one no longer has access to the state of the latter.~The relative entropy in the second term quantifies how the state of the environment is changed by the dynamical interactions with $S$.~This is an irreversible process since any information contained locally in the state of the complement of $S$ is disregarded after partial tracing over it.

Entropy flow is then defined via the entropy balance equation \eqref{entbal2} and (for unitary global dynamics as here) is given by 
\be \label{entflow}
\Phi_{\ibar} =\Delta S_j+S[\rho_j(t)||\rho_j(0)]\,, \ee
with $\Delta S_j = S_{\rm vN}[\rho_j(t)]-S_{\rm vN}[\rho_j(0)]$ the change in von Neumann entropy of the other frame $R_j$ under dynamical evolution.

Like with heat and work, an unambiguous definition for entropy production in quantum thermodynamics is an open problem (and may not exist).~With our aim to take the first step towards inquiring about the QRF dependence of entropy production in mind, we shall consider here the proposal in \eqref{entprod} because it is a physically reasonable choice; for instance, in view of the second law (cfr.~footnote \ref{fn:2ndlaw}), it is non-negative  for generic ``other frame'' initial states, unlike the alternative proposal of open quantum systems theory \cite{Breuer:2002pc}, which is non-negative for stationary states \cite{Esposito_2010} only.\footnote{In fact, the former reduces to the latter under a weak coupling or high temperature limit \cite{Esposito_2010}, and for dynamical maps with global fixed points (e.g.~thermal operations) \cite{RevModPhys.93.035008}.}~Furthermore, when the ``other frame'' is initially in a thermal state $\rho_j(0) = \rho_j^\beta = e^{-\beta_j H_j}/Z$, then \eqref{entprod}-\eqref{entflow} reduce to their conventional forms with $\Phi_{\ibar}= \beta_j \Delta q_j$, where $\Delta q_j = \Tr(H_j(\rho_j(t)-\rho_j^\beta))$ \cite{RevModPhys.93.035008}.~We anticipate, however, that several of the qualitative observations below will also hold for other definitions of entropy production.

We note though that in the present context of QRFs, the proposal \eqref{entprod} is quite constraining due to its use of factorised global initial states.~In order to apply the same notion \eqref{entprod} of entropy production (and hence of entropy flow \eqref{entflow}) in both perspectives, the initial states must be restricted to product states which remain product states under QRF transformation.~As we shall discuss later in Sec.~\ref{Sec:TDproc2}, this restricts the allowed initial states to a small subset.~We stress however, that at some later time the state $\rho_{\ibar}(t)$ does not need to be of product form and can get entangled via the presence of interactions.~The fact that, for a generic product state in a given frame perspective, the proposals \eqref{entprod} and \eqref{entflow} may not even be simultaneously applicable in the other perspective is already a somewhat extreme form of QRF dependence of this notion of entropy production and flow.

\subsubsection{QRF-relative heat and work exchange}\label{Sec:TDproc1}

With the above setup at our disposal, we can now investigate the QRF-dependence of thermodynamic processes in terms of heat, work, and entropy exchanges between the subsystem $S$ and ``the other frame''.~Here we shall focus on heat and work, while entropy production and entropy flow are discussed in the next subsection.

 First, note that the different prescriptions for effective Hamiltonians obeying Definition~\ref{def_prescription} will generally give different answers for the energetics and thermodynamics of the subsystems not only in one perspective, but also, for what concerns us here, their QRF-dependence.~Indeed, as we discuss in App.~\ref{App:TDrates}, this is the case even for the simplest class of TPS-invariant Hamiltonians, whose local and non-local terms are individually TPS-invariant.~In that case, the prescription of \cite{rezakhani} leads to coinciding energetics in the two QRF perspectives, while this is not in general true for the one of \cite{weimer,Hossein_Nejad_2015}.~Below, however, we shall focus on situations where the different prescriptions obeying Definition~\ref{def_prescription} agree on the QRF-(in-)dependence.

Let us start by inquiring into the frame-dependence of the notion of isolated and closed thermodynamic system.~As in previous sections, our strategy will be to first characterise when the energy exchange (in terms of heat and work) between ``the other frame'' and $S$-subsystem appears the same or not in the two perspectives \emph{for all solutions}, before studying what happens for a given solution, when this is not true for all of them.

In Sec.~\ref{Sec:qrfdepopenclosed} we characterised the conditions for a dynamically closed subsystem in one perspective to remain closed or become open in the other perspective.~These have direct consequences for the thermodynamic description of the system relative to the two frames.~For dynamically closed subsystems in  $R_i$-perspective,~i.e.~when the Hamiltonian $H_{\ibar}$ is non-interacting, $H_{S|R_i}^{\eff}=H_{S|R_i}$, $H_{j}^{\eff}=H_{j}$, and $H_{jS}^{\eff}=0$ for any given prescription and any state.~In absence of interactions, neither heat nor work (for time independent Hamiltonians) is exchanged by the subsystems.~In other words, the system is thermodynamically isolated in the chosen frame perspective.~This is no longer true for all states and all times when interactions are present and the system becomes thermodynamically closed.~Lemma~\ref{lemma:closed2closed}
thus directly entails the following.
\begin{corollary}
[\textbf{Isolated-to-isolated~\&~isolated-to-closed}]\label{TD:iso2iso&iso2clos}
Suppose $H_{\ibar}=H_j\otimes\mathds1_S+\mathds1_j\otimes H_{S|R_i}$.~Then, the system is thermodynamically \emph{isolated} in both $R_i$- and $R_j$-perspective for \emph{any} given state, at \emph{any} time, i.e.
\be\label{eq:zeroqw}
\begin{aligned}
\dot{q}_{S|R_j}&=\dot{q}_{S|R_i}=0\;,\\ \dot{w}_{S|R_j}&=\dot{w}_{S|R_i}=0\;,\qquad\forall\,\rho_{\ibar}\,,\,t\\ \dot{E}^*_{S|R_j}&=\dot{E}^*_{S|R_i}=0\;,
\end{aligned}
\ee
and similarly for ``the other frame'', if and only if
\be\label{TD:HPidPit}
H_{\ibar}=H_j^{\mbd}\otimes\mathds1_S+\mathds1_j\otimes H_S^{\mbt}\,.
\ee
In other words, an \emph{isolated} thermodynamic system in $R_i$-perspective becomes \emph{closed} (but not isolated) in $R_j$-perspective if and only if $H_{\ibar}$ satisfies at least one of the following,
\be\label{TD:Hpidpitperp}
H_{j}^{\mbd\perp}\otimes\mathds1_S\neq0\,,\quad\mathds1_j\otimes H_{S}^{\mbt\perp}\neq0\,.
\ee
The quantities in \eqref{eq:zeroqw} are as defined in Sec.~\ref{Sec:TDsetup}.
\end{corollary}

Thus, when the total non-interacting Hamiltonian satisfies at least one of~\eqref{TD:Hpidpitperp}, depending on the state, there can be non-zero energy exchange among the subsystems in $R_j$-perspective, owing to the interactions generated in the Hamiltonian $H_{\jbar}$ by the QRF-transformation.
This is illustrated in Example~\ref{ex:TDisolatedvsclosed} below.

\begin{example}\label{ex:TDisolatedvsclosed}
\emph{\textbf{(Isolated-to-isolated~vs.~isolated-to-closed)}}~Consider three qubits $R_1R_2S$ and the two Hamiltonians 
\begin{align}
    &(1)\, H_{\bar 1}=\sigma_2^z\otimes\mathds1_S+\mathds1_2\otimes\sigma_S^x,\\
    &(2)\, H_{\bar 1}=\sigma_2^x\otimes\mathds1_S+\mathds1_2\otimes\sigma_S^x\label{eq2:hbar1}
\end{align}
of $R_2S$ relative to $R_1$.~$(1)$ is of the form \eqref{TD:HPidPit}, $(2)$ is not, as it satisfies the first condition in \eqref{TD:Hpidpitperp}.~In absence of interactions, $H_{S|R_1}^{\eff}=H_{S|R_1}=\sigma_S^x$ and $H_{2}^{\eff}=H_2=\sigma_2^z,\sigma_2^x$, respectively, for $(1),(2)$ in any given prescription.~In both cases, owing to the subsystem's unitary evolution and the time-independence of the effective Hamiltonians, neither heat nor work is exchanged between $R_2$ and $S$ for any given state $\rho_{\bar1}$, at any time.~

Setting the orientations before and after QRF transformations to $g_1=0$ and $g_2=0$, respectively, the Hamiltonian $H_{\bar2}=\hat{V}^{0,0}_{1\to2}(H_{\bar1})$ in $R_2$-perspective reads
\begin{align}
    &(1)\, H_{\bar 2}=\sigma_1^z\otimes\mathds1_S+\mathds1_1\otimes\sigma_S^x,\\
    &(2) \, H_{\bar 2}=\mathds1_1\otimes\sigma_S^x+\sigma_1^x\otimes\sigma_S^x.\label{eq2:H2bar}
\end{align}
In case $(1)$, $H_{\bar 2}$ remains non-interacting and there is no heat and work exchange between $R_1$ and $S$ for any $\rho_{\bar2}$, at any time $t$.~In case $(2)$, there are $R_1S$ interactions, and heat and work exchanges are not zero for any $\rho_{\bar2}$ and any $t$.~Indeed, let us consider the following initial states
\begin{align}
\ket{\psi(0)}_{\bar2}&=\ket{x_+,x_+}_{\bar2}\quad\text{with}\quad\sigma^x\ket{x_+}=\ket{x_+}\;,\label{eq2:state1}\\
\ket{\varphi(0)}_{\bar2}&=\frac{1}{\sqrt{2}}(\ket{00}_{\bar2}+\ket{11}_{\bar2}),\label{eq2:state2}
\end{align}
and the prescription of effective Hamiltonians \cite{rezakhani}
\begin{align}
H_{S|R_2}^{\eff}(t)&=H_{S|R_2}+\tilde{H}_{S|R_2}(t)\nonumber\\
&-\alpha_S\Tr_{\bar2}(H_{1S}(\rho_1(t)\otimes\rho_{S|R_2}(t))\cdot\mathds1_S,\label{eq2:heff}
\end{align}
where $\tilde{H}_{S|R_2}(t)$ is defined as in \eqref{eq:HStilde}, and $\alpha_S$ is a real parameter such that $\alpha_S+\alpha_1=1$, with $\alpha_1$ the parameter in the analogous expression for $H_{1}^{\eff}$.~In this prescription $\dot{E}_{S|R_2}^*=0$ (cfr.~\eqref{star}), and the definitions \eqref{convenWQ1},~\eqref{convenWQ2} and \eqref{eq:QWdef} of heat and work rates coincide, respectively.

Clearly,~\eqref{eq2:state1} is an eigenstate of \eqref{eq2:H2bar} and $\rho_{S|R_2}(t)=\ket{x_+}\!\bra{x_+}=\rho_{1}(t)$, for all $t$.~Thus, for the state \eqref{eq2:state1} and Hamiltonian \eqref{eq2:H2bar}, Eq.~\eqref{eq2:heff} yields $H_{S|R_2}^{\eff}=2\sigma_S^x-\alpha_S\mathds1_S$ and $\dot{q}_{S|R_2}=\dot{w}_{S|R_2}=0$ at any $t$ (cfr.~Eqs.~\eqref{convenWQ1},~\eqref{convenWQ2}).~For the state \eqref{eq2:state2}, we have
\be\label{state2:results}
\begin{aligned}
\rho_{1}(t)&=\frac{1}{2}\Bigl[\ket{0}\!\bra{0}_1+\ket{1}\!\bra{1}_1+\\
&\quad+\sin(2t)\sin(t)\bigl(e^{it}\ket{0}\!\bra{1}_1+e^{-it}\ket{1}\!\bra{0}_1\bigr)\Bigr]\,,\\
\rho_{S|R_2}&=\frac{1}{2}\Bigl[\ket{0}\!\bra{0}_S+\ket{1}\!\bra{1}_S+\\
&\quad+\sin(2t)\sin(t)\bigl(e^{-it}\ket{0}\!\bra{1}_S+e^{it}\ket{1}\!\bra{0}_S\bigr)\Bigr]\,,\\
H_{S|R_2}^{\eff}&=\left(1+\frac{1}{2}\sin^2(2t)\right)\sigma_S^x-\frac{\alpha_S}{4}\sin^2(2t)\mathds1_S\,,\\
\dot{q}_{S|R_2}&=\sin(2t)\bigl[2-\sin(t)(1+3\sin^2(t))\bigr]\,,\\
\dot{w}_{S|R_2}&=\sin(4t)\sin^2(2t)\left(\frac{1}{2}-\alpha_S\right)\,,
\end{aligned}
\ee
and the heat and work rates do not vanish at all times.

The above considerations hold true also with the effective Hamiltonian prescription of \cite{weimer,Hossein_Nejad_2015}, where $H_{S|R_2}^{\eff}(t)=H_{S|R_2}+\mdsh_{S|R_2}(t)$ with $\mdsh_{S|R_2}(t)$ the part of $\tilde{H}_{S|R_2}(t)$ that commutes with $H_{S|R_2}$.~In fact,  for $H_{\bar2}$ in \eqref{eq2:H2bar}, we have $H_{S|R_2}=\sigma_S^x$ and, by Eq.~\eqref{eq:HStilde}, $\tilde{H}_{S|R_2}(t)=\sigma_S^x$ for $\rho_{\bar2}$ in \eqref{eq2:state1} and $\tilde{H}_{S|R_2}(t)=\frac{1}{2}\sin^2(2t)\sigma_S^x$ for $\rho_{\bar2}$ in \eqref{eq2:state2}.~In either case, $\mdsh_{S|R_2}(t)=\tilde{H}_{S|R_2}(t)$ and the results for the given state in this prescription are obtained by setting $\alpha_S=0$ in the expressions above.
\end{example}

For a given state, the heat and work exchanged between the subsystems can thus differ in the two perspectives.~Let us now instead seek conditions, again for a given dynamical trajectory, such that the energetics agree in both perspectives, \emph{regardless} of whether interactions are present.~This brings us back to the discussion of Sec.~\ref{Sec:Dyn}, which highlights the central role of TPS-invariance of the global state trajectory for the QRF-independence of the subsystem evolution and correlations.~Lemma~\ref{lemma:finitetimeinvsoln} clarifies how this depends on the interplay of initial state and Hamiltonian.~In particular, we recall that for $\rho_{\ibar}(t)\in\mcA_{\ibar}^{X}$ in some time interval $[t_0,t_1]$, the imported Hamiltonian $H_{j\to i}=\hat X\hat{\mathcal I}_{j\to i}(H_{\jbar})$ from $R_j$-perspective can equivalently be used as the generator of dynamics for the given state in the interval $[t_0,t_1]$ in $R_i$-perspective.

In fact, in this time interval, $H_{j\to i}$ can be \emph{also} equivalently used to define the total energy of the system for the states in question.~Indeed, at any $t$ for which $\rho_{\ibar}(t)\in\mcA_{\ibar}^{X}$, we have
\begin{align}\label{energyUXinvstate}
E_{\ibar}(t)&=\Tr_{\ibar}\left(H_{\ibar}\,\rho_{\ibar}(t)\right)\nonumber\\
&=\Tr_{\ibar}\bigl(\PiUX(H_{\ibar})\,\rho_{\ibar}(t)\bigr)\nonumber\\
&=\Tr_{\ibar}\bigl(\PiUX(H_{j\to i})\,\rho_{\ibar}(t)\bigr)\nonumber\\
&=\Tr_{\ibar}\left(H_{j\to i}\,\rho_{\ibar}(t)\right)\;,
\end{align}
where we used $\rho_{\ibar}=\PiUX(\rho_{\ibar})$ for $\rho_{\ibar}\in\mcA_{\ibar}^{X}$, with $\PiUX$ the projector onto $\mcA_{\ibar}^{X}$, Hilbert-Schmidt orthogonality of the decomposition \eqref{AUXdecofAibar} of $\mcA_{\ibar}$, and the property \eqref{eq:hagree} that $\PiUX(H_{j\to i})=\PiUX(H_{\ibar})$.~This also holds for fluctuations of the total energy.

As a consequence, since effective Hamiltonian prescriptions modify the individual contributions to the total energy, while keeping the latter invariant, one could alternatively define the energetics in $R_i$-perspective using the imported Hamiltonian $H_{j\to i}$, rather than $H_{\ibar}$ (while still obeying Definition~\ref{def_prescription} with $H_{j\to i}$ instead).~Note, however, that the individual effective energy contributions using $H_{j\to i}$ or $H_{\ibar}$ will not agree in general, even for the same prescription.\footnote{They might agree though for certain states, as illustrated in Example~\ref{ex:stateQRFdepqw}.}~Indeed, as we have seen in Sec.~\ref{subsec:globaldyn}, $H_{j\to i}$ and $H_{\ibar}$ can be very different operators.~In particular, due to the QRF-transformation, new interactions can be present in $H_{\jbar}$ and thereby $H_{j\to i}$.~These do not contribute to the effective energies derived from $H_{\ibar}$, but will generically contribute to those derived from $H_{j\to i}$.

As one may expect, the following result, proven in App.~\ref{App:TDrates}, establishes that heat and work agree in the two QRF perspectives for a global TPS-invariant trajectory, provided one defines the effective Hamiltonian in one perspective with the bare Hamiltonian of that perspective and in the other with the correspondingly imported one (see Fig.~\ref{Fig:invTD} for an illustration).

\begin{lemma} \label{lemma:invTDrates}
Let $\rho_{\ibar}(t)\in\mathcal{A}_{\ibar}^X$ be a TPS-invariant trajectory for some interval $[t_0,t_1]$.~Choose \emph{any} prescription of effective Hamiltonians obeying Definition~\ref{def_prescription} and apply it to both $H_{\jbar}$ in $R_j$-perspective and to the imported Hamiltonian $H_{j\to i}$ in $R_i$-perspective ($i\neq j$) for defining the energetics.~Then, over $[t_0,t_1]$, the heat flow into and work done on $S$ are equal in both perspectives relative to this prescription.~This is \emph{independent} of whether work and heat are defined via~\eqref{convenWQ1} \&~\eqref{convenWQ2} or via~\eqref{eq:QWdef}.~That is, for any $t\in[t_0,t_1]$,
    \be\label{sysqts}
    \dot{q}_{S|R_j}=\dot{q}_{S|R_i}, \quad \dot{w}_{S|R_j}=\dot{w}_{S|R_i}, \quad \dot{E}^*_{S|R_j}=\dot{E}^*_{S|R_i},
    \ee
and similarly, for ``the other frame'' subsystem
    \be\label{frameqts}
    \dot q_{i}=\dot q_{j}, \;\; \dot{w}_{i}=\dot{w}_{j}, \quad \dot{E}^*_{i}=\dot{E}^*_{j}.
    \ee
This holds for any frame orientations $g_i,g_j$ before and after QRF transformation, respectively.
\end{lemma}
\begin{figure*}[t!]
\centering\includegraphics[height=6.5cm,width=0.9\textwidth]{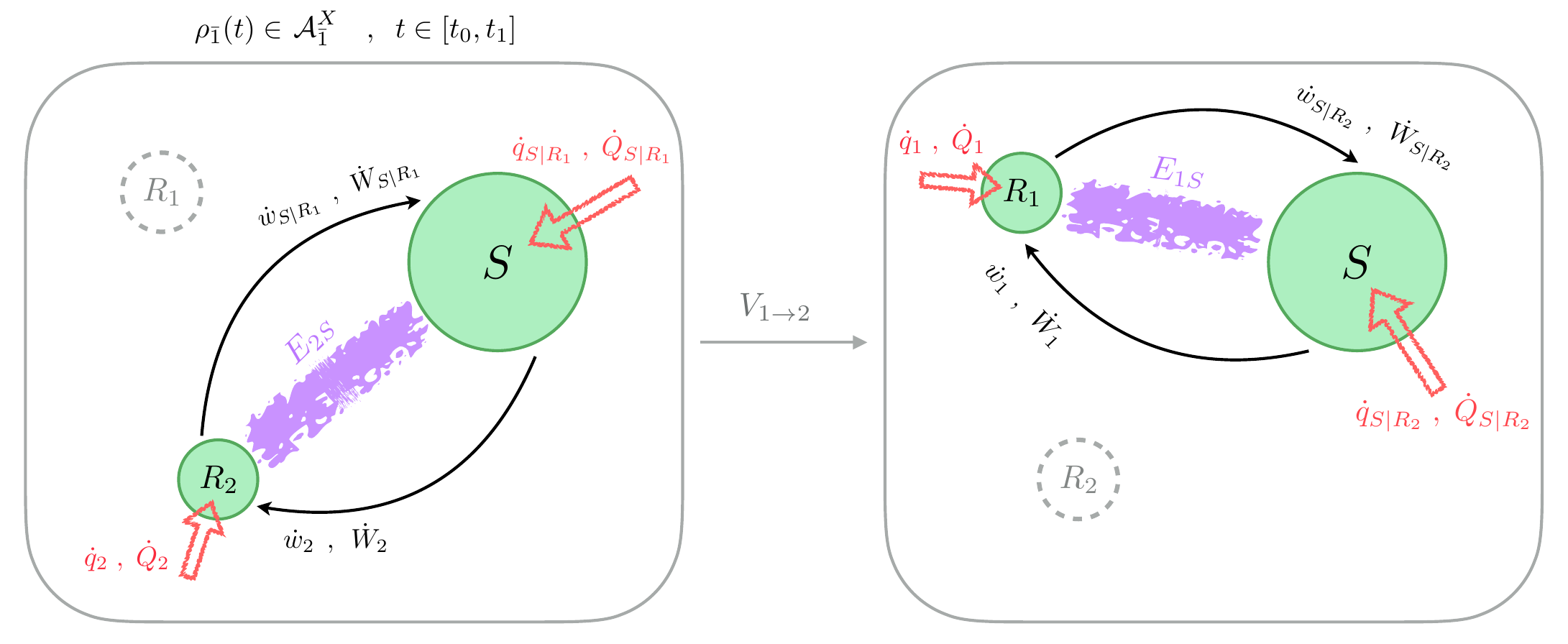}
\caption{Pictorial representation of Lemma~\ref{lemma:invTDrates}.~When the state $\rho_{\bar1}(t)$ of $R_2S$ relative to $R_1$ lies in a TPS-invariant subalgebra $\mcA_{\bar1}^{X}$ during some time interval $[t_0,t_1]$, the energetics of the total system is the same relative to the frames $R_1$ and $R_2$ (provided one uses the imported Hamiltonian in one perspective).~Specifically, the heat and work exchanged between ``the other frame'' and S-subsystem, and the energy stored in correlations, are the same in the perspectives of both frames during the interval $[t_0,t_1]$, independently of the prescription used to define them.}
\label{Fig:invTD}
\end{figure*}

Importantly, however, under the same premises, the conclusion will generally not hold, if we replace the imported Hamiltonian $H_{j\to i}$ by $H_{\ibar}$ in $R_i$-perspective. In other words, even for TPS-invariant global trajectories, the QRF perspectives will not in general agree on the energetics if both use the corresponding bare Hamiltonian in their perspective for defining effective Hamiltonians (even if both otherwise adhere to the same prescription).

Clearly, the precise form of the thermodynamic quantities involved and their magnitude in a given perspective depend on the specific prescription.~However, the above result holds for any given prescription of effective Hamiltonians obeying Definition~\ref{def_prescription} and is independent of the specific definition of heat and work rates (\eqref{convenWQ1},~\eqref{convenWQ2} or \eqref{eq:QWdef}).~Moreover, Lemma~\ref{lemma:invTDrates} does not require us to restrict to a specific form either of the initial state (separable or not) or the Hamiltonian (e.g.~weak coupling).~Note that, to obtain the same result for mixtures of TPS-invariant pure states, one would have to adapt the imported Hamiltonian to the bilocal unitary of each such pure state.~This would be rather cumbersome and may not be desirable in general.

A generic state will not meet the conditions of Lemma~\ref{lemma:invTDrates}, in which case, by Lemma~\ref{lemma:finitetimeinvsoln}, the imported Hamiltonian cannot be used for defining the energetics and time evolution, and one has to resort to the true Hamiltonian instead.~In that case, we expect that the heat and work exchange between ``the other frame'' and $S$-subsystem will generically differ in the two perspectives, but can still agree in special cases.~The following example illustrates this observation.
\begin{example}\label{ex:stateQRFdepqw}
    Consider again case $(2)$ of Example~\ref{ex:TDisolatedvsclosed} and the initial states \eqref{eq2:state1} and \eqref{eq2:state2} in $R_2$-perspective.~As for \eqref{eq2:state1}, we note that $\mbfU_{\bar2}^{0,0}\ket{x_+,x_+}_{\bar2}=\ket{x_+,x_+}_{\bar2}$,~i.e.~$\rho_{\bar2}(0)\in\mcA_{\bar2}^{\mathds1}$.~Moreover, the state is an eigenstate of both $H_{\bar2}$ in \eqref{eq2:H2bar} and $H_{1\to2}=\hat{\mathcal I}_{1\to2}(H_{\bar1})$ with $H_{\bar1}$ given in \eqref{eq2:hbar1} and $\mathcal I_{1\to2}=(\ket0_1\otimes\bra0_2+\ket1_1\otimes\bra1_2)\otimes\mathds1_S$.~Thus, $\rho_{\bar2}(t)=\ket{x_+,x_+}\!\bra{x_+,x_+}_{\bar2}\in\mcA_{\bar2}^{\mathds1}$ for any $t$.~Compatibly with Lemma~\ref{lemma:invTDrates}, for such a state, heat and work are zero in both $R_1$- and $R_2$-perspective, regardless of the prescription to compute them and, in this specific case, also of whether we use the imported or the true Hamiltonian.

    The state \eqref{eq2:state2} instead does not lie in any TPS-invariant subalgebra at any time, so that we cannot invoke the imported Hamiltonian.~As discussed earlier in Example~\ref{ex:TDisolatedvsclosed}, heat and work are zero at all times for the QRF-transformed state $\ket{\varphi(t)}_{\bar1}=V_{2\to1}^{0,0}\ket{\varphi(t)}_{\bar2}$ in $R_1$-perspective (isolated system).~In $R_2$-perspective, the thermodynamic system is closed and the heat and work rates for the given state $\ket{\varphi(t)}_{\bar2}$ are not zero at all $t$ (cfr.~Eqs.~\eqref{state2:results}).~They both vanish however at certain times, namely at $t=n\pi/2$, $n\in\mathbb N\cup\{0\}$.~In other words, even though the system is thermodynamically closed in the given perspective, it can still be \emph{effectively isolated} for some states and at some time.
\end{example}
We leave open what the necessary conditions are for heat and work exchange to coincide in the two perspectives and move our discussion to entropy production and entropy flow.

\subsubsection{QRF-relative entropy production and entropy flow}\label{Sec:TDproc2}

As emphasised at the end of Sec.~\ref{Sec:TDsetup}, for the notions of entropy production \eqref{entprod} and entropy flow \eqref{entflow} to be applicable in both perspectives, the initial global states must be restricted to those product states which remain product under QRF-transformation.~For pure states, Theorem~\ref{claim:purestatesrhoSZS} and its Cor.~\ref{claim:Renyipure} inform us that a necessary and sufficient condition for this to be the case is $\rho_{\ibar}(0)=\rho_j(0)\otimes \rho_{S|R_i}(0)\in\mcA_{\ibar}^{X}$,
for some bilocal unitary $X=Y_j\otimes Z_S$.~For mixed states, by linearity of the QRF-transformation and Lemma~\ref{Lemma:tpseq}, a sufficient condition is when $\rho_{\ibar}(0)$ belongs to a TPS-invariant subalgebra and/or it admits a pure state decomposition $\rho_{\ibar}(0)=\sum_{A,B}p_Aq_B\rho_{j}^A(0)\otimes\rho_{S|R_i}^B(0)$ with $\rho_{j}^A(0)\otimes\rho_{S|R_i}^B(0)\in\mcA_{\ibar}^{X^{AB}}$, for possibly distinct $X^{AB}=Y_j^A\otimes Z_S^B$.~Either way, this restricts the allowed initial states to a small subset.~The fact that the same notion of entropy production and entropy flow is generically not applicable in both perspectives simultaneously signals already a somewhat extreme form of the QRF-relativity of these quantities.

In this last subsection, we now aim to take first steps to investigate the QRF-dependence of subsystem entropy exchanges when the definitions \emph{do} apply to both perspectives simultaneously.~In what follows we shall thus restrict to initial states which are of product form in both perspectives.

Similarly to what was done with heat and work in Sec.~\ref{Sec:TDproc1}, let us start by asking when, given that there is no entropy production within the subsystem $S$ and no entropy flow out of it in one perspective, these quantities remain zero also in the other perspective.~A characterisation for pure global states is given by Lemma~\ref{lemma:zerosigmaphi} below which also provides us with sufficient conditions for mixed global states.~We refer to App.~\ref{App:TDrates} for the proof.

\begin{lemma}[\textbf{No entropy production \& flow}]\label{lemma:zerosigmaphi}
For pure global states $\rho_{\ibar}$, the entropy production and entropy flow of the subsystem $S$ at time $t=t_1>0$ are zero in both perspectives (for any given frame orientations $g_i,g_j\in\mathcal G$)
    \be\label{eq:noentex}
\Sigma_{\ibar}(t_1)=0=\Sigma_{\jbar}(t_1)\quad,\quad\Phi_{\ibar}(t_1)=0=\Phi_{\jbar}(t_1)\,,
    \ee
    if and only if, at the initial $t_0=0$ and the later time $t_1$
\be\label{cond1:zerosigmaphi}
    \rho_{\ibar}(t_a)=\rho_j(0)\otimes\rho_{S|R_i}(t_a)\in\mcA_{\ibar}^{X^{a}}\,,
    \ee
    for some bilocal unitary $X^a=Y_j^0\otimes Z_S^a$, $a=0,1$.
    For $\rho_{\ibar}$ mixed,~\eqref{eq:noentex} holds true if $\rho_{\ibar}(t_a)$, $a=0,1$, satisfy at least one of the following:\\

\noindent
    $\bullet$ they obey~\eqref{cond1:zerosigmaphi},\\

\noindent    $\bullet$ they admit pure state decompositions of the form
\be\label{MS2:zerosigmaphi}
    \rho_{\ibar}(t_a)=\sum_{A,B}p_A\,q_B^a\,\rho_{j}^A(0)\otimes\rho_{S|R_i}^{B}(t_a)\;,
\ee
with $\rho_{j}^A(0)\otimes\rho_{S|R_i}^B(t_a)\in\mcA_{\ibar}^{X^{AB}}$, for possibly distinct $X^{AB}=Y_j^{A,0}\otimes Z_S^{B,a}$.
\end{lemma}

When the conditions of Lemma~\ref{lemma:zerosigmaphi} are satisfied, the entropy variation of both ``the other frame'' and $S$ subsystems are zero in the two perspectives, i.e.~$\Delta S_{S|R_i}=0=\Delta S_{S|R_j}$ and $\Delta S_{j}=0=\Delta S_{i}$, as is clear from the entropy balance relation \eqref{entbal2} and definition \eqref{entflow} of entropy flow.

Conditions \eqref{cond1:zerosigmaphi}--\eqref{MS2:zerosigmaphi} tell us that, for the entropy production and entropy flow of the subsystem $S$ at time $t$ to be zero relative to both frames, not only must the global state at time $t$ be a product state in both perspectives, but also ``the other frame'' reduced states at time $t$ must coincide with their corresponding initial states.\footnote{Rewriting \eqref{entprod} as $\Sigma_{\ibar}=S[\rho_{\ibar}(t)\|\rho_{j}(0)\otimes\rho_{S|R_i}(t)]$, and similarly for $\Sigma_{\jbar}$, the well-known property that $S[\rho||\sigma]=0$ if and only if $\rho=\sigma$ tells us that $\rho_{\ibar}(t)=\rho_{j}(0)\otimes\rho_{S|R_i}(t)$ and $\rho_{\jbar}(t)=\rho_{i}(0)\otimes\rho_{S|R_j}(t)$ give necessary and sufficient conditions for the entropy production to be zero in both perspectives also with global mixed states.~Whether the latter condition necessarily implies that $\rho_{\ibar}(t)$ resides in a TPS-invariant subalgebra or can be decomposed into pure states that do (so being equivalent to~\eqref{cond1:zerosigmaphi} or~\eqref{MS2:zerosigmaphi}), is currently not clear to us.~It does imply this, however, in certain cases, e.g.\ when $\rho_{\ibar}(t)$ has non-degenerate spectrum.}~Clearly, these are rather restrictive conditions for the states and their dynamics in view of the discussion in Secs.~\ref{Sec:states&entanglement} and \ref{Sec:Dyn}, showing that the structure of correlations and interactions can be significantly affected by the QRF transformation.~Thus, for a generic Hamiltonian and initial global state that is product in both perspectives, even though there is no entropy produced within and flowing out of the subsystem $S$ in one perspective, there will be nonzero entropy production and/or entropy flow in the other perspective.~This is illustrated in the following example. 
\begin{example}{\bf(Zero-to-nonzero~entropy~production)}\label{ex1:sigmaphi}
Consider three qubits $R_1R_2S$, the Hamiltonian $H_{\bar1}=\sigma_2^z\otimes\mathds1_S+\mathds1_2\otimes\sigma_S^z$ of $R_2S$ relative to $R_1$ (in orientation $g_1=0$), and the initial state
\be\label{eq:ex1rho1barin}
    \rho_{\bar1}(0)=\frac{e^{-\beta\sigma_2^z}}{Z(\beta)}\otimes\ket{x_+}\!\bra{x_+}_S\;,
\ee
where $Z(\beta)=2\cosh(\beta)$, and $\sigma^x\ket{x_\pm}=\pm\ket{x_\pm}$.~Note that $\rho_{\bar1}(0)=\Pid\Pit(\rho_{\bar1}(0))$ so that, by Lemma~\ref{lemma:pinpd-uin}, $\rho_{\bar1}(0)\in\mcA_{\bar1}^{\mathds1}$ (for $g_2=0$ also).~Thus, the initial state $\rho_{\bar2}(0)$ in $R_2$-perspective is also a product state as in \eqref{eq:ex1rho1barin} with the frame labels $1$ and $2$ exchanged, and the definitions \eqref{entprod},~\eqref{entflow} of $\Sigma$,~$\Phi$ can be applied in both perspectives.

The state $\rho_{\bar1}(t)$ at time $t$ is given by
\be\label{eq:ex1rho1bart}
\rho_{\bar1}(t)=\frac{e^{-\beta\sigma_2^z}}{Z(\beta)}\otimes\ket{\varphi(t)}\!\bra{\varphi(t)}_{S|R_1}\,,
\ee
with $\ket{\varphi(t)}=\cos(t)\ket{x_+}-i\sin(t)\ket{x_-}$ and hence is of the form $\rho_{\bar1}(t)=\rho_2(0)\otimes\rho_{S|R_1}(t)$.~Thus, at any $t$, there is no entropy produced within and flowing out of the subsystem $S$, $\Sigma_{\bar1}=0=\Phi_{\bar1}$, as well as no entropy change $\Delta S_2=0=\Delta S_{S|R_1}$ between $t=0$ and any $t>0$.

In $R_2$-perspective, the Hamiltonian becomes interacting, $H_{\bar2}=\hat{V}_{1\to2}^{0,0}(H_{\bar1})=\sigma_1^z\otimes\mathds1_S+\sigma_1^z\otimes\sigma_S^z$, and the state at time $t$ is given by
\begin{align}
\rho_{\bar2}(t)=&\frac{e^{-\beta}}{Z(\beta)}\ket{0}\!\bra{0}_1\otimes\ket{\varphi(t)}\!\bra{\varphi(t)}_{S|R_2}\nonumber\\
&+\frac{e^{+\beta}}{Z(\beta)}\ket{1}\!\bra{1}_1\otimes\sigma_S^x\ket{\varphi(t)}\!\bra{\varphi(t)}_{S|R_2}\sigma_S^x\,.\label{eq:ex1rho2bart}
\end{align}
Thus, $\rho_{1}(t)=\Tr_S(\rho_{\bar2}(t))=\frac{e^{-\beta\sigma_1^z}}{Z(\beta)}=\rho_{1}(0)$, so that $\Phi_{\bar2}=0$ and $\Delta S_1=0$, similarly to $R_1$-perspective.~However, from the entropy balance equation \eqref{entbal2}, we have $\Sigma_{\bar2}=\Delta S_{S|R_2}=S_{\rm vN}[\rho_{S|R_2}(t)]$, where we used the fact that $\rho_{S|R_2}(0)=\ket{x_+}\!\bra{x_+}_S$ is pure.~But, $\Tr_S\left[\rho_{S|R_2}(t)^2\right]=\frac{1}{2}\left[1+\cos^2(2t)+\tanh^2(\beta)\sin^2(2t)\right]$, which is equal to 1 only for $t=\frac{n\pi}{2}$, $n\in\mathbb N\cup\{0\}$, i.e.~the state $\rho_{S|R_2}(t)$ is pure at $t=\frac{n\pi}{2}$ and mixed otherwise.~Thus, $\Sigma_{\bar2}$ and $\Delta S_{S|R_2}$ vanish for $t=\frac{n\pi}{2}$, and are nonzero otherwise.

In line with Lemma~\ref{lemma:zerosigmaphi}, the global state \eqref{eq:ex1rho2bart} is separable, but not a product state for any $t\neq\frac{n\pi}{2}$.~Thus, at those $t$, the product state $\rho_{\bar1}(t)$ in \eqref{eq:ex1rho1bart} neither belongs to any TPS-invariant subalgebra $\mcA_{\bar1}^X$, nor does it admit a decomposition as in \eqref{MS2:zerosigmaphi}.~It admits a decomposition into TPS-invariant pure states, as can be seen by using the spectral decomposition of $\rho_2$ in \eqref{eq:ex1rho1bart}, but this is not of the form \eqref{MS2:zerosigmaphi}.~Instead, we have $\rho_{\bar1}(t)=\sum_{A=0,1}p_A\ket{A}\!\bra{A}_2\otimes\ket{\varphi(t)}\!\bra{\varphi(t)}_S$ with $\ket{A}\!\bra{A}_2\otimes\ket{\varphi(t)}\!\bra{\varphi(t)}_S\in\mcA_{\bar1}^{X^A}$, $X^A=\mathds1_2\otimes Z_S^A$, where $Z_S^0=\mathds1_S$ and $Z_S^1=\sigma_S^x$.~It is only at $t=\frac{n\pi}{2}$ that $\rho_{\bar1}(t)=\Pid\Pit(\rho_{\bar1}(t))\in\mcA_{\bar1}^{\mathds1}$ (cfr.~\eqref{eq:ex1rho1bart}), also $\ket{A}\!\bra{A}_2\otimes\ket{\varphi(\frac{n\pi}{2})}\!\bra{\varphi(\frac{n\pi}{2})}_S\in\mcA_{\bar1}^{\mathds1}$, $A=0,1$.~Thus, both conditions \eqref{cond1:zerosigmaphi} and \eqref{MS2:zerosigmaphi} are satisfied for $t=\frac{n\pi}{2}$ and at these times there is no entropy generation for the subsystem $S$ in both perspectives.~This demonstrates that the decomposability of a state into TPS-invariant pure states is not in general sufficient to guarantee coincidence of the entropy production (and flow) in the two perspectives.
\end{example}

Let us briefly comment on the last observation, namely that a decomposition into pure TPS-invariant states is not sufficient for the two QRFs to agree on entropy production and flow.~While we saw in Sec.~\ref{Sec:states&entanglement} that, for the entropies associated with subsystem $S$ to be the same in the two perspectives, it is sufficient that the global state is decomposable into TPS-invariant pure states, this was only true when these pure states all agreed on the $S$-unitary $Z_S$ in $X^A=Y_j^A\otimes Z_S$.~The TPS-invariant decomposition of $\rho_{\ibar}(t)$ in the example above already violates this condition.~More generally, entropy production \eqref{entprod} and entropy flow \eqref{entflow} involve the entropies of \emph{both} $S$ and ``the other frame'', so that decompositions into TPS-invariant pure states will have to get restricted even further in order to admit a coinciding entropy balance in both perspectives.

More generally, for a given state in one perspective, we can now ask when the entire entropy balance (subsystem entropy variation, as well as entropy production and flow), which need not vanish, is the same in the perspective of the two frames.~The following simple observation, proven in App.~\ref{App:TDrates}, again emphasises the relevance of TPS-invariance of initial and final states, similar to Lemma~\ref{lemma:zerosigmaphi}.

\begin{lemma}[\textbf{QRF-invariance of entropy balance}]\label{lemma:inventbal}
Suppose the global state is TPS-invariant at $t_0$ and $t_1>t_0$, i.e.\ $\rho_{\ibar}(t_a)\in\mcA_{\ibar}^{X^a}$, $a=0,1$, for some bilocal unitaries $X^a=Y_j^a\otimes Z_S^a$.~Then, between these times,
\begin{itemize}
    \item the change in the entropy of ``the other frame'' and $S$ subsystems are the same in both perspectives, i.e.
\be\label{eq:invDeltaS}
\Delta S_{i}=\Delta S_{j}\quad,\quad\Delta S_{S|R_j}=\Delta S_{S|R_i}\,.
\ee
When $\rho_{\ibar}(t_0)$ is pure, we have an additional equality, $\Delta S_j=\Delta S_{S|R_i}$, and so all entropy changes are equal.

\item the entropy production in the subsystem $S$ and the entropy flow out of $S$ are equal in both perspectives,~i.e.
\be\label{eq:invSigmaPhi}
\Sigma_{\jbar}=\Sigma_{\ibar}\quad,\quad \Phi_{\jbar}=\Phi_{\ibar}
\ee
if and only if the initial state is of product form and
\be\label{cond01:invSigmaPhi}
S[\hat{Y}_j^{01}(\rho_j(t_1))\|\rho_j(t_0)]=S[\rho_j(t_1)\|\rho_j(t_0)]\,,
\ee
with $Y_j^{01}=Y_j^0Y_j^{1\dagger}$.
\end{itemize}
\end{lemma}

Condition \eqref{cond01:invSigmaPhi} tells us that, for the entropy produced within and flowing out of the subsystem $S$ to be the same relative to the two frames, ``the other frame'' state $\rho_j(t_1)$ at the later time $t_1$ is as distinguishable from its initial state $\rho_j(t_0)$ as the unitarily related state $\hat{Y}_j^{01}(\rho_j(t_1))$.~This is the case when e.g.~$Y_j^{01}$ belongs to the stabiliser group of $\rho_j(t_1)$.~A special case of this is when $\rho_j(t_1)$ and $\rho_j(t_0)$ share the same subalgebra $\mcA_{\ibar}^X$ and so $Y_j^{01}=\mathds1_j$.~These are non-trivial constraints and thus, even for initial and final global states that are TPS-invariant, entropy production and flow will not in general coincide in the two perspectives.~In particular, this encompasses the situation in which the system is in a non-equilibrium steady-state (NESS) in $R_i$-perspective, but not in $R_j$-perspective.~This is the case, for instance, when \eqref{eq:invDeltaS} holds with zero entropy variation, but $\Sigma_{\ibar}=\Phi_{\ibar}\neq0$, while $\Sigma_{\jbar}=\Phi_{\jbar}=0$.~In this case, \eqref{cond01:invSigmaPhi} is not satisfied.

The conditions of Lemma~\ref{lemma:inventbal} are satisfied, e.g.\ by an initial product state whose trajectory lies for the entire interval $[t_0,t_1]$ in a fixed TPS-invariant subalgebra, a case encompassed by Lemma~\ref{lemma:finitetimeinvsoln}.~It is also clear that the pure state case of Lemma~\ref{lemma:zerosigmaphi} satisfies the conditions of Lemma~\ref{lemma:inventbal}.~In particular, $\rho_j(t_1)=\rho_j(t_0)$ (cfr.~\eqref{cond1:zerosigmaphi}) implies that \eqref{cond01:invSigmaPhi} is automatically satisfied.~We emphasise that, for the equality \eqref{eq:invSigmaPhi} to hold for nonzero values, the state at later times need not be a product state.

Example~\ref{ex2:entbalance} in App.~\ref{App:TDrates} illustrates Lemma~\ref{lemma:inventbal} with a pure state trajectory that oscillates between two distinct TPS-invariant subalgebras and is otherwise not TPS-invariant.~At the discrete times when the trajectory is TPS-invariant,~\eqref{cond01:invSigmaPhi} is obeyed and the entropy balance in the two QRF perspectives coincides (including the case when both entropy production and flow are non-vanishing), while it disagrees for the intervals in-between.

This concludes our initial exploration of the QRF relativity of quantum thermodynamics.~Our results establish conditions for thermodynamic processes to appear the same in distinct QRF perspectives and moreover demonstrate that generally their description depends on the choice of QRF.~This is rooted in subsystem relativity and the ensuing QRF relativity of correlations and interactions.~We leave open what the precise necessary and sufficient conditions are for different QRFs to agree or disagree on the description of quantum thermodynamic processes.~In view of the fact that this necessarily involves non-equilibrium processes and generic mixed states, identifying these is a challenging task. 

\section{Discussion and outlook}\label{conc}

In this work, we expanded on the quantum frame relativity of subsystems\,---\,the fact that distinct internal QRFs partition the total system in distinct ways into subsystems\,---\,and initiated a systematic analysis of its physical consequences.~While this subsystem relativity was first noticed for (both ideal and non-ideal) QRFs associated with abelian groups in \cite{PHquantrel}, before being generalised to ideal QRFs for arbitrary unimodular Lie groups in \cite{delaHamette:2021oex,Castro-Ruiz:2021vnq}, here we refined and simplified its explication both technically and conceptually.~We unraveled novel algebraic structures to characterise when operators such as states and observables remain invariant under QRF changes, including when they are invariant up to local subsystem unitaries.~We then utilised such tools to launch the first investigation into the QRF covariance of subsystem states, entropies, dynamics and thermodynamics, and to significantly deepen previous studies on the QRF relativity of correlations.

\subsection{Summary}

\noindent
Our main results can be summarised as follows:
\begin{itemize}
    \item We began by revealing the intimate relation between QRF and special covariance.~We showed that the capability to mimic the covariance structures underlying special relativity\,---\,and thereby to extend them into the quantum theory\,---\,singles out the perspective-neutral framework \cite{delaHamette:2021oex,PHquantrel,PHtrinity,Hoehn:2020epv,Vanrietvelde:2018pgb,Vanrietvelde:2018dit,Hohn:2018toe,Hohn:2018iwn,PHinternalQRF,PHQRFassymmetries,delaHamette:2021piz,Giacomini:2021gei,CastroRuizQuantumclockstemporallocalisability,Suleymanov:2023wio,yang2020switching} used in this work among the approaches to QRF covariance.~While here we illustrated this for ideal QRFs for finite abelian groups, the argument extends beyond this setting.~In particular, it was already shown in \cite{delaHamette:2021oex} that the perspective-neutral framework mimics the covariance structures of special relativity also for non-ideal frames and arbitrary unimodular Lie groups. 

    We emphasised that this observation does not invalidate other approaches as QRF covariance can have different meanings in different contexts, especially depending on whether access to an external frame can be physically realised.~Furthermore, the purely perspective-dependent approach to QRF covariance \cite{GiacominiQMcovariance:2019,Giacomini:2018gxh,delaHammetteQRFsforgeneralsymmetrygroups,delaHamette:2021iwx,Ballesteros:2020lgl,streiter2021relativistic,Mikusch:2021kro} is equivalent to the perspective-neutral one for ideal QRFs \cite{Vanrietvelde:2018pgb,PHQRFassymmetries,PHinternalQRF,PHtrinity,delaHamette:2021oex}, not, however, beyond this case \cite{delaHamette:2021oex}.~Similarly, for ideal frames for abelian groups, the effective  \cite{Bojowald:2010qw,Bojowald:2010xp,Hohn:2011us} and algebraic approaches \cite{Bojowald:2019mas,Bojowald:2022caa} are also equivalent to the perspective-neutral one \cite{JulianArtur} (the former in the semiclassical regime),\footnote{In fact, the algebraic approach is somewhat more general in the sense that the current formulation of the perspective-neutral framework is one Hilbert space realisation of it \cite{JulianArtur}; predictions are equivalent.~However, the algebraic approach is not yet developed for non-ideal QRFs and non-abelian groups.} with the relation beyond this setting currently unknown.~Thanks to our restriction to ideal QRFs for abelian groups, the results of this article thus also apply to these approaches.

    \item We exhibited the relation between subsystem relativity and the relativity of simultaneity in special relativity.~We demonstrated that subsystem relativity is nothing characteristic of QRFs and also arises in special relativity with internal frames, where it \emph{implies} the relativity of simultaneity.~Subsystem relativity is thus a generalisation of the latter and, from the point of view of frame covariance, its consequences should thus be seen in similar light.~Just like the relativity of simultaneity is the root of essentially all characteristic special relativistic phenomena, \emph{so is the relativity of subsystems the essence of QRF covariance}.
    
    \item We stressed that there are two equivalent ways of ``jumping into an internal frame perspective'' both for special relativity with tetrads and QRFs:\\ (i) gauge fixing to align the (possibly fictitious) external frame with the internal one;\\ (ii) building gauge-invariant relational observables to describe the degrees of interest relative to the internal frame.~While one can thus associate a choice of gauge with an internal frame, the gauge-invariant option (ii) highlights that the choice of internal frame itself it not a gauge.~The choice of frame isolates which degrees of freedom become fixed, but does not determine to what configuration they become fixed.~In particular, the subsystem relativity is \emph{not} a gauge artifact and for QRFs can be seen in the two equivalent ways:

    (i) For ideal frames, we showed that implementing an internal QRF perspective via conditioning on a frame orientation (gauge fixing) is formally nothing but a TPS on the (external-frame-independent) perspective-neutral Hilbert space $\mathcal{H}_{\rm phys}$, which by itself does not come with a TPS and thus is TPS-neutral too.~Similarly, (gauge induced) QRF transformations are nothing but changes of TPS on this space and the gauge-invariant algebra $\Aphys$.
    
    (ii) Relational observables describing the subsystem $S$ and ``the other frame'', respectively, define distinct commuting subalgebras of the gauge-invariant algebra $\Aphys$ for different choices of internal QRF; the two relational algebras describing $S$ relative to $R_1$ and $R_2$ are isomorphic and overlap, yet do not coincide. While this was already demonstrated in \cite{PHquantrel,delaHamette:2021oex}, here we refined this statement by characterising the overlap as consisting of all \emph{internal} relational observables of $S$ (cfr.~Fig.~\ref{Fig:relSR}).~These are the only $S$-observables which remain invariant under QRF transformation.~Conversely, this characterises QRF-contingent $S$-observables as those not residing in the overlap. A similar observation holds for the ``other frame'' observables.
     
    \item Noting that QRF perspectives are TPSs on $\Hphys$, we introduced the notion of TPS-invariant subalgebras.~These are subalgebras consisting of operators that are up to local unitaries invariant under QRF transformations, thus maintaining their degree of locality.~We fully characterised when subsystem local operators reside in such algebras and partially did so for composite operators.~This characterisation is crucial as operators not belonging to such subalgebras encode the QRF-dependent properties.
   
    \item Applying these algebraic tools to density operators allowed us to investigate how reduced subsystem states transform under QRF changes.~While global states transform unitarily, this is generically not the case for subsystem states due to subsystem relativity.~However, we showed that for pure global states, subsystem states are up to local unitaries QRF invariant if and only if they reside in a TPS-invariant subalgebra.~This thus also characterises when the entanglement spectrum and $\alpha$-R\'enyi entropies are invariant.~For global mixed states the situation is more subtle and we provided sufficient conditions.~Using these tools, we explored how correlations within a subsystem and with its complement can change under QRF transformations.
    
    \item Coming then to dynamics, we explored how Hamiltonians transform and determined necessary and sufficient conditions for a free Hamiltonian to remain free or become interacting under QRF changes.~This in turn entails a QRF dependence of the notion of closed or open dynamical subsystems.~We also characterised when an individual dynamical trajectory appears up to local unitaries the same in the two QRF perspectives, thus featuring the same correlation structures in both.~Conversely, this also characterises when a trajectory is QRF dependent.
    
    \item Finally, we focused on the thermodynamic description of subsystems relative to distinct QRFs.~First, we studied the frame-dependence of local observable averages (hence subsystem macrostates), subsystem thermal equilibrium and stationarity.~In particular, we demonstrated that, while subsystem Gibbs states and thus the notion of thermal equilibrium are generically QRF relative, QRF-invariant Gibbs states do exist.~Further, we showed that QRF transformations can map from one thermal equilibrium to another with distinct entropy.~Specifically, QRF transformations can not only change the temperature of the subsystem, but even flip its sign.~We then took the first steps to also investigate the QRF covariance of non-equilibrium thermodynamic processes, exhibiting that different QRFs generically refer to different types of thermodynamic subsystems in that heat, work, entropy production and flow depend on the choice of QRF.~To this end, we utilised the framework of quantum thermodynamics, which is the most natural one in a QRF context, where a single particle can serve as internal QRF and equilibrium and interactions are frame-dependent.
\end{itemize}

Since we were interested in quantum information-theoretic quantities, such as entropies, and quantum thermodynamics, which are best developed for finite-dimensional quantum systems, we restricted our exposition to QRFs associated with finite abelian groups as they naturally lead to such a setting.~To put the spotlight on physical properties, we further restricted to ideal QRFs.~However, with a bit more machinery, non-ideal QRFs and arbitrary unimodular Lie groups can be treated very similarly \cite{delaHamette:2021oex}, so that we expect our results to extend qualitatively to those cases as well.

Notwithstanding, for non-ideal QRFs, the reduced Hilbert spaces of the internal QRF perspectives will no longer coincide with the kinematical Hilbert space tensor factor of the complement of the frame as in~\eqref{eq:iReduction}, but will be a subspace of it \cite{delaHamette:2021oex,PHtrinity,Hoehn:2020epv,PHquantrel,reldynperclock}.~If the complement is kinematically also a composite system, then the perspectival Hilbert space is a subspace of a tensor product of Hilbert spaces, which is known as a \emph{fusion} or \emph{entangling} product.~This can usually be written as a direct sum of tensor products.~In other words, a non-ideal QRF perspective would not define a TPS on $\Hphys$ and $\Aphys$, but rather a direct sum of TPSs.~For example, relative to $R_1$ it would be of the form
\begin{align}
\Hphys&=\bigoplus_c\mathcal{H}^c_{R_2|R_1}\otimes\mathcal{H}^c_{S|R_1}\,,\nonumber\\\Aphys&=\bigoplus_c\mcA_{R_2|R_1}^c\otimes\mcA_{S|R_1}^c\,,\label{directsum}
\end{align}
where $c$ labels some charge sectors.~On the former this would happen via gauge fixing on a frame orientation, on the latter via the gauge-invariant subalgebras of relational observables describing, respectively, ``the other frame'' and $S$ relative to the chosen QRF, which fail to commute for non-ideal frames \cite{PHquantrel}.

\subsection{Relational subsystems and entropies vs.\ subsystems and entropies in gauge theories}

We take this opportunity to stress that the relational definition of subsystem and its entanglement entropies invoked in this work is distinct from the ones usually employed in gauge theories\footnote{To avoid misunderstanding already at this point: this is as far as computing entanglement entropies in gauge theories is concerned, which essentially takes only the \emph{internal} gauge-invariant data of a region into account.~Defining the gauge-invariant content of subregions relative to edge mode frames (and thus looking at a larger set of gauge-invariant data) \cite{Donnelly:2016auv,Geiller:2019bti,SCPHedgemodesasreferenceframes,Carrozza:2022xut}, however, is equivalent to what we do here.~But this larger data is usually not invoked in entropy calculations.} and constitutes an alternative proposal for that setting. 

In gauge theories, where locality is naturally defined in terms of the spacetime background, one typically associates a subsystem with some subregion $V$ and wishes to compute its entanglement with the complement $\bar{V}$. However, the gauge-invariant degrees of freedom in $V$ and in $\bar{V}$ do \emph{not} form a tensor product because, thanks to the gauge constraints, the gauge-invariant algebras associated with both overlap in a non-trivial center 
\be\label{eq:center}
\mcA_V\cap\mcA_{\bar{V}} = \mathcal{Z}\neq\mathbb{C}\mathbf{1}\,,
\ee 
which commutes with all elements in the regional algebras (e.g., see \cite{Casini:2013rba,Delcamp:2016eya}); here $\mcA_{\bar{V}}=(\mcA_V)'$ is the commutant of $\mcA_V$. This means that $\mcA_V$ and $\mcA_{\bar{V}}$ are \emph{not} factors of the full algebra $\Aphys$ of gauge-invariant observables (cfr.~(i) in Definition~\ref{def:TPSreleq}). This also holds on the lattice and is thus distinct from the usual arguments that subregions in QFT do not factorise into tensor products due to the type III nature of their observable algebras.

To circumnavigate this issue, one then usually builds extended regional Hilbert spaces $\mathcal{H}_{V'}$ and $\mathcal{H}_{\bar{V}'}$, where the prime denotes that these spaces not only contain the regional gauge-invariant degrees of freedom, but also additional boundary degrees of freedom \cite{Buividovich:2008gq,Donnelly:2011hn,Donnelly:2016auv,Delcamp:2016eya,Aoki:2015bsa,Ghosh:2015iwa,Donnelly:2014fua,Donnelly:2014gva,VanAcoleyen:2015ccp}, often called edge modes.~The physical Hilbert space is  a subspace of the tensor product $\mcH_{V'}\otimes\mcH_{\bar{V}'}$; more precisely, it is the entangling product of these spaces with respect to a surface symmetry group \cite{Donnelly:2016auv,Geiller:2019bti}.~In order to compute the entanglement entropy of $V$, one traces over $\bar{V}'$ in $\mcH_{V'}\otimes\mcH_{\bar{V}'}$, resulting in a superselection across different representations of the surface symmetry on $\mcH_{V'}$, and applies the usual von Neumann formula to each such superselection sector, and finally averages over the result \cite{Buividovich:2008gq,Donnelly:2011hn,Delcamp:2016eya,Aoki:2015bsa,Ghosh:2015iwa,Donnelly:2014fua,Donnelly:2014gva,VanAcoleyen:2015ccp}.~In \cite{Casini:2013rba}, it was shown that this procedure is equivalent to a (non-unique) algebraic center construction, building up on~\eqref{eq:center}.~In particular, it amounts to taking the algebra generated by the regional ones and decomposing it into charge sectors labeled by irreps $z$ of the center,
\be 
\mcA_V\vee\mcA_{\bar{V}}=\bigoplus_z\mcA^z_V\otimes\mcA^z_{\bar{V}}\subsetneq\Aphys\,.
\ee 
This, in particular, does not generate the full invariant algebra as it lacks the cross-boundary gauge-invariant data (e.g.\ Wilson loops).~One then decomposes the global physical state into these charge sectors, ${\rho_{V\bar{V}}=\oplus_z p_z\,\rho^z_{V\bar{V}}}$, and within each one traces over $\bar{V}$ to obtain ${\rho_V=\oplus_zp_z\,\rho_V^z}$.~Applying the usual von Neumann entropy formula to this state yields
\be \label{GTentropy}
S_{\rm vN}(V)=H(\{p_z\})+\sum_zp_z\,S_{\rm vN}(\rho_V^z)\,,
\ee 
where $H(\{p_z\})=-\sum_z p_z\ln p_z$ is the classical Shannon entropy associated with the probability distribution over the charge sectors.~As emphasised in \cite{Casini:2013rba,VanAcoleyen:2015ccp}, this entropy does not have the usual interpretation of entanglement entropy in quantum information theory because $H(\{p_z\})$ constitutes a non-distillable constribution.

Let us now compare this to our setting, which can be viewed as a toy theory for lattice gauge theory (with global, rather than local gauge action given by $U_{12S}^g$).~For contradistinction, let us first mimic the above in our setting.~If one identifies $V$ with $S$ in our case, then its complement is $\bar{V}=R_1R_2$.~The gauge-invariant ``regional'' subalgebra for $V=S$ is then $\mcA_V =  \mcA^{\rm phys}_{S|R_1}\cap\mcA_{S|R_2}^{\rm phys}$, i.e.\ the purely internal relational observables of $S$ (cfr.~Theorem~\ref{lem_TPSineqalg} and Fig.~\ref{Fig:relSR}).~By contrast, the ``regional'' subalgebra associated with $\bar{V}=R_1R_2$ is  $\mcA_{\bar{V}}=\mcA_{R_2|R_1}^{\rm phys}\vee\mcA^{\rm phys}_{R_1|R_2}$, i.e.\ all the relational observables between the two QRFs.\footnote{This can be seen from the expression~\eqref{RelDO} for relational observables relative to $R_1$ (and the same expression with $R_1$ and $R_2$ exchanged) and noting that gauge-invariant observables in $\mcA_{\bar{V}}$ should have no non-trivial dependence on the $S$-tensor factor.~That is, they should be of the form $O_{12}\otimes\mathds1_S$.}~The algebra generated by both, $\mcA_{V}\vee\mcA_{\bar{V}}$, is a strict subalgebra of $\Aphys$,\footnote{This can be best seen by reducing $\mcA_V$ and $\mcA_{\bar{V}}$ into one of the two QRF perspectives, say $R_1$. $\mcA_{\bar{V}}$ does not yield any purely $S$-local observables relative to $R_1$, while $\mcA_V$ yields only translation-invariant $S$-local observables relative to $R_1$.~That is, all the non-translation-invariant $S$-observables relative to $R_1$ cannot be generated; they correspond to $\mcA^{\rm phys}_{S|R_1}\setminus(\mcA^{\rm phys}_{S|R_1}\cap\mcA^{\rm phys}_{S|R_2})$.} and we have a non-trivial center
\be
\mcA_V\cap\mcA_{\bar{V}} = \{\mathbf{1},\Piphys^{12}\otimes\Piphys^S\}\,,
\ee
where $\Piphys^{12}=\frac{1}{|\mcG|}\sum_{g\in\mcG}U_{12}^g$ and $\Piphys^S=\frac{1}{|\mcG|}\sum_{g\in\mcG}U_S^g$. This can be seen from the fact that the ``regional'' observables commute with the latter two projectors and
\be
(\Piphys^{12}\otimes\mathds1_S)\Piphys=(\mathds1_{12}\otimes\Piphys^S)\Piphys=\Piphys^{12}\otimes\Piphys^S\,.\nonumber
\ee 
$\mcA_{V}\vee\mcA_{\bar{V}}$ can then be decomposed into a direct sum over irreps of this center.~The ``entanglement entropy'' of $S$ that one would then compute according to~\eqref{GTentropy} contains the non-distillable contribution and is  associated with the internal relational observables of $S$, $\mcA^{\rm phys}_{S|R_1}\cap\mcA_{S|R_2}^{\rm phys}$.~In particular, it is computed with states that only have support in a strict \emph{sub}algebra $\mcA_{V}\vee\mcA_{\bar{V}}$ of $\Aphys$.

Clearly, this is very different from what we do in this article (and more complicated).~Instead of defining subsystems and entropies via kinematical complements $S$ and $R_1R_2$, we defined them via complements \emph{relative to a QRF}, e.g.\ $S$ and $R_2$ relative to $R_1$.~As such, rather than defining the gauge-invariant content of $S$ to be only its internal relational observables, $\mcA^{\rm phys}_{S|R_1}\cap\mcA_{S|R_2}^{\rm phys}$, we defined it to be \emph{all} of $\mcA^{\rm phys}_{S|R_1}$.~Similarly, the gauge-invariant content of the complement of $S$ in $R_1$-perspective is given by $\mcA^{\rm phys}_{R_2|R_1}$.~The crucial point is that these together span \emph{all} of $\Aphys$ and, for ideal QRFs, we have a genuine tensor product $\Aphys\simeq\mcA^{\rm phys}_{R_2|R_1}\otimes\mcA^{\rm phys}_{S|R_1}$.~This TPS permitted us to compute the entanglement entropy of $S$ in the standard way without any non-distillable contributions and therefore with the usual interpretation as in quantum information theory.~The key difference to the gauge theory setting is thus that we choose a larger gauge-invariant algebra associated with the system and that states can be decomposed in terms of the \emph{full} invariant algebra $\Aphys$.~Furthermore, since the subsystem and entanglement entropy are defined relative to a QRF, we also obtain QRF-dependent entropies, in contrast to the usual gauge theory setting.

This also means that the entangling product appearing in the extended Hilbert space construction \cite{Donnelly:2016auv,Geiller:2019bti} is distinct to the one leading to~\eqref{directsum}, which is non-trivial only for non-ideal QRFs.~In particular, the entangling product in $R_1$-perspective leading to~\eqref{directsum} is between the \emph{unextended} kinematical tensor factors $\mcH_{2}$ and $\mcH_S$.~For such a direct sum decomposition one can compute the entanglement entropy analogously to the procedure leading to~\eqref{GTentropy} \cite{Bianchi:2019stn,Bianchi:2021aui}.~Since the techniques are applied to different entangling products, this will also lead to distinct entropies for non-ideal frames.

Altogether, this suggests an alternative for defining subsystems and their entanglement entropies in gauge theories.~Rather than defining a subsystem via the purely regional gauge-invariant observables in analogy to $\mcA^{\rm phys}_{S|R_1}\cap\mcA^{\rm phys}_{S|R_2}$, we can gauge-invariantly define the subsystem via the relational observables describing the subregion relative to a QRF in analogy to $\mcA^{\rm phys}_{S|R_1}$.~Indeed, the above mentioned edge modes were shown to constitute dynamical or quantum reference frames of the type considered in this article \cite{SCPHedgemodesasreferenceframes,Carrozza:2022xut,Kabel:2023jve} (see also \cite{Gomes:2016mwl,Gomes:2019otw,Gomes:2021nwt,Gomes:2022kpm,Riello:2020zbk,Riello:2021lfl} for related observations).~In particular, there exists no unique edge mode frame \cite{SCPHedgemodesasreferenceframes,Carrozza:2022xut} and thus a frame-dependent notion of subsystem and entropy will also arise \cite{SubregionRel}.~This means a slight departure from the background notion of locality.~The local subsystem is now not defined purely in terms of a local subregion, but in terms of the relational observables describing that subregion relative to a frame.~This frame can be non-locally constructed from the field content \cite{SCPHedgemodesasreferenceframes,Carrozza:2022xut,Goeller:2022rsx} and so locality is defined relationally.

\newpage

\subsection{Outlook}

Besides realising a relational definition of subsystems and entropies for gauge theories and gravity, as just suggested, and generalising to non-ideal frames and general groups, there are a few other interesting lines of research.~These pertain to the physical consequences of subsystem relativity, of which we have studied only a fraction.~It is clear that there are many more physical properties of subsystems that will be QRF relative.~For example, we can expect that chaos as defined via out-of-time-ordered correlators, Markovianity of a dynamics, complexity of states and processes, fluctuation theorems, etc.\ 
will become contingent on the choice of internal QRF.~This is already clear for fluctuation theorems such as the Jarzynski equality \cite{Campisietal}, which relies on an initial Gibbs state.~As we have seen, Gibbs states and thermal equilibrium depend on the QRF and so a Jarzynski equality may hold in one QRF perspective, but not in another.~This is similar to our discussion of entropy production and flow, where the initial state needed to be of product form, which is not in general invariant under QRF transformations.~It would also be interesting to consider open thermodynamic systems and how matter exchange may depend on the choice of QRF. 

\begin{acknowledgments}

We would like to thank Florio Ciaglia, Stefan Eccles, Josh Kirklin, Leon Loveridge, Markus M\"uller, Germain Tobar, and Mischa Woods for discussion and comments.~PAH and FMM are grateful for the hospitality of the high-energy physics group at EPF Lausanne, where the final stages of this work were carried out.~This work was supported by funding from Okinawa Institute of Science and Technology Graduate University.~This project/publication was also made possible through the support of the ID\# 62312 grant from the John Templeton Foundation, as part of the \href{https://www.templeton.org/grant/the-quantum-information-structure-of-spacetime-qiss-second-phase}{\textit{`The Quantum Information Structure of Spacetime'} Project (QISS)}.~The opinions expressed in this project/publication are those of the author(s) and do not necessarily reflect the views of the John Templeton Foundation.
\end{acknowledgments}

\newpage

\bibliography{refqrftd} 

\clearpage

\onecolumngrid

\appendix

\section*{Appendices}
\renewcommand{\thesubsection}{\Alph{subsection}}

\subsection{Symmetry-induced QRF transformations and QRF-invariant relational $S$-observables}\label{App:relcondreorientations}

The frame change transformations $V_{1\to 2}^{g_1,g_2}$ discussed in Sec.~\ref{sec:qrfphystps} (cfr.~\eqref{eq:Vitoj}) allow us to change the reduced description of a \emph{given} relational observable from one frame to the other.~These are also referred to as \emph{gauge-induced} QRF transformations \cite{delaHamette:2021oex}.~In fact, as the reduction maps $\mathcal R_{i}^{g_i}$ into the perspective of frame $R_i$ involve a conditioning $\bra{g_i}_i$ on frame orientations, they are the quantum analogue of gauge-fixing conditions and the reduced perspectival descriptions correspond to different gauge-fixing conditions.~The transformation $V_{1\to 2}^{g_1,g_2}=\mathcal R_2^{g_2}\circ(\mathcal R_1^{g_1})^{-1}$ thus amounts to a change $\bra{g_1}_1\to\bra{g_2}_2$ of both the conditioning and the frame therein which can then be interpreted as a \emph{relation-conditional gauge transformation}, i.e., a gauge transformation that depends on the relation between the two frames.~These are \emph{passive} transformations which do not change neither the physical state of the system nor the observable we evaluate in the given physical state (i.e., the physical question that is being asked) but only their gauge-fixed internal description from one frame to another.

As shown in \cite{delaHamette:2021oex}, there exists a second type of QRF transformations, which are \emph{symmetry-induced} and \emph{active} in the sense that they \emph{change} relational observables from one frame to another. (They are the quantum versions of relational observable transformations that have appeared in classical gauge theories and gravity \cite{Carrozza:2022xut,Goeller:2022rsx}.) More precisely, they transform between the algebras of relational observables relative to frames $R_1$ and $R_2$, e.g.\ mapping $O^{g_1}_{f_{2S}|R_1}$ into $O^{g_2}_{f_{2S}|R_2}$, and in particular were used in \cite{delaHamette:2021oex} to explain the relativity of subsystems. Here, we shall use this transformation to provide an alternative characterization of the overlap of $\mathcal{A}_{S|R_1}^{\rm phys}\cap\mathcal{A}_{S|R_2}^{\rm phys}$, described in Theorem~\ref{lem_TPSineqalg}, as those relational observables describing $S$ that are invariant under this second type of QRF transformation.

To build up some intuition, it is instructive to consider a simple example of $N$ particles in one spatial dimension subject to a continuous translation invariance.~Let us take the frames $R_1$ and $R_2$ to be any two particles and let the system $S$ consists of the remaining $N-2$ particles.~Consider now the relational observables $O^{q_1=0}_{q_k|R_1}=q_k-q_1\in\mcA^{\text{phys}}_{S|R_1}$, $k\neq1,2$, corresponding to the positions of any of the particles in $S$ when $R_1$ is in the origin.~Clearly, translating $R_1$ by its relative distance $q_2-q_1$ from $R_2$ will transform $O^{q_1=0}_{q_k|R_1}$ into the corresponding observables $O^{q_2=0}_{q_k|R_2}=q_k-q_2\in\mcA^{\text{phys}}_{S|R_2}$ relative to frame $R_2$.~This is a reorientation of frame $R_1$ which depends on the relation between the frames, that is a \emph{relation-conditional frame reorientation}.
This transformation is not changing the physical state of the system but only the question that is being asked.~For example, we change from asking ``what is the position of particle $k$ when $R_1$ is in the origin?'' to ``what is the position of particle $k$ when $R_2$ is in the origin?'', while referring to the same kinematical property of $S$ (the position of particle $k$). 

For finite Abelian groups as considered in this work, the analogue of the translation group-valued relative distance $q_2-q_1$ is given by $g_1^{-1}g_2$.~A reorientation of $R_1$ (which is a symmetry on $\Hphys$) conditioned on $g_1^{-1}g_2$ can be thus written as the following super-operator (cfr.~\cite{delaHamette:2021oex},~Sec.~8 for the generalisation to the regular representation of arbitrary unimodular Lie groups) 
\be\label{symmQRFtransf}
\hat{\mathcal V}_{1\to 2}^{g_1,g_2}(\bullet)=\sum_{g,g'\in\mathcal G}\left(\ket{g}\!\bra{g}_1\otimes\ket{gg'}\!\bra{gg'}_2\otimes\mathds1_S\right)\left(U_1^{g_1^{-1}g_2g'^{-1}}\otimes\mathds 1_2\otimes\mathds 1_S\right)\bullet\left(U_1^{g'g_2^{-1}g_1}\otimes\mathds 1_2\otimes\mathds 1_S\right)\,,
\ee
where $O_{\ket{g'}\!\bra{g'}_2\otimes\mathds1_S|R_1}^{g_1=0}=\sum_{g\in\mathcal G}\ket{g}\!\bra{g}_1\otimes\ket{gg'}\!\bra{gg'}_2\otimes\mathds1_S$ is the projection onto states with $g_1^{-1}g_2=g'$.~The super-operator \eqref{symmQRFtransf} maps the subalgebra $\mcA_{S|R_1}^{\text{phys}}\subset\Aphys$ of relational observables describing $S$ relative to frame $R_1$ into the subalgebra $\mcA_{S|R_2}^{\text{phys}}\subset\Aphys$ of relational $S$-observables relative to $R_2$.~Specifically, recalling the expression \eqref{RelDO} of relational observables together with the property $\Piphys A\Piphys=\hat\Pi_{\rm inv}(A)\Piphys$ for any $A\in\mathcal L(\Hkin)$ (cfr.~\eqref{relationpiphyspiinv}), with $\hat\Pi_{\rm inv}(A)$ the incoherent group averaging given in \eqref{Gtwirl}, one can check that
\be\label{symmindtransf:Slocal}
\hat{\mathcal V}_{1\to 2}^{g_1,g_2}(O^{g_1}_{\mathds 1_2\otimes f_S|R_1})=O^{g_2}_{\mathds 1_1\otimes f_S|R_2}\qquad f_S\in\mcA_S\;,
\ee
on the physical Hilbert space \cite{delaHamette:2021oex}.~Therefore, as pictorially depicted in Fig.~\ref{Fig:symmSobs}, the symmetry-induced QRF transformation \eqref{symmQRFtransf} maps the relational observable describing the value of a given kinematical observable $f_S\in\mcA_S$ of the system $S$ when frame $R_1$ is in orientation $g_1$ into the relational observable describing the value of the \emph{same} $S$-observable $f_S$ when frame $R_2$ is in orientation $g_2$.

The QRF transformations \eqref{symmQRFtransf} provide us with a further characterisation of the elements in the overlap $\mcA_{S|R_1}^{\text{phys}}\cap\mcA_{S|R_2}^{\text{phys}}$ as those relational $S$-observables which are invariant under them.
\begin{figure}[!t]
\centering\includegraphics[scale=0.45]{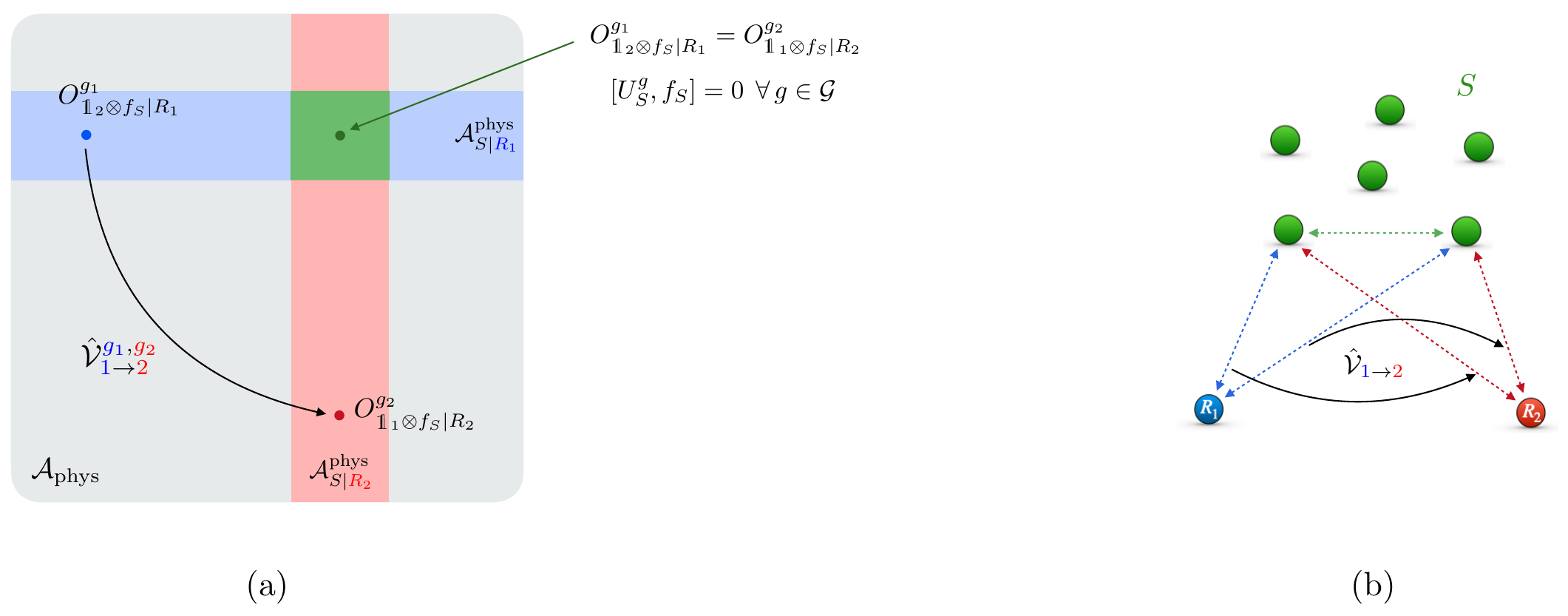}
\caption{(a)~The symmetry-induced QRF transformations \eqref{symmQRFtransf} map the relational observables $O^{g_1}_{\mathds 1_2\otimes f_S|R_1}$ describing the value of the $S$-observable $f_S$ when frame $R_1$ is in orientation $g_1$ into the relational observables $O^{g_2}_{\mathds 1_1\otimes f_S|R_2}$ describing the value of the \emph{same} $S$-observable $f_S$ when frame $R_2$ is in orientation $g_2$.~For example, in a system of particles subject to translation invariance depicted in (b), the relative distance of a particle in $S$ from $R_1$ (blue dashed arrow) is transformed into the corresponding relative distance of that particle from $R_2$ (red dashed arrow).~The observables in $\mcA_{S|R_1}^{\text{phys}}\cap\mcA_{S|R_2}^{\text{phys}}$ (green region in (a)), which by Thm.~\ref{lem_TPSineqalg} corresponds to translation invariant $f_S$ observables, are invariant under these QRF changes.~These correspond to internal $S$-relations and identify coincident relational observables relative to the two frames $R_1$ and $R_2$ (e.g.~the green dashed arrow in (b)).}
\label{Fig:symmSobs}
\end{figure}
\begin{claim} 
    Let $f_S\in\mcA_S$ be a kinematical $S$-observable.~Then, the associated relational observable $O_{\mathds1_2\otimes f_S|R_1}^{g_1}\in\mcA_{S|R_1}^{\emph{phys}}$ is invariant under the symmetry-induced QRF transformation \eqref{symmQRFtransf}, that is $\hat{\mathcal V}_{1\to 2}^{g_1,g_2}(O_{\mathds1_2\otimes f_S|R_1}^{g_1})=O_{\mathds1_2\otimes f_S|R_1}^{g_1}$, if and only if $f_S$ is translation-invariant, $[f_S,U_S^g]=0$, $\forall\,g\in\mathcal G$, i.e., if and only if $O_{\mathds1_2\otimes f_S|R_1}^{g_1}\in\mcA_{S|R_1}^{\emph{phys}}\cap\mcA_{S|R_2}^{\emph{phys}}$. 
\end{claim}
\begin{proof}
    By Theorem~\ref{lem_TPSineqalg}, the relational observables $O_{\mathds1_2\otimes f_S|R_1}^{g_1}$ with translation-invariant $f_S$ are invariant under reorientations of both frames.~Thus, the action \eqref{symmQRFtransf} of $\hat{\mathcal V}_{1\to 2}^{g_1,g_2}$ on $O_{\mathds1_2\otimes f_S|R_1}^{g_1}$ reduces to
    \be
    \hat{\mathcal V}_{1\to 2}^{g_1, g_2}(O_{\mathds1_2\otimes f_S|R_1}^{g_1})=\sum_{g\in\mathcal G}(\ket{g}\!\bra{g}_1\otimes\mathds1_2\otimes\mathds1_S)O_{\mathds1_2\otimes f_S|R_1}^{g_1}=O_{\mathds1_2\otimes f_S|R_1}^{g_1}\;.
    \ee

Vice versa, if $\hat{\mathcal V}_{1\to 2}^{g_1, g_2}(O_{\mathds1_2\otimes f_S|R_1}^{g_1})=O_{\mathds1_2\otimes f_S|R_1}^{g_1}$, then $O_{\mathds1_2\otimes f_S|R_1}^{g_1}$ is invariant under reorientations of frame $R_1$.~By construction, $O_{\mathds1_2\otimes f_S|R_1}^{g_1}$ is also invariant under reorientations of frame $R_2$.~Thus, by the proof of Theorem~\ref{lem_TPSineqalg} $O_{\mathds1_2\otimes f_S|R_1}^{g_1}\in\mcA_{S|R_1}^{\text{phys}}\cap\mcA_{S|R_2}^{\text{phys}}$ and $[f_S,U_S^g]=0$ for any $g\in\mathcal G$.
\end{proof}

For illustration, consider once more the example of a continuous translation-invariant system of particles and Fig.~\ref{Fig:symmSobs}.~Consider the relative distance between any two particles in $S$, say $q_{k+1}-q_k$ with $k\neq 1,2$.~This is a translation invariant observable internal to $S$ which can be equivalently written as $q_{k+1}-q_k=(q_{k+1}-q_1)-(q_k-q_1)=(q_{k+1}-q_2)-(q_k-q_2)$, thus yielding the same relational observable $O_{q_{k+1}-q_k|R_1}^{q_1=0}=O_{q_{k+1}-q_k|R_2}^{q_2=0}$ relative to $R_1$ and $R_2$.~Note also that it corresponds to the same reduced $S$-observable when the position $q_i$ of the frame $R_i$ is fixed.~Moreover, this observable is clearly invariant under reorientations of both frames.~In particular, it is also invariant under a translation of $R_1$ by the relative distance $q_2-q_1$ between $R_2$ and $R_1$.

\subsection{Proof: Theorem \ref{lem_TPSineqalg}}\label{App:TPSlocu2}

\noindent{\bf Theorem~\ref{lem_TPSineqalg}.}
\emph{The algebras of relational observables describing solely $S$ relative to the two frames are distinct, but isomorphic subalgebras of $\mathcal{A}_{\rm phys}$, i.e.\ $\mathcal{A}_{S|R_1}^{\rm phys}\neq\mathcal{A}_{S|R_2}^{\rm phys}$. Their overlap $\mathcal{A}_{S|R_1}^{\rm phys}\cap\mathcal{A}_{S|R_2}^{\rm phys}$ is non-trivial and consists of all operators of the form $(\mathds{1}_1\otimes\mathds1_2\otimes f_S)\,\Pi_{\rm phys}=\mathds{1}_1\otimes\mathds1_2\otimes f_S\restriction{\mathcal{H}_{\rm phys}}$, where $f_S\in\mathcal{A}_S$ is translation-invariant $[f_S,U^g_S]=0$, $\forall\,g\in\mathcal{G}$. Hence, observables in the overlap commute with reorientations of \emph{both} frames.}

\begin{proof}
The reduction theorem~\eqref{eq:obsredthm} implies that both $\mathcal{A}_{S|R_i}^{\rm phys}$, $i=1,2$, are isomorphic to $\mathcal{A}_S$. By construction, $O^{g_i}_{\mathds1_{j}\otimes f_S|R_i}$ is invariant under reorientations of frame $R_j$, i.e.~$\big[O^{g_i}_{\mathds1_{j}\otimes f_S|R_i},\mathds1_{i}\otimes U^g_j\otimes\mathds1_S\big]=0$, $\forall\,g\in\mathcal{G}$ (cfr.~Eq.~\eqref{RelDO}). Hence, all elements in $\mathcal{A}_{S|R_1}^{\rm phys}\cap\mathcal{A}_{S|R_2}^{\rm phys}$ must be invariant under reorientations of \emph{both} frames and so commute with $U^g_1\otimes U^{g'}_2\otimes\mathds1_S$, $\forall\,g,g'\in\mathcal{G}$. Since $(U^g_1\otimes U^{g}_2\otimes\mathds1_S)\,\Piphys = (\mathds1_{1}\otimes\mathds1_{2}\otimes U^{g^{-1}}_S)\,\Piphys$, this means that all these elements must also commute with $\mathds1_{1}\otimes\mathds1_{2}\otimes U^{g}_S$, $\forall\,g\in\mathcal{G}$. The expression in~\eqref{RelDO} then entails $O^{g_i}_{\mathds1_{j}\otimes f_S|R_i}=O^{g_i}_{\mathds1_{j}\otimes U^g_S f_SU^{g^{-1}}_S|R_i}$ for all elements in $\mathcal{A}_{S|R_1}^{\rm phys}\cap\mathcal{A}_{S|R_2}^{\rm phys}$. Now applying the reduction theorem~\eqref{eq:obsredthm} yields $f_S=U^g_S f_S U^{g^{-1}}_S$. For such translation-invariant $S$-observables \eqref{RelDO} implies $O^{g_i}_{\mathds1_{j}\otimes f_S|R_i}=(\mathds1_{12}\otimes f_S)O^{g_i}_{\mathds1_{j}\otimes\mathds1_S|R_i}$ and noting that $O^{g_i}_{\mathds1_{j}\otimes\mathds1_S|R_i}=\Piphys$ means that all elements in $\mathcal{A}_{S|R_1}^{\rm phys}\cap\mathcal{A}_{S|R_2}^{\rm phys}$ are of the claimed form. Conversely, it is clear that every observable of the claimed form resides in this overlap.
\end{proof}

\subsection{Subalgebra of $\mbfUX$-invariant operators} \label{app:subalgebraAU}

\noindent
\begin{claim}
The proper subset $\mcA_{\ibar}^{X}\subset\mathcal A_{\overline{i}}$ of $\mbfUX$-invariant operators defined in~\eqref{Uinvsubalgebra} is a unital $^*$-subalgebra of $\mathcal A_{\overline{i}}$.
\end{claim}
\begin{proof}
Given the subset \eqref{Uinvsubalgebra}, one can easily check that $\mcA_{\ibar}^{X}$ is a vector subspace of $\mcA_{\ibar}$ over $\mathbb{C}$, closed under linear combination and multiplication, by using unitarity of $\mbfU_{\ibar}^{g_i,g_j}$ and $X$.~Next, if $A_{\ibar} \in \mcA_{\ibar}^{X}$, then $\hat{X}\circ\hat{\mbfU}_{\ibar}^{g_i,g_j}(A_{\ibar}^\dagger)=(\hat{X}\circ\hat{\mbfU}_{\ibar}^{g_i,g_j}(A_{\ibar}))^\dagger= A_{\ibar}^\dagger$, and so $A_{\ibar}^\dagger \in \mcA_{\ibar}^{X}$. Lastly, clearly $\hat{X}\circ\hat{\mbfU}_{\ibar}^{g_i,g_j}(\mathds1_{\ibar}) = \mathds1_{\ibar}$, thus $\mathds 1_{\ibar}\in\mcA_{\ibar}^{X}$.
\end{proof}

\noindent\textbf{Lemma~\ref{lem:pix}.}
 \emph{We have $\PiUX\circ\hat X\circ\hat U_{\ibar}^{g_i,g_j}=\PiUX$. }

\begin{proof}
  Let $A_{\ibar}\in\mcA_{\ibar}$ be arbitrary. Then, using that $\PiUX$ is a Hilbert-Schmidt orthogonal projector, we can write 
  \be
\PiUX\left(\hat X\hat U_{\ibar}^{g_i,g_j}(A_{\ibar})\right) = \sum_a \Tr\left[\left(\hat X\hat \mbfU_{\ibar}^{g_i,g_j}(A_{\ibar})\right)^\dag\,E^a_{\ibar}\right]\,E_{\ibar}^a = \sum_a \Tr\left[A^\dag_{\ibar}\,\hat X\hat\mbfU_{\ibar}^{g_i,g_j}(E_{\ibar}^a)\right]\,E_{\ibar}^a = \sum_a \Tr\left[A^\dag_{\ibar}\,E_{\ibar}^a\right]\,E_{\ibar}^a = \PiUX(A_{\ibar})\,,\nonumber
  \ee
  where $E_{\ibar}^a$ is any basis of $\mcA_{\ibar}^X$, whose elements by definition obey $\hat X\hat\mbfU_{\ibar}^{g_i,g_j}(E_{\ibar}^a)=E_{\ibar}^a$.
\end{proof}

\noindent\textbf{Lemma~\ref{lem_Agg}.} \emph{
If $f_{\ibar}\in\mathcal{A}_{\ibar}^{X,g_i,g_j}$ for some fixed QRF orientations $g_i,g_j$ and bilocal unitary $X=Y_j\otimes Z_S$, then, for arbitrary $g_i',g'_j\in\mathcal{G}$, also $f_{\ibar}\in\mathcal{A}_{\ibar}^{X',g'_i,g'_j}$ with $X'=Y'_j\otimes Z_S'$ given by
\be
Y'_j=Y_j\,U_j^{(g'_ig'_j){}^{-1}g_ig_j}\,,\qquad\qquad Z_S'=Z_S\, U_S^{g_jg'_j{}^{-1}}\,.\nonumber
\ee
}

\begin{proof}
The condition that $f_{\ibar}$ resides in both $\mathcal{A}_{\ibar}^{X,g_i,g_j}$ and $\mathcal{A}_{\ibar}^{X',g'_i,g'_j}$ entails by \eqref{Uinvsubalgebra} that
\be
\left(\hat{\mathbf{U}}_{\ibar}^{g_i,g_j}\right)^\dag\hat X^\dag(f_{\ibar})=
\left(\hat{\mathbf{U}}_{\ibar}^{g'_i,g'_j}\right)^\dag\hat X^{\prime\dag}(f_{\ibar})\,.\label{uxf}
\ee
Invoking \eqref{TPSschangemap}, we have
\begin{align} 
\mathbf{U}_{\ibar}^{g_i,g_j}\left(\mathbf{U}_{\ibar}^{g'_i,g'_j}\right)^\dag&=\sum_{g,g'}\ket{g_ig}\!\braket{g_jg^{-1}|g'^{-1}g_j'}\!\bra{g_i'g'}_j\otimes U_S^{gg'{}^{-1}}\\
&=\sum_g\ket{g_ig}\!\bra{g'_ig_j{}^{-1}g_j'g}_j\otimes U_S^{g_jg'_j{}^{-1}}\\
&=U_j^{g_i}\left(\sum_g\ket{g}\!\bra{g}_j\right)U_j^{g'_i{}^{-1}g_jg'_j{}^{-1}}\otimes U_S^{g_jg_j'{}^{-1}}\\
&=U_j^{(g'_ig'_j){}^{-1}g_ig_j}\otimes U_S^{g_jg_j'{}^{-1}}\,.
\end{align}
We can thus solve condition~\eqref{uxf} if we choose 
\be 
X^{\prime\dag}=\bigl(U_j^{(g_ig_j){}^{-1}g'_ig'_j}\otimes U_S^{g_j{}^{-1}g'_j}\bigr)X^\dag\,.
\ee 
This proves the claim.   
\end{proof}

\subsection{Algebraic properties related to projectors $\hat{\Pi}_{\mbt}, \hat{\Pi}_{\mbd}$}\label{App:algebrasplitting}

Here we provide details on the statements made in section \ref{Sec:Uinv&break} surrounding $\hat{\Pi}_{\mbt}$ and $\hat{\Pi}_{\mbd}$, as defined in equations \eqref{pitdef} and \eqref{piddef} respectively. 

\begin{claim} \label{claim:projPit}
The map $\hat{\Pi}_{\mbt}: \mcA_{\ibar} \to \mcA_{\ibar}$ defined by
\be \label{pitdef2}
\hat{\Pi}_{\mbt} (f_{\ibar}) := \frac{1}{|\mcG|}\sum_{g \in \mcG}  (\mathds1_{j}\otimes U^g_S)\, f_{\ibar} \, (\mathds1_{j}\otimes {U^{g}_S}^\dag)\,, \qquad f_{\ibar} \in \mcA_{\ibar}
\ee
is an orthogonal projector, with complement $\hat{\Pi}_{\mbt}^{\perp} := \hat{\mathds1}_{\ibar} - \hat{\Pi}_{\mbt}$ defined with respect to the Hilbert-Schmidt inner product on $\mcA_{\ibar}$.
\end{claim}

\begin{proof}
The proof is straightforward by direct computation. For idempotence: for any $f_{\ibar} \in \mcA_{\ibar}$, we have
\be
\Pit^2(f_{\ibar}) = \frac{1}{|\mcG|^2} \sum_{g',g} (\mathds1_{j}\otimes U_S^{g'g}) \, f_{\ibar} \, (\mathds1_{j}\otimes U_S^{g'g\dag}) = \frac{1}{|\mcG|^2} \sum_{g',g} (\mathds1_{j}\otimes U_S^{g}) \, f_{\ibar} \, (\mathds1_{j}\otimes U_S^{g\dag}) = \Pit(f_{\ibar}) \,,
\ee
so that $\Pit^2 = \Pit$. This also implies $(\Pit^{\perp})^2  = \Pit^{\perp}$. Thus, $\Pit$ and $\Pit^{\perp}$ are projectors on $\mcA_{\ibar}$. For orthogonality: consider the Hilbert-Schmidt inner product on $\mcA_{\ibar}$ given by $(A_{\ibar},B_{\ibar}) := \Tr(A_{\ibar}^\dagger B_{\ibar})$, where $\Tr$ refers to the full trace over the underlying Hilbert space $\Hibar$. Then, for any $A_{\ibar},B_{\ibar} \in \mcA_{\ibar}$:
\be
    (\Pit(A_{\ibar}),B_{\ibar}) = \frac{1}{|\mcG|}\sum_{g} \Tr \Big( (\mathds1_j \otimes U_S^{g\dag})A_{\ibar}^\dagger (\mathds1_j \otimes U_S^g) B_{\ibar} \Big) = (A_{\ibar},\Pit(B_{\ibar}))
\ee
using cyclicity and linearity of $\Tr$; also, $(\Pit^{\perp}(A_{\ibar}),B_{\ibar}) = (A_{\ibar},B_{\ibar}) - (\Pit(A_{\ibar}),B_{\ibar}) = (A_{\ibar},\Pit^{\perp}(B_{\ibar}))$, using orthogonality of $\Pit$. Finally, it is easy to see that $\Pit$ and $\Pit^{\perp}$ are complements with respect to the Hilbert-Schmidt inner product: for any $A_{\ibar} \in {\rm im}(\Pit), B_{\ibar} \in {\rm im}(\Pit^{\perp})$, we have
\be
(A_{\ibar},B_{\ibar}) = (\Pit(A_{\ibar}),\Pit^{\perp}(B_{\ibar})) = (\Pit(A_{\ibar}),B_{\ibar}) - (\Pit(A_{\ibar}),\Pit(B_{\ibar})) = 0
\ee
using the fact that $\Pit$ is an orthogonal projector, in the third equality.
\end{proof}

\noindent {\bf Remark.} From Claim \ref{claim:projPit} above, it directly follows that ${\rm im}(\Pit) = {\rm ker}(\Pit^\perp)$, and vice-versa. We thus have the algebra decomposition $\mcA_{\ibar} = \mcA_{\ibar}^{\mbt} \oplus \mathcal{C}_{\ibar}^{\mbt}$ given in \eqref{algtdecomp}, with $\mcA_{\ibar}^{\mbt} := {\rm im}(\Pit) = {\rm ker}(\Pit^\perp)$ and $\mathcal{C}_{\ibar}^{\mbt} := {\rm im}(\Pit^\perp) = {\rm ker}(\Pit)$.

\begin{claim} \label{claim:algat}
$\mcA_{\ibar}^{\mbt} := {\rm im}(\Pit)$ is a unital \emph{*}-subalgebra of $\mcA_{\ibar}$. Its Hilbert-Schmidt complement, $\mathcal{C}_{\ibar}^{\mbt} := {\rm im}(\Pit^\perp)$ is a subspace of $\mcA_{\ibar}$ but does not close under multiplication as an algebra.
\end{claim}

\begin{proof}
Given Claim \ref{claim:projPit}, it is straightforward to check that $\mcA_{\ibar}^{\mbt}$ is closed under complex linear combinations and multiplication. It is also easy to verify that $\mcA_{\ibar}^{\mbt}$ is closed under the adjoint, i.e.\ given any $A_{\ibar} \in \mcA_{\ibar}^{\mbt}$, then $\Pit^{\perp}(A_{\ibar}^\dagger) = 0$, thus $A_{\ibar}^\dagger \in \mcA_{\ibar}^{\mbt}$. The multiplicative identity  satisfies $\Pit^{\perp}(\mathds1_{\ibar})=0$, hence is also an element of $\mcA_{\ibar}^{\mbt}$. 

Similarly, one checks that $\mathcal{C}_{\ibar}^{\mbt}$ is a vector subspace of $\mcA_{\ibar}$ over $\mathbb{C}$, i.e.~closed under linear combinations.~However, $\mathcal{C}_{\ibar}^{\mbt}$ does not close under multiplication.~We show this by contradiction.~For any $A_{\ibar},B_{\ibar} \in \mathcal{C}_{\ibar}^{\mbt}$, let us assume that $(A_{\ibar}B_{\ibar}) \in \mathcal{C}_{\ibar}^{\mbt}$.~By Hilbert-Schmidt orthogonality, this implies that for any $f_{\ibar} \in \mcA_{\ibar}^{\mbt}$, it holds that $(f_{\ibar},A_{\ibar}B_{\ibar}) = 0$.~This holds in particular also for $f_{\ibar} = \mathds1_{\ibar}$, i.e.~$0 = (\mathds1_{\ibar},A_{\ibar}B_{\ibar}) = \Tr(A_{\ibar}B_{\ibar}) = (A_{\ibar}^{\dagger},B_{\ibar})$.~This implies that for any $B_{\ibar} \in \mathcal{C}_{\ibar}^{\mbt}$, then $A_{\ibar}^{\dagger} \in \mcA_{\ibar}^{\mbt}$.~Using the result above that $\mcA_{\ibar}^{\mbt}$ is closed under the adjoint, we get $A_{\ibar} \in \mcA_{\ibar}^{\mbt}$.~This is a contradiction.
\end{proof}

\noindent {\bf Remark.} Since $\mathds1_{\ibar} \in \mcA_{\ibar}^{\mbt}$, we have $\Tr(f_{\ibar}) = (\mathds1_{\ibar},f_{\ibar})=0$, for any $f_{\ibar} \in \mathcal{C}_{\ibar}^{\mbt}$. That is, all elements of $\mathcal{C}_{\ibar}^{\mbt}$ are traceless. But the converse is not true i.e.\ not all traceless operators are elements of $\mathcal{C}_{\ibar}^{\mbt}$. For example, suppose $\mcG=\mathbb{Z}_2$, then $A_j \otimes \sigma^x_{S} \in \mcA_{\ibar}^{\mbt}$ and is traceless, for any $A_j \in \mcA_{j}$.

\begin{claim} \label{claim:ugs-inv}
For any $f_{\ibar} \in \mcA_{\ibar}$, it holds that \emph{(cfr.~Eq.~\eqref{ugscomm})}
\be [\Pit(f_{\ibar}), \mathds1_j \otimes U_S^g] = 0 \,, \quad \forall g\in \mcG \,. \ee
\end{claim}

\begin{proof}
By direct computation via sum relabelling: for any $g\in \mcG$, we have
$(\mathds1_j \otimes U_S^g)\Pit(f_{\ibar}) (\mathds1_j \otimes U_S^{g\dag}) = \frac{1}{|\mcG|}\sum_{g'} (\mathds1_j \otimes U_S^{gg'}) f_{\ibar} (\mathds1_j \otimes U_S^{gg'\dag}) = \Pit(f_{\ibar})$.
\end{proof}

\noindent {\bf Remark.} We further have, for any $f_{\ibar}\in \mcA_{\ibar}$,
\be \label{pitcommute}
\Pit(f_{\ibar}) = f_{\ibar} \quad \Leftrightarrow \quad [f_{\ibar}, \mathds1_j \otimes U_S^g] = 0 \,, \; \forall g\in \mcG \ee
where the forward implication is a corollary of Claim \ref{claim:ugs-inv} above, and the converse follows easily by direct computation of $\Pit(f_{\ibar})$. Thus, the subalgebra $\mcA_{\ibar}^{\mbt}$ is completely characterised by invariance under local $U_S^{g}$-translations.

\begin{claim} \label{claim:projPid}
The map $\Pid: \mcA_{\ibar} \to \mcA_{\ibar}$ defined by
\be \label{piddef2}
\hat{\Pi}_{\mbd}(f_{\ibar}) := \sum_{g \in \mcG} (\ket{g}\!\bra{g}_j \otimes \mathds1_S) \, f_{\ibar} \, (\ket{g}\!\bra{g}_j \otimes \mathds1_S)\,, \qquad f_{\ibar} \in \mcA_{\ibar}
\ee
is an orthogonal projector, with complement $\Pid^{\perp} := \hat{\mathds1}_{\ibar} - \Pid$ defined with respect to the Hilbert-Schmidt inner product on $\mcA_{\ibar}$.
\end{claim}

\begin{proof}
Proof again follows straightforwardly by direct computation. For idempotence: for any $f_{\ibar} \in \mcA_{\ibar}$,
\be \Pid^2(f_{\ibar}) = \sum_{g,g'} \delta_{g,g'}^2 (\ket{g'}\!\bra{g}_j \otimes \mathds1_S) f_{\ibar} (\ket{g}\!\bra{g'}_j \otimes \mathds1_S) = \Pid(f_{\ibar})  \ee
so that, $\Pid^2 = \Pid$. Also, $(\Pid^{\perp})^2 = (\hat{\mathds1}_{\ibar} - \Pid)^2 = \Pid^{\perp}$, using idempotence of $\Pid$. Thus, $\Pid$ and $\Pid^{\perp}$ are projectors on $\mcA_{\ibar}$. For orthogonality with respect to the Hilbert-Schmidt inner product (using the same notation as in Claim \ref{claim:projPit} above): for any $A_{\ibar},B_{\ibar} \in \mcA_{\ibar}$, we have
\be (\Pid(A_{\ibar}), B_{\ibar}) = \sum_{g} \Tr \Big( (\ket{g}\!\bra{g}_j \otimes \mathds1_S)A_{\ibar}^\dagger (\ket{g}\!\bra{g}_j \otimes \mathds1_S) B_{\ibar} \Big) = (A_{\ibar},\Pid(B_{\ibar})) \ee
using cyclicity and linearity of $\Tr$. Then, $(\Pid^{\perp}(A_{\ibar}),B_{\ibar}) = (A_{\ibar},B_{\ibar}) - (\Pid(A_{\ibar}),B_{\ibar}) = (A_{\ibar},\Pid^{\perp}(B_{\ibar}))$, using orthogonality of $\Pid$. Lastly, $\Pid$ and $\Pid^{\perp}$ are Hilbert-Schmidt complements: for any $A_{\ibar} \in {\rm im}(\Pid), B_{\ibar} \in {\rm im}(\Pid^{\perp})$, we have $(A_{\ibar},B_{\ibar}) = (\Pid(A_{\ibar}),B_{\ibar}) - (\Pid(A_{\ibar}),\Pid(B_{\ibar})) = 0$, using that $\Pid$ is an orthogonal projector.
\end{proof}

\noindent {\bf Remark.} From Claim \ref{claim:projPid}, it follows that ${\rm im}(\Pid) = {\rm ker}(\Pid^\perp)$, and vice-versa. We thus have the algebra decomposition $\mcA_{\ibar}=\mcA_{\ibar}^{\mbd}\oplus\mathcal{C}_{\ibar}^{\mbd}$ given in \eqref{algddecomp}, with $\mcA_{\ibar}^{\mbd} := {\rm im}(\Pid) = {\rm ker}(\Pid^\perp)$ and $\mathcal{C}_{\ibar}^{\mbd} := {\rm im}(\Pid^\perp) = {\rm ker}(\Pid)$.

\begin{claim}
$\mcA_{\ibar}^{\mbd} := {\rm im}(\Pid)$ is a unital \emph{*}-subalgebra of $\mcA_{\ibar}$. Its Hilbert-Schmidt complement, $\mathcal{C}_{\ibar}^{\mbd} := {\rm im}(\Pid^\perp)$ is a subspace of $\mcA_{\ibar}$ but does not close under multiplication as an algebra.
\end{claim}

The proof is completely analogous to that for Claim \ref{claim:algat}, and thus  omitted.\\

\noindent {\bf Remark.} Since $\mathds1_{\ibar} \in \mcA_{\ibar}^{\mbd}$, we again have that $\Tr(f_{\ibar}) = (\mathds1_{\ibar},f_{\ibar})=0$, for any $f_{\ibar} \in \mathcal{C}_{\ibar}^{\mbd}$. That is, all elements of $\mathcal{C}_{\ibar}^{\mbd}$ are traceless. But the converse is not true. For example, suppose $\mcG=\mathbb{Z}_2$, then $\sigma_j^z\otimes A_S \in \mcA_{\ibar}^{\mbd}$ and is traceless, for any $A_S \in \mcA_{S}$.

\begin{claim} \label{claim:gginv}
For any $f_{\ibar}\in \mcA_{\ibar}$, it holds that \emph{(cfr.~Eq.~\eqref{pvmgcomm})}
\be [\Pid(f_{\ibar}), \ket{g}\!\bra{g}_j \otimes \mathds1_S] = 0 \,, \quad \forall g\in \mcG \,. \ee
\end{claim}
\begin{proof}
By direct computation: for any $g \in \mcG$, we have $[\Pid(f_{\ibar}), \ket{g}\!\bra{g}_j \otimes \mathds1_S] = \sum_{g'} [\ket{g'}\!\bra{g'}_j,\ket{g}\!\bra{g}_j] \otimes \braket{g'|f_{\ibar}|g'}_j = \sum_{g'} \delta_{g,g'}(\ket{g'}\!\bra{g}_j - \ket{g}\!\bra{g'}_j) \otimes \braket{g'|f_{\ibar}|g'}_j = 0$. \end{proof}

\noindent {\bf Remark.} We further have, for any $f_{\ibar}\in \mcA_{\ibar}$,
\be \Pid(f_{\ibar}) = f_{\ibar} \quad \Leftrightarrow \quad [f_{\ibar}, \ket{g}\!\bra{g}_j \otimes \mathds1_S] = 0 \,, \; \forall g\in \mcG \ee
where the forward implication follows from Claim \ref{claim:gginv} above, and the converse follows by direct computation of $\Pid(f_{\ibar})$.

\begin{claim}
The following commutation relations hold on $\mcA_{\ibar}$, for any $g_j \in \mcG$,
\be [\hat{\mbfU}_{\ibar}^{g_i,g_j}, \hat{\Pi}_{\mbd}]=0, \quad [\hat{\mbfU}_{\ibar}^{g_i,g_j}, \hat{\Pi}_{\mbt}]=0, \quad [\hat{\Pi}_{\mbt}, \hat{\Pi}_{\mbd}]=0 \,. \ee
\end{claim}

\begin{proof}
Let $f_{\ibar} \in \mcA_{\ibar}$ be any operator.~Using equations \eqref{TPSschangemap} and \eqref{piddef}, by direct computation and sum relabelings we have that: $\hat{\mbfU}_{\ibar}^{g_i,g_j}\hat{\Pi}_{\mbd}(f_{\ibar}) = \sum_{g}\ket{gg_jg_i}\!\bra{gg_jg_i}_j \otimes U_S^{gg_j}\braket{g^{-1}|f_{\ibar}|g^{-1}}_j U_S^{gg_j\dag} = \hat{\Pi}_{\mbd}\hat{\mbfU}_{\ibar}^{g_i,g_j}(f_{\ibar})$.~From \eqref{TPSschangemap} and \eqref{pitdef}, it directly follows that $\hat{\mbfU}_{\ibar}^{g_i,g_j}\hat{\Pi}_{\mbt}(f_{\ibar}) = \hat{\Pi}_{\mbt}\hat{\mbfU}_{\ibar}^{g_i,g_j}(f_{\ibar})$, since the factors $(\ket{gg_i}\!\bra{g_jg^{-1}}_j \otimes U_S^g)$ and $(\mathds1_j \otimes U_S^{g'})$ commute for Abelian $\mcG$, any $g,g'\in \mcG$.~Finally, by definitions \eqref{pitdef} and \eqref{piddef}, it is clear that these two projectors commute, since their corresponding generating factors $(\ket{g}\!\bra{g}_j \otimes \mathds1_S)$ and $(\mathds1_j \otimes U_S^{g'})$ commute for any $g,g'\in \mcG$.
\end{proof}

\begin{claim}
The algebra decompositions \eqref{algtdecomp} and \eqref{algddecomp} are preserved under the QRF transformation, i.e.\ on $\mcA_{\ibar}$, for any $g_i,g_j \in \mcG$, we have
\be 
\hat{\Pi}_{\mbt} \circ \hat{V}_{i \to j}^{g_i,g_j} = \hat{V}_{i \to j}^{g_i,g_j} \circ \hat{\Pi}_{\mbt} \,, \quad
\hat{\Pi}_{\mbd} \circ \hat{V}_{i \to j}^{g_i,g_j} = \hat{V}_{i \to j}^{g_i,g_j} \circ \hat{\Pi}_{\mbd} 
\ee
where, the projectors $\Pit$ and $\Pid$ on the LHS of each equation are implicitly understood to act on $\mcA_{\jbar} = \hat{V}_{i \to j}^{g_i,g_j}(\mcA_{\ibar})$, while on the RHS they act on $\mcA_{\ibar}$.
\end{claim}

\begin{proof}
Let $f_{\ibar} \in \mcA_{\ibar}$ be any operator. Using equations \eqref{eq:Vitoj}, \eqref{pitdef} and that $\mcG$ is Abelian, we have
\begin{align}
    \hat{\Pi}_{\mbt} \hat{V}_{i \to j}^{g_i,g_j}(f_{\ibar}) &= \frac{1}{|\mcG|}\sum_{g,g',g''} \ket{g_i g'}\!\bra{g_i g''^{-1}}_i \otimes U_S^{gg'} \braket{g_j g'^{-1}|f_{\ibar}|g_jg''}_j U_S^{g'' g^{-1}} \nonumber \\
    &= \frac{1}{|\mcG|}\sum_{g',g''} \ket{g_i g'}\!\bra{g_i g''^{-1}}_i \otimes U_S^{g'} \bra{g_j g'^{-1}} \, \sum_g (\mathds1_j \otimes U_S^g) f_{\ibar} (\mathds1_j \otimes U_S^{g\dag}) \,  \ket{g_jg''}_j U_S^{g''} \nonumber \\
    &= \hat{V}_{i \to j}^{g_i,g_j}\hat{\Pi}_{\mbt} (f_{\ibar})\,.
\end{align}
Similarly for $\Pid$ using equations \eqref{eq:Vitoj} and \eqref{piddef}, it can be shown by direct computation that, $\hat{\Pi}_{\mbd} \hat{V}_{i \to j}^{g_i,g_j}(f_{\ibar}) = \sum_g \ket{g_i g}\!\bra{g_i g}_i \otimes U_S^g\braket{g_jg^{-1}|f_{\ibar}|g_jg^{-1}}_j U_S^{g\dag}  =  \hat{V}_{i \to j}^{g_i,g_j}\hat{\Pi}_{\mbd}(f_{\ibar})$.
\end{proof}

\subsection{Proofs: Theorems \ref{thm:Slocals} and \ref{thm:jlocals} and Lemma \ref{lemma:pinpd-uin}} \label{app:thmproof-S&jlocals}

\noindent\textbf{Theorem \ref{thm:Slocals}} ($\mathbf{S}$-\textbf{operators})\textbf{.}
\emph{Let $\mathds1_j\otimes f_S \in \mcA_{\ibar}$ be some $S$-local operator relative to frame $R_i$. Then $\mathds1_j\otimes f_S\in\mcA_{\ibar}^{X}$ for some bilocal unitary $X$ (and some orientations $g_i,g_j\in\mcG$) if and only if it lies in the image of $\Pit$, i.e.\ if and only if $f_S$ is translation invariant, $[f_S,U_S^g]=0$, $\forall\,g\in\mcG$. In particular, $\mathds1_j\otimes f_S\in\mcA_{\ibar}^\mbt$ resides in exactly all those subalgebras $\mcA_{\ibar}^{X}$, with $X$ of the form $X=Y_j\otimes Z_S$, where the $S$-unitary $Z_S$ commutes with $f_S$, $[Z_S,f_S]=0$, and $Y_j$ is an arbitrary $R_j$-unitary.}

\begin{proof}
$\mathds1_j\otimes f_S\in\mcA_{\ibar}^{X}$ with $X=Y_j\otimes Z_S$ means 
\be\label{inaux}
\hat{\mbfU}_{\ibar}^{g_i,g_j}(\mathds1_j \otimes f_S) = \mathds1_j \otimes \hat{Z}_{S}^\dagger(f_S)\,.
\ee
Now using \eqref{TPSschangemap} and the decomposition \eqref{algtdecomp}, we find
\be\label{inaux2}
\hat{\mbfU}_{\ibar}^{g_i,g_j}(\mathds1_j \otimes f_S) = \mathds1_j \otimes f_S^{\mbt} +  \sum_{g}\ket{gg_i}\!\bra{gg_i}_j \otimes \hat U_S^{g} (f_S^{\mbt\perp})\,, 
\ee
writing for simplicity $(\mathds1_j\otimes f_S)^\mbt=\mathds1_j \otimes f_S^{\mbt}$ and exploiting that, due to \eqref{pitcommute}, $\mathds1_j \otimes f_S^{\mbt}$ commutes with $\mbfU_{\ibar}^{g_i,g_j}$. Importantly, the second term cannot be written in the form of a local $S$-operator since, thanks to \eqref{pitcommute}, $[f_S^{\mbt\perp},U_S^g]\neq 0$ for at least one $g\in \mcG$. This implies, by linear independence, that we must have both,
\be 
f_S^\mbt=\hat Z^\dag_S(f_S)\,,\qquad\qquad\text{and}\qquad\qquad f_S^{\mbt\perp}=0\,,
\ee
in order for \eqref{inaux} to be true. Since therefore $f_S=f_S^\mbt$, we also have $[Z_S,f_S]=0$. 

Conversely, inserting $f_S^{\mbt\perp}=0$ into \eqref{inaux2} yields $\hat{\mbfU}_{\ibar}^{g_i,g_j}(\mathds1_j \otimes f_S) = \mathds1_j \otimes f_S$. Hence, $\mathds1_j\otimes f_S\in\mcA_{\ibar}^{X}$ with $X=Y_j\otimes Z_S$, where $[Z_S,f_S]=0$ and $Y_j$ is an arbitrary $R_j$-unitary.
\end{proof}

\noindent\textbf{Theorem \ref{thm:jlocals}} (\textbf{Frame-operators})\textbf{.}
\emph{Let $f_j \otimes \mathds1_S \in \mcA_{\ibar}$ be some $R_j$-local operator relative to frame $R_i$. Then $f_j\otimes\mathds1_S\in\mcA_{\ibar}^{X}$ for some bilocal unitary $X$ (and some orientations $g_i,g_j\in\mcG$) if and only if it lies in the image of $\Pid$,i.e.\ if and only if $f_j$ is diagonal in the frame orientation basis, $f_j=\sum_g f_g\ket{g}\!\bra{g}_j$. In particular, $f_j\otimes\mathds1_S\in\mcA_{\ibar}^\mbd$ resides in exactly all those subalgebras $\mcA_{\ibar}^{X}$ with $X=Y_j\otimes Z_S$ such that $Z_S$ is an arbitrary $S$-unitary and the action of the $R_j$-unitary $Y_j$ on $f_j$ is equivalent to that of the (unitary) \emph{parity-swap operator}
\be\label{eq:Pswap}
P_j^{g_i,g_j}:=\sum_g\ket{g_ig}\!\bra{g_jg^{-1}}_j
\ee
in the form 
\be \label{eq:Pswap2}
\hat Y^\dag_j(f_j)=\hat{P}_j^{g_i,g_j}(f_j)\,.
\ee
}

\begin{proof}
$f_j\otimes\mathds1_S\in\mcA_{\ibar}^{X}$ with $X=Y_j\otimes Z_S$ means 
\be 
\label{inaux3}
\hat{\mbfU}_{\ibar}^{g_i,g_j}(f_j \otimes \mathds1_S) = \hat Y_j^\dag(f_j) \otimes \mathds1_S\,.
\ee 
Invoking the decomposition \eqref{algddecomp} and \eqref{TPSschangemap} yields
\be \label{inaux4}
\hat{\mbfU}_{\ibar}^{g_i,g_j}(f_j \otimes \mathds1_S)=\sum_{g\in\mcG}\braket{g_jg^{-1}|f_j^\mbd|g_jg^{-1}}\ket{gg_i}\!\bra{gg_i}_j\otimes\mathds1_S+\sum_{\substack{g,g'\in\mcG\\ g\neq g'}}\braket{g_jg^{-1}|f_j^{\mbd\perp}|g_jg'^{-1}}\ket{gg_i}\!\bra{g'g_i}_j\otimes U_S^{gg'^{-1}}\,,
\ee
recalling that $f_j^\mbd$ and $f_j^{\mbd\perp}$ are diagonal and purely off-diagonal in the $\mcG$-basis. The second term cannot be written as a purely $R_j$-local expression due to linear independence and because $U_S^{gg'^{-1}}\neq\mathds1_S$ whenever $g\neq g'$. \eqref{inaux3} can thus only hold provided that 
\be 
\sum_{g\in\mcG}\braket{g_jg^{-1}|f_j^\mbd|g_jg^{-1}}\ket{gg_i}\!\bra{gg_i}_j = \hat Y_j^\dag(f_j)\,,\qquad\qquad\text{and}\qquad\qquad f_j^{\mbd\perp}=0\,.
\ee 
Hence, $f_j=f_j^\mbd$. Conjugating the left equation with the parity-swap operator \eqref{eq:Pswap} restores $f_j^\mbd=f_j$ on the left hand side, giving
\be
f_j^\mbd=f_j=(\hat{P}_j^{g_i,g_j})^\dag\hat Y_j^\dag(f_j)\,,
\ee 
that is $[f_j,Y_jP_j^{g_i,g_j}]=0$.

Conversely, inserting $f_j^{\mbd\perp}=0$, and so $f_j=f_j^\mbd$, into \eqref{inaux4} and again using the parity-swap \eqref{eq:Pswap} produces ${\hat{\mbfU}_{\ibar}^{g_i,g_j}(f_j \otimes \mathds1_S) =\hat P_j^{g_i,g_j}(f_j)\otimes\mathds1_S}$. That is, $f_j\otimes\mathds1_S\in\mcA_{\ibar}^{X}$ for all $X=Y_j\otimes Z_S$ such that $Z_S$ is an arbitrary $S$-unitary and $Y_j$ satisfies $\hat Y_j^\dag(f_j)=\hat P_j^{g_i,g_j}(f_j)$, or equivalently,  $[f_j,Y_jP_j^{g_i,g_j}]=0$.
\end{proof}

\noindent\textbf{Lemma~\ref{lemma:pinpd-uin}.}
\emph{
For any $f_{\ibar} \in \mcA_{\ibar}$ and frame orientations $g_i, g_j \in \mcG$, it holds that
\be \hat{\mbfU}_{\ibar}^{g_i,g_j} \hat{\Pi}_{\mbd} \hat{\Pi}_{\mbt} (f_{\ibar}) = (\hat {P}_{j}^{g_i,g_j}\otimes\hat{\mathds1}_S)\,\hat{\Pi}_{\mbd} \hat{\Pi}_{\mbt} (f_{\ibar}) \,, 
\ee
where $P_j^{g_i,g_j}$ is the parity-swap operator \eqref{Pswap}.~In other words, for any $f_{\ibar} \in \mcA_{\ibar}$, we have 
$\hat{\Pi}_{\mbd} \hat{\Pi}_{\mbt} (f_{\ibar})\in\mcA_{\ibar}^{X}$ with $X=(P_j^{g_i,g_j})^\dag\otimes\mathds1_S$.
}

\begin{proof} 
The proof is straightforward. Using \eqref{ugscomm}, the fact that $\braket{g|\Pid(A_{\ibar})|g'}_j=\delta_{g,g'}\braket{g|\Pid(A_{\ibar})|g}_j$ is diagonal for any $A_{\ibar}\in\mcA_{\ibar}$ and the definition of $\mbfU_{\ibar}^{g_i,g_j}$ in \eqref{TPSschangemap},
we immediately find for any $f_{\ibar}\in \mcA_{\ibar}$
\begin{align}
    \hat{\mbfU}_{\ibar}^{g_i,g_j}  \hat{\Pi}_{\mbd}\hat{\Pi}_{\mbt}(f_{\ibar}) 
    &= \sum_{g,g'} \left(\ket{g_ig}\!\bra{g_jg^{-1}}_j \otimes U_S^g\right) 
    \hat{\Pi}_{\mbd}\hat{\Pi}_{\mbt}(f_{\ibar}) \left(\ket{g_jg'^{-1}}\!\bra{g_ig'}_j\otimes U_S^{g'^{-1}}\right)\\
     &=\sum_{g}\left(\ket{g_ig}\!\bra{g_jg^{-1}}_j\otimes\mathds1_S\right)\hat{\Pi}_{\mbd}\hat{\Pi}_{\mbt}(f_{\ibar})
     \left(\ket{g_jg^{-1}}\!\bra{g_ig}_j\otimes\mathds1_S\right)
     \\
       &=\sum_{g,g'}\left(\ket{g_ig}\!\bra{g_jg^{-1}}_j\otimes\mathds1_S\right)\hat{\Pi}_{\mbd}\hat{\Pi}_{\mbt}(f_{\ibar})
     \left(\ket{g_jg'^{-1}}\!\bra{g_ig'}_j\otimes\mathds1_S\right)
     \\
    &=(\hat {P}_{j}^{g_i,g_j}\otimes\hat{\mathds1}_S)\,\hat{\Pi}_{\mbd}\hat{\Pi}_{\mbt}(f_{\ibar}) \,.
\end{align}
In the last step we invoked the definition \eqref{Pswap} of the parity-swap operator.
\end{proof}

\subsection{Subsystem states and entropies: Proofs}\label{APP:statesandentropy}

\noindent
\textbf{Theorem~\ref{claim:purestatesrhoSZS} (Unitarily related subsystem states).}
\emph{Let $\rho_{\ibar}\in\mathcal S(\mathcal H_{\ibar})$ be a pure state relative to frame $R_i$ in orientation $g_i$.~Let $\rho_{\jbar}=\hat{V}_{i\to j}^{g_i,g_j}(\rho_{\ibar})\in\mathcal S(\mathcal H_{\jbar})$ be the corresponding QRF-related pure state relative to frame $R_j$ in orientation $g_j$.~Then, the $S$-subsystem states $\rho_{S|R_i}=\Tr_j(\rho_{\ibar})$ and $\rho_{S|R_j}=\Tr_i(\rho_{\jbar})$ in the two frame perspectives are unitarily related, i.e.~$\exists\, Z_S\in\mathcal U(\mathcal H_S)$ s.t.
\be\label{rhoSZSpure}
\rho_{S|R_j}=\hat Z_S^{\dagger}(\rho_{S|R_i})
\ee
if and only if $\rho_{\ibar}\in\mcA_{\ibar}^{X}$ with $X=Y_j\otimes Z_S$ for some local unitary $Y_j \in \mathcal{U}(\mcH_{j})$, i.e.
\be\label{PSrhoibarinAU}
\hat\mbfU_{\ibar}^{g_i,g_j}(\rho_{\ibar})=(\hat Y_{j}^\dagger\otimes\hat{Z}_S^{\dagger})(\rho_{\ibar})\;.
\ee
}

\begin{proof} $\eqref{rhoSZSpure}\Rightarrow\eqref{PSrhoibarinAU}$:
Let $\rho_{\ibar}=\ket\psi\!\bra\psi_{\ibar}$ be a pure state describing the bipartite
system $R_j+S$ relative to frame $R_i$, $i=1,2$, $j\neq i$, in orientation $g_i\in\mathcal G$.~Let then $\ket\psi_{\jbar}=V_{i\to j}^{g_i,g_j}\ket\psi_{\ibar}$ be the corresponding QRF-related state relative to frame $R_j$ in orientation $g_j$, and let us assume that \eqref{rhoSZSpure} holds.~Writing $\ket\psi_{\ibar}$ and $\ket\psi_{\jbar}$ in their (a priori distinct) Schmidt decomposition,
\be\label{eq:SDiandjbar}
\ket\psi_{\ibar}=\sum_{a}\sqrt{\lambda_a}\,\ket{e_a}_{j}\otimes\ket{e_a}_{S}\qquad,\qquad \ket\psi_{\jbar}=\sum_{a}\sqrt{\zeta_a}\,\ket{\ell_a}_{i}\otimes\ket{\ell_a}_{S},
\ee
the ensuing reduced $S$ states $\rho_{S|R_i}$ and $\rho_{S|R_j}$ are diagonal
\be\label{SVDrhoSreltoiandj1}
\rho_{S|R_i}=\sum_{a}\lambda_a\,\ket{e_a}\!\bra{e_a}_{S}\qquad,\qquad \rho_{S|R_j}=\sum_{a}\zeta_a\,\ket{\ell_a}\!\bra{\ell_a}_{S}\;.
\ee 
Since unitarily equivalent operators have the same singular values, the requirement $\rho_{S|R_j}=\hat{Z}_S^{\dagger}(\rho_{S|R_i})$ as in \eqref{rhoSZSpure} demands that there exists an ordering of the summands in which $\lambda_a=\zeta_a$ in \eqref{SVDrhoSreltoiandj1} and $\ket{\ell_a}_{S}=Z_S^{\dagger}\ket{e_a}_{S}$. Should the spectrum be degenerate, one can always choose the non-unique basis elements of the degeneracy subspaces such that the second relation holds.
Hence, also the reduced ``other frame'' states $\rho_{j}=\Tr_S(\rho_{\ibar})$ and $\rho_{i}=\Tr_S(\rho_{\jbar})$ have the same spectrum so that there exists a unitary map $X_{i\to j}:\mathcal H_{j}\to\mathcal H_{i}$ such that $X_{i\to j}\ket{e_a}_{j}=\ket{\ell_a}_{i}$.~The states $\ket\psi_{\ibar}$ and $\ket\psi_{\jbar}$ in \eqref{eq:SDiandjbar} are thus related as
\be\label{eq:uritoj}
\ket\psi_{\jbar}=V_{i\to j}^{g_i,g_j}\ket\psi_{\ibar}=(X_{i\to j}\otimes Z_S^{\dagger})\ket\psi_{\ibar}\;,
\ee
from which, using the relation \eqref{IUVrelation} between $\mbfU_{\ibar}^{g_i,g_j}$, $V_{i\to j}^{g_i,g_j}$, and $\mathcal I_{j\to i}$, it follows that
\be\label{eq:UrhoibarA}
\mbfU_{\ibar}^{g_i,g_j}\ket\psi_{\ibar}=\mathcal I_{j\to i}(X_{i\to j}\otimes Z_S^{\dagger})\ket\psi_{\ibar}=(Y_j^{\dagger}\otimes\ Z_S^{\dagger})\ket\psi_{\ibar}\;.
\ee
In the second equality we used the fact that $\mathcal I_{j\to i}\circ(X_{i\to j}\otimes\mathds 1_S)$ yields a local $R_j$-unitary (cfr.~Eq.~\eqref{eq:Ijtoi}).~Therefore, $\rho_{\ibar}=\ket\psi\!\bra\psi_{\ibar}$ belongs to the operator subalgebra $\mcA_{\ibar}^{X}$, with $X=Y_j\otimes Z_S$.

$\eqref{PSrhoibarinAU}\Rightarrow\eqref{rhoSZSpure}$:~The sufficient implication can be proved by direct computation by taking the partial trace $\Tr_j$ of both sides of \eqref{PSrhoibarinAU}.~Due to the invariance of partial trace $\Tr_j$ under conjugation of its argument by any local $R_j$-unitary, the RHS of \eqref{PSrhoibarinAU} simply yields $\Tr_j[(\hat{Y}_j^\dagger\otimes\hat{Z}_S^\dagger)(\rho_{\ibar})]=\hat{Z}_S^\dagger(\rho_{S|R_i})$.~As for the LHS, we have
$$
\Tr_j[\hat{\mbfU}_{\ibar}^{g_i,g_j}(\rho_{\ibar})]=\Tr_j[\hat{\mathcal I}_{j\to i}(\rho_{\jbar})]=\Tr_i(\rho_{\jbar})=\rho_{S|R_j}\;,
$$
where, in the first equality, we used the relation \eqref{IUVrelation} between $\mbfU_{\ibar}^{g_i,g_j}$, $V_{i\to j}^{g_i,g_j}$, and $\mathcal I_{j\to i}$, while in the second equality we used that $\Tr_i(\mathcal I_{i\to j}\bullet\mathcal I_{j\to i})=\Tr_j(\bullet)$, which follows from the expression \eqref{eq:Ijtoi} for $\mathcal I_{j\to i}$.
\end{proof}

\noindent
\textbf{Corollary~\ref{claim:Renyipure} (QRF-invariant subsystem R\'enyi entropies).}
\emph{Let $\rho_{\ibar}$ be a pure state relative to frame $R_i$ in orientation $g_i$.~Let $\rho_{\jbar}=\hat{V}_{i\to j}^{g_i,g_j}(\rho_{\ibar})$ be the corresponding state relative to frame $R_j$ in orientation $g_j$.~Then, the $\alpha$-R\'enyi entropies of the $S$-subsystem states $\rho_{S|R_i}=\Tr_j(\rho_{\ibar})$ and $\rho_{S|R_j}=\Tr_i(\rho_{\jbar})$ in the two frame perspectives are the same for any value of $\alpha$, i.e.
\be\label{eq:RenyiSstate}
S_\alpha[\rho_{S|R_i}]=S_\alpha[\rho_{S|R_j}], \;\; \forall \alpha\in(0,1)\cup(1,\infty)
\ee
where
\be\label{eq:alpharenyi}
S_\alpha[\rho_{S|\bullet}]=\frac{1}{1-\alpha}\log\bigl(\Tr_S(\rho_{S|\bullet}^\alpha)\bigr),
\ee
if and only if $\rho_{\ibar} \in \mcA_{\ibar}^{X}$ for some bilocal unitary $X=Y_j\otimes Z_S$.~In other words, the entanglement spectrum agrees in both QRF perspectives, if and only if the global state lies in a TPS-invariant subalgebra.}

\begin{proof}~$(\Rightarrow)$: From the equality \eqref{eq:RenyiSstate} of $\alpha$-R\'enyi entropies, it follows that $\Tr_S(\rho_{S|R_i}^\alpha)=\Tr_S(\rho_{S|R_j}^\alpha)$ for all $\alpha\in(0,1)\cup(1,\infty)$ (cfr.~Eq.~\eqref{eq:alpharenyi}). In particular, this holds true for all integer values of $\alpha$ in $(0,1)\cup(1,\infty)$. But we also know that $\Tr_{S}(\rho_{S|R_i})=1=\Tr_{S}(\rho_{S|R_j})$ and $\Tr_{S}(\mathds1_{S|R_i})=\dim\mathcal H_S=\Tr_{S}(\mathds1_{S|R_j})$ so that the equality of the trace of powers of the $S$-reduced states in the two perspectives holds true also for $\alpha=0,1$. Then, recalling that for Hermitian operators such as $\rho_{S|R_i}$ and $\rho_{S|R_j}$, the equality of all \emph{words} $\Tr_S(\mathscr W(\rho_{S|R_i},\rho_{S|R_i}^{\dagger}))=\Tr_S(\mathscr W(\rho_{S|R_j},\rho_{S|R_j}^\dagger))$ amounts to any non-negative integer power of the two operators to have equal traces, by Specht's theorem \cite{Specht1940}, there exists a unitary $Z_S$ such that $\rho_{S|R_j}=\hat Z_S^\dagger(\rho_{S|R_i})$.~Then, by the necessary implication of Theorem~\ref{claim:purestatesrhoSZS}, we conclude that $\rho_{\ibar}\in\mathcal A_{\ibar}^{X}$ for some local unitary $X=Y_j\otimes Z_S$.

$(\Leftarrow)$: By the sufficient implication of Theorem~\ref{claim:purestatesrhoSZS}, we have that $\rho_{S|R_j}=\hat Z_S^\dagger(\rho_{S|R_i})$ for some local $S$ unitary $Z_S$. Thus, $\rho_{S|R_j}^\alpha=\exp(\alpha\log(\rho_{S|R_j}))=\hat Z_S^\dagger(\rho_{S|R_i}^\alpha)$ for any $\alpha$, from which the equality \eqref{eq:RenyiSstate} of $\alpha$-R\'enyi entropies follows.
\end{proof}

\noindent
\textbf{Lemma~\ref{claim:Uinv&tinvpure}.}
\emph{Let $\rho_{\ibar}$ be any $S$-translation invariant pure state relative to frame $R_i$ in orientation $g_i$, i.e.\ $[\mathds 1_{j}\otimes U^g_S, \rho_{\ibar}]=0$ for any $g\in\mathcal G$ (equivalently, $\rho_{\ibar}\in\mcA_{\ibar}^{\mbt}$).~Then, $\rho_{\ibar}\in\mcA_{\ibar}^{X}$ with bilocal unitary $X=Y_j\otimes\mathds1_S$,
\be\label{eq:YjUgSinv}
Y_j(g_i,g_j)=\sum_{g\in\mathcal G}\overline{q(g)}\ket{g^{-1}g_j}\!\bra{g_ig}_j\,,
\ee
and $q(g)\in\mathbb C$, $|q(g)|^2=1$.~The converse is however not true, i.e.\ $\rho_{\ibar}\in\mathcal{A}_{\ibar}^X$ does not imply $\rho_{\ibar}\in\mcA_{\ibar}^{\mbt}$}.

\begin{proof}~First, let us notice that for a pure state $\rho_{\ibar}=\ket\psi\!\bra\psi_{\ibar}$, local $S$-translation invariance $[\mathds 1_{j}\otimes U^g_S, \rho_{\ibar}]=0$ is equivalent to $(\mathds 1_{j}\otimes U^g_S)\ket\psi_{\ibar}=q(g)\ket\psi_{\ibar}$ for any $g\in\mathcal G$, with $q(g)\in\mathbb C$ such that $|q(g)|^2=1$.~Conjugation of $\rho_{\ibar}\in\mathcal A_{\ibar}^{\mbt}$ by $\mbfU_{\ibar}^{g_i,g_j}$ thus yields (cfr. Eq.~\eqref{TPSschangemap})
\begin{align}
    \hat{\mbfU}_{\ibar}^{g_i,g_j}(\rho_{\ibar})&=\sum_{g,g'\in\mathcal G}(\ket{gg_i}\!\bra{g_jg^{-1}}_j\otimes\mathds1_S)(\mathds1_j\otimes U_S^g)\rho_{\ibar}(\mathds1_j\otimes U^{g'^{-1}}_S)(\ket{g'^{-1}g_j}\!\bra{g_ig'}_j\otimes\mathds1_S)\notag\\
    &=\Bigl(\sum_{g\in\mathcal G}q(g)\ket{gg_i}\!\bra{g_jg^{-1}}_j\otimes\mathds1_S\Bigr)\rho_{\ibar}\Bigl(\sum_{g'\in\mathcal G}\overline{q(g')}\ket{g'^{-1}g_j}\!\bra{g_ig'}_j\otimes\mathds1_S\Bigr)\notag\\
    &=:(\hat{Y}_j(g_i,g_j)^\dagger\otimes\hat{\mathds1}_S)(\rho_{\ibar})\;.
\end{align}
That is, $\rho_{\ibar}\in\mcA_{\ibar}^{X}$ with local unitaries $X=Y_j\otimes Z_S$ such that $Z_S=\mathds1_S$ and $Y_j(g_i,g_j)$ given above (cfr.~Eq.~\eqref{eq:YjUgSinv}). As can be checked by direct computation, the converse statement is not true in general and we provided simple counterexamples of states $\rho_{\ibar}\in\mcA_{\ibar}^{X}$, but $\rho_{\ibar}\not\in\mathcal A_{\ibar}^{\mbt}$, in Examples~\ref{Ex:Wstate} and~\ref{Ex:nontranslinv} in Sec.~\ref{Sec:states&entanglement}.
\end{proof}

\noindent
\textbf{Theorem~\ref{thm_globindep} (QRF-invariant subsystem states).}
\emph{Given some subsystem state $\rho_{S|R_i}\in\mcS(\mcH_S)$ in $R_i$-perspective, its counterpart in $R_j$-perspective ($i\neq j$), $\rho_{S|R_j}=\Tr_i\bigl(\hat{V}_{i\to j}^{g_i,g_j}(\rho_{\ibar})\bigr)$, is \emph{independent of the global state} $\rho_{\ibar}$ from which $\rho_{S|R_i}$ originates, if and only if $\rho_{S|R_i}$ is translation-invariant.~In this case, we further have exact QRF-invariance, $\rho_{S|R_j}=\rho_{S|R_i}$.}

\begin{proof}
We first show by example that $\rho_{S|R_j}$ depends on $\rho_{\ibar}$ whenever $\rho_{S|R_i}$ is not translation-invariant.~Consider product states of the form
\be
\rho_{\ibar}=\ket{g}\!\bra{g}_j\otimes \rho_{S|R_i}=\ket{g}\!\bra{g}_j\otimes \left(\Pit^S(\rho_{S|R_i})+\Pit^{S\perp}(\rho_{S|R_i})\right)\,,
\ee
where $\Pit^S(\bullet) := \frac{1}{|\mcG|}\sum \limits_{g \in \mcG} U_S^g(\bullet) U^{g^{-1}}_S$ denotes the $G$-twirl over $\mcG$ for $S$ only (restriction of~\eqref{pitdef} to the $S$-tensor factor only), which is an orthogonal projector onto the translation-invariant sector, and $\Pit^{S\perp}=\hat{\mathds{1}}_S-\Pit^S$ is the projector onto its Hilbert-Schmidt orthogonal complement. From~\eqref{counterexample},~\eqref{rhoSUgSinv}  and Lemma~\ref{Lemma:tpseq}, we infer 
\be 
\rho_{\jbar}=\ket{g_ig_jg^{-1}}\!\bra{g_ig_jg^{-1}}_i\otimes\left(\Pit^S(\rho_{S|R_i})+\hat{U}_S^{g_jg^{-1}}\bigl(\Pit^{S\perp}(\rho_{S|R_i})\bigr)\right)\,.
\ee 
Now choose $\tilde g\in\mcG$ such that $\hat{U}_S^{\tilde g}\bigl(\Pit^{S\perp}(\rho_{S|R_i})\bigr)\neq \Pit^{S\perp}(\rho_{S|R_i})$.~For every fixed frame orientations $g_i,g_j\in\mcG$ (that determine the QRF-transformation), we can then build convex mixtures
\be 
\rho_{\ibar}=\left(p_e\,\ket{g_j}\!\bra{g_j}_j+p_{\tilde g}\,\ket{g_j\tilde g^{-1}}\!\bra{g_j\tilde g^{-1}}_j\right)\otimes\rho_{S|R_i}\,,
\ee 
with $p_e+p_{\tilde g}=1$ and $p_e,p_{\tilde g}\in[0,1]$, such that the QRF transformed state reads
\be 
\rho_{\jbar}=p_e\,\ket{g_i}\!\bra{g_i}_i\otimes \rho_{S|R_i}+p_{\tilde g}\,\ket{g_i\tilde g}\!\bra{g_i\tilde g}_i\otimes\left(\Pit^S(\rho_{S|R_i})+\hat{U}_S^{\tilde g}\bigl(\Pit^{S\perp}(\rho_{S|R_i})\bigr)\right)\,.
\ee 
Clearly, for different values of $p_e,p_{\tilde g}$, the resulting subsystem state $\rho_{S|R_j}=\Tr_i(\rho_{\jbar})$ will be different whenever $\rho_{S|R_i}$ is not translation-invariant, i.e.\ ${\Pit^{S\perp}(\rho_{S|R_i})\neq 0}$.~Hence, if $\rho_{S|R_j}$ is independent of the global state $\rho_{\ibar}$ such that $\rho_{S|R_i}=\Tr_j(\rho_{\ibar})$, then $\rho_{S|R_i}$ must be translation-invariant.

Conversely, suppose $\rho_{S|R_i}$ is translation-invariant.~Recall the observation in the main text that $\Pid(\rho_{\ibar})$ is $S$-translation-invariant in this case, i.e.\ ${\Pid(\rho_{\ibar})=\Pit\Pid(\rho_{\ibar})}$, which is the piece that matters for the trace.~Thanks to~\eqref{pivcomp} and Lemma~\ref{lemma:pinpd-uin}, this means that \emph{every} global state $\rho_{\ibar}$, such that $\Tr_j(\rho_{\ibar})=\rho_{S|R_i}$, gives rise to the same QRF-invariant subsystem state, $\rho_{S|R_j}=\rho_{S|R_i}$.
\end{proof}

\subsection{Further illustrative examples for Sec.~\ref{Sec:states&entanglement}}\label{app:Vex}

Here, we shall provide a few further illustrating examples for the results obtained in Sec.~\ref{Sec:states&entanglement}, complementing those already discussed in that section.\\

\begin{example}\label{Ex:GBS}\textbf{\emph{(Generalised Bell States).}}~Given a bipartite system $AB$ described by the Hilbert space $\mathcal H=\mathcal H_A\otimes\mathcal H_B$ with $\mathcal H_I\cong\mathbb C^n$, $I=A,B$, GB states are an orthonormal basis of $n^2$ maximally entangled states $\{\ket{\psi_{l,m}}\}_{l,m=0,\dots,n-1}$ defined as\footnote{For $n=2$, Eq.~\eqref{GBSdef} reproduces the well-known Bell states for bipartite qubit systems \cite{NielsenQIbook2010}: $\ket{\psi_{0,0}}=\frac{1}{\sqrt{2}}(\ket{0,0}+\ket{1,1})$, $\ket{\psi_{0,1}}=\frac{1}{\sqrt{2}}(\ket{0,0}-\ket{1,1})$, $\ket{\psi_{1,0}}=\frac{1}{\sqrt{2}}(\ket{0,1}+\ket{1,0})$, and $\ket{\psi_{1,1}}=\frac{1}{\sqrt{2}}(\ket{0,1}-\ket{1,0})$.} \cite{BennettPRLGBS}
\be\label{GBSdef}
\ket{\psi_{l,m}}=\frac{1}{\sqrt{n}}\sum_{j=0}^{n-1}e^{2\pi i\left(\frac{j\,m}{n}\right)}\ket{j,j+l}\;.
\ee
Consider a $\mathcal G=\mathbb Z_n$ translation-invariant system of $N=4$ particles on the discrete circle with bipartite 2-particle subsystem $S$ to be observed by the remaining two particles $R_1$ and $R_2$ ($g_1,g_2=e$). The group basis plays the role of the usual computational basis $\ket{l,m}$ and the exponentials in \eqref{GBSdef} are given by the group characters $\chi_k(g)=e^{2\pi i\left(\frac{kg}{n}\right)}$, $g\in\mathbb Z_n$, $k\in\{0,1,\dots,n-1\}$, i.e.
\be\label{ZnGBS}
\ket{\psi_{h,k}}=\frac{1}{\sqrt{|\mathcal G|}}\sum_{g\in\mathcal G}\chi_k(g)\ket{g,g+h}=:\ket{h;\chi_k}\;.
\ee
Let then $\rho_{\bar1}=\ket{\psi}\bra{\psi}_{\bar1}$ be a pure state with $\ket{\psi}_{\bar1}=\ket{\varphi}_{2}\otimes\ket{h;\chi_k}_{S}$, for some arbitrary state $\ket{\varphi}_2=\sum_{g\in\mathcal G}\varphi(g)\ket{g}_{2}$ and the system $S$ prepared in a maximally entangled GB state \eqref{ZnGBS}.~Then, using the property $U^g_S\ket{h;\chi_k}_{S}=(U^g\otimes U^g)\ket{h;\chi_k}_{S}=\chi_k(g^{-1})\ket{h;\chi_k}_{S}$, GB states are translation-invariant up to a phase given by the corresponding $\mathbb Z_n$ characters.~As proved in Lemma~\ref{claim:Uinv&tinvpure}, this means that the pure state $\rho_{\bar1}$ belongs to $\mcA_{\bar1}^{X}$ for some local unitary $X=Y_2\otimes Z_S$ with $Z_S=\mathds1_S$. Specifically, recalling the 
expression \eqref{TPSschangemap} for $\mbfU_{\bar1}$, we have
\be
\mbfU_{\bar1}\ket\psi_{\bar1}=\sum_{g\in\mathcal G}\varphi(g^{-1})\chi_k(g^{-1})\ket{g}_{2}\otimes\ket{h;\chi_k}_{S}=(Y_2^\dagger\otimes\mathds1_S)\ket\psi_{\bar1}\;,
\ee
with (cfr.~\eqref{YjUgSinv} with $q(g)=\chi_k(g^{-1})$ and $g_1,g_2=e$)
\be
Y_2=\sum_{g\in\mathcal G}\chi_k(g)\ket{g^{-1}}\!\bra{g}_2\;.
\ee
Note that unitarity of $Y_2$ follows from the property $\chi_k(g)\overline{\chi_k(g)}=\chi_k(g)\chi_k(g^{-1})=1$.~Correspondingly, the QRF transformed state relative to particle $R_2$ reads 
\begin{align}
\ket{\psi}_{\bar2}=V_{1\to 2}^{e,e}\ket{\psi}_{\bar1}&=\Bigl(\sum_{g\in\mathcal G}\ket{g}_1\otimes\bra{g^{-1}}_2\otimes U_S^g\Bigr)\ket{\psi}_{\bar1}\notag\\
&=\Bigl(\sum_{g\in\mathcal G}\varphi(g^{-1})\chi_k(g^{-1})\ket{g}_{1}\Bigr)\otimes\ket{h;\chi_k}_{S}\notag\\
&=\mathcal I_{1\to 2}(Y_2^\dagger\otimes\mathds1_S)\ket{\psi}_{\bar1}\;.
\end{align}
That is, compatibly with Theorem~\ref{claim:purestatesrhoSZS}, $\ket{\psi}_{\bar2}$ will still be a separable pure state with the subsystem $S$  in the same maximally entangled GB state as in $R_i$-perspective.~Moreover, in line with Corollary~\ref{claim:Renyipure}, the amount of correlations in the total state relative to both frames is the same and the entanglement entropy as well as all $\alpha$-R\'enyi entropies of the $S$-reduced states are equal in both perspectives.~Specifically, the total pure state remains separable before and after the QRF transformation, the $S$-subsystem state is a pure GB state, and its von Neumann entropy vanishes.

In the example just considered, only $S$ was in a GB state and the global separable $R_2+S$ state was $\mbfU_{\bar1}$-invariant up to local unitaries.~Let us now consider the case $N=3$ with particles $R_2$ and $S$ together in a GB state relative to particle $R_1$, namely $\ket\psi_{\bar1}=\ket{\psi_{h,k}}_{\bar1}$.~Then, 
\be
\mbfU_{\bar1}\ket\psi_{\bar1}=\Bigl(\frac{1}{\sqrt{|\mathcal G|}}\sum_{g\in\mathcal G}\chi_k(g^{-1})\ket{g}_2\Bigr)\otimes\ket{h}_S
\ee
that is, $\rho_{\bar1}=\ket\psi\!\bra\psi_{\bar1}$ does not reside in any of the subalgebras $\mcA_{\bar1}^{X}$.~The global maximally entangled state in $R_1$-perspective becomes separable in $R_2$-perspective.~Compatibly with Theorem~\ref{claim:purestatesrhoSZS} and Corollary~\ref{claim:Renyipure}, $\rho_{S|R_1}=\frac{1}{|\mathcal G|}\mathds1_S$, $S_\alpha[\rho_{S|R_1}]=\log|\mathcal G|$ while $\rho_{S|R_2}=\ket h\!\bra h_{S}$, $S_\alpha[\rho_{S|R_2}]=0$.
\end{example}

\begin{example}\label{Ex:GHZ}\textbf{\emph{(GHZ States).}}~As a further example with $S$ now being a multipartite system, let us consider a discrete $\mathbb Z_n$-system of $N>4$ particles. Let then $g_1=g_2=e$ and let the total state $\ket{\psi}_{\bar1}\in\mathcal H_{\bar1}\simeq\mathcal H^{\otimes(N-1)}$ relative to some particle $R_1$ be a pure separable state
\be\label{eq:totstateGHZex}
\ket{\psi}_{\bar1}=\ket{\varphi}_{2}\otimes\ket{GHZ}_{S}\;,
\ee
with the subsystem $S$ in a $(N-2)$-partite GHZ state \cite{GHZpaper} as e.g.
\be\label{Gsystemghz}
\ket{GHZ}_{S}=\frac{1}{\sqrt{|\mathcal G|}}\sum_{g\in\mathcal G}\ket{g}^{\otimes(N-2)}\;,
\ee
and particle $R_2$ in a generic state vector $\ket{\varphi}_{2}$.~The states of the form \eqref{Gsystemghz} are clearly $U_g^{\otimes(N-2)}$-invariant and, as can be checked by direct computation using the expression \eqref{eq:Vitoj} of the QRF transformation $V_{1\to 2}$, the $S$-subsystem state relative to particle $R_2$ is the same as \eqref{Gsystemghz}.~As expected from Theorem~\ref{claim:purestatesrhoSZS} and Lemma~\ref{claim:Uinv&tinvpure}, the state $\rho_{\bar1}=\ket\psi\!\bra\psi_{\bar1}$ with $\ket\psi_{\bar1}$ given in \eqref{eq:totstateGHZex} belongs to the operator subalgebra $\mcA_{\bar1}^{X}$ with $X=Y_2\otimes\mathds1_S$ given by
\be\label{GHZex:Yj}
\mbfU_{\bar1}\ket\psi_{\bar1}=\Bigl(\sum_{g\in\mathcal G}\ket{g}\!\bra{g^{-1}}_2\otimes\mathds1_S\Bigr)\ket\psi_{\bar1}=:(Y_2^\dagger\otimes\mathds1_S)\ket\psi_{\bar1}\;.
\ee
Correspondingly, the total state relative to particle $R_2$ reads as
\be
\ket\psi_{\bar2}=V_{1\to 2}\ket\psi_{\bar1}=\mathcal I_{1\to 2}(Y_2^\dagger\otimes\mathds1_S)\ket{\psi}_{\bar1}
\ee
i.e., it is also separable and the reduced $S$
state is a pure GHZ state in both perspectives.~This is also compatible with monogamy of entanglement according to which the QRF-transformed state in $R_2$-perspective is still a pure separable state with subsystem $S$ in a maximally entangled GHZ state.~No correlations between $S$ and the remaining DoF are thus present in both perspectives and, compatibly with Corollary~\ref{claim:Renyipure}, the $\alpha$-R\'enyi entropy of the $S$-subsystem is the same for all $\alpha$.~In particular, there is no entanglement across the $R_2+S$ and $R_1+S$ bipartitions relative to frames $R_1$ and $R_2$, respectively.~Note that in \eqref{Gsystemghz} we used a specific example of multiparticle GHZ state, but other examples of such states are known in the literature (see e.g.~\cite{ZeilingerFoundPhys99} for the case of 3 qubits).~These are all $U_S^g$-invariant, hence $\mbfU_{\bar1}$-invariant up to local $R_2$ unitaries according to Lemma~\ref{claim:Uinv&tinvpure}, so that the above considerations still hold.
\end{example}

\noindent
Combining the examples with pure states discussed above and in the main text, let us give a quick example with mixed states to illustrate also the results of Corollary~\ref{claim:rhoSZSmixed}.

\begin{example}\label{Ex:mixWGHZ}\textbf{\emph{(Mixed States).}}~Let us consider again the $\mathcal G=\mathbb Z_2$-system of $N>4$ qubits and let the total state $\rho_{\bar1}$ of the remaining $1+(N-2)$ particles relative to particle $R_1$ be a separable mixed state of the form
\be
\rho_{\bar1}=p_1\rho_{\bar1}^{(1)}+p_2\rho_{\bar1}^{(2)}\qquad\quad p_1,p_2>0\;\;,\;\;p_1+p_2=1
\ee
with pure states $\rho_{\bar1}^{(1)}=\ket{\psi^{(1)}}\!\bra{\psi^{(1)}}_{\bar1}$ and $\rho_{\bar1}^{(2)}=\ket{\psi^{(2)}}\!\bra{\psi^{(2)}}_{\bar1}$ given by
\begin{align}
\ket{\psi^{(1)}}_{\bar1}&=\ket{1}_2\otimes\ket{W}_{N-2}\;,\\
\ket{\psi^{(2)}}_{\bar1}&=(a\ket{0}_2+b\ket{1}_2)\otimes\ket{GHZ}_{N-2}\quad,\quad a,b\in\mathbb C,\;,\;|a|^2+|b|^2=1\,.
\end{align}
Note that $\braket{\psi^{(1)}|\psi^{(2)}}_{\bar1}=0$.~As discussed in Sec.~\ref{Sec:states&entanglement} and in Example \ref{Ex:GHZ} above, $\rho_{\bar1}^{(1)}\in\mcA_{\bar1}^{X^{(1)}}$ with $X^{(1)}=Y^{(1)}_2\otimes Z^{(1)}_S=\mathds1_2\otimes\sigma_x^{\otimes(N-2)}$ and $\rho_{\bar1}^{(2)}\in\mcA_{\bar1}^{X^{(2)}}$ with $X^{(2)}=Y^{(2)}_2\otimes Z^{(2)}_S=\mathds 1_2\otimes\mathds1^{\otimes(N-2)}$ (cfr.~Lemma~\ref{claim:Uinv&tinvpure} and \eqref{GHZex:Yj} for $\mathcal G=\mathbb Z_2$),~i.e.~$\rho_{\bar1}^{(2)}\in\mcA_{\bar1}^{\mathds1}$.\footnote{Note that here the fact that $Y_2^{(1)}=\mathds1_2=Y_2^{(2)}$ is simply due to our choice of frame orientations $g_1=g_2=0$ and the resulting form of $\mbfU^{g_1,g_2}_{\bar1}$, namely $\mbfU^{0,0}_{\bar1}=\ket{0}\!\bra{0}_2\otimes\mathds1^{\otimes(N-2)}+\ket{1}\!\bra{1}_2\otimes\mathds\sigma_x^{\otimes(N-2)}$.~For example, with orientation $g_1=0$, $g_2=1$, we have $\mbfU^{0,1}_{\bar1}=\ket{0}\!\bra{1}_2\otimes\mathds1^{\otimes(N-2)}+\ket{1}\!\bra{0}_2\otimes\mathds\sigma_x^{\otimes(N-2)}$ and correspondingly $Y_2^{(1)}=\sigma_2^x=Y_2^{(2)}$}~Then, in agreement with our discussion in Sec.~\ref{Sec:states&entanglement} and Lemma~\ref{Lemma:tpseq}, the total state relative to particle $R_2$ is $\rho_{\bar2}=p_1\rho_{\bar2}^{(1)}+p_2\rho_{\bar2}^{(2)}$ with \be
\ket{\psi^{(1)}}_{\bar2}=\ket{1}_1\otimes\sigma_x^{\otimes(N-2)}\ket{W}_{N-2}\qquad,\qquad\ket{\psi^{(2)}}_{\bar2}=(a\ket{0}_1+b\ket{1}_1)\otimes\ket{GHZ}_{N-2}
\ee
i.e.~still unentangled.~The $S$-subsystem states relative to particle $R_1$ and $R_2$ are respectively given by
\be
\rho_{S|R_1}=p_1\ket{W}\!\bra{W}_{N-2}+p_2\ket{GHZ}\!\bra{GHZ}_{N-2}\;,
\ee
and
\be
\rho_{S|R_2}=p_1\sigma_x^{\otimes(N-2)}\ket{W}\!\bra{W}_{N-2}\sigma_x^{\otimes(N-2)}+p_2\ket{GHZ}\!\bra{GHZ}_{N-2}\;.
\ee
Moreover, using $U_g^{\otimes(N-2)}$-invariance of the GHZ state \eqref{Gsystemghz}, $\rho_{S|R_2}$ can be rewritten as $\rho_{S|R_2}=\sigma_x^{\otimes(N-2)}\rho_{S|R_2}\,\sigma_x^{\otimes(N-2)}$.~The total state $\rho_{\bar1}$ can be then thought of as a mixture of pure states belonging to subalgebras $\mcA_{\bar1}^{X^A}$ with same $Z_S$, here just $\mcA_{\bar1}^{X^{(1)}}$ as $X^{(2)}=\mathds 1_2\otimes\mathds1^{\otimes(N-2)}$ and $[X^{(1)},\rho_{\bar1}^{(2)}]=0$ due to $U_S^g$-invariance of $\rho_{\bar1}^{(2)}$ so that $\rho_{\bar1}^{(2)}\in\mcA_{\bar1}^{X^{(1)}}\cap\mcA_{\bar1}^{X^{(2)}}$.~As such, it satisfies the conditions \eqref{eq:PSdecomp},~\eqref{rhoAibarinAU} of Corollary~\ref{claim:rhoSZSmixed} and the $S$-subsystem states in the two perspectives are unitarily related, thus giving rise to identical R\'enyi and von Neumann entropies.
\end{example}

Finally, we also provide two further examples illustrating the observations in Sec.~\ref{ssec:intScorr} on the behaviour of internal $S$-correlations under QRF transformations. The first example illustrates how internal entanglement of $S$ can remain invariant, even if the entanglement between $S$ and ``the other frame'' changes under QRF transformation.

\begin{example}[{\bf Invariant internal, but changing external $S$-entanglement}]\label{Ex:BSinternaltoS}
Consider $N=5$ qubits with $R_1$ and $R_2$ our internal reference frames ($g_1,g_2=0$) and the subsystem $S$ consisting of the remaining three qubits $S=(ABC)$.~Suppose the total state relative to $R_1$ is of the form
\be\label{Bell:BC-DE}
\ket\psi_{\bar{1}}=\frac{1}{\sqrt{2}}(\ket{00}_{2A}+\ket{11}_{2A})\otimes\frac{1}{\sqrt{2}}(\ket{00}_{BC}-\ket{11}_{BC})
\ee
i.e., in $R_1$ perspective, two of the three particles in $S$ are in a maximally entangled pair ($B$ and $C$) and the other particle $A$ in $S$ is maximally entangled with the other frame $R_2$.\footnote{Similar considerations will also apply to a total state of the form \eqref{Bell:BC-DE} for other choices of Bell states as well as for qudits generalised Bell states.}~Then,
\be\label{UBell:BC-DE}
\mbfU_{\onebar}\ket\psi_{\bar 1}=\frac{1}{\sqrt{2}}(\ket{0}_{2}-\ket{1}_{2})\otimes\ket0_A\otimes\frac{1}{\sqrt{2}}(\ket{00}_{BC}-\ket{11}_{BC})\;.
\ee
Clearly, the state \eqref{Bell:BC-DE} does not belong to any of the subalgebras $\mcA_{\bar{1}}^{X}$.~In the total state $\ket\psi_{\bar 2}=V_{1\to 2}\ket\psi_{\bar 1}$ relative to $R_2$, particle $A$ is not entangled anymore with any other particle, but particles $B$ and $C$ are still in a Bell state.~Thus, in this case, correlations between the subsystem $S$ and the other frame are different in the two perspectives, but the correlations between particles $B$ and $C$ inside $S$ and, by monogamy, also between (BC) and their complement inside $S$ remain the same. 
\end{example}

Next, we show an example where a QRF transformation can lead to non-trivial permutations of local subsystems within the subsystem $S$ (while leaving correlations invariant).

\begin{example}[{\bf QRF transformations and local $S$-permutations}]\label{Ex:permut}
    To illustrate the role of internal $S$ permutations, consider a $\mathbb Z_2$-system of $N=4$ qubits $R_1R_2AB$.~Let $R_1$ and $R_2$ be our internal reference frames with orientations $g_1,g_2=0$ so that the subsystem $S$ of interest consists of the pair $AB$.~Consider the following global pure state of $R_2AB$ relative to particle $R_1$
\be\label{fullWrelA}
\ket\psi_{\bar 1}=a\ket{100}_{\bar 1}+b\ket{010}_{\bar 1}+c\ket{001}_{\bar 1}\;,
\ee
where the three positions in the kets refer to particle $R_2$, $A$, and $B$ in that order, and $|a|^2+|b|^2+|c|^2=1$.~Then,
\be\label{UfullW}
\mbfU_{\onebar}\ket\psi_{\bar 1}=(\mathds1_2\otimes\sigma_x^{\otimes2})(\mathds1_2\otimes\mathcal P_{AB})\ket\psi_{\bar 1}\;,
\ee
where $\mathcal P_{AB}$ is the permutation operator acting on particles $A$ and $B$.~The state \eqref{fullWrelA} thus belongs to the subalgebra $\mcA_{\bar{1}}^{X}$ with $Y_2=\mathds1_2$ and $Z_S=\sigma_x^{\otimes2}\mathcal P_{AB}$ with respect to both the TPS $\mathbf T_{\bar{1}}:\mathcal H_{\emph{phys}}\to\mathcal H_2\otimes\mathcal H_{S}$ and the refined TPS in which $S$ comes to be a composite two particle system with internal TPS $\mathbf T_S:\mathcal H_{S}\to\mathcal H_A\otimes\mathcal H_B$.~The global state $\ket\psi_{\bar 2}=V_{1\to 2}\ket\psi_{\bar 1}$ of $R_1AB$ relative to $R_2$ is thus the same as \eqref{UfullW} with the labels $1$ and $2$ exchanged.~The entanglement between the ``other frame'' and $S=(AB)$, as well as within $S$ are thus  the same in both $R_1$ and $R_2$ perspectives.
\end{example}

\subsection{Perspectival dynamics: Proofs and further examples} \label{app:dyn}

In this appendix, we collect the proofs of the main statements from Sec.~\ref{Sec:Dyn} on global and subsystem dynamics.~Further examples complementing those of the main text are also detailed.
\begin{center}
\textbf{Global dynamics}
\end{center}
\noindent\textbf{Lemma \ref{lemma:finitetimeinvsoln}.}
\emph{Consider a Hamiltonian $H_{\ibar} \in \mcA_{\ibar}$ and an initial state $\rho_{\ibar}(t_0) \in \mathcal{S}(\mcH_{\ibar})$.~Then, $\rho_{\ibar}(t) \in \mcA_{\ibar}^{X}, \;\; \forall t \in[t_0,t_1]
$, for some bilocal unitary $X$ (and frame orientations $g_i,g_j \in \mcG$), if and only if the initial state resides in this subalgebra, $\rho_{\ibar}(t_0)\in\mcA_{\ibar}^X$, and
\be\label{AUXt2}
[H_{\ibar}-H_{j \to i},\rho_{\ibar}(t)]=0 , \;\; \forall t\in [t_0,t_1]\,,
\ee
which is equivalent to $[\hat{\Pi}_{\ibar}^{X\perp}(H_{\ibar}-H_{j\to i}),\rho_{\ibar}(t)]=0$ for $t\in [t_0,t_1]$.
}

\begin{proof}~By the definition \eqref{Uinvsubalgebra} of the subalgebras $\mcA_{\ibar}^{X}\subset\mcA_{\ibar}$, $\rho_{\ibar}(t)\in\mcA_{\ibar}^X$ for $t\in[t_0,t_1]$ implies
\begin{align}
0&=[\rho_{\ibar}(t_0),e^{iH_{\ibar}(t-t_0)}X\mbfU_{\ibar}^{g_i,g_j}e^{-iH_{\ibar}(t-t_0)}]=[\rho_{\ibar}(t_0), e^{iH_{\ibar}(t-t_0)}e^{-i\hat{X}\hat{\mbfU}_{\ibar}^{g_i,g_j}(H_{\ibar})(t-t_0)}X\mbfU_{\ibar}^{g_i,g_j}]\nonumber\\
&=[\rho_{\ibar}(t_0), e^{iH_{\ibar}(t-t_0)}e^{-i H_{j\to i}(t-t_0)}X\mbfU_{\ibar}^{g_i,g_j}]\nonumber\\
&= [\rho_{\ibar}(t_0), e^{iH_{\ibar}(t-t_0)}e^{-iH_{j\to i}(t-t_0)}]X\mbfU_{\ibar}^{g_i,g_j} + e^{iH_{\ibar}(t-t_0)}e^{-iH_{j\to i}(t-t_0)}[\rho_{\ibar}(t_0), X\mbfU_{\ibar}^{g_i,g_j}] \nonumber \\
      &= [\rho_{\ibar}(t_0), e^{iH_{\ibar}(t-t_0)}e^{-iH_{j\to i}(t-t_0)}]X\mbfU_{\ibar}^{g_i,g_j}\,.\label{eq:commAUXt}
\end{align}
In the second line we invoked definition~\eqref{defHimport} and in the last line we used that $\rho_{\ibar}(t_0)\in\mcA_{\ibar}^X$.~Eq.~\eqref{eq:commAUXt} can be equivalently written as $\rho_{\ibar}(t)=e^{-iH_{j\to i}(t-t_0)}\rho_{\ibar}(t_0)e^{iH_{j\to i}(t-t_0)}$ from which, taking the time derivative, the result \eqref{AUXt2} follows.~The converse follows straightforwardly along similar lines as above, and is thus omitted.~The final equivalent statement is implied by~\eqref{eq:hagree}.
\end{proof}

We illustrate Lemma~\ref{lemma:finitetimeinvsoln} with another example, in this case showcasing an automorphism of a TPS-invariant algebra, as well as cases when the trajectory is not TPS-invariant, even if the Hamiltonian is.

\begin{example}[{\bf TPS-invariant and non-invariant trajectories}]\label{Ex:HrhoAUX}
Consider again the three qubits $R_1R_2S$ setup with orientations $g_1=g_2=0$, and the Hamiltonian 
\be \label{ex:2qIZZ}
H_{\bar 1}= \sigma_2^z\otimes\mathds1_S+\mathds1_2\otimes\sigma_S^z + \sigma_2^z\otimes\sigma_S^z \ee
which resides in $\mcA_{\bar 1}^{\mathds1}$ (cfr.~\eqref{3qham} in Example \ref{Ex:AUqubits}), so that $H_{\bar{2}}=\hat{\mathcal I}_{1\to 2}(H_{\bar{1}})$.~Next, consider a pure initial state $\ket{\psi(0)}_{\bar{1}}$ at time $t=0$ in a Schmidt-decomposition given by
\be\label{ex:2qSD}
\ket{\psi(0)}_{\bar{1}} = \cos\theta\,\ket{00}_{\bar{1}} + \sin\theta\,\ket{11}_{\bar{1}}\,.
\ee
$\theta\in[0,\frac{\pi}{4}]$ parametrises the degree of entanglement: for $\theta=0$, $\ket{\psi(0)}_{\bar{1}}=\ket{00}_{\bar{1}}$ is separable; for $\theta=\pi/4$, $\ket{\psi(0)}_{\bar{1}}=\frac{1}{\sqrt{2}}(\ket{00}_{\bar{1}}+\ket{11}_{\bar{1}})$ is maximally entangled; and we get intermediate entangled states for $0<\theta<\pi/4$.~This state, $\rho_{\bar 1}(0)=\ket{\psi(0)}\bra{\psi(0)}_{\bar{1}}$, belongs to $\mcA_{\bar 1}^{\mathds1}$ for $\theta=0$, while it does not belong to any $\mcA_{\bar 1}^{X}$ (for any $X$) for $\theta\in(0,\frac{\pi}{4}]$.~The initial state in $R_2$-perspective $\ket{\psi(0)}_{\bar{2}} =V_{1\to 2}\ket{\psi(0)}_{\bar{1}}$ is separable for all $\theta$,
\be
\ket{\psi(0)}_{\bar{2}} =(\cos\theta\,\ket{0}_{1} + \sin\theta\,\ket{1}_{1})\otimes \ket{0}_S.
\ee
Under dynamical evolution, the states in the two perspectives at later times are,
\begin{align}
&\ket{\psi(t)}_{\bar{1}}= e^{-3it}\,\cos\theta\,\ket{00}_{\bar{1}} +e^{it}\, \sin\theta\,\ket{11}_{\bar{1}}\,,\label{ex11rho1bart}\\
&\ket{\psi(t)}_{\bar{2}}=(e^{-3it}\,\cos\theta\,\ket{0}_{1} +e^{it}\, \sin\theta\,\ket{1}_{1})\otimes \ket{0}_S\,.\label{ex11rho2bart}
\end{align}
Thus, we see that, only for $\theta=0$, the full dynamical trajectory appears the same in the two perspectives i.e.~$\rho_{\bar 1}(t) \in \mcA_{\bar 1}^{\mathds1}$ for all $t \geq 0$.~In this case, the initial state and the Hamiltonian belong to the same TPS-invariant subalgebra, and $\rho_{\bar 1}(t)$ remains therein at all later times, constituting an automorphism of $\mcA_{\onebar}^{\mathds1}$.
\end{example}

Finally, let us demonstrate an example in which the dynamical evolution of the global state $\rho_{\onebar}(t)$ oscillates between two distinct TPS-invariant subalgebras, so that the two QRFs will only periodically agree on the correlations between $S$ and ``other frame''.

\begin{example} \label{ex:inoutAUX}
Consider a three qubit $R_1R_2S$ system subject to the interacting Hamiltonian (cfr.~\eqref{HamBJZZ} in the main text)
\be\label{ex:HamBJZZ}
H_{\bar 1}=B(\sigma_2^z\otimes\mathds1_S+\mathds1_2\otimes\sigma_2^z)+2J\sigma_2^z\otimes\sigma_S^z\,,
\ee
which only for $B=2J$ resides in $\mcA_{\onebar}^{\mathds1}$ and otherwise not in any other TPS-invariant subalgebra (cfr.~Example~\ref{Ex:AUqubits}). Let the initial state be given by
\be\label{initialstateinout1}
\ket{\psi(0)}_{\bar 1}=\ket{x_+}_2\otimes\ket{x_+}_S\,,
\ee
where $\ket{x_+}=\frac{1}{\sqrt{2}}(\ket0+\ket1)$ denotes the eigenstate of $\sigma_x$ with eigenvalue $+1$.~Then the global state at time $t$ is
\be\label{timestateinout}
\ket{\psi(t)}_{\bar 1}=\frac{1}{2}\Bigl(e^{-2i(J+B)t}\ket{00}_{\bar 1}+e^{2iJt}\ket{01}_{\bar 1}+e^{2iJt}\ket{10}_{\bar 1}+e^{-2i(J-B)t}\ket{11}_{\bar 1}\Bigr).
\ee
Recalling that for the frame orientations $g_1,g_2=0$, $\mbfU_{\onebar}^{0,0}=\ket{0}\!\bra0_{2}\otimes\mathds1_S+\ket{1}\!\bra1_{2}\otimes\sigma_S^x \equiv \mbfU_{\onebar}$, we notice that $\mbfU_{\onebar}\ket{00}_{\bar 1}=\ket{00}_{\bar 1},\mbfU_{\onebar}\ket{01}_{\bar 1}=\ket{01}_{\bar 1}, \mbfU_{\onebar}\ket{10}_{\bar 1}=\ket{11}_{\bar 1}$, and $\mbfU_{\onebar}\ket{11}_{\bar 1}=\ket{10}_{\bar 1}$.~In particular we have that $\mbfU_{\onebar}(\ket{10}_{\bar 1}+\ket{11}_{\bar 1})=\ket{10}_{\bar 1}+\ket{11}_{\bar 1}$.~Therefore at any time $t_n =\frac{\pi n}{2J-B}$, $n\in\mathbb N \cup \{0\}$, the last two coefficients in \eqref{timestateinout} become equal and the state  $\ket{\psi(t_n)}\!\bra{\psi(t_n)}_{\bar 1}\in\mcA_{\bar 1}^{\mathds1}$ becomes $\mbfU_{\mathds1}$-invariant. The global dynamical trajectory thus enters the subalgebra $\mcA_{\bar 1}^{\mathds1}$ at certain instants of time while it lies outside of it at all times $t \neq t_n$.

This example further illustrates the scenario where the trajectory oscillates between two different invariant subalgebras under dynamical evolution.~Notice that for $X=\mathds1_2\otimes\sigma_S^x$, we have $X\mbfU_{\onebar}\ket{10}_{\bar 1}=\ket{10}_{\bar 1}$, $X\mbfU_{\onebar}\ket{11}_{\bar 1}=\ket{11}_{\bar 1}$, and $X\mbfU_{\onebar}(\ket{00}_{\bar 1}+\ket{01}_{\bar 1})=\ket{00}_{\bar 1}+\ket{01}_{\bar 1}$.~Therefore at any $\tilde{t}_n =\frac{\pi n}{2J+B}$, $n \in\mathbb N$, the first two coefficients in \eqref{timestateinout} become equal and $\ket{\psi(\tilde{t}_n)}\!\bra{\psi(\tilde{t}_n)}_{\bar 1}\in\mcA_{\bar 1}^{X}$ with $X=\mathds1_2\otimes\sigma_S^x$.~Thus under time evolution, the state \eqref{timestateinout} goes into and out of the subalgebras $\mcA_{\bar 1}^{\mathds1}$ and $\mcA_{\bar 1}^{X}$.~This is illustrated in Fig.~\ref{Fig:inandoutAUXs}.
\end{example}

\begin{figure}[!t]
\centering\includegraphics[scale=0.3]{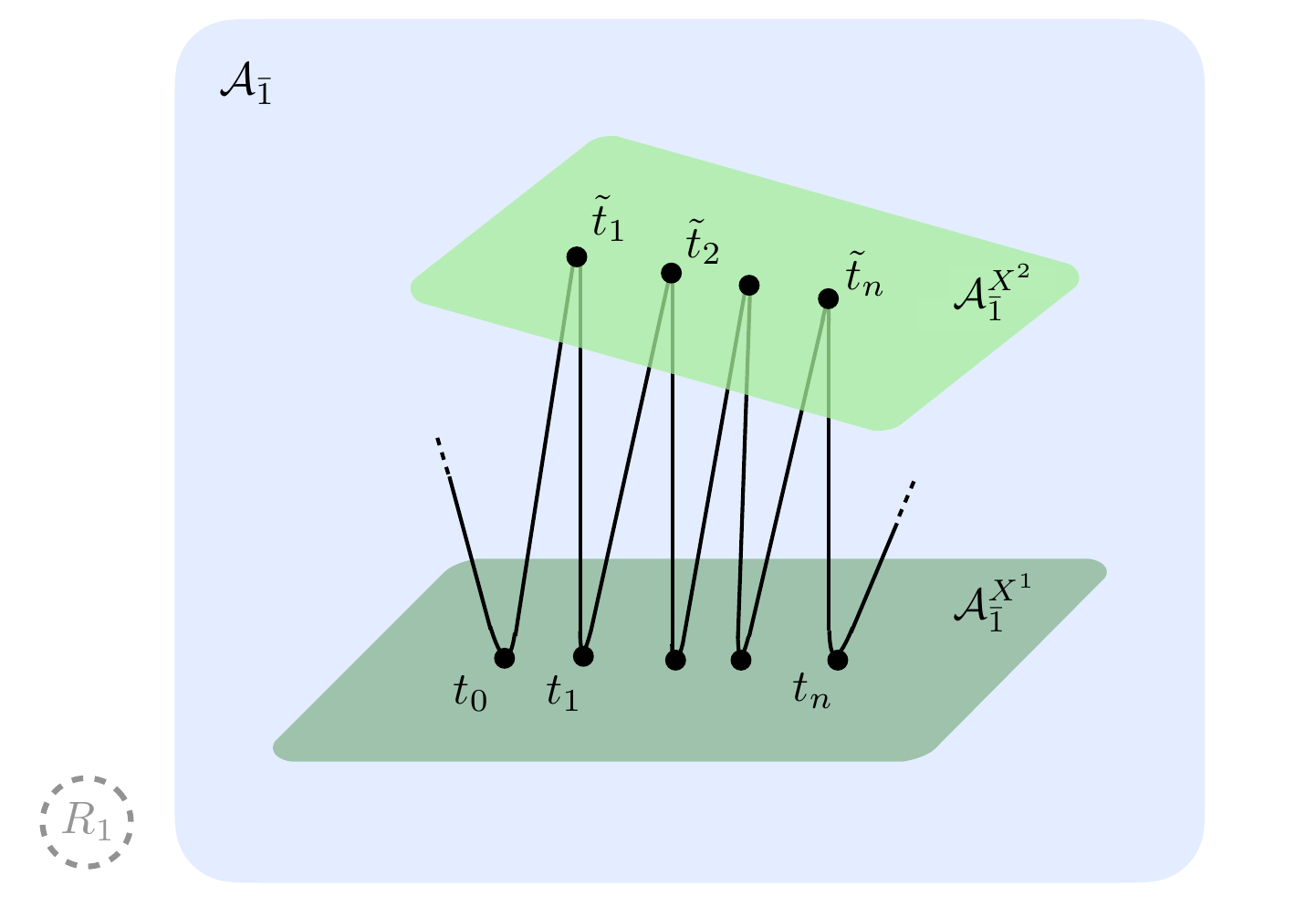}
\caption{Pictorial representation of the global dynamical trajectory in Example~\ref{ex:inoutAUX}.~The state $\rho_{\bar 1}(t)$ relative to frame $R_1$ enters two different invariant subalgebras, $\mcA_{\bar 1}^{X^1}$ and $\mcA_{\bar 1}^{X^2}$ with $X^1=\mathds1_{\bar 1}, X^2=\mathds1_2\otimes\sigma_S^x$, at different instants of time.~Thus also the state $\rho_{\bar 1}(t)$ exhibits the same correlation structure in the two perspectives at these instants of time (cfr.~Example~\ref{inandout:Sdyn} below).}
\label{Fig:inandoutAUXs}
\end{figure}
\begin{center}
\textbf{Subsystem dynamics}
\end{center}
\noindent
\textbf{Lemma \ref{lemma:closed2closed}.}~\textbf{(Closed-to-closed \& closed-to-open)}.
\emph{Let the subsystem $S$ be \emph{dynamically closed} relative to frame $R_i$ (in some orientation $g_i$), i.e.~the total Hamiltonian $H_{\ibar}=H_{j}\otimes\mathds1_S+\mathds1_j\otimes H_{S}$ relative to $R_i$ is non-interacting.~Then, the subsystem $S$ is \emph{dynamically closed} also relative to frame $R_j$ (in any orientation $g_j$) if and only if $H_{\ibar}$ is of the form
\be\label{eq:invHclosed}
H_{\ibar}=H_j^{\mbd}\otimes\mathds1_S+\mathds1_j\otimes H_S^{\mbt}\,.
\ee
In other words, $S$ is \emph{dynamically open} relative to frame $R_j$ if and only if at least one of the following holds
\be
H_{j}^{\mbd\perp}\otimes\mathds1_S\neq0\,,\quad\mathds1_j\otimes H_{S}^{\mbt\perp}\neq0\,.
\ee
}

\begin{proof}
$(\Leftarrow)$ By Theorems~\ref{thm:Slocals}~and~\ref{thm:jlocals} in Sec.~\ref{Sec:Uinv&break}, the QRF transform of \eqref{eq:invHclosed} yields (for $i\neq j$) $H_{\jbar}=H_i^{\mbd}\otimes\mathds1_S+\mathds1_i \otimes H_S^{\mbt}$ with $H_i^{\mbd}\otimes\mathds1_S=\hat{\mathcal I}_{i\to j}(\hat{P}_j^{g_i,g_j}(H_j^{\mbd})\otimes\mathds1_S)$.~Therfore, for a total Hamiltonian of the form \eqref{eq:invHclosed}, the dynamics of the subsystem $S$ is unitary and generated by $H_S^{\mbt}$ in both perspectives (for arbitrary $g_i,g_j\in\mathcal{G}$).

$(\Rightarrow)$ Consider a generic Hamiltonian $H_{\ibar}$ in $R_i$-perspective for which the subsystem $S$ is closed in that perspective, $H_{\ibar}=H_j\otimes\mathds1_S+\mathds1_j\otimes H_S$.~Then, in $R_j$-perspective we have
\be\label{qrfHam}
H_{\jbar}=H_i^{\mbd}\otimes\mathds1_S+\mathds1_i \otimes H_S^{\mbt}+\hat{V}_{i\to j}^{g_i,g_j}(H_i^{\mbd\perp}\otimes\mathds1_S+\mathds1_i \otimes H_S^{\mbt\perp})\;,
\ee
where we separated the $\Pid$-,~$\Pit$-invariant and non-invariant $\Pid^{\perp}$-,~$\Pit^{\perp}$-terms by means of the decomposition \eqref{fulldecomp}.~By Theorems~\ref{thm:Slocals}~and~\ref{thm:jlocals}, the last term on the RHS of \eqref{qrfHam} is necessarily non-local, i.e.\ an interaction (if non-vanishing).~Thus, in order for the subsystem $S$ to be closed also in the new perspective we must have $\hat{V}_{i\to j}^{g_i,g_j}(H_i^{\mbd\perp}\otimes\mathds1_S+\mathds1_i \otimes H_S^{\mbt\perp})=0$.~As the QRF transformation is unitary, this implies the result.
\end{proof}

We close this appendix with three examples illustrating the observations of Sec.~\ref{Sec:dynrhoSandentropies} about QRF (in)dependence of subsystem correlations over time for pure and mixed global states.~We begin with an example where entropies are constant along the dynamical trajectory.

\begin{example}\label{Ex:subsystemdyn}
We continue with Example~\ref{Ex:HrhoAUX}.~Recall that the global state relative to $R_1$ evolves according to \eqref{ex11rho1bart}, specifically that it is separable for $\theta=0$ and entangled for $\theta\in(0,\frac{\pi}{4}]$.~The subsystem state is given by $\rho_{S|R_1}(t)=\cos^2\theta\,\ket{0}\!\bra{0}_S+\sin^2\theta\, \ket{1}\!\bra{1}_S$
and the entropies by $S_\alpha[\rho_{S|R_1}(t)]=\frac{1}{1-\alpha}\log(\cos^{2\alpha}\theta+\sin^{2\alpha}\theta)$.~The global state in $R_2$-perspective, given by \eqref{ex11rho2bart}, is separable for all $t$ and $\theta$ with subsystem state $\rho_{S|R_2}(t)=\ket{0}\!\bra{0}_S$, so that $S_\alpha[\rho_{S|R_2}(t)]=0$ at any time $t$.~Thus, the $S$-subsystem density operator is the same in both perspectives only for $\theta=0$.~As discussed in Example~\ref{Ex:HrhoAUX}, this is the case wherein the trajectory is an authomorphism of  $\mcA^{\mathds1}_{\bar 1}$.
\end{example}

The next example illustrates a situation in which the correlations between $S$ and ``the other frame'' depend on time and where the trajectory oscillates between two distinct TPS-invariant subalgebras.~The two QRFs thereby agree only periodically on the entropies and disagree otherwise.

\begin{example}\label{inandout:Sdyn}
We continue with Example~\ref{ex:inoutAUX}.~The global state evolution in $R_1$-perspective given by \eqref{timestateinout} leads to
\be\label{rhoS1inout}
\rho_{S|R_1}(t)=\frac{1}{2}\Bigl[\ket0\!\bra0_S+\ket1\!\bra1_S+\cos(4Jt)\left(e^{-2iBt}\ket0\!\bra1_S+e^{2iBt}\ket1\!\bra0_S\right)\Bigr].
\ee
The $\alpha$-R\'enyi entropies read $S_\alpha[\rho_{S|R_1}(t)]=\frac{1}{1-\alpha}\log\left(\lambda_{+|R_1}(t)^{\alpha}+\lambda_{-|R_1}(t)^{\alpha}\right)$, where $\lambda_{\pm|R_1}(t)=\frac{1}{2}(1\pm|\cos(4Jt)|)$ are the eigenvalues of \eqref{rhoS1inout}.~On the other hand, in $R_2$-perspective, the reduced $S$ state is 
\be\label{rhoS2inout}
\rho_{S|R_2}(t)=\frac{1}{2}\Bigl[\ket0\!\bra0_S+\ket1\!\bra1_S+\cos(2Bt)\left(e^{-4iJt}\ket0\!\bra1_S+e^{4iJt}\ket1\!\bra0_S\right)\Bigr]
\ee
with R\'enyi entropies $S_\alpha[\rho_{S|R_2}(t)]=\frac{1}{1-\alpha}\log\left(\lambda_{+|R_2}(t)^{\alpha}+\lambda_{-|R_2}(t)^{\alpha}\right)$, where now $\lambda_{\pm|R_2}(t)=\frac{1}{2}(1\pm|\cos(2Bt)|)$ are the eigenvalues of \eqref{rhoS2inout}.

At any time $t_n=\frac{\pi n}{2J-B}, \, n\in\mathbb N\cup\{0\}$, we have that $\cos(4Jt_n)=\cos(2Bt_n)$ and $\sin(4Jt_n)=\sin(2Bt_n)$.~Therefore, at these instants of time, the reduced states coincide exactly $\rho_{S|R_2}(t_n)=\rho_{S|R_1}(t_n)$, and $\rho_{\bar 1}(t_n)\in\mcA_{\onebar}^{\mathds1}$.~Similarly, at any time $\tilde t_n=\frac{\pi n}{2J+B}, \, n\in\mathbb N$, we instead have $\cos(4Jt_n)=\cos(2Bt_n)$ and $\sin(4Jt_n)=-\sin(2Bt_n)$.~Thus, at these other instants of time, the reduced states are unitarily related as $\rho_{S|R_2}(\tilde t_n)=\sigma_S^x\rho_{S|R_1}(\tilde t_n)\sigma_S^x$, and $\rho_{\bar 1}(\tilde t_n)\in\mcA_{\onebar}^{X}$ with $X=\mathds1_2\otimes\sigma_S^x$.~Hence, for both $t=t_n$ and $t=\tilde t_n$ we have that $S_\alpha[\rho_{S|R_2}(t)]=S_\alpha[\rho_{S|R_1}(t)]$ for all $\alpha$.~At any other time $t\notin \{t_n,\tilde t_n\}_{n\in\mathbb N}$, the global state $\rho_{\bar 1}(t)$ lies outside any of the invariant subalgebras (cfr.~Fig.~\ref{Fig:inandoutAUXs}) and the entropies between $S$ and the other frame will differ in the two perspectives. 
\end{example}

Finally, we also discuss an example for mixed global states in which the two QRFs also only agree periodically on the entropies and disagree otherwise.

\begin{example}\label{ex:mixedstatedynamics}~Consider the total system $R_1R_2S$ to be a composite $N$ qubit system ($\mathcal G=\mathbb Z_2$), two of which, say particles $1$ and $2$, serve as our internal reference frames $R_1$ and $R_2$, and the subsystem $S$ consists of the remaining $(N-2)$ particles.~Suppose the initial state of $R_2S$ relative to $R_1$ is given by a rank-2 mixed state $\rho_{\bar1}(0)=p\rho^{(1)}_{\bar1}(0)+q\rho^{(2)}_{\bar1}(0)$ with $p,q>0$, $p+q=1$, and $\rho_{\bar1}^{(A)}(0)=\ket{\psi^{(A)}(0)}\!\bra{\psi^{(A)}(0)}_{\bar1}$ with
\begin{align}
    \ket{\psi^{(1)}(0)}_{\bar1}&=\ket0_2\otimes\ket{W}_{N-2}\,,\\
    \ket{\psi^{(2)}(0)}_{\bar1}&=(a\ket0_2+b\ket1_2)\otimes(\ket{0\dots0}_{N-2}-\ket{1\dots1}_{N-2})\,.
\end{align}
Here, $\ket{W}_{N-2}$ is the generalised W-state, and $a=\frac{1}{\sqrt{8(2-\sqrt{3})}}$, $b=\sqrt{\frac{2-\sqrt{3}}{8}}$.~Restricting, as in previous examples, to $\mbfU_{\onebar}^{g_1=0,g_2=0}=\mbfU_{\onebar}=\textsf{CNOT}$, we have $\rho^{(1)}_{\bar1}(0)\in\mcA_{\bar1}^{\mathds1}$ and $\rho^{(2)}_{\bar1}(0)\in\mcA_{\bar1}^{X^{(2)}}$ with $X^{(2)}=Y_2\otimes\mathds1^{\otimes(N-2)}$, $Y_2=\frac{\sqrt{3}}{2}\mathds1_2+\frac{1}{2}i\sigma_2^y$.

Consider now the $(N-1)$-particle Hamiltonian in $R_1$-perspective to be an interacting Hamiltonian with nearest neighbour $ZZ$-interactions, namely
\be
H_{\bar1}=\sum_{k=2}^{N-1}\sigma^z_{k}\otimes\sigma^z_{k+1}\otimes\mathds1^{\otimes(N-3)}=H_{2S}+\mathds1_2\otimes H_S\;,
\ee
where $H_{2S}=\sigma_2^z\otimes\sigma^z_{3}\otimes\mathds1^{\otimes(N-3)}$ and $H_S=\sum_{k=3}^{N-2}\sigma^z_{k}\otimes\sigma^z_{k+1}\otimes\mathds1^{\otimes(N-4)}$  contains the internal interactions in $S$.~The time evolved pure states in the mixture read
\begin{align}
    \ket{\psi^{(1)}(t)}_{\bar1}&=\ket0_2\otimes e^{it}\bigl(\cos(t)\ket{W}_{N-2}-i\sin(t)(\sigma^z_{3}\otimes\mathds1^{\otimes(N-3)})\ket{W}_{N-2}\bigr)\,,\label{psi1oft}\\
    \ket{\psi^{(2)}(t)}_{\bar1}&=(a\ket0_2+be^{2it}\ket1_2)\otimes\ket{0\dots0}_{N-2}-(ae^{2it}\ket0_2+b\ket1_2)\otimes\ket{1\dots1}_{N-2}\,.\label{psi2oft}
\end{align}
The state \eqref{psi1oft} still belongs to $\mcA_{\bar1}^{\mathds1}$ at all $t$, but \eqref{psi2oft} does not lie in any subalgebra $\mcA_{\bar1}^{X}$ except at $t=t_{n}=\pi n$, $n\in\mathbb N\cup\{0\}$, and $t=\tilde{t}_n=\frac{\pi}{2}(2n-1)$, $n\in\mathbb N$, when it belongs to $\mcA_{\bar1}^{X^{(2)}}$ and $\mcA_{\bar1}^{\mathds1}$, respectively.~At these instants of time, the state $\rho_{\bar1}(t)$ meets the conditions of Cor.~\ref{claim:rhoSZSmixed}.~While the initial separable mixed state gets entangled along its evolution, it returns to being separable precisely at those instants of time.~Lastly, the state in $R_2$-perspective is given by $\rho_{\bar2}(t)=p\rho_{\bar2}^{(1)}(t)+q\rho_{\bar2}^{(2)}(t)$, with
\begin{align}
    \ket{\psi^{(1)}(t)}_{\bar2}&=\ket0_1\otimes e^{it}\bigl(\cos(t)\ket{W}_{N-2}-i\sin(t)(\sigma^z_{3}\otimes\mathds1^{\otimes(N-3)})\ket{W}_{N-2}\bigr)\,,\\
    \ket{\psi^{(2)}(t)}_{\bar2}&=(a\ket0_1-b\ket1_1)\otimes(\ket{0\dots0}_{N-2}-e^{2it}\ket{1\dots1}_{N-2})\,.\label{psi2oftR2}
\end{align}
i.e., it is separable at all $t$.\footnote{This is due to the fact that no interactions between $S$ and $R_1$ are present in the Hamiltonian $H_{\bar 2}=\mathds1_1\otimes\sigma_{z,3}\otimes\mathds1^{\otimes(N-3)}+\mathds1_1\otimes H_S$ of $R_2$-perspective.~$H_{\bar2}$ satisfies the second condition of \eqref{cl2opconditions} in the main text and the subsystem $S$ is thus open in $R_1$- and closed in $R_2$-perspective.}~Thus, from \eqref{psi1oft}-\eqref{psi2oftR2}, we see that  $\rho_{S|R_1}(t)=\rho_{S|R_2}(t)$ and $S_{\alpha}[\rho_{S|R_1}(t)]=S_{\alpha}[\rho_{S|R_2}(t)]$ only at $t=t_n,\tilde{t}_n$.
\end{example}

\subsection{Proofs: Lemmas \ref{Lem:avgs} and \ref{Lemma:NegBProduct}} \label{app:thermo}

\noindent
\textbf{Lemma \ref{Lem:avgs}.}~\emph{Let $A_S \in \mcA_{S}$ be some $S$-local operator.~Then, for \emph{any} global state $\rho_{\ibar} \in \mathcal{S}(\mcH_{\ibar})$,
\be \label{eqavgspr}
\Tr(\rho_{S|R_j} A_S) = \Tr(\rho_{S|R_i} A_S)\,,
\ee
if and only if ${[A_S,U^g_S] = 0, \, \forall g\in \mcG}$.~Here $\rho_{S|R_j} = \Tr_i(\rho_{\jbar})$, $\rho_{S|R_i} = \Tr_j(\rho_{\ibar})$,~and $\rho_{\jbar} = \hat{V}_{i \to j}^{g_i,g_j}(\rho_{\ibar})$ for some frame orientations $g_i,g_j \in \mcG$.}

\begin{proof} By definition of partial trace and cyclic property, we have
\be
\Tr_S(\rho_{S|R_j}A_S) =\Tr_{\jbar}(\hat{V}^{g_i,g_j}_{i\to j}(\rho_{\ibar})\mathds1_i \otimes A_S) =\Tr_{\jbar}\bigl(\hat{\mathcal I}_{i\to j}(\hat{\mbfU}_{\ibar}^{g_i,g_j}(\rho_{\ibar}))\hat{\mathcal I}_{i\to j}(\mathds1_j \otimes A_S)\bigr) =\Tr_{\ibar}\bigl(\rho_{\ibar}\,\bigl(\hat{\mbfU}_{\ibar}^{g_i,g_j}\bigr)^\dagger(\mathds1_j \otimes A_S)\bigr)\,,\nonumber
\ee
where we used the relation \eqref{IUVrelation} between $V^{g_i,g_j}_{i\to j}$, $\mbfU_{\ibar}^{g_i,g_j}$, and the unitary operator $\mathcal I_{i\to j}$ given in \eqref{eq:Ijtoi}, and the fact that $\Tr_{\jbar}(\mathcal I_{i\to j}\bullet\mathcal I_{j\to i})=\Tr_{\ibar}(\bullet)$. Thus, \eqref{eqavgspr} is equivalent to ${\Tr_{\ibar}(\rho_{\ibar}(\hat{\mbfU}_{\ibar}^{g_i,g_j})^\dagger(\mathds1_j \otimes A_S))=\Tr_S(\rho_{S|R_i}A_S)=\Tr_{\ibar}(\rho_{\ibar}\,\mathds1_j \otimes A_S)}$. This holds true for \emph{any} $\rho_{\ibar}$ if and only if $(\hat{\mbfU}_{\ibar}^{g_i,g_j})^{\dagger}(\mathds1_j \otimes A_S) = \mathds1_j\otimes A_S$, which, according to Theorem~{\ref{thm:Slocals}}, is equivalent to $[A_S,U_S^g]=0$, $\forall\,g\in \mcG$.
\end{proof}

\noindent
\textbf{Lemma \ref{Lemma:NegBProduct}.}~\emph{Let the system relative to $R_i$ in some orientation $g_i\in \mcG$ be prepared in a global state $\rho_{\ibar} \in \mathcal{S}(\mcH_{\ibar})$ of the form
\be\label{eq:i-separable-gibbs}
\rho_{\ibar} = \rho_{j} \otimes \frac{1}{Z}e^{-\beta H_{S|R_i}}
\ee
and the Hamiltonian $H_{S|R_i}$ be such that,
\begin{align}
\exists \, \mcG_{\rm a} \subsetneq \mcG \;\;\; &\;{\rm s.t.} \;\, \{U_S^h, H_{S|R_i}\} = 0 \;\; \forall h\in \mcG_{\rm a}, \label{eq:negcond1} \\
&{\rm and} \;\;\; [U_S^g, H_{S|R_i}] \,= 0 \;\; \forall g \in \mcG_{\rm c}  \label{eq:negcond2}
\end{align}
where, $\mcG_{\rm a}$ is a strict\footnote{This is because there always exists one element of $\mcG$, the identity, which only ever satisfies \eqref{negcond2} and thus always belongs to $\mcG_{\rm c}$.} subset of $\mcG$, $\{\cdot,\cdot\}$ denotes an anti-commutator, and $\mcG_{\rm c} = \mcG \backslash \mcG_{\rm a}$.~Then, the subsystem state relative to $R_j$ in some orientation $g_j \in \mcG$ is
\be 
\rho_{S| R_j} = q_{\rm a}^{g_j} \frac{1}{Z} \, e^{+\beta H_{S|R_i}}  + (1 - q_{\rm a}^{g_j}) \frac{1}{Z} \, e^{-\beta H_{S|R_i}} 
\ee
where, $q_{\rm a}^{g_j} = \sum \limits_{h \in \mcG_{\rm a}} \bra{g_j h^{-1}}\rho_{j}\ket{g_j h^{-1}}_j$.}

\begin{proof} Given $\rho_{\ibar}$ in \eqref{eq:i-separable-gibbs} and the QRF transformation $V_{i\to j}^{g_i,g_j}$, for any $g_i,g_j \in \mcG$, it is easy to compute that the $S$-subsystem state in $R_j$-perspective is given by, $\rho_{S|R_j} = \sum_{g\in\mcG} \bra{g_jg^{-1}}\rho_{j}\ket{g_jg^{-1}}_j\frac{1}{Z}\, e^{-\beta U_S^g H_{S|R_i} U_S^{g^{-1}}}$. Let us split this sum according to elements $h \in \mcG_{\rm a}$ satisfying the anti-commutation property \eqref{eq:negcond1}, and elements $k \in \mcG_{\rm c}$ satisfying the commutation property \eqref{eq:negcond2}, to get 
\be
\rho_{S|R_j}= \sum_{h \in \mcG_{\rm a}} \bra{g_jh^{-1}}\rho_{j}\ket{g_jh^{-1}}_j \frac{1}{Z}e^{-\beta U_S^h H_{S|R_i} U^{h^{-1}}_S} + \sum_{k \in \mcG_{\rm c}}  \bra{g_jk^{-1}}\rho_{j}\ket{g_jk^{-1}}_j \frac{1}{Z} e^{-\beta U^{k}_S H_{S|R_i} U^{k^{-1}}_S} \,.
\ee
The result follows immediately by using conditions \eqref{eq:negcond1}-\eqref{eq:negcond2} for both exponents, and
$\sum \limits_{g\in \mcG} \bra{g_jg^{-1}}\rho_{j}\ket{g_jg^{-1}}_j = 1$ along with $\mcG = \mcG_{\rm a} \cup \mcG_{\rm c}$.
\end{proof}

\subsection{Perspectival energy balance and correlations as heat source}\label{App:TDroleofcorrelations}

In this appendix, complementing the exposition of Sec.~\ref{Sec:TDsetup}, we discuss in further detail the energy balance of a composite, interacting, quantum thermodynamical system as described relative to a given internal reference frame, and highlight the role of correlations for the energetics of the system. In particular, we shall detail the various contributions to the heat flowing into the subsystems and the work they perform onto one another in a given perspective, and illustrate that dynamical correlations can act as an energy source and contribute to the heat exchange between $S$ and ``the other frame'', as expected from studies in quantum thermodynamics. 

To begin, let us notice that from \eqref{totinten} and the definitions \eqref{convenWQ1},~\eqref{eq:QWdef} (and their analogues for the ``other frame'' subsystem), energy conservation for the total isolated system in the $R_i$-perspective can be written as
\begin{align}
0=\dot E_{\ibar}=\dot{E}_{j}+\dot{E}_{S|R_i}+\dot{E}_{jS}&=\dot q_{j}+\dot q_{S|R_i}+\dot w_{j}+\dot w_{S|R_i}+\dot{E}_{jS}\nonumber\\
&=\dot Q_{j}+\dot Q_{S|R_i}+\dot W_{j}+\dot W_{S|R_i}+\dot{E}_{jS}.\label{totEdotqwQW}
\end{align}
The specific forms of the various quantities appearing on the RHS of \eqref{totEdotqwQW} (and their magnitudes) depend on the prescription for defining quantum heat and work and the effective Hamiltonians therein.~To understand better the various contributions to \eqref{totEdotqwQW} and make the role of correlations more transparent, it is thus instructive to look at a specific prescription.~Consider some total Hamiltonian $H_{\ibar}=H_{j}\otimes\mathds1_S+\mathds1_j\otimes H_{S|R_i}+H_{jS}$, relative to $R_i$.~Following the prescription of \cite{rezakhani}, along with \eqref{hibareff} and the surrounding discussion, let us write $H_{\ibar}$ equivalently as
\begin{align}
H_{\ibar}
&=\left(H_j+\tilde{H}_j(t)-\alpha_j\tr\left[H_{jS}(\rho_j(t)\otimes\rho_{S|R_i}(t))\right]\mathds1_j\right)\otimes\mathds1_S\nonumber\\
&+\mathds1_j\otimes\Bigl(H_{S|R_i}+\tilde{H}_{S|R_i}(t)-\alpha_S\tr\left[H_{jS}(\rho_j(t)\otimes\rho_{S|R_i}(t))\right]\mathds1_S\Bigr)\nonumber\\
&+H_{jS}-\tilde{H}_j(t)\otimes\mathds1_S-\mathds1_j\otimes\tilde{H}_{S|R_i}(t)+\tr\left[H_{jS}(\rho_j(t)\otimes\rho_{S|R_i}(t))\right]\mathds1_{\ibar} \label{eq:Hibaralphas} \\
&=: H_{j}^{\eff}(t) \otimes \mathds1_S + \mathds1_j \otimes H_{S|R_i}^{\eff}(t) + H_{jS}^{\eff}(t) \,,\label{heffpr}
\end{align}
where $\tilde{H}_{S|R_i}(t) = \Tr_j\!\big(H_{jS} (\rho_{j}(t) \otimes \mathds1_S) \big)$ as in \eqref{eq:HStilde}, an analogous definition holds for $\tilde{H}_j(t)$, and $\alpha_S$, $\alpha_j$ are real, scalar parameters s.t.~$\alpha_j + \alpha_S = 1$, to be specified depending on the physical conditions of the (non-equilibrium) thermodynamic system under consideration.~Comparing \eqref{eq:Hibaralphas} with \eqref{heffpr}, we see that in our setup of Sec.~\ref{Sec:TDsetup}, the prescription of \cite{rezakhani} consists of the following identifications (cfr.~\eqref{effhs}), 
\be\label{hjrez}
\begin{aligned}
\mathds{h}_{S|R_i}(t)&=\tilde{H}_{S|R_i}(t)-\alpha_S\tr\left(H_{jS}(\rho_j(t)\otimes\rho_{S|R_i}(t))\right)\mathds1_S\;,\\
\mathds{h}_{j}(t)&=\tilde{H}_{j}(t)-\alpha_j\tr\left(H_{jS}(\rho_j(t)\otimes\rho_{S|R_i}(t))\right)\mathds1_j\;.
\end{aligned}
\ee
The main of purpose of rewriting the total Hamiltonian as \eqref{hibareff}, like in \eqref{eq:Hibaralphas} above following \cite{rezakhani}, is 
to divide the total internal energy of a composite system into local parts ($H_{S|R_i}^{\eff}$ and $H_j^{\eff}$), which are associated to and accessible by each subsystem, and a non-local part ($H_{jS}^{\eff}$), which is locally inaccessible by the subsystems\footnote{\label{fn:rez}The local inaccessibility requirement in \cite{rezakhani} implies that the mean value of the effective interactions in the separable state $\rho_{j}(t)\otimes\rho_{S}(t)$ of the corresponding uncoupled system vanishes.} and can be attributed only to the composite system as a whole.

With the above choice of $H_{S|R_i}^{\eff}$, it is easy to check that, $\dot{E}^*_{S|R_i}(t) = -i \Tr\!\big(H_{S|R_i}^{\eff}(t)\big[H_{S|R_i} + \tilde{H}_{S|R_i}(t), \rho_{S|R_i}(t) \big]\big)=0$, cfr.~Eq.~\eqref{star}.~Thus in this prescription for effective Hamiltonians, the expressions \eqref{eq:QWdef} reduce to those in \eqref{convenWQ1}-\eqref{convenWQ2}, and there is no discrepancy in the prescription for heat and work.~Explicitly, we get
\begin{align}
    \dot{W}_{S|R_i}= \dot{w}_{S|R_i}&=\Tr(\dot{\mathds h}_{S|R_i}(t)\rho_{S|R_i}(t))\nonumber\\
    &=\Tr\left(H_{jS}(\dot{\rho}_{j}(t)\otimes\rho_{S|R_i}(t))\right)-\alpha_S\tr\left(H_{jS}(\dot{\rho}_j(t)\otimes\rho_{S|R_i}(t)+\rho_{j}(t)\otimes\dot{\rho}_{S|R_i}(t))\right)\label{explicitwdot}\\
    \dot{Q}_{S|R_i}= \dot{q}_{S|R_i}&=\Tr\left((H_{S|R_i}+\mathds h_{S|R_i}(t))\dot{\rho}_{S|R_i}(t)\right)\nonumber\\
    &=\Tr\left((\mathds1_j\otimes H_{S|R_i})\dot{\rho}_{\ibar}\right)+\Tr\left(H_{jS}(\rho_{j}(t)\otimes\dot{\rho}_{S|R_i}(t))\right) \label{explicitqdot}
\end{align}
by using properties of partial traces, Eqs.~\eqref{heffpr}-\eqref{hjrez}, and $\Tr(\dot{\rho})=0$ for any state.~Similar results hold for $R_j$, with $\dot{W}_{j}= \dot{w}_{j}$ and $\dot{Q}_{j}= \dot{q}_{j}$.~Then, inserting \eqref{explicitwdot}-\eqref{explicitqdot} and the analogous quantities for $R_j$, into \eqref{totEdotqwQW} we get\footnote{Alternatively, we can get \eqref{eq:EjSdotomegajS} by direct calculation, by simplifying $\dot{E}_{jS}=\Tr(\dot{H}_{jS}^{\eff}\rho_{\ibar})+\Tr(H_{jS}^{\eff}\dot{\rho_{\ibar}})$, using $H_{jS}^{\eff}$ defined in \eqref{heffpr}.}
\be\label{eq:EjSdotomegajS}
\dot{E}_{jS}=\Tr(H_{jS}\,\dot{\omega}_{jS})\,,
\ee
using $\alpha_j+\alpha_S=1$, Eq.~\eqref{rhocorr}, $H_j\otimes\mathds1_S+\mathds1_j\otimes H_S=H_{\ibar}-H_{jS}$, and $\Tr_{\ibar}(H_{\ibar}\dot{\rho}_{\ibar})=\dot{E}_{\ibar}=0$.\footnote{We recall that in the present work we consider time-independent total Hamiltonians $H_{\ibar}$.} Furthermore, from \eqref{explicitwdot} and the analogous expression for the work onto $R_j$, a similar computation yields
\be\label{eq:workbalance}
\dot{w}_{j}+\dot{w}_{S|R_i}=0\,,
\ee
that is, in $R_i$-perspective, the work performed by $R_j$ onto $S$ is equal and opposite to the work performed by $S$ onto $R_j$.~Finally, combining \eqref{eq:EjSdotomegajS} and \eqref{eq:workbalance} with the total energy conservation \eqref{totEdotqwQW}, we have
\be\label{eq:heatbalance}
\dot{q}_{S|R_i}=-\dot{q}_{j}-\Tr(H_{jS}\,\dot{\omega}_{jS})\,,
\ee
so that the heat $\dot{q}_{S|R_i}$ ($=\dot{Q}_{S|R_i}$) flowing into $S$ comes not only from the heat $-\dot{q}_{j}$ ($=-\dot{Q}_{j}$)  flowing out of ``the other frame'', but also from energy stored in correlations.

The above discussion illustrates that there are three key players relevant for the energetics of a composite, quantum thermodynamical system relative to the chosen frame: the subsystems $S$  and ``the other frame'', and the correlations between them in an interacting dynamics. In quantum thermodynamical processes, the internal energy of one of the subsystems, say $S$ relative to $R_i$, can change as a result of the work done on $S$ by $R_j$, heat exchange with $R_j$, and energy exchange via non-trivial dynamical couplings between $R_j$ and $S$ (due to the interactions $H_{jS}$ in the total Hamiltonian and the resulting correlations $\omega_{jS}(t)$ in the global $R_jS$ state $\rho_{\ibar}(t)$).

As a final remark, we see that in a weak coupling limit, $H_{jS}\to 0$, we have $\dot q_{S|R_i}=-\dot q_{j}$, $\dot w_{S|R_i}=-\dot w_{j}$, $\dot{E}_{jS}=0$.~That is, as expected in a weak coupling limit, the correlations generated by the interactions are negligible and do not act as an additional source in the energy balance of the system. 

\subsection{Thermodynamic processes: proofs and further discussion}\label{App:TDrates}

Here, we complement the content of Secs.~\ref{Sec:TDproc1} and \ref{Sec:TDproc2} with further discussion on the effective Hamiltonians and proofs of the main statements.

First, let us illustrate how different prescriptions for the effective Hamiltonians might give different answers for the QRF-dependence of the energetics and thermodynamic quantities of the system when assumptions are made on the Hamiltonian but not on the state under consideration.~For concreteness, let us focus on the prescriptions of \cite{rezakhani} and \cite{weimer,Hossein_Nejad_2015}, and consider a class of TPS-invariant Hamiltonians given by $H_{\ibar}=\Pid\Pit(H_{\ibar})$, with $\Pit$ and $\Pid$ as defined in \eqref{pitdef} and \eqref{piddef}.~As we shall see below, for such a Hamiltonian, the prescription of \cite{rezakhani} yields the equality of subsystem energies, heat, and work exchanges in $R_i$- and $R_j$-perspective at all times and for all states.~This is not the case for the prescription of \cite{weimer,Hossein_Nejad_2015}.~The following claim collects the results obtained using the former prescription.

\begin{claim}\label{claim:pidpitprescription}
    Let $H_{\ibar}=\Pid\Pit(H_{\ibar})$.~Then, using the prescription \eqref{heffpr} for the effective Hamiltonians as in \cite{rezakhani}, the individual contributions to the energy of the total system, the heat flow into and work done onto ``the other frame'' and $S$ subsystems are equal in $R_i$- and $R_j$-perspectives
    \begin{align}
    E_{S|R_j}(t)=E_{S|R_i}(t)\quad,\quad E_{i}(t)&=E_{j}(t)\quad,\quad E_{iS}(t)=E_{jS}(t)\,,\label{eq:pidpitenergy}\\
    \dot{q}_{S|R_j}=\dot{q}_{S|R_i}\quad,\quad \dot{w}_{S|R_j}&=\dot{w}_{S|R_i}\quad,\quad\dot{E}^*_{S|R_j}=\dot{E}^*_{S|R_i}\,,\label{eq:Sqwestar}\\
    \dot{q}_{i}=\dot{q}_{j} \quad,\quad \dot{w}_{i}&=\dot{w}_{j} \quad,\quad \dot{E}^*_{i}=\dot{E}^*_j\,,\label{eq:frameqwestar}
    \end{align}
    at any time $t$ and for any given global state $\rho_{\ibar}$.~The quantities in \eqref{eq:pidpitenergy}-\eqref{eq:frameqwestar} are as defined in Sec.~\ref{Sec:TDsetup}.
\end{claim}

\begin{proof}
    When the Hamiltonian $H_{\ibar}=H_j\otimes\mathds 1_S+\mathds1_j\otimes H_{S|R_i}+H_{jS}$ is in the image of the projectors $\Pid,\Pit$, it means that the individual terms in it are as well.~That is,
    \be\label{eq:Hpidpit}
H_{\ibar}=H_j^{\mbd}\otimes\mathds 1_S+\mathds1_j\otimes H_{S|R_i}^{\mbt}+\Pid\Pit(H_{jS})\,,
    \ee
    where, recalling that $\Pid,\Pit$ act locally on the $R_j$ and $S$ tensor factors, respectively, we used the shorthand notation $H_j^{\mbd}\otimes\mathds 1_S=\Pid(H_j\otimes\mathds 1_S)$ and $\mathds1_j\otimes H_{S|R_i}^{\mbt}=\Pit(\mathds1_j\otimes H_{S|R_i})$.~Using then \eqref{eq:Hpidpit} in the expressions \eqref{heffpr} for the effective Hamiltonians in the prescription of \cite{rezakhani}, it follows that the effective subsystem Hamiltonians and interaction in such a prescription are also $\Pid\Pit$-invariant, i.e.
\be\label{eq:effHpidpit}
H_{\ibar}=(H_j^{\eff})^{\mbd}\otimes\mathds 1_S+\mathds1_j\otimes (H_{S|R_i}^{\eff})^{\mbt}+\Pid\Pit(H_{jS}^{\eff})\,.
    \ee
    By Lemma~\ref{lemma:pinpd-uin} and \ref{Lemma:tpseq}, the Hamiltonian in $R_j$-perspective is given by $H_{\jbar}=\hat{\mathcal{I}}'_{i\to j}\hat{P}_j^{g_i,g_j}(H_j^{\mbd})\otimes\mathds1_S+\mathds1_i\otimes H_{S|R_i}^{\mbt}+\hat{\mathcal I}_{i\to j}\hat{X}^{\dagger}(\Pid\Pit(H_{jS}))$, with $\mathcal{I}_{i\to j}$ the identity \eqref{eq:Ijtoi} from $R_i$- to $R_j$-perspective, $\hat{\mathcal{I}}'_{i\to j}$ its restriction to the frame degrees of freedom, $X=P_j^{g_i,g_j}\otimes\mathds1_S$, and $P_j^{g_i,g_j}$ the parity-swap operator \eqref{Pswap}.~Owing to the Hilbert-Schmidt orthogonality of the algebra decompositions \eqref{algtdecomp} and \eqref{algddecomp} and their QRF-invariance (cfr.~Eqs.~\eqref{pivcomp0},~\eqref{pivcomp}), only the $\Pid$- and $\Pit$-invariant components of the state contribute to the effective Hamiltonians \eqref{heffpr}.~Invoking again Lemma~\ref{lemma:pinpd-uin} and \ref{Lemma:tpseq}, we have that \be\label{eq:pitpidstates}
    \Pid\Pit(\rho_{\jbar})=\hat{\mathcal I}_{i\to j}\hat{X}^{\dagger}(\Pid\Pit(\rho_{\ibar}))\qquad\Rightarrow\qquad\rho_{S|R_j}^{\mbt}=\rho_{S|R_i}^{\mbt}\quad,\quad\rho_{i}^{\mbd}=\hat{\mathcal{I}}'_{i\to j}\hat{P}_j^{g_i,g_j}(\rho_{j}^{\mbd})
    \ee
    for any QRF-related states $\rho_{\ibar}$ and $\rho_{\jbar}$, and at any time.~Thus, using the expression for $H_{\jbar}$ and \eqref{eq:pitpidstates} in the analogous prescription \eqref{heffpr} for the effective Hamiltonians in $R_j$-perspective, we have
\be\label{eq:heffjpersp}
H_i^{\eff}=\hat{\mathcal{I}}'_{i\to j}\hat{P}_j^{g_i,g_j}((H_j^{\eff})^{\mbd})\quad,\quad H_{S|R_j}^{\eff}=(H_{S|R_i}^{\eff})^{\mbt}\quad,\quad H_{iS}^{\eff}=\hat{\mathcal I}_{i\to j}\hat{X}^{\dagger}(\Pid\Pit(H_{jS}^{\eff}))\,,
    \ee
    which coincide with the direct QRF transformation of the effective Hamiltonian contributions in~~\eqref{eq:effHpidpit}.~Again by Hilbert-Schmidt orthogonality, the $\Pid$- and $\Pit$-invariant components of the state are the only contributing also to the effective energies in the two perspectives.~The results in Eqs.~\eqref{eq:pidpitenergy}-\eqref{eq:frameqwestar} then follow from inserting \eqref{eq:effHpidpit}--\eqref{eq:heffjpersp} into the expressions \eqref{totinten},~\eqref{convenWQ1}, and \eqref{star}.
\end{proof}

\noindent
\textbf{Remark.}~The key step in the above proof is that, for a $\Pid\Pit$-invariant Hamiltonian, the effective Hamiltonians defined with the prescription of \cite{rezakhani} are also $\Pid\Pit$-invariant (cfr.~Eq.~\eqref{eq:effHpidpit}).~This is, however, not necessarily the case for the prescription of \cite{weimer,Hossein_Nejad_2015}.~In this prescription, the effective Hamiltonian for, say, the $S$-subsystem is defined as $H_{S|R_i}^{\eff}(t)=H_{S|R_i}+\mdsh_{S|R_i}(t)$, where $\mdsh_{S|R_i}(t)$ is the part of $\tilde{H}_{S|R_i}(t)$ which commutes with $H_{S|R_i}$, with $\tilde{H}_{S|R_i}(t)=\Tr_j\!\big(H_{jS} (\rho_{j}(t) \otimes \mathds1_S) \big)$ defined as in Eq.~\eqref{eq:HStilde}.~Clearly, $\tilde{H}_{S|R_i}=\tilde{H}^{\mbt}_{S|R_i}$ for $H_{\ibar}=\Pid\Pit(H_{\ibar})$, owing to Eq.~\eqref{eq:Hpidpit}.~However, for $H_{S|R_i}$ and $\tilde{H}_{S|R_i}$ to be translation invariant, the same need not be true for $\mdsh_{S|R_i}$ and thereby $H_{S|R_i}^{\eff}$, only that their $\Pit^{\perp}$-components are equal.~Therefore, unlike in Claim~\ref{claim:pidpitprescription} with the prescription of \cite{rezakhani}, in this other prescription it is not only the $\Pid$-,~$\Pit$-components of the states that contribute to the energetics of the system.~Thus, the effective energies, heat, and work exchanges of subsystems will generically not necessarily be equal in the two perspectives.\\

\noindent
In the remainder of the appendix we shall prove Lemmas~\ref{lemma:invTDrates},~\ref{lemma:zerosigmaphi},~and \ref{lemma:inventbal}.\\

\noindent
\textbf{Lemma~\ref{lemma:invTDrates}.}~\emph{Let $\rho_{\ibar}(t)\in\mathcal{A}_{\ibar}^X$ be a TPS-invariant trajectory for some interval $[t_0,t_1]$.~Choose \emph{any} prescription of effective Hamiltonians obeying Definition~\ref{def_prescription} and apply it to both $H_{\jbar}$ in $R_j$-perspective and to the imported Hamiltonian $H_{j\to i}$ in $R_i$-perspective ($i\neq j$) for defining the energetics.~Then, over $[t_0,t_1]$, the heat flow into and work done on $S$ are equal in both perspectives relative to this prescription. This is \emph{independent} of whether work and heat are defined via~\eqref{convenWQ1} \&~\eqref{convenWQ2} or via~\eqref{eq:QWdef}.~That is, for any $t\in[t_0,t_1]$,
    \be\label{sysqtspr}
    \dot{q}_{S|R_j}=\dot{q}_{S|R_i}, \quad \dot{w}_{S|R_j}=\dot{w}_{S|R_i}, \quad \dot{E}^*_{S|R_j}=\dot{E}^*_{S|R_i},
    \ee
and similarly, for the ``the other frame'' subsystem
    \be\label{frameqtspr}
    \dot q_{i}=\dot q_{j}, \;\; \dot{w}_{i}=\dot{w}_{j}, \quad \dot{E}^*_{i}=\dot{E}^*_{j}.
    \ee
This holds for any frame orientations $g_i,g_j$ before and after QRF transformation, respectively.
    }
\begin{proof}  
We recall that, for $\mbfUX$-invariant states, the imported Hamiltonian $H_{j\to i}=\hat X\hat{\mathcal I}_{j\to i}(H_{\jbar})$ can be equivalently used to define the dynamics (cfr.~Lemma~\ref{lemma:finitetimeinvsoln}) and the total energy (cfr.~Eq.~\eqref{energyUXinvstate}) for the global state under consideration over the interval $[t_0,t_1]$.~For $\rho_{\ibar}(t)\in\mcA_{\ibar}^{X}$ with $X=Y_j\otimes Z_S$, we also have by Lemma~\ref{Lemma:tpseq} that $
\rho_{\ibar}(t)=\hat X\hat{\mathcal I}_{j \to i}(\rho_{\jbar}(t))$.

Next, fix any prescription for constructing effective Hamiltonians that abides by Definition~\ref{def_prescription} and apply it \emph{both} to $H_{\jbar}$ in $R_j$-perspective and to $H_{j\to i}$ in $R_i$-perspective. Since all permissible operator ingredients in the two perspectives are related by conjugation with the bilocal unitaries $X\mathcal{I}_{j\to i}$, property (c) of Definition~\ref{def_prescription} directly implies that the ensuing effective Hamiltonians transform likewise between the two perspectives,\footnote{That is, property (c) of Def.~\ref{def_prescription} here ensures consistency between applying the prescription before or after applying $\hat{X}\hat{\mcI}_{j\to i}$ to $H_{\jbar}$.} i.e.
\be\label{effHimported}
H_{j \to i}=(\hat Y_j\otimes\hat{\mathds 1}_S)\hat{\mathcal I}_{j\to i}(H_{i}^{\eff}(t)\otimes\mathds1_S)+\mathds1_j\otimes\hat{Z}_S(H_{S|R_j}^{\eff}(t))+\hat X\hat{\mathcal I}_{j\to i}(H_{iS}^{\eff}(t))\;,
\ee
where $H_{\jbar}=H_{i}^{\eff}(t)\otimes\mathds1_S+\mathds1_i\otimes H_{S|R_j}^{\eff}(t)+H_{iS}^{\eff}(t)$ is the resulting splitting into effective Hamiltonians in $R_j$-perspective.~(This holds regardless of whether $H_{\ibar}$ is $\mbfUX$-invariant.)~Plugging \eqref{effHimported} into the total energy \eqref{energyUXinvstate} relative to frame $R_i$, we get
\be\label{effenergisUXrho}
    E_{\ibar}(t)=\Tr_{\ibar}\bigl((\hat Y_j\otimes\hat{\mathds 1}_S)\hat{\mathcal I}_{j\to i}(H_{i}^{\eff}(t)\otimes\mathds1_S)\,\rho_{\ibar}(t)\bigr)+\Tr_{\ibar}\bigl((\mathds1_j\otimes\hat{Z}_S(H_{S|R_j}^{\eff}(t)))\,\rho_{\ibar}(t)\bigr)+\Tr_{\ibar}\bigl(\hat X\hat{\mathcal I}_{j\to i}(H_{iS}^{\eff}(t))\,\rho_{\ibar}(t)\bigr)\,,
\ee
for $t\in[t_0,t_1]$.~The three terms in \eqref{effenergisUXrho} are naturally understood as ``the other frame'', the $S$-subsystem, and the interaction energies, respectively, as defined by the effective Hamiltonians built from the imported one.~Therefore, using again the fact that
$\rho_{\jbar}(t)=\hat{\mathcal I}_{i\to j}\hat X^{\dagger}(\rho_{\ibar}(t))$, and so $\rho_{S|R_j}(t)=\hat{Z}_S^{\dagger}(\rho_{S|R_i}(t))$, it is straightforward to see that the individual contributions to the energy of the total system over the interval $[t_0,t_1]$ are the same in the two perspectives,
\be
E_{j}(t)=E_{i}(t)\;,\qquad E_{S|R_i}(t)=E_{S|R_j}(t)\;,\qquad E_{jS}(t)=E_{iS}(t)\;.
\ee
Now, from the definitions \eqref{convenWQ1}, \eqref{convenWQ2} of the work and heat rates for the subsystem $S$ we have
\begin{align}
\dot w_{S|R_i}&=\Tr_S\bigl(\hat{Z}_S(\dot{H}_{S|R_j}^{\eff}(t))\rho_{S|R_i}(t)\bigr)=\Tr_S\bigl(\dot{H}_{S|R_j}^{\eff}(t)\rho_{S|R_j}(t)\bigr)=\dot w_{S|R_j},\label{eq:wequality}\\
\dot q_{S|R_i}&=\Tr_S\bigl(\hat{Z}_S(H_{S|R_j}^{\eff}(t))\dot{\rho}_{S|R_i}(t)\bigr)=\Tr_S\bigl(H_{S|R_j}^{\eff}(t)\dot{\rho}_{S|R_j}(t)\bigr)=\dot q_{S|R_j}, \label{eq:qequality}
\end{align}
where we used the time independence of the unitary $Z_S$ together with $\rho_{S|R_j}(t)=\hat{Z}_S^{\dagger}(\rho_{S|R_i}(t))$.~Similarly, one finds $\dot w_j=\dot w_i$, $\dot q_j=\dot q_i$ for ``the other frame'' work and heat rates over the time interval $[t_0,t_1]$.

Finally, notice that, by definition of the imported Hamiltonian, we have $H_{j\to i} = \hat{X}\hat{\mcI}_{j\to i}(H_i \otimes \mathds1_S + \mathds1_i \otimes H_{S|R_j} + H_{iS}) \equiv H'_j \otimes \mathds1_S + \mathds1_j \otimes H'_{S|R_i} + H'_{jS}$, where we have denoted the ensuing bare contributions by primes to distinguish them from the bare contributions to the generally distinct $H_{\ibar}$. Using these bare pieces and the fact that $\rho_{\ibar}(t) \in \mcA_{\ibar}^X$, it follows that $H'_{S|R_i} = \hat{Z}_S(H_{S|R_j})$ and $\tilde{H}_{S|R_i}'=\hat{Z}_S(\tilde{H}_{S|R_j})$.~Thus, the quantity $\dot{E}^*_{S|R_i}$ defined in Eq.~\eqref{star} reads
\be\label{Estarequality}
\dot{E}^*_{S|R_i}=-i\Tr_S\bigl(H_{S|R_j}^{\eff}(t)[\hat{Z}_S^\dag(H_{S|R_i}'+\tilde{H}_{S|R_i}'(t)),\rho_{S|R_j}(t)]\bigr)=\dot{E}^*_{S|R_j}\;,\qquad t\in[t_0,t_1]\,.\ee
Combining then Eqs.~\eqref{eq:wequality},~\eqref{eq:qequality},~and \eqref{Estarequality}, from the definitions \eqref{eq:QWdef} it follows that
\be
\dot{W}_{S|R_i}=\dot{W}_{S|R_j}\qquad,\qquad \dot{Q}_{S|R_i}=\dot{Q}_{S|R_j}\;\;,\qquad t\in[t_0,t_1]\,.
\ee
Similarly for $\dot{E}^*_j=\dot{E}^*_i$ and $\dot{W}_{j}=\dot{W}_{i}$,~$\dot{Q}_{j}=\dot{Q}_{i}$. 
\end{proof}

\noindent
\textbf{Lemma~\ref{lemma:zerosigmaphi}~(No entropy production and flow).}~\emph{For pure global states $\rho_{\ibar}$, the entropy production and entropy flow of the subsystem $S$ at time $t=t_1>0$ are zero in both perspectives (for any given frame orientations $g_i,g_j\in\mathcal G$)
    \be\label{noentexchange}
\Sigma_{\ibar}(t_1)=0=\Sigma_{\jbar}(t_1)\quad,\quad\Phi_{\ibar}(t_1)=0=\Phi_{\jbar}(t_1)\,,
    \ee
    if and only if, at the initial $t_0=0$ and the later time $t_1$
\be\label{pure:zerosigmaphi}
    \rho_{\ibar}(t_a)=\rho_j(0)\otimes\rho_{S|R_i}(t_a)\in\mcA_{\ibar}^{X^a}\,,
    \ee
    for some bilocal unitary $X^a=Y_j^0\otimes Z_S^a$, $a=0,1$.~For $\rho_{\ibar}$ mixed,~\eqref{noentexchange} holds true if $\rho_{\ibar}(t_a)$, $a=0,1$, satisfy at least one of the following:}\\

\noindent
    \emph{$\bullet$ they obey~\eqref{pure:zerosigmaphi},}\\
    
\noindent
\emph{$\bullet$ they admit pure state decompositions of the form
\be\label{mixed2:zerosigmaphi}
    \rho_{\ibar}(t_a)=\sum_{A,B}p_A\,q_B^a\,\rho_{j}^A(0)\otimes\rho_{S|R_i}^{B}(t_a)\;,
\ee
with $\rho_{j}^A(0)\otimes\rho_{S|R_i}^B(t_a)\in\mcA_{\ibar}^{X^{AB}}$, for possibly distinct $X^{AB}=Y_j^{A,0}\otimes Z_S^{B,a}$.
}
\begin{proof} \eqref{noentexchange} $\Rightarrow$ \eqref{pure:zerosigmaphi}: As emphasised in the main text, to apply the definitions \eqref{entprod}, \eqref{entflow} of entropy production and flow in both perspectives the initial global pure states $\rho_{\ibar}(0)$ and $\rho_{\jbar}(0)=\hat{V}_{i\to j}(\rho_{\jbar}(0))$ must be product states.~By Corollary~\ref{claim:Renyipure}, this means that $\rho_{\ibar}(0)=\rho_{j}(0)\otimes\rho_{S|R_i}(0)\in\mcA_{\ibar}^X$, for some bilocal unitary $X=Y_j\otimes Z_S$.~Then, inserting \eqref{noentexchange} into the entropy balance equation \eqref{entbal2}, we have $0=\Delta S_{S|R_i}=S_{\rm vN}[\rho_{S|R_i}(t_1)]$, where in the second equality we used the property $S_{\rm vN}[\rho_{S|R_i}(0)]=0$ for $\rho_{S|R_i}(0)$ pure.~Similarly, in $R_j$-perspective we have $S_{\rm vN}[\rho_{S|R_j}(t_1)]=0$.~Therefore, the $S$-subsystem state at time $t_1$ must be pure in both perspectives and, invoking again Corollary~\ref{claim:Renyipure}, the global pure state $\rho_{\ibar}(t_1)$ has to be a product state in a TPS-invariant subalgebra, namely $\rho_{\ibar}(t_1)=\rho_j(t_1)\otimes\rho_{S|R_i}(t_1)\in\mcA_{\ibar}^{X^1}$ for some bilocal unitary $X^1=Y_j^1\otimes Z_S^1$.~Finally, rewriting the definition \eqref{entprod} of $\Sigma_{\ibar}$ equivalently as $\Sigma_{\ibar}=S[\rho_{\ibar}(t_1)\|\rho_j(0)\otimes\rho_{S|R_i}(t_1)]$ and recalling the well-known result that $S[\rho\|\sigma]=0$, if and only if $\rho=\sigma$, it follows that $\rho_{\ibar}(t_1)=\rho_{j}(0)\otimes\rho_{S|R_i}(t_1)$ which completes the proof of the necessary direction of the claim.

\eqref{pure:zerosigmaphi} $\Rightarrow$ \eqref{noentexchange}:~By Lemma~\ref{Lemma:tpseq}, \eqref{pure:zerosigmaphi} implies that the global pure states are product states in both perspectives with $\rho_j(t_1)=\rho_j(0)$ and $\rho_i(t_1)=\rho_i(0)=\hat{\mathcal I}_{i\to j}'\hat{Y}_j(\rho_j(0))$, $\hat{\mathcal{I}}'_{i\to j} = \sum_{g\in\mcG}\ket{g}_i\otimes \bra{g}_j$.~The result \eqref{noentexchange} immediately follows then from inserting these relations into the definitions \eqref{entprod} and \eqref{entflow} of $\Sigma$ and $\Phi$ in the two perspectives.

For $\rho_{\ibar}$ mixed, when (at least one of) the conditions \eqref{pure:zerosigmaphi},~\eqref{mixed2:zerosigmaphi} hold(s), the global states at the initial and the later time $t$ are product states in both perspectives by Lemma~\ref{Lemma:tpseq}.~Moreover, $\rho_j(t_1)=\rho_j(0)$ and $\rho_i(t_1)=\rho_i(0)$.~Thus, the definitions \eqref{entprod},~\eqref{entflow} of entropy production and flow are applicable in both perspectives and vanish.
\end{proof}    

\noindent
\textbf{Lemma~\ref{lemma:inventbal}~(QRF-invariance of entropy balance).}~\emph{Suppose the global state is TPS-invariant at $t_0$ and $t_1>t_0$, i.e.\ $\rho_{\ibar}(t_a)\in\mcA_{\ibar}^{X^a}$, $a=0,1$, for some bilocal unitaries $X^a=Y_j^a\otimes Z_S^a$.~Then, between these times,} \begin{itemize}
    \item \emph{the change in the entropy of ``the other frame'' and $S$ subsystems are the same in both perspectives, i.e.}
\be\label{invDeltaS}
\Delta S_{i}=\Delta S_{j}\quad,\quad\Delta S_{S|R_j}=\Delta S_{S|R_i}\,.
\ee
\emph{When $\rho_{\ibar}(t_0)$ is pure, we have an additional equality, $\Delta S_j=\Delta S_{S|R_i}$, and so all entropy changes are equal.}

\item \emph{the entropy production in the subsystem $S$ and the entropy flow out of $S$ are equal in both perspectives,~i.e.}
\be\label{invSigmaPhi}
\Sigma_{\jbar}=\Sigma_{\ibar}\quad,\quad \Phi_{\jbar}=\Phi_{\ibar}
\ee
\emph{if and only if the initial state is of product form and}
\be\label{invSigmaPhi:cond01}
S[\hat{Y}_j^{01}(\rho_j(t_1))\|\rho_j(t_0)]=S[\rho_j(t_1)\|\rho_j(t_0)]\,,
\ee
\emph{with $Y_j^{01}=Y_j^0Y_j^{1\dagger}$.}
\end{itemize}

\begin{proof}
   When $\rho_{\ibar}(t_a) \in \mcA_{\ibar}^{X^a}$, $a=0,1$, by Lemma \ref{Lemma:tpseq} we have that $\rho_{\jbar}(t_a) = \hat{\mathcal{I}}_{i \to j} \widehat{X^a}^\dag (\rho_{\ibar}(t_a))$.~Thus, the ``other frame'' and $S$ reduced states at $t_0$ and $t_1$ in the two perspectives are related unitarily, respectively, as
    \be \label{unisubspr} 
    \rho_{i}(t_a)=\hat{\mathcal{I}}'_{i\to j} \widehat{Y^a_j}^\dag (\rho_j (t_a))\qquad,\qquad \rho_{S|R_j}(t_a) = \widehat{Z^a_S}^\dag(\rho_{S|R_i}(t_a))\,, \ee 
    where $\hat{\mathcal{I}}'_{i\to j} = \sum_{g\in\mcG}\ket{g}_i\otimes \bra{g}_j$.~Therefore, $S_{\rm vN}[\rho_{i}(t_a)]=S_{\rm vN}[\rho_{j}(t_a)]$ and $S_{\rm vN}[\rho_{S|R_j}(t_a)]=S_{\rm vN}[\rho_{S|R_i}(t_a)]$, from which the equality \eqref{invDeltaS} of the subsystem entropy variations follows. When $\rho_{\ibar}(t_0)$ is furthermore pure, then so is $\rho_{\ibar}(t_1)$, and we have $S_{\rm vN}[\rho_{S|R_i}(t_a)]=S_{\rm vN}[\rho_j(t_a)]$, implying also $\Delta S_{S|R_i}=\Delta S_j$.
    
   For the second statement, recall first from the main text that, in order to apply the definition of entropy production and flow, the initial state has to be of product form. Next, the mutual information between the ``other frame'' and $S$ at time $t_1$ is equal in the two perspectives,
    \be \label{eqmutinfpr}
        I[\rho_{\jbar}(t_1)] = S_{\rm vN}[\rho_i(t_1)]+S_{\rm vN}[\rho_{S|R_j}(t_1)]-S_{\rm vN}[\rho_{\jbar}(t_1)] = S_{\rm vN}[\rho_j(t_1)]+S_{\rm vN}[\rho_{S|R_i}(t_1)]-S_{\rm vN}[\rho_{\ibar}(t_1)] = I[\rho_{\ibar}(t_1)]
    \ee
    since each term on the RHS of the first equality is equal to the respective term on the RHS of the second equality, due to \eqref{unisubspr}.~Inserting then \eqref{invDeltaS} and \eqref{eqmutinfpr} into the definitions \eqref{entprod},~\eqref{entflow} of entropy production and entropy flow, these are equal in the two perspectives if and only if
    \be\label{eqrelentpr}
    S[\rho_i(t_1) || \rho_i(t_0)] = S[\rho_j(t_1) || \rho_j(t_0)]\;.
    \ee
    Owing to \eqref{unisubspr} and the invariance of quantum relative entropy under unitary conjugation of both of its arguments, this is equivalent to \eqref{invSigmaPhi:cond01}.
\end{proof}

We close with an illustration of Lemma~\ref{lemma:inventbal}, in which a global state trajectory oscillates between two distinct TPS-invariant subalgebras and is otherwise not TPS-invariant. The entropy balance in the two perspectives agrees only at these discrete times, in which case also~\eqref{invSigmaPhi:cond01} holds.

\begin{example}{\bf(QRF-relativity of entropy balance)}\label{ex2:entbalance}
Consider three qubits $R_1R_2S$ with frames $R_1$ and $R_2$ in orientations $g_1=0$ and $g_2=1$ before and after QRF transformation, respectively.~In this case, we have $\mbfU_{\bar1}^{0,1}=\ket{0}\!\bra{1}_2\otimes\mathds1_S+\ket{1}\!\bra{0}_2\otimes\sigma^x_S$.~Suppose the initial state is $\ket{\psi(0)}_{\bar1}=\ket{1,0}_{\bar1}$ of $R_2S$ relative to $R_1$.~Then, $\rho_{\bar1}(0)\in\mcA_{\bar1}^{X^0}$ with $X^0=\sigma_2^x\otimes\mathds1_S$, and the corresponding initial state of $R_1S$ relative to $R_2$ is $\ket{\psi(0)}_{\bar2}=V_{1\to2}^{0,1}\ket{\psi(0)}_{\bar1}=\ket{0,0}_{\bar2}$.

Consider now the Hamiltonian $H_{\bar1}=\sigma_2^x\otimes\mathds1_S+\mathds1_2\otimes\sigma_S^x$ in $R_1$-perspective.~The Hamiltonian in $R_2$-perspective is $H_{\bar2}=\mathds1_1\otimes\sigma_S^x+\sigma_1^x\otimes\sigma_S^x$, thus interacting.~The global states at time $t>0$ in the two perspectives are given by
\begin{align}
\ket{\psi(t)}_{\bar1}&=(\cos(t)\ket{1}_2-i\sin(t)\ket{0}_2)\otimes(\cos(t)\ket{0}_S-i\sin(t)\ket{1}_S)\,,\label{psit1bar}\\
\ket{\psi(t)}_{\bar2}&=\cos(t)\ket{0}_1\otimes(\cos(t)\ket{0}_S-i\sin(t)\ket{1}_S)-i\sin(t)\ket{1}_1\otimes(\cos(t)\ket{1}_S-i\sin(t)\ket{0}_S).\label{psit2bar}
\end{align}
Thus, the state in $R_1$-perspective remains unentangled at all times, while the state in $R_2$-perspective is entangled at any $t$ except for $t=n\frac{\pi}{2}$, $n\in\mathbb{N}$, when the global state $\rho_{\bar1}(t)$ resides in a TPS-invariant subalgebra.~Specifically, we have that $\rho_{\bar1}(t)\in\mcA_{\bar1}^{X^0}$ for $t=t_n=n\pi$ and $\rho_{\bar1}(t)\in\mcA_{\bar1}^{X^1}$ with $X^1=\sigma_2^x\otimes\sigma_S^x$ for $t=\tilde{t}_n=(2n-1)\frac{\pi}{2}$.~At any other time $\rho_{\bar1}(t)\notin\mcA_{\bar1}^{X}$, for any $X$.

Let us now look at the entropy balance in the two perspectives.~Since $\rho_{\bar1}(t)$ is a separable pure state for all $t$, there is no subsystem entropy variation in $R_1$-perspective at all times, that is $\Delta S_{S|R_1}=0=\Delta S_2$.~In $R_2$-perspective, we have $\Delta S_1=S_{\rm vN}[\rho_1(t)]=S_{\rm vN}[\rho_{S|R_2}(t)]=\Delta S_{S|R_2}$, owing to the global state being pure and the initial state being separable.~Thus, $\Delta S_1$ and $\Delta S_{S|R_2}$ are nonzero at all $t$ ($\rho_{\bar2}$ entangled), except for $t=t_n,\tilde{t}_n$ when, as discussed, $\rho_{\bar1}$ resides in a $TPS$-invariant subalgebra and the entropy variations vanish as in $R_1$-perspective.

Finally, using \eqref{psit1bar} in the definitions \eqref{entprod} and \eqref{entflow}, we have that $\Sigma_{\bar1}=\Phi_{\bar1}=S[\rho_2(t)\|\rho_2(0)]$, which vanishes at $t=t_n$ ($\rho_2(t_n)=\rho_2(0)$, $n\in\mathbb N$) and diverges at any other $t$.~For $t=t_n,\tilde{t}_n$, the reduced states are pure also in $R_2$-perspectives, and we have $\Sigma_{\bar2}=\Phi_{\bar2}=S[\rho_1(t)\|\rho_1(0)]$, which vanishes for $t=t_n$ and diverges for $t=\tilde{t}_n$.~Thus, $\Sigma_{\bar1}=\Sigma_{\bar2}$ and $\Phi_{\bar1}=\Phi_{\bar2}$ for $t=t_n,\tilde{t}_n$, when in fact condition \eqref{cond01:invSigmaPhi} is satisfied with $Y_2^{01}=(\sigma_2^x)^2=\mathds1_2$.~For $t\neq t_n,\tilde{t}_n$, $\Delta_{S|R_2}$ is nonzero so that, by the entropy balance relation \eqref{entbal2}, both $\Sigma_{\bar2}$ and $\Phi_{\bar2}$ are finite and at least one of them is nonzero.~Therefore, $\Sigma_{\bar1}\neq\Sigma_{\bar2}$ and $\Phi_{\bar1}\neq\Phi_{\bar2}$ for $t\neq t_n,\tilde{t}_n$ when $\rho_{\bar1}$ does not reside in any TPS-invariant subalgebra.
\end{example}
\end{document}